\documentclass[11pt]{article}
\pdfoutput=1
\usepackage{yrz}

\usepackage{tikz}

\usetikzlibrary{shapes, calc, shapes, arrows, positioning}
\tikzstyle{qedge}=[->,thick,black]
\tikzstyle{neuron}=[draw,circle,minimum size=25pt,inner sep=0pt, fill=black!10]
\tikzstyle{hidden}=[draw,circle,minimum size=25pt,inner sep=0pt, fill=white]

\usetikzlibrary{arrows.meta}
\tikzset{>={Latex[width=3mm,length=2mm]}}

\tikzstyle{arr}=[->, thick, black]

\usetikzlibrary{decorations.text}

\tikzstyle{learned}=[text=red]

\usetikzlibrary{shadows,shadings,shapes.symbols}
\tikzset{
    double color fill/.code 2 args={
        \pgfdeclareverticalshading[%
            tikz@axis@top,tikz@axis@middle,tikz@axis@bottom%
        ]{diagonalfill}{100bp}{%
            color(0bp)=(tikz@axis@bottom);
            color(50bp)=(tikz@axis@bottom);
            color(50bp)=(tikz@axis@middle);
            color(50bp)=(tikz@axis@top);
            color(100bp)=(tikz@axis@top)
        }
        \tikzset{shade, left color=#1, right color=#2, shading=diagonalfill}
    }
}

\usepackage{hyperref}
\usepackage{xurl}

\usepackage[margin=1in]{geometry}

\usepackage{listings}

\usepackage{caption}
\usepackage{subcaption}

\usepackage{upgreek}
\usepackage[only,llbracket,rrbracket]{stmaryrd}

\def\PG{\mathrm{PG}}
\def\perm{\mathrm{Perm}}
\def\marg{\mathrm{Marg}}
\def\shuff{\mathrm{Shuff}}
\def\law{\mathrm{Law}}
\def\lex{\prec_{\mathrm{lex}}}
\def\logistic{\mathrm{logistic}}
\def\logit{\mathrm{logit}}

\def\tz{\tilde{z}}
\def\Categorical{\mathrm{Categorical}}
\def\softmax{\mathrm{softmax}}
\def\Bernoulli{\mathrm{Bernoulli}}
\def\Poisson{\mathrm{Poisson}}
\def\Beta{\mathrm{Beta}}
\def\mix{\mathrm{mix}}

\def\WAIC{\textsf{WAIC}}
\def\lppd{\textsf{lppd}}
\def\var{\mathrm{var}}

\begin{document}
\title{Bayesian Latent Class Regression with Interpretable Binary Profiles}
\author{
Yuren Zhou
\thanks{Department of Statistical Science, Duke University, Durham, NC, USA.}
\and
Yuqi Gu
\thanks{Department of Statistics, Columbia University, New York, NY, USA.}
\and
David B. Dunson
\thanks{Department of Statistical Science \& Mathematics, Duke University, Durham, NC, USA.}
}
\date{\today}
\maketitle

\begin{abstract}
High-dimensional categorical data arise in diverse scientific domains and are often accompanied by covariates.
Latent class regression models are routinely used in such settings, reducing dimensionality by assuming conditional independence of the categorical variables given a single latent class that depends on covariates through a logistic regression model. 
However, such methods become unreliable as the dimensionality increases. 
To address this, we propose Bayesian latent class regression with interpretable binary profiles (BLIP), a flexible family of models that introduces a binary latent-attribute layer between the covariate-dependent latent class and the observed categorical responses.
\texttt{BLIP} satisfies key theoretical properties, including identifiability and posterior consistency, and we establish a Bayes oracle clustering property that ensures robustness against the curse of dimensionality.
We develop efficient posterior computation methods, validate them through simulation studies, and use \texttt{BLIP} to infer regions of common profile in ecological data.

\end{abstract}


\section{Introduction}\label{sec:intro}

High-dimensional categorical data play a pivotal role in a wide range of scientific disciplines \citep{agresti2012categorical}.
In ecology, the joint distribution of species at various sampling locations is represented by binary vectors, with each entry indicating the presence or absence of a species \citep{ovaskainen2020joint, tikhonov2020joint}.
In genetics, DNA and RNA sequences are expressed as categorical vectors of nucleotides, encapsulating the molecular blueprint of organisms \citep{zito2023inferring}. In sociology and political studies, survey data often consist of high-dimensional categorical responses that reflect individual choices and opinions \citep{bolck2004estimating}.
These datasets are frequently accompanied by covariate information, such as geographical and environmental factors of species sampling locations \citep{warton2015so}, the biological characteristics of the genetic study subjects \citep{amfoh1994use}, or the demographic information of survey participants \citep{vermunt2010latent}.

To gain deeper insights into the underlying ecological, biological, or social processes, we often seek to form data into clusters that integrate information from both categorical observations and their associated covariates. Latent class regression models are routinely used for this purpose, typically characterizing categorical response data as conditionally independent given a categorical latent class, whose distribution depends on covariates through a logistic regression model \citep{goodman1974exploratory, dayton1988concomitant, garrett2000latent, magidson2004latent, chung2006latent}.
Latent class allocations induce a clustering of the samples into covariate-dependent groups. Building on this literature, our focus is on developing a Bayesian generative model that characterizes the conditional distribution of the high-dimensional categorical data given the covariates while simultaneously inferring interpretable discrete \emph{fine-grained} latent structures underlying the data.

A variety of extensions have been developed for latent class and latent class regression models. Mixed membership models \citep{russo2022multivariate, gu2023dimension} incorporate subject-specific weight vectors that represent partial memberships in latent classes. Robust profile clustering \citep{stephenson2020robust} introduces a global-local structure of latent classes to account for nested sub-populations. To integrate sample similarities into the model, tree-structured shrinkage and regularization approaches have been proposed \citep{li2023integrating, li2023tree}. In addition to the observed covariates, latent factors can also be included in latent class regression \citep{guo2006latent}. More broadly, these models belong to the mixture model family, which spans finite mixture models \citep{mclachlan2000finite, rousseau2011asymptotic}, Dirichlet process mixture models \citep{antoniak1974mixtures, escobar1995bayesian}, and other Bayesian nonparametric methods \citep{lijoi2020pitman, dunson2008kernel} with covariate dependent priors \citep{horiguchi2024tree}.

However, it has been increasingly recognized that mixture models face challenges when applied to high-dimensional data, falling prey to the curse of dimensionality. Recent work \citep{chandra2023escaping} demonstrates that Bayesian Gaussian mixture models tend to favor an excessive or insufficient number of clusters as the data dimensionality increases. A similar issue arises for high-dimensional categorical data, as illustrated in Section \ref{sec:high_dim}, where a well-specified Bayesian Bernoulli mixture model asymptotically assigns each observation to a separate singleton cluster as the data dimensionality grows to infinity.
This behavior underscores the importance of designing models and inference procedures to handle high-dimensional categorical data, ensuring robust and meaningful inference of latent classes while avoiding the pitfalls of the curse of dimensionality.

In recent years, there has been growing interest in deep learning methods for categorical data, including those that account for covariates and generate latent classes. Examples include conditional variational autoencoders with vector quantization \citep{sohn2015learning, van2017neural} and conditional normalizing flows or diffusion models designed for discrete data \citep{winkler2019learning, tran2019discrete, dhariwal2021diffusion}.
Despite their flexibility, deep learning methods typically suffer from a lack of identifiability and limited interpretability. They also tend to rely on large datasets. Hence, it is important to develop statistical approaches that prioritize a balance between flexibility and interpretability, especially in applications where identifying latent classes and uncovering interpretable structures are key goals.

With this motivation, we propose Bayesian Latent class regression with Interpretable binary Profiles (\texttt{BLIP}).
This model is illustrated by the Bayesian network in Figure \ref{fig:dag}. In the bottom layer are the observed categorical variables, whose distribution is modeled through sparse edges to a latent layer of binary attributes. The distribution of these latent binary attributes is further governed by a deeper latent class that is dependent on the covariates. Importantly, the deep latent class, the vector of binary latent attributes, and the model parameters are all easily interpretable.

Under sparsity constraints on the Bayesian network edges, we establish a comprehensive identifiability theory for our model, including strict and generic notions of identifiability \citep{allman2009identifiability, gu2023bayesian, zhou2024bayesian}.
Building on this foundation, we prove posterior consistency of \texttt{BLIP} using the frameworks of Doob’s and Schwartz’s theories \citep{doob1949application, schwartz1965bayes}.
Furthermore, we demonstrate that our model satisfies a Bayes oracle clustering property \citep{chandra2023escaping}, ensuring that, as the data dimensionality increases, the deep latent class has posterior probabilities converging to its conditional probabilities given the unobserved vector of binary attributes.
Since the number of binary attributes is finite and fixed, this property implies that our model effectively escapes the curse of dimensionality, enabling reliable inference on the deep latent class regardless of the observed data dimensionality.

We design hierarchical prior distributions for the sparse edges in the Bayesian network and the continuous model parameters.
To perform posterior inference, we propose a Gibbs sampler using Polya-Gamma augmentation \citep{polson2013bayesian}.
Simulation studies are conducted to empirically validate the robustness of model selection, the consistency of posterior inference, and the Bayes oracle clustering property of our model.

As one of the motivating examples of high-dimensional categorical data, we apply our model to joint species distribution modeling (\texttt{JSDM}) in ecology \citep{ovaskainen2020joint, ovaskainen2016using, ovaskainen2016uncovering}, a field where understanding species co-occurrence patterns is essential for biodiversity conservation and environmental management.
Specifically, we analyze the Finnish bird dataset \citep{ovaskainen2020joint}, which records the presence and absence of bird species at multiple sampling locations in Finland, along with associated environmental covariates such as temperature and habitat type. Our methodology provides a novel approach to understanding variation across sampling sites in ecological community structure, improving regions of common profile approaches while maintaining predictive performance comparable to the state-of-the-art Hierarchical Modeling of Species Communities \citep[\texttt{HMSC},][]{ovaskainen2017species, ovaskainen2017make}, 

To enhance intuition and reduce notational complexity, we mainly present \texttt{BLIP} for multivariate binary observations.
We generalize our model to handle arbitrary multivariate categorical and count observations in the Supplementary Material Section \ref{supp_sec:gen}.
The remainder of this paper is organized as follows. In Section \ref{sec:model}, we introduce our model for high-dimensional binary observations. Section \ref{sec:theo} provides comprehensive theoretical results of model identifiability and posterior consistency. In Section \ref{sec:high_dim}, we first illustrate how the multivariate Bernoulli mixture model succumbs to the curse of dimensionality and then establish the Bayes oracle clustering property of \texttt{BLIP}. Section \ref{sec:post} presents the posterior computation, including the hierarchical prior distributions, the MCMC sampling algorithms, and simulation studies. We apply \texttt{BLIP} to the Finnish bird dataset in Section \ref{sec:appl}. Finally, we outline potential future directions in Section \ref{sec:discuss}.

\textbf{Notations.}
We use bold capital letters (e.g. $\Gb, \Bb$) to denote matrices, bold letters (e.g. $\xb, \yb$) to denote vectors, and nonbold letters (e.g. $z$) to denote scalars.
For integers $k_1, k_2 \in \ZZ$, we let $[k_1, k_2]:=\{k_1, k_1 + 1, \ldots, k_2\}$ and abbreviate $[k_2] := [1, k_2]$.
For a sequence of random variables $\{Z^{(n)}\}_{n = 1}^\infty$, we let $Z^{(N_1:N_2)}$ denote the collection $\{Z^{(n)}\}_{n = N_1}^{N_2}$ and let $Z^{(1:\infty)}$ denote the whole sequence. The logit function is  $\logit(p) := \log\frac{p}{1 - p}$ for $p \in (0, 1)$.
The cardinality of a set $\cB$ is $\card(\cB)$.
For a matrix $\Gb$ and a subset of row indices $\cI$, we let $\Gb_{\cI}$ denote the submatrix of $\Gb$ formed by the rows in $\cI$.

\section{Model Formulation}\label{sec:model}

Let the observed outcome data be a $p$-dimensional binary vector $\yb \in \{0, 1\}^p$ and the covariates be a $p_x$-dimensional vector $\xb \in \RR^{p_x}$. We introduce $q$ binary attributes $\wb \in \{0, 1\}^q$ that capture the essential patterns in $\yb$, so that each entry $y_i$ is conditionally independent given $\wb$. Each $y_i$ depends on a subset of $\wb$ characterized by the $p \times q$ matrix  $\Gb$, where the $(i, j)$th entry $g_{i, j} = 1$ indicates that the distribution of $y_i$ depends on the attribute $w_j$
and a continuous parameter $\beta_{i, j}$ provides a weight on this edge. Starting from the baseline $\beta_{i, 0}$, we adjust the logit of the probability of $y_i = 1$ by $\beta_{i, j}$ for each attribute $w_j = 1$ that $y_i$ depends on:
$$
\logit~ \PP(y_i = 1 ~|~ \wb, \Gb, \Bb)
=
\log \frac{\PP(y_i = 1 ~|~ \wb, \Gb, \Bb)}{\PP(y_i = 0 ~|~ \wb, \Gb, \Bb)}
=
\beta_{i, 0} + \sum_{j = 1}^q g_{i, j} w_j \beta_{i, j}
,
$$
where $\Bb$ denotes the $p \times (q + 1)$ matrix whose $(i, j)$th entry is $\beta_{i, j}$ for $j \in [0, q]$.
Using $\gb_i \in \{0, 1\}^q$ to denote the $i$th row of $\Gb$ and letting $\bbeta_i := (\beta_{i, 1}, \ldots, \beta_{i, q}) \in \RR^q$, the conditional distribution of $\yb$ given $\wb$ can be expressed as
\begin{equation}\label{eq:model_y_w}
\PP(\yb ~|~ \wb, \Gb, \Bb)
=
\prod_{i = 1}^p \PP(y_i ~|~ \wb, \gb_i, \beta_{i, 0}, \bbeta_i)
=
\prod_{i = 1}^p \frac{\exp(y_i (\beta_{i, 0} + (\gb_i \circ \bbeta_i)^\top \wb))}{1 + \exp(\beta_{i, 0} + (\gb_i \circ \bbeta_i)^\top \wb)}
.
\end{equation}

Building on this foundation, we further structure the latent attribute vector $\wb$ by introducing a deeper latent class $z$ taking values in $[d] := \{1, \ldots, d\}$.
The attributes $w_j$ are modeled as conditionally independent given the deep latent class $z$, following Bernoulli distributions with parameters $\alpha_{j, z} \in (0, 1)$.
This gives the conditional distribution of $\wb$ given $z$:
\begin{equation}\label{eq:model_w_z}
\PP(\wb ~|~ z, \Ab)
=
\prod_{j = 1}^q \PP(w_j ~|~ z, \balpha_j)
=
\prod_{j = 1}^q 
\alpha_{j, z}^{w_j} (1 - \alpha_{j, z})^{1 - w_j}
,
\end{equation}
where $\Ab$ is the $q \times d$ matrix with $\alpha_{j, z}$ as its $(j, z)$th entry and $\balpha_j$ denotes the $j$th row in $\Ab$.

To incorporate the influence of covariates $\xb$ into our model, we further allow the deep latent class $z$ to depend on $\xb$.
The conditional distribution of $z$ given $\xb$ is modeled using a multinomial logistic regression with parameter $\bGamma \in \RR^{d \times (p_x + 1)}$, defined as
\begin{equation}\label{eq:model_z_x}
\PP(z ~|~ \xb, \bGamma)
=
\frac{\exp(\gamma_{z, 0} + \bgamma_z^\top \xb)}{\sum_{h = 1}^d \exp(\gamma_{h, 0} + \bgamma_h^\top \xb)}
,
\end{equation}
where $\gamma_{h, 0} \in \RR$ and $\bgamma_h \in \RR^{p_x}$ are the intercept and coefficients for class $h$, respectively, and $(\gamma_{h, 0}, \bgamma_h) \in \RR^{p_x + 1}$ forms the $h$th row of $\bGamma$. To ensure identifiability, we designate class $d$ as the baseline class and fix its regression intercept and coefficients to $\gamma_{d, 0} = 0, \bgamma_d = \zero$.

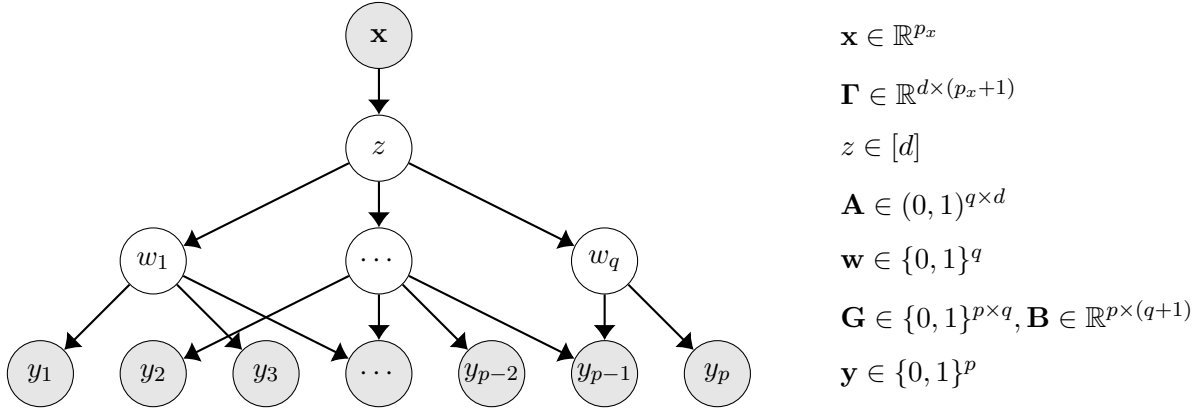
\begin{figure}
\centering
\begin{tikzpicture}[scale = 1.5]
\node (x)[neuron] at (3, 3) {$\xb$};
\node (z)[hidden] at (3, 2) {$z$};
\node (w1)[hidden] at (1, 1) {$w_1$};
\node (w2)[hidden] at (3, 1) {$\cdots$};
\node (w3)[hidden] at (5, 1) {$w_q$};
\node (y1)[neuron] at (0, 0) {$y_1$};
\node (y2)[neuron] at (1, 0) {$y_2$};
\node (y3)[neuron] at (2, 0) {$y_3$};
\node (y4)[neuron] at (3, 0) {$\cdots$};
\node (y5)[neuron] at (4, 0) {$y_{p - 2}$};
\node (y6)[neuron] at (5, 0) {$y_{p - 1}$};
\node (y7)[neuron] at (6, 0) {$y_p$};
\draw[arr] (x) -- (z);
\draw[arr] (z) -- (w1);
\draw[arr] (z) -- (w2);
\draw[arr] (z) -- (w3);
\draw[arr] (w1) -- (y1);
\draw[arr] (w1) -- (y3);
\draw[arr] (w1) -- (y4);
\draw[arr] (w2) -- (y2);
\draw[arr] (w2) -- (y4);
\draw[arr] (w2) -- (y5);
\draw[arr] (w2) -- (y6);
\draw[arr] (w3) -- (y6);
\draw[arr] (w3) -- (y7);
\node[anchor=west] (y) at (7, 0) {$\yb \in \{0, 1\}^p$};
\node[anchor=west] (Gbeta) at (7, 0.5) {$\Gb \in \{0, 1\}^{p \times q}, \Bb \in \RR^{p \times (q + 1)}$};
\node[anchor=west] (w) at (7, 1) {$\wb \in \{0, 1\}^q$};
\node[anchor=west] (alpha) at (7, 1.5) {$\Ab \in (0, 1)^{q \times d}$};
\node[anchor=west] (z) at (7, 2) {$z \in [d]$};
\node[anchor=west] (gamma) at (7, 2.5) {$\bGamma \in \RR^{d \times (p_x + 1)}$};
\node[anchor=west] (x) at (7, 3) {$\xb \in \RR^{p_x}$};
\end{tikzpicture}
\caption{
Bayesian deep latent class regression model.
Shading/unshading represents observed/latent variables.
Covariates $\xb$, deep latent class $z$, binary latent attributes $\wb$, and observed variables $\yb$ are specific to each sample.
Model parameters $\Ab, \Bb, \bGamma, \Gb$ are shared across the population.
}
\label{fig:dag}
\end{figure}

The hierarchical structure of \texttt{BLIP} is depicted as a Bayesian network in Figure \ref{fig:dag}.
From the bottom layer to the top layer, the Bayesian network consists of the observed data $\yb$, the latent attribute vector $\wb$, the deep latent class $z$, and the covariates $\xb$.
The sparse dependency between $\wb$ and $\yb$ is visually represented in the bipartite network with sparse edges that point from attributes $w_j$ to observed data entries $y_i$.

For a sample of $N$ observations $\yb^{(1:N)} = \{\yb^{(1)}, \ldots, \yb^{(N)}\}$ and corresponding covariates $\xb^{(1:N)} = \{\xb^{(1)}, \ldots, \xb^{(N)}\}$, our model infers the latent attribute vector $\wb^{(n)}$ and the deep latent class $z^{(n)}$ specific to each sample $n$. Parameters $\Ab, \Bb, \bGamma, \Gb$ are population-level quantities shared by all samples. Assuming that the covariates $\xb^{(1:N)}$ are independently drawn from $\cP_x$, our model has the following independent data generating process for each sample $n = 1, \ldots, N$:
\begin{equation}\label{eq:model}\begin{aligned}
&
z^{(n)} \sim \Categorical\left(
\frac{
\exp(\gamma_{1, 0} + \bgamma_1^\top \xb^{(n)})
}{
\sum_{h = 1}^d \exp(\gamma_{h, 0} + \bgamma_h^\top \xb^{(n)})
}
,
\ldots
,
\frac{
\exp(\gamma_{d, 0} + \bgamma_d^\top \xb^{(n)})
}{
\sum_{h = 1}^d \exp(\gamma_{h, 0} + \bgamma_h^\top \xb^{(n)})
}
\right)
,\\&
w_j^{(n)} \stackrel{ind}{\sim} \Bernoulli\left(
\alpha_{j, z^{(n)}}
\right)
\quad
\forall j \in [q]
,\\&
y_i^{(n)} \stackrel{ind}{\sim} \Bernoulli\left(\frac{
\exp(\beta_{i, 0} + (\gb_i \circ \bbeta_i)^\top \wb^{(n)})
}{
1 + \exp(\beta_{i, 0} + (\gb_i \circ \bbeta_i)^\top \wb^{(n)})
}\right)
\quad
\forall i \in [p]
.
\end{aligned}\end{equation}
This hierarchical process models how the covariates $\xb^{(n)}$ influence the deep latent class $z^{(n)}$, which subsequently governs the generation of the latent attribute vector $\wb^{(n)}$ and, ultimately, the observed data $\yb^{(n)}$.
While \eqref{eq:model} focuses on multivariate binary $\yb$, the framework generalizes by incorporating appropriate generalized linear models for the conditional distribution of $\yb$ given $\wb$; see Supplementary Material Section \ref{supp_sec:gen}.

\section{Theoretical Properties}\label{sec:theo}

In this section, we establish theoretical properties of our model, including strict identifiability, generic identifiability, and posterior consistency of suitable Bayesian estimators.

\subsection{Identifiability}\label{ssec:iden}

We first develop a theory to ensure the unique identification of our model parameters from the marginal distribution of the observed data $\yb$ and the covariates $\xb$. As discussed in Section \ref{sec:model}, the dependency of $\yb$ on $\wb$ is sparse. We are interested in recovering the parameter $\beta_{i, j}$ only when $y_i$ depends on $w_j$, which is equivalent to recovering the entrywise matrix product $(\one ~ \Gb) \circ \Bb$, where $(\one ~ \Gb)$ denotes the $p \times (q + 1)$ binary matrix that left appends an all-ones column to $\Gb$.
As in latent class models, our model exhibits a trivial form of non-identifiability caused by label switching \citep{gyllenberg1994non}, arising from permutations over the $d$ deep latent classes and permutations over the $q$ latent attributes. However, this issue can be resolved using post-processing methods \citep{stephens2000dealing}, and we focus on studying the identifiability of parameters up to permutations on $z$ and dimensions of $\wb$.

Let $\sS$ denote the parameter space of $(\Ab, \Bb, \bGamma, \Gb)$. For any $(\Ab, \Bb, \bGamma, \Gb) \in \sS$, we define the equivalence class of parameters within $\sS$ induced by its permutations over the $d$ classes of $z$ and the $q$ dimensions of $\wb$ as
\begin{equation}\label{eq:perm}\begin{aligned}
\perm_{\sS}(\Ab, \Bb, \bGamma, \Gb)
:=
\Big\{
(\Ab', \Bb', \bGamma', \Gb') \in \sS:~
\exists~ \Pb_z, \Pb_w \text{ s.t. }
\bGamma' = \Pb_z \bGamma
,~
\Ab' = \Pb_w \Ab \Pb_z^\top
,&\\
\Gb' = \Gb \Pb_w^\top
,~
(\one ~ \Gb') \circ \Bb' = ((\one ~ \Gb) \circ \Bb) \diag(1, \Pb_w)^\top
&\Big\}
,
\end{aligned}\end{equation}
where $\Pb_z, \Pb_w$ are arbitrary $d \times d$ and $q \times q$ permutation matrices, and $\diag(1, \Pb_w)$ denote the $(q + 1) \times (q + 1)$ block diagonal matrix with its upper-left block being 1 and lower-right block being $\Pb_w$.
For $(\Ab, \Bb, \bGamma, \Gb) \in \sS$, we also define the set of parameters in $\sS$ generating the same marginal distribution of $\xb, \yb$ by
\begin{equation}\label{eq:marg}
\marg_{\sS}(\Ab, \Bb, \bGamma, \Gb)
:=
\Big\{
(\Ab', \Bb', \bGamma', \Gb') \in \sS:~
\forall \xb, \yb,~
\PP(\xb, \yb ~|~ \Ab', \Bb', \bGamma', \Gb')
=
\PP(\xb, \yb ~|~ \Ab, \Bb, \bGamma, \Gb)
\Big\}
.
\end{equation}
With $\cP_x$ denoting the distribution of covariates $\xb$ as in \eqref{eq:model}, generating the same marginal distribution of $\xb, \yb$ is equivalent to generating the same conditional distribution of $\yb ~|~ \xb$ for $\cP_x$-almost every $\xb$, that is,  $\marg_{\sS}(\Ab, \Bb, \bGamma, \Gb) =\big\{ \Ab', \Bb', \bGamma', \Gb':~ \forall \yb,~ \allowbreak \PP(\yb ~|~ \xb, \Ab', \Bb', \bGamma', \Gb') = \PP(\yb ~|~ \xb, \Ab, \Bb, \bGamma, \Gb),~ \cP_x\text{-almost every } \xb \big\}$.

Since the conditional distribution of $\yb ~|~ \xb$ is invariant under permutations on the $d$ classes of $z$ and the $q$ dimensions of $\wb$, we have $\perm_{\sS}(\Ab, \Bb, \bGamma, \Gb) \subseteq \marg_{\sS}(\Ab, \Bb, \bGamma, \Gb)$. Viewing the parameters within $\perm_{\sS}(\Ab, \Bb, \bGamma, \Gb)$ as the same, we define a parameter $(\Ab, \Bb, \bGamma, \Gb)$ as identifiable if $\perm_{\sS}(\Ab, \Bb, \bGamma, \Gb) = \marg_{\sS}(\Ab, \Bb, \bGamma, \Gb)$ and nonidentifiable if $\perm_{\sS}(\Ab, \Bb, \bGamma, \Gb) \subsetneq \marg_{\sS}(\Ab, \Bb, \bGamma, \Gb)$. When discussing the identifiability of our model with a parameter space $\sS$, we try to establish whether all parameters in $\sS$ are identifiable or, at minimum, that almost every parameter in $\sS$ is identifiable. Following \citet{allman2009identifiability}, we refer to the stronger first notion as strict identifiability, and the weaker second notion as generic identifiability.

\begin{definition}[Strict Identifiability]\label{defi:strict}
The model with parameter space $\sS$ is strictly identifiable if $\perm_{\sS}(\Ab, \Bb, \bGamma, \Gb) = \marg_{\sS}(\Ab, \Bb, \bGamma, \Gb)$ for all parameters $(\Ab, \Bb, \bGamma, \Gb) \in \sS$.
\end{definition}

We equip the parameter space $\sS$ with the $L_1$ distance and its induced topology.
Let $\uplambda$ denote the Lebesgue measure on the space of continuous parameters $\Ab, \Bb, \bGamma$, and let $\upmu$ denote the counting measure on the space of discrete parameters $\Gb$. We define their product measure $\uplambda \times \upmu$ as the base measure of $\sS$. Although strict identifiability in Definition \ref{defi:strict} requires all parameters in $\sS$ to be identifiable, the following definition of generic identifiability relaxes this requirement to almost every parameter in $\sS$ with respect to the measure $\uplambda \times \upmu$.

\begin{definition}[Generic Identifiability]\label{defi:generic}
The model with parameter space $\sS$ is generically identifiable if
$
(\uplambda \times \upmu)\big\{
(\Ab, \Bb, \bGamma, \Gb) \in \sS:~
\perm(\Ab, \Bb, \bGamma, \Gb) \ne \marg(\Ab, \Bb, \bGamma, \Gb)
\big\}
=
0
.
$
\end{definition}

Generic identifiability has been widely studied \citep{allman2009identifiability, gu2023bayesian, zhou2024bayesian}.
For practical data analysis, the distinction between strict and generic identifiability is negligible, as both notions ensure reliable parameter estimation. However, the conditions required for strict identifiability are typically more stringent.

The distribution $\cP_x$ of the covariates $\xb$ has full rank if it is not supported within a lower-dimensional linear subspace of $\RR^{p_x}$, that is, if and only if for any non-zero $(p_x + 1)$-dimensional vector $\vb$,  $\big\{\xb:~ \vb^\top (1, \xb) \ne 0\big\} \cap \supp(\cP_x) \ne \varnothing$. For any two sets of vectors $\{\ab_1, \ldots, \ab_m\}$ and $\{\bbb_1, \ldots, \bbb_n\}$, where all $\ab_i$'s and $\bbb_j$'s have the same dimension, we define their entrywise product as $\{\ab_1, \ldots, \ab_m\} \bigcirc \{\bbb_1, \ldots, \bbb_n\} := \{\ab_i \circ \bbb_j:~ i \in [m], j \in [n]\}$, which extends to an arbitrary number of sets of vectors of the same dimension.
A set of $d$-dimensional vectors $\{\ab_1, \ldots, \ab_m\}$ has full rank if $m \ge d$ and the $m \times d$ matrix formed by rows $\ab_1, \ldots, \ab_m$ has full rank $d$.

\begin{theorem}\label{theo:strict}
If the distribution $\cP_x$ has full rank, then our model with the parameter space $\sS_1$ is strictly identifiable, for $\sS_1$ consisting of parameters satisfying the following conditions:
\begin{enumerate}[(i)]
\item
There exists a partition of $[q]$ as $\cQ_1 \cup \cQ_2 \cup \cQ_3$ with $\card(\cQ_1), \card(\cQ_2) \ge \log_2 d$, such that the two sets of vectors $\bigcirc_{j \in \cQ_1} \{\one, \balpha_j\}$, $\bigcirc_{j \in \cQ_2} \{\one, \balpha_j\}$ have full ranks and the submatrix $\Ab_{\cQ_3}$ has distinct columns;
\item
The matrix $\Gb$ contains three distinct identity blocks $\Ib_q$, i.e.
$
(\Pb \Gb)^\top
=
\big(\begin{matrix}
\Ib_q & \Ib_q & \Ib_q & \Gb_{\cI}^\top
\end{matrix}\big)
$
for some $p \times p$ permutation matrix $\Pb$ and subset $\cI \subset [p]$ with $\card(\cI) = p - 3q$;
\item
$\forall i \in [p], j \in [q]$, if $g_{i, j} = 1$ then $\beta_{i, j} \ne 0$.
\end{enumerate}
\end{theorem}

The proof of Theorem \ref{theo:strict} is provided in the Supplementary Material Section \ref{supp_sec:iden}, with a brief overview included in Section \ref{supp_sec:intuit}.
Intuitively, the full-rank condition of $\cP_x$ ensures the unique identification of the parameters $\bGamma$ from the conditional distribution $z ~|~ \xb$.
Conditions (i) and (ii) facilitate the unique identification of $\Ab$ and $(\one ~ \Gb) \circ \Bb$ from the conditional distributions $\wb ~|~ z$ and $\yb ~|~ \wb$, respectively. Condition (iii) is naturally imposed, as $\beta_{i, j} = 0$ implies that the attribute $w_j$ has no influence on the observed data entry $y_i$, regardless of the value of $g_{i, j}$.

The conditions in Theorem \ref{theo:strict} can often be verified using simpler sufficient conditions.
For example, a sufficient condition for $\cP_x$ to have full rank is that its support has positive Lebesgue measure, $\uplambda(\supp(\cP_x)) > 0$.
When treating the covariates as fixed rather than random, this reduces to requiring the design matrix to have full rank.
For condition (i), the set of vectors $\bigcirc_{j \in \cQ_1} \{\one, \balpha_j\}$ has full rank if the submatrix $\Ab_{\cQ_1}$ has full row rank.

In Theorem \ref{theo:strict} on strict identifiability, condition (ii) inherently requires $p \ge 3q$.
For generic identifiability, a weaker requirement of $p > 2q$ suffices, as shown in the next theorem.
Its proof is deferred to the Supplementary Material Section \ref{supp_sec:iden}.

\begin{theorem}\label{theo:generic}
If the distribution $\cP_x$ has full rank, then our model with the parameter space $\sS_2$ is generically identifiable, for $\sS_2$ consisting of parameters satisfying the following condition:
\begin{itemize}
\item[($*$)]
$\Gb$ has distinct columns and takes the blockwise form
$
(\Pb \Gb)^\top
=
\big(\begin{matrix}
\Gb_{\cI_1}^\top & \Gb_{\cI_2}^\top & \Gb_{\cI_3}^\top
\end{matrix}\big)
$
for some $p \times p$ permutation matrix $\Pb$, with two $q \times q$ submatrices $\Gb_{\cI_1}, \Gb_{\cI_2}$ having all-one diagonals and the remaining $(p - 2q) \times q$ submatrix $\Gb_{\cI_3}$ having no all-zero columns, i.e. $\forall j \in [q]$, $(\Gb_{\cI_1})_{j, j} = (\Gb_{\cI_2})_{j, j} = 1$ and $\sum_{i \in \cI_3} g_{i, j} > 0$.
\end{itemize}
\end{theorem}

The conditions in Theorem \ref{theo:generic} for generic identifiability are significantly weaker than those in Theorem \ref{theo:strict} for strict identifiability, which implies $\sS_2 \supset \sS_1$. For generic identifiability, conditions (i) and (iii) of Theorem \ref{theo:strict}, which impose constraints on the parameters $\Ab, \Bb$, are no longer required. Moreover, condition (ii) in Theorem \ref{theo:strict}, which mandates that $\Gb$ contains three identity blocks, can be relaxed to require only two blocks with all-one diagonals and arbitrary off-diagonal entries. These blockwise conditions in Theorems \ref{theo:strict} and \ref{theo:generic} are common in establishing identifiability for constrained latent class models \citep[e.g.,][]{gu2020partial, gu2023bayesian}.

\subsection{Posterior Consistency}\label{ssec:post_cons}

Next, we establish the posterior consistency of our model. 
The standard definition of posterior consistency requires that, for any $\epsilon$-neighborhood $\cO_\epsilon(\Ab^*, \Bb^*, \bGamma^*, \Gb^*)$ around the true parameter $(\Ab^*, \Bb^*, \bGamma^*, \Gb^*)$, the posterior probability of its complement $\cO_\epsilon^c(\Ab^*, \Bb^*, \bGamma^*, \Gb^*)$ asymptotically converges to zero. We modify this definition to consider an equivalence class-based $\epsilon$-neighborhood $\tilde{\cO}_\epsilon$, which surrounds all permutations of the true parameter $(\Ab^*, \Bb^*, \bGamma^*, \Gb^*)$. Specifically, we define
\begin{equation}\label{eq:eps_nbhd_perm}
\tilde{\cO}_\epsilon(\Ab^*, \Bb^*, \bGamma^*, \Gb^*)
:=
\bigcup_{(\Ab, \Bb, \bGamma, \Gb) \in \perm_{\sS}(\Ab^*, \Bb^*, \bGamma^*, \Gb^*)} \cO_\epsilon(\Ab, \Bb, \bGamma, \Gb)
.
\end{equation}
This adaptation ensures that posterior consistency is defined in a manner that accommodates the permutation invariance of our model.

\begin{definition}\label{defi:post}
The model is posterior consistent at parameter $(\Ab^*, \Bb^*, \bGamma^*, \Gb^*)$ if for any $\epsilon > 0$,
$$
\lim_{N \to \infty}
\PP\Big(
(\Ab, \Bb, \bGamma, \Gb) \in \tilde{\cO}_\epsilon^c(\Ab^*, \Bb^*, \bGamma^*, \Gb^*)
~\Big|~
\xb^{(1:N)}, \yb^{(1:N)}
\Big)
=
0
,
$$
$\PP(\xb^{(1:\infty)}, \yb^{(1:\infty)} ~|~ \Ab^*, \Bb^*, \bGamma^*, \Gb^*)$-almost surely.
\end{definition}

Definition \ref{defi:post} essentially states that, for an independent sequence of observed data $\yb^{(n)}$ and covariates $\xb^{(n)}$ generated from our model under parameters $(\Ab^*, \Bb^*, \bGamma^*, \Gb^*)$, asymptotically, the posterior probability of the complement of any $\epsilon$-neighborhood of $\perm_{\sS}(\Ab^*, \Bb^*, \bGamma^*, \Gb^*)$ should almost surely vanish.
For $\sS$ being a separable metric space, this is equivalent to the posterior distribution being asymptotically supported on the permutations of true parameters in $\perm_{\sS}(\Ab^*, \Bb^*, \bGamma^*, \Gb^*)$.

Recall that $\sS_1$ and $\sS_2$ are the parameter spaces defined in Theorems \ref{theo:strict} and \ref{theo:generic}, ensuring strict and generic identifiability of our model, respectively.
We say that the distribution $\cP_x$ of the covariates $\xb$ has a finite first moment if $\int_{\RR^{p_x}} \|\xb\|_\infty \cP_x(\ud \xb) < \infty$.
Let $\uppi(\Ab, \Bb, \bGamma, \Gb)$ denote the prior distribution over the parameters $(\Ab, \Bb, \bGamma, \Gb)$, supported within the parameter space $\sS$.
Although our recommended choice of prior will be introduced later in Section \ref{ssec:prior}, the following theorem establishes the posterior consistency of our model for any prior $\uppi$ that is absolutely continuous with respect to the base measure $\uplambda \times \upmu$.

\begin{theorem}\label{theo:post}
Let the prior $\uppi$ be absolutely continuous w.r.t. $\uplambda \times \upmu$.
Suppose $\cP_x$ has full rank and a finite first moment, then our model with parameter space $\sS$ is posterior consistent at
(i) any identifiable parameter in $\supp(\uppi)$ if $\sS$ is compact;
(ii) any parameter in $\supp(\uppi)$ if $\sS \subset \sS_1$;
(iii) $\uppi$-almost every parameter in $\sS$ if $\sS \subset \sS_2$.
\end{theorem}

The proof of Theorem \ref{theo:post} is in the Supplementary Material Section \ref{supp_sec:post_cons}, with an intuitive overview in Section \ref{supp_sec:intuit}.
The scenarios in Theorem \ref{theo:post} encompass virtually all parameter spaces of practical interest. For model interpretability, we can require that each attribute $w_j$ affects at least two data entries $y_i$. For high-dimensional data $\yb$, this suggests a parameter space $\sS \subset \sS_2$. Moreover, we can truncate the range of each continuous parameter in $\bGamma, \Bb$ to be bounded by a sufficiently large constant to facilitate theoretical analysis without having a practical impact.

\section{Limiting Behavior in High-Dimensional Clustering}\label{sec:high_dim}

Bayesian inference for mixture models has been shown to suffer from a curse of dimensionality.
\citet{chandra2023escaping} show that as the dimensionality of the data increases, the posterior distributions of Gaussian mixture models asymptotically assign all the data to a single shared cluster or place each data in its own separate cluster.
We prove a similar result for multivariate Bernoulli mixture models in Section \ref{ssec:curse}. 

\texttt{BLIP} is designed to overcome this curse of dimensionality, by introducing the middle layer of fine-grained latent attributes $\wb$ having dimension $q \ll p$.
We formalize this idea in Section \ref{ssec:oracle_prob} by defining the Bayes oracle clustering probability for our model, which is invariant to the dimension $p$ of the observed data $\yb$. Then in Section \ref{ssec:escape}, under interpretable diversity and separation conditions, we prove that the posterior probability of the deep latent class $z$ in our model asymptotically converges to its Bayes oracle clustering probability as $p \to \infty$. This result suggests that our model successfully escapes the curse of dimensionality as $p \to \infty$.

\subsection{Curse of Dimensionality in Bayesian Mixture Models}\label{ssec:curse}

In this subsection, we analyze Bayesian Bernoulli mixture models for clustering $N$ observations of $p$-dimensional binary vectors $\yb^{(1:N)}$, under the asymptotic regime of $N$ fixed and $p \to \infty$.
To avoid confusion with previous work on the consistency and inconsistency of Dirichlet mixture models under various prior distributions \citep{miller2014inconsistency, ascolani2023clustering}, we emphasize the distinction that they study a different asymptotic regime with $p$ fixed and $N \to \infty$.

To demonstrate the curse of dimensionality, assume for simplicity that observations $\yb^{(1:N)}$ are independent realizations from a single-cluster Bernoulli mixture model. That is, for a sequence of true parameters $\{a_i^*\}_{i = 1}^\infty \subset (0, 1)$, we have the data generating distribution
\begin{equation}\label{eq:mix_true}
\PP_{\mix}^*(\yb^{(1:N)})
=
\prod_{i = 1}^p (a_i^*)^{\sum_{n = 1}^N y_i^{(n)}} (1 - a_i^*)^{\sum_{n = 1}^N (1 - y_i^{(n)})}.
\end{equation}
As a specific example, we let the true parameters be $a_i^* = \frac12$ for all $i \ge 1$.

We consider clustering the $N$ observations $\yb^{(1:N)}$ using a Bayesian Bernoulli mixture model. Let $z^{(n)}$ denote the latent class of observation $n$, and let $a_{i, z}$ be the Bernoulli parameter associated with the conditional distribution of data entry $y_i^{(n)}$ given the latent class $z^{(n)} = z$.
We denote the collection of all parameters $a_{i, z}$ by $\ba$. Under the Bernoulli mixture model with parameters $\ba$, the conditional distribution of $\yb^{(1:N)} ~|~ z^{(1:N)}$ is given by
\begin{equation}\label{eq:mix_model}
\PP_{\mix}(\yb^{(1:N)} ~|~ z^{(1:N)}, \ba)
=
\prod_{n = 1}^N \prod_{i = 1}^p \PP_{\mix}(y_i^{(n)} ~|~ z^{(n)}, a_{i, z^{(n)}})
=
\prod_{n = 1}^N \prod_{i = 1}^p a_{i, z^{(n)}}^{y_i^{(n)}} (1 - a_{i, z^{(n)}})^{1 - y_i^{(n)}}
.
\end{equation}

We place an arbitrary prior distribution on the cluster assignments $z^{(1:N)}$, subject to the constraint that every possible partition of $[N]$ is assigned a positive prior probability. This encompasses finite mixture models with a prior on the number of clusters, Bayesian nonparametric mixture models, and mixture models with covariate dependent clusters. As a common default in mixture models for binary data, we choose uniform priors for the parameters $\ba$, where each $a_{i, z}$ independently follows $a_{i, z} \sim \mathrm{Uniform}(0, 1)$.

Given our prior choices and the data-generating distribution $\PP_{\mix}^*$ in \eqref{eq:mix_true}, the following proposition characterizes the asymptotic behavior of the Bernoulli mixture model \eqref{eq:mix_model} as $p \to \infty$.
We defer intuitive and detailed proofs to the Supplementary Material Sections \ref{supp_sec:intuit} and \ref{supp_sec:mix}, respectively.

\begin{proposition}\label{prop:curse_example}
Let $z^{(1:N)} = \{\{1\}, \{2\}, \ldots, \{N\}\}$ denote the partition of observations $[N]$ into $N$ separate clusters. Then for any other partition $\sZ$ of $[N]$, we have
$$
\liminf_{p \to \infty} \frac{\PP_{\mix}(z^{(1:N)} = \{\{1\}, \{2\}, \ldots, \{N\}\} ~|~ \yb^{(1:N)})}{\PP_{\mix}(z^{(1:N)} = \sZ ~|~ \yb^{(1:N)})}
=
\infty
,
$$
$\PP_{\mix}^*$-almost surely.
\end{proposition}

A key implication of Proposition \ref{prop:curse_example} is $\lim_{p \to \infty} \PP_{\mix}(z^{(1:N)} = \{\{1\}, \{2\}, \ldots, \{N\}\} ~|~ \yb^{(1:N)}) = 1$, $\PP_{\mix}^*$-almost surely.
This result indicates that the Bayesian Bernoulli mixture model asymptotically assigns each observation to its own separate cluster, despite the data being generated from a well-specified true model with only one cluster. Intuitively, as the dimensionality of the data $p$ increases, the mixture model starts to misinterpret the randomness in the observed data as meaningful differences, ultimately leading to over-clustering. A discussion of related frequentist literature is provided in the Supplementary Material Section \ref{supp_sec:lit}.
Although Proposition \ref{prop:curse_example} presents a specific case of Bayesian Bernoulli mixture models with a random and unbounded number of clusters, it sheds light on how the resulting clusters in general Bayesian mixture models can become increasingly unreliable as the dimension $p$ increases.

\subsection{Bayes Oracle Clustering Probability}\label{ssec:oracle_prob}

Inspired by \citet{chandra2023escaping}, we define the Bayes oracle clustering probability for our \texttt{BLIP} model. We assume that the observed data $\yb^{(1:N)}$ follow the distribution in \eqref{eq:model} and that there is an oracle that has access to the true values $\wb_*^{(1:N)}$ of the latent binary attributes. Given this knowledge, the oracle probability of deep latent clusters $z^{(1:N)}$ is defined as their conditional probability given the covariates $\xb^{(1:N)}$ and the true binary attributes $\wb_*^{(1:N)}$, as formally stated in the following definition. The oracle probability does not depend on the observed data or its dimensionality $p$.

\begin{definition}\label{defi:oracle_prob}
Given the true attribute vectors $\wb_*^{(1:N)}$, the oracle probability of $z^{(1:N)}$ is
$$
\PP(z^{(1:N)} ~|~ \xb^{(1:N)}, \wb_*^{(1:N)})
=
\frac{
\PP(\wb_*^{(1:N)} ~|~ z^{(1:N)})
\PP(z^{(1:N)} ~|~ \xb^{(1:N)})
}{\sum_{\tilde{z}^{(1:N)} \in [d]^N}
\PP(\wb_*^{(1:N)} ~|~ \tilde{z}^{(1:N)})
\PP(\tilde{z}^{(1:N)} ~|~ \xb^{(1:N)})
}
.
$$
\end{definition}

In Definition \ref{defi:oracle_prob}, $\PP(\wb_*^{(1:N)} ~|~ z^{(1:N)})$ and $\PP(z^{(1:N)} ~|~ \xb^{(1:N)})$ represent the conditional probabilities of $\wb_*^{(1:N)} ~|~ z^{(1:N)}$ and $z^{(1:N)} ~|~ \xb^{(1:N)}$, respectively, with the parameters $\bGamma, \Ab$ marginalized out using their prior distributions.
By assuming that the priors on $\bGamma$ and $\Ab$ are independent, as in our recommended prior in Section \ref{ssec:prior}, we have $\PP(\wb_*^{(1:N)} ~|~ z^{(1:N)}) = \int_{[0, 1]^{q \times d}} \PP(\Ab) \prod_{n = 1}^N \PP(\wb_*^{(n)} ~|~ z^{(n)}, \Ab) \ud \Ab$ and $\PP(z^{(1:N)} ~|~ \xb^{(1:N)}) = \int_{\RR^{d \times p_x}} \PP(\bGamma) \prod_{n = 1}^N \PP(z^{(n)} ~|~ \xb^{(n)}, \bGamma) \ud \bGamma$.

\subsection{Escaping the Curse of Dimensionality}\label{ssec:escape}

In this section, we show under interpretable conditions that, as the data dimensionality $p \to \infty$, the posterior probability of $z^{(1:N)} ~|~ \xb^{(1:N)}, \yb^{(1:N)}$ converges to the Bayes oracle clustering probability defined in Definition \ref{defi:oracle_prob}.
As discussed in Section \ref{ssec:oracle_prob}, we let $(\Ab^*, \Bb^*, \bGamma^*, \Gb^*)$ and $\wb_*^{(1:N)}$ denote the true parameters and the true attribute vectors of the data generating process.
We start by making some assumptions about $(\one ~ \Gb^*) \circ \Bb^*$ and $\wb_*^{(1:N)}$.

\begin{assumption}\label{assu:oracle_w}
The set of true attribute vectors $\{\wb_*^{(n)}:~ n \in [N]\} \supset \{0, 1\}^q$.
\end{assumption}

The Assumption \ref{assu:oracle_w} essentially requires that the collection of $\wb_*^{(1:N)}$ is diverse, covering all $2^q$ possible combinations of binary attributes. For example, the collection of $\wb_*^{(1:4)} = (0, 0)$, $(0, 1)$, $(1, 0)$, $(1, 1)$ satisfies Assumption \ref{assu:oracle_w} with $q = 2$.

\begin{assumption}\label{assu:oracle_linear}
There exists some $\upsilon_0^* \in \RR$, vector $\bupsilon^* \in \RR^q$, and $\delta > 0$ such that
$$
\limsup_{p \to \infty} \left| \logit\left(
\frac{1}{p} \sum_{i = 1}^p \frac{\exp(\beta_{i, 0}^* + (\gb_i^* \circ \bbeta_i^*)^\top \wb_*^{(n)})}{1 + \exp(\beta_{i, 0}^* + (\gb_i^* \circ \bbeta_i^*)^\top \wb_*^{(n)})}
\right)
-
\upsilon_0^*
-
(\bupsilon^*)^\top \wb_*^{(n)}
\right|
<
\delta
,\quad
\forall n \in [N]
.
$$
\end{assumption}

Assumption \ref{assu:oracle_linear} is a mild condition that ensures that $(\one ~ \Gb^*) \circ \Bb^*$ is asymptotically well-conditioned. There are simpler sufficient conditions for this assumption, such as $\limsup_{p \to \infty} \allowbreak \frac{1}{p} \sum_{i = 1}^p \beta_{i, 0}^* = \upsilon_0^*$ and $\limsup_{p \to \infty} \frac{1}{p} \sum_{i = 1}^p (\gb_i^* \circ \bbeta_i^*) = \bupsilon^*$.
We provide details in the Supplementary Material Section \ref{supp_sec:oracle}.

\begin{assumption}\label{assu:oracle_inequality}
For some known constant $\tau^*$, we have
$$
\inf\left\{
|\cbb^\top \bupsilon^*|:~
\cbb \in \{-2, -1, 0, 1, 2\}^q,~
\cbb \ne 0
\right\}
\ge
\tau^*
>
8\delta
.
$$
\end{assumption}

We note that $\inf\big\{ |\cbb^\top \bupsilon^*|:~ \cbb \in \{-2, -1, 0, 1, 2\}^q,~ \cbb \ne 0 \big\}$ is equivalent to the minimum of $\big| (\bupsilon^*)^\top (\wb_*^{(n_1)} - \wb_*^{(n_2)} - \wb_*^{(n_3)} + \wb_*^{(n_4)}) \big|$ over all $n_1, n_2, n_3, n_4 \in [N]$ with $\wb_*^{(n_1)} - \wb_*^{(n_2)} - \wb_*^{(n_3)} + \wb_*^{(n_4)} \ne \zero$, which is positive for almost every $\bupsilon^* \in \RR^q$.
This quantity appears in the proof of the following theorem, where we construct an estimator for $(\bupsilon^*)^\top \wb_*^{(n)}$ using $\yb^{(n)}$.
The assumption $\tau^* > 8\delta$ provides theoretical control over the accuracy of this estimator. In general, Assumption \ref{assu:oracle_inequality} is a mild condition for sequences of $\gb_i^*, \beta_{i, 0}^*, \bbeta_i^*$ that satisfy Assumption \ref{assu:oracle_linear} with a small $\delta$.

A prior distribution that is independent over $\Ab, \Bb, \bGamma, \Gb$, is said to be i.i.d. and symmetric on $\Ab$ if $\PP(\Ab) = \prod_{j = 1}^q \prod_{h = 1}^d \PP(\alpha_{j, h})$, where $\PP(\alpha_{j, h})$ is identical for all $j \in [q], h \in [d]$ and satisfies the symmetry condition $\PP(\alpha_{j, h}) = \PP(1 - \alpha_{j, h})$. For example, $\alpha_{j, h} \stackrel{iid}{\sim} \mathrm{Beta}(b, b)$ for some $b > 0$ specifies an i.i.d. and symmetric prior to $\Ab$. 

\begin{theorem}\label{theo:oracle}
Let Assumptions \ref{assu:oracle_w}, \ref{assu:oracle_linear}, \ref{assu:oracle_inequality} hold and the prior distribution be i.i.d. and symmetric on $\Ab$, then
$$
\lim_{p \to \infty} \PP(z^{(1:N)} ~|~ \xb^{(1:N)}, \yb^{(1:N)})
=
\PP(z^{(1:N)} ~|~ \xb^{(1:N)}, \wb_*^{(1:N)})
$$
holds $\PP(\yb^{(1:N)} ~|~ \wb_*^{(1:N)}, \Gb^*, \Bb^*)$-almost surely.
\end{theorem}

Theorem \ref{theo:oracle} states that as the data dimensionality $p \to \infty$, the posterior probability of the deep latent clusters $z^{(1:N)}$ becomes asymptotically equivalent to their oracle probability.
Intuitively, this occurs because, for sufficiently large $p$, the high-dimensional observations $\yb^{(1:N)}$ contain enough information to recover the unobserved attributes $\wb_*^{(1:N)}$, up to certain label switching that does not affect the conditional distribution of $z^{(1:N)} ~|~ \wb^{(1:N)}$.
The detailed proof is provided in the Supplementary Material Section \ref{supp_sec:oracle}.

Crucially, the oracle probability $\PP(z^{(1:N)} ~|~ \xb^{(1:N)}, \wb_*^{(1:N)})$ remains invariant and reliable regardless of the data dimensionality.
Thus, for clustering high-dimensional binary observations, Theorem \ref{theo:oracle} provides a theoretical guarantee of the reliability of posterior inference for deep latent clusters $z^{(1:N)}$.
This result demonstrates that our model effectively avoids the curse of dimensionality faced by Bayesian inference for mixture models, as established in Proposition \ref{prop:curse_example}.

The proof suggests that Theorem \ref{theo:oracle} remains broadly valid for models with arbitrary priors for $z^{(1:N)}$: as long as the model has the dependency structure $z \to \wb \to \mathbf{y}$ and satisfies assumptions \ref{assu:oracle_w}, \ref{assu:oracle_linear}, \ref{assu:oracle_inequality}, the posterior probability of $z^{(1:N)}$ converges to the oracle probability as $p \to \infty$.

\section{Posterior Computation}\label{sec:post}

In this section, we present a hierarchical prior distribution for our model and develop a data-augmented Gibbs sampler to perform posterior inference.
We further numerically demonstrate the consistency of our posterior inference and model selection through a simulated example.

\subsection{Prior Specifications}\label{ssec:prior}

We first specify prior distributions for the parameters $(\Ab, \Bb, \bGamma, \Gb)$ in model \eqref{eq:model}.
Our model incorporates observation-specific covariates $\xb$ through the conditional distribution of the deep latent cluster $z ~|~ \xb$ in \eqref{eq:model_z_x}.
In many applications, one also has access to variable- instead of sample-specific features
\citep{norberg2019comprehensive}. For example, in ecology, each species has its biological traits; in genetics, each DNA nucleotide has a corresponding location; and in survey design, each question has its semantic features. We denote the collection of such meta covariates by $p \times p_t$ matrix $\Tb$, where the $i$th row $\tb_i$ represents the features for the $i$th entry.

As discussed in Section \ref{sec:model}, the binary matrix $\Gb$ encodes the dependency of each $y_i$ on the latent attributes in $\wb$.
We design a hierarchical prior distribution that encourages data entries $y_i$ with similar meta covariates $\tb_i$ to exhibit similar dependency structures on $\wb$:
\begin{equation}\label{eq:prior_G}
g_{i, j} ~|~ \tb_i, \theta_{j, 0}, \btheta_j
\stackrel{ind}{\sim}
\mathrm{Bernoulli}\left(
\frac{
\exp(\theta_{j, 0} + \btheta_j^\top \tb_i)
}{
1 + \exp(\theta_{j, 0} + \btheta_j^\top \tb_i)
}
\right)
,\quad
\theta_{j, 0}
\stackrel{iid}{\sim}
N(0, 1)
,\quad
\btheta_j
\stackrel{iid}{\sim}
N(0, \Ib)
,
\end{equation}
where for each $j \in [q]$, $\theta_{j, 0} \in \RR$ and $\btheta_j \in \RR^{p_t}$ are hyperparameters associated with attribute $w_j$ serving as the regression intercept and coefficients.
Intuitively, we can interpret each $\btheta_j$ as a summary of the effects of the meta covariates relevant to $w_j$, such that a data entry $y_i$ has a higher prior probability of depending on $w_j$ when $\btheta_j^\top \tb_i$ is larger.

For the continuous parameters $\Ab, \Bb, \bGamma$, we adopt independent, weakly informative priors \citep{gelman2006prior} by default, though structured priors can be applied in specific cases.
A detailed discussion is provided in the Supplementary Material Section \ref{supp_sec:post}.

\subsection{Posterior Sampling and Simulation Studies}\label{ssec:sample_sim}

To perform posterior inference for our model, we design a Gibbs sampler using Polya-Gamma data augmentation \citep{polson2013bayesian}.
The details are provided in the Supplementary Material Section \ref{supp_sec:post}, along with an approximate variant to enhance computational efficiency in ultra-high-dimensional settings.
The Gibbs sampler enables efficient and fast posterior computation, with diagnostic trace plots indicating good mixing.

Additionally, we conduct a series of simulation studies in the Supplementary Material Section \ref{supp_sec:sim} to evaluate our posterior computation framework.
These studies assess model selection using Watanabe–Akaike Information Criterion \citep{watanabe2010asymptotic}, posterior inference of model parameters and latent variables, and asymptotic behavior as data dimensionality increases.
Collectively, these results confirm the efficiency and reliability of our inference framework, while numerically validating the theoretical results on posterior consistency and the Bayes oracle clustering property developed in Sections \ref{sec:theo} and \ref{sec:high_dim}.

\section{Application to Joint Species Distribution Modeling}\label{sec:appl}

Joint Species Distribution Model (\texttt{JSDM}) characterize the joint distribution of the occurrences of multiple species and their dependence on environmental covariates \citep{ovaskainen2016uncovering, ovaskainen2024common}.
Traditional \texttt{JSDM}s use continuous latent variables to capture dependencies between species occurrences \citep{norberg2019comprehensive}.
As a demonstration, we apply our model to the Finnish bird dataset, a benchmark in \texttt{JSDM} research \citep{ovaskainen2020joint}.

The Finnish bird dataset records the presence or absence of 50 bird species in various sampling locations in Finland over multiple years. The dataset includes environmental covariates for each sampling location, such as the mean temperature in April and May and habitat type, classified into broadleaf forests, coniferous forests, open habitats, urban habitats, and wetlands.
Additionally, meta covariates, which structure our prior in Section \ref{ssec:prior}, include biological traits of each bird species, such as the log-transformed typical body mass and migration type (classified as long term, short term, and resident).

We used data recorded from 2011 to 2013 as in-sample observations and data recorded in year 2014 as out-of-sample observations, comprising 363 and 137 sampling locations, respectively. The dataset also contains spatial coordinates for each sampling location and a phylogenetic tree of the 50 bird species. Although our model does not explicitly use these data, they provide valuable context for interpreting the deep latent classes and latent attributes inferred by our model.

\subsection{Interpretation}\label{ssec:interp}

\begin{figure}[ht!]
\centering
\includegraphics[width = 0.5\textwidth]{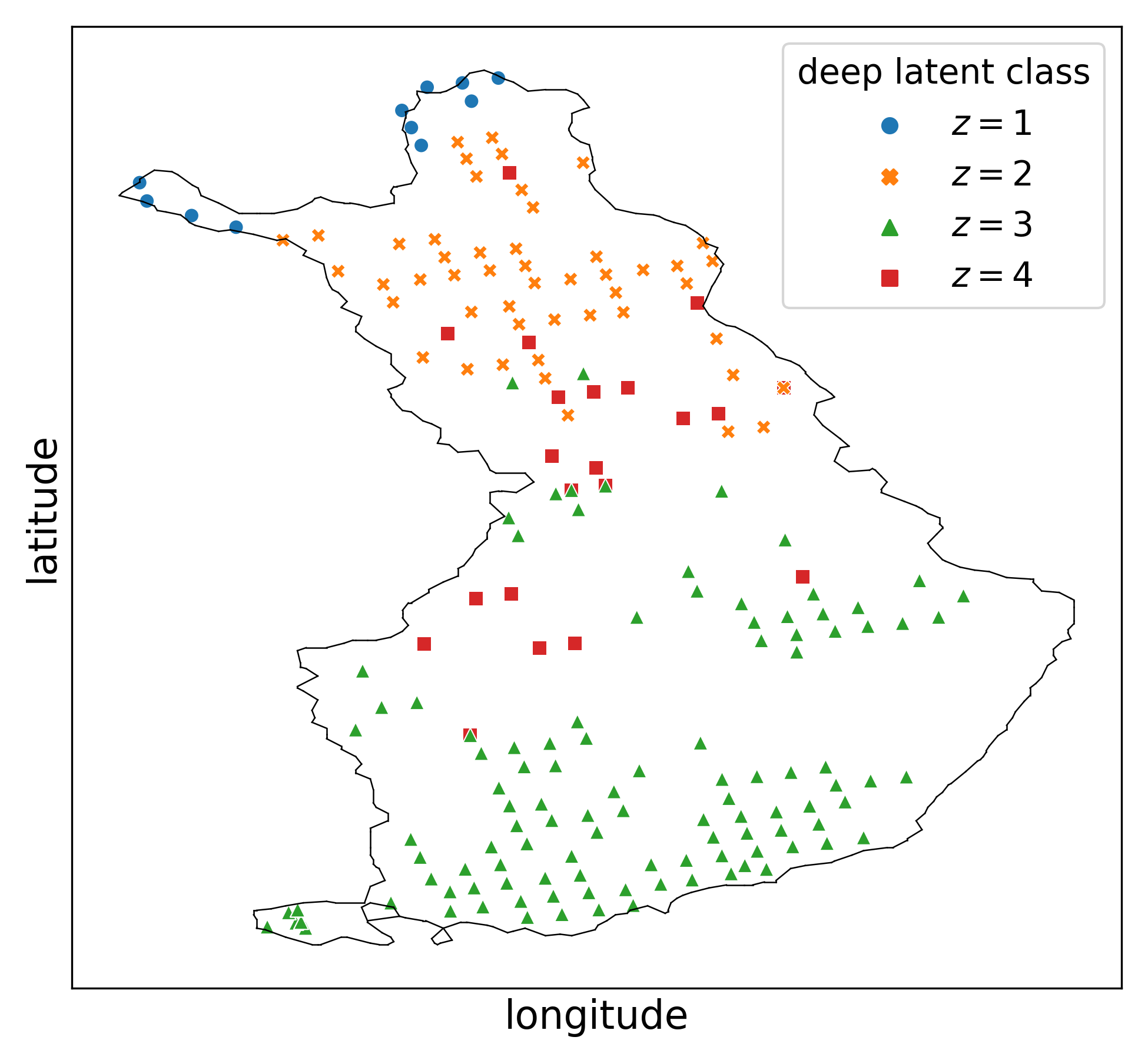}
\caption{
Spatial distribution of bird species sampling locations in Finland, color-coded by their inferred deep latent class $z$. In ecology, such clusters are referred to as regions of common profile. 
}
\label{fig:app_z_map}
\end{figure}

\begin{figure}[ht!]
\centering
\includegraphics[width = \textwidth]{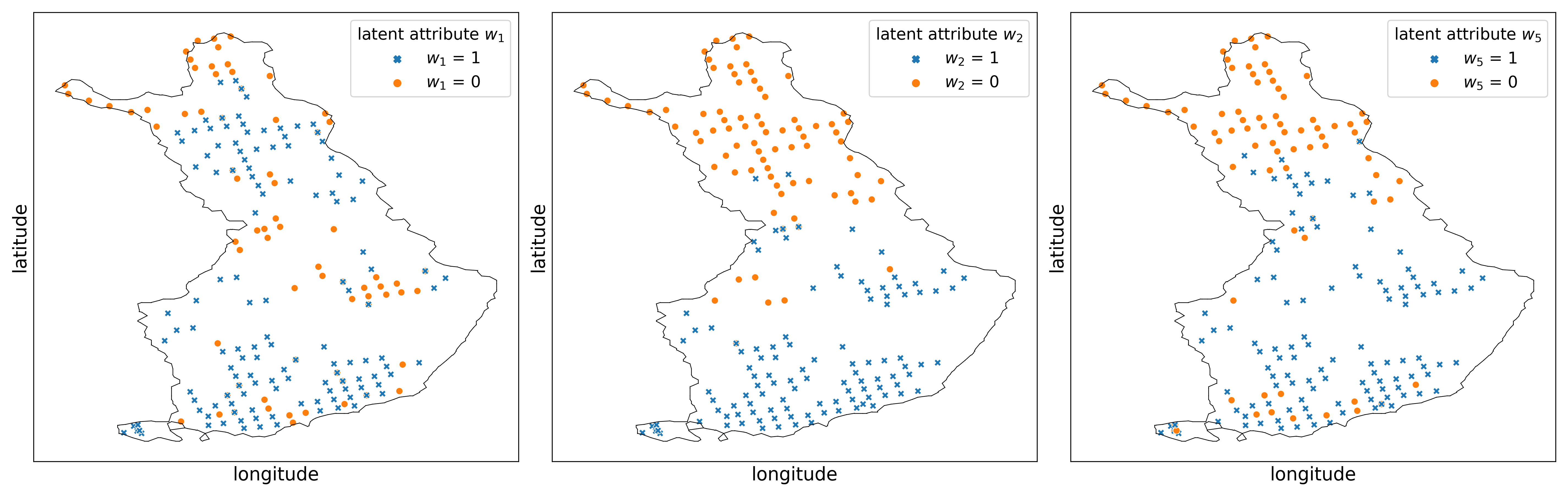}
\caption{
Spatial distribution of bird species sampling locations on the map of Finland, color-coded by the values of their binary latent attributes $w_1, w_2, w_5$, respectively.
For latent attributes $w_3, w_4$, see Supplementary Material Figure \ref{fig:app_w_map_full}.
}
\label{fig:app_w_map}
\end{figure}

\begin{figure}[ht!]
\centering
\begin{subfigure}[b]{0.32\textwidth}
\centering
\includegraphics[width = \textwidth, trim = 0 3 0 0, clip]{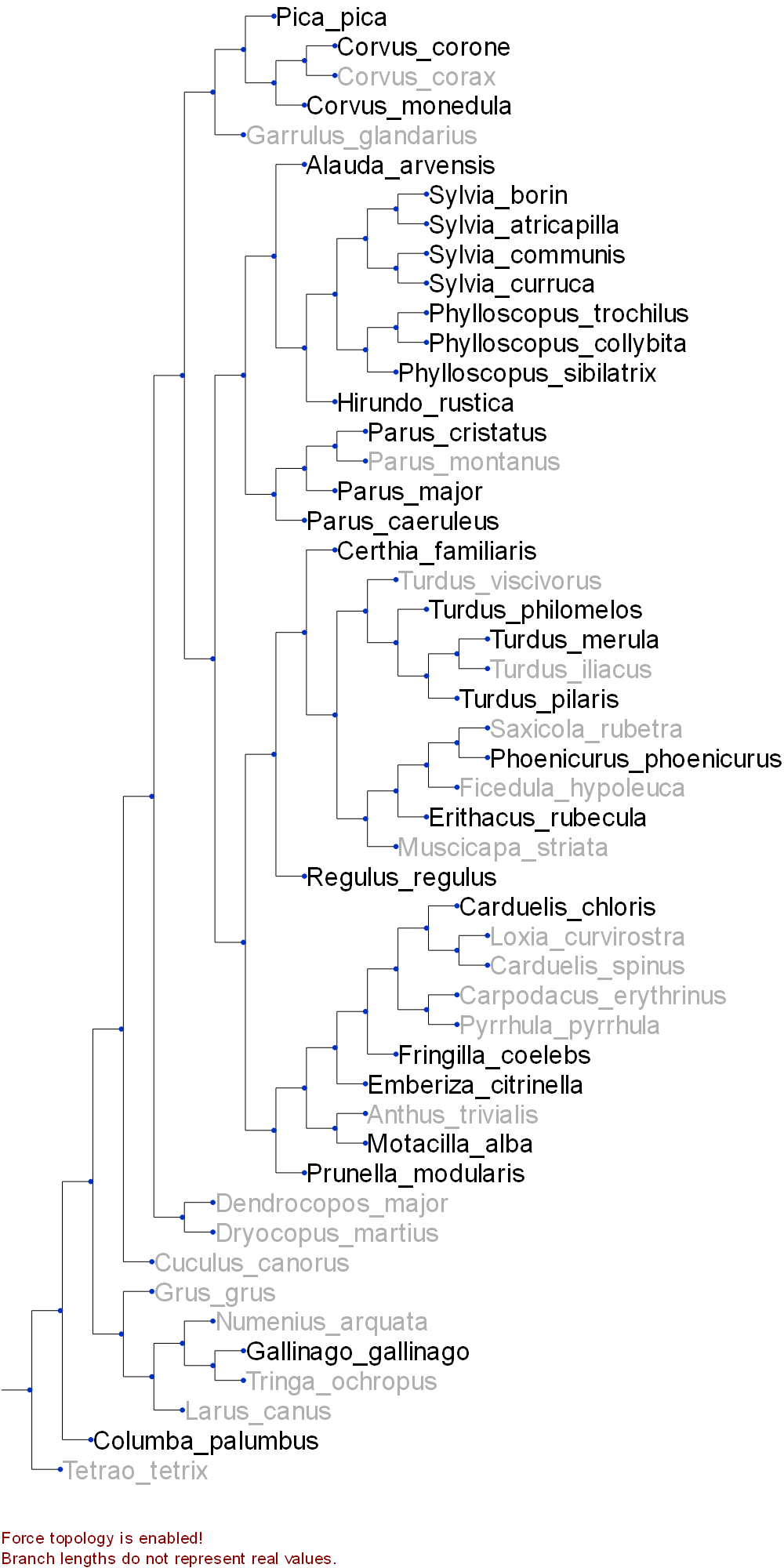}
\caption{Species dependent on $w_2$.}
\label{sub_fig:phylo2}
\end{subfigure}
\begin{subfigure}[b]{0.32\textwidth}
\centering
\includegraphics[width = \textwidth, trim = 0 3 0 0, clip]{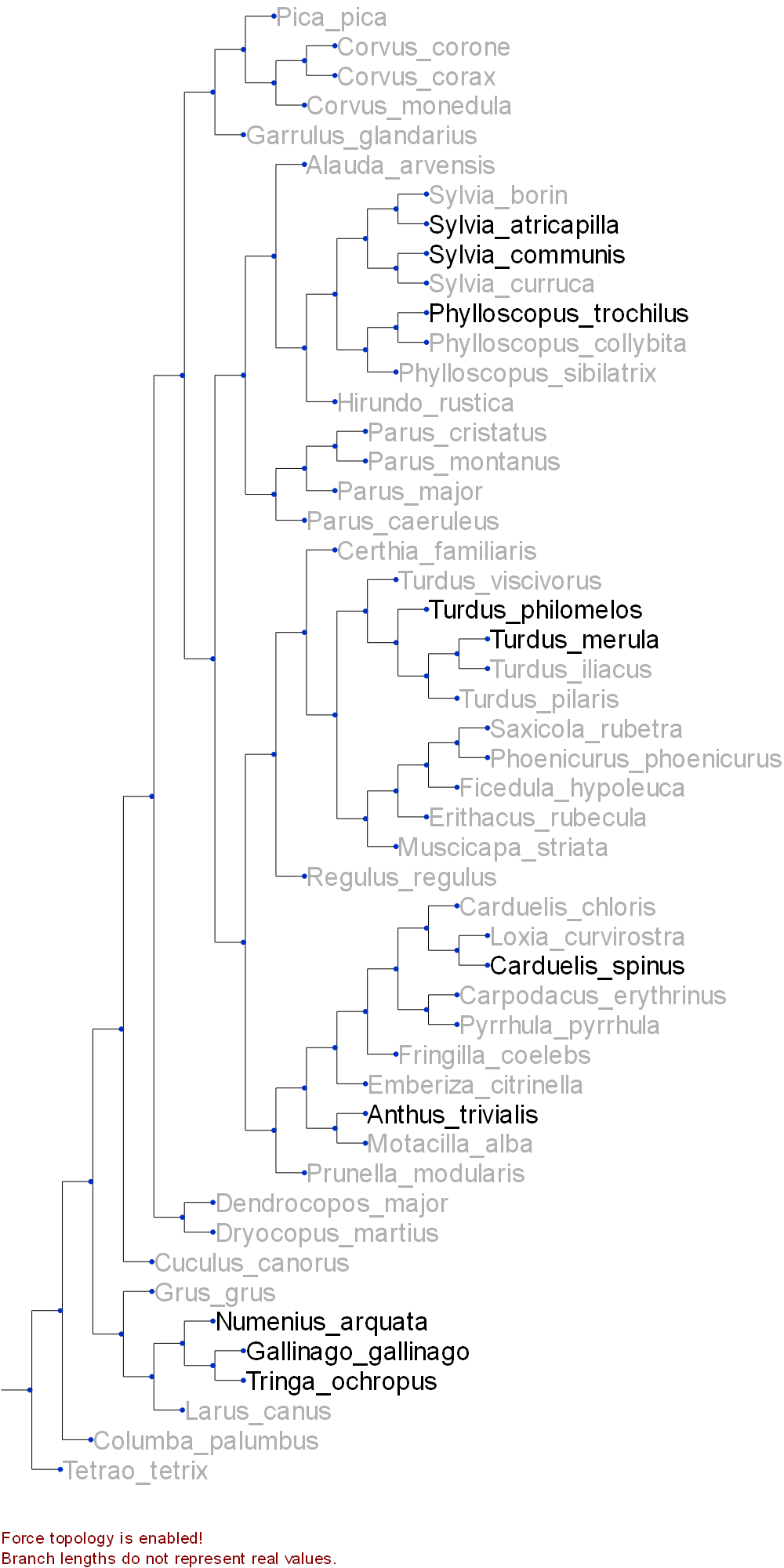}
\caption{Species dependent on $w_4$.}
\label{sub_fig:phylo4}
\end{subfigure}
\begin{subfigure}[b]{0.32\textwidth}
\centering
\includegraphics[width = \textwidth, trim = 0 3 0 0, clip]{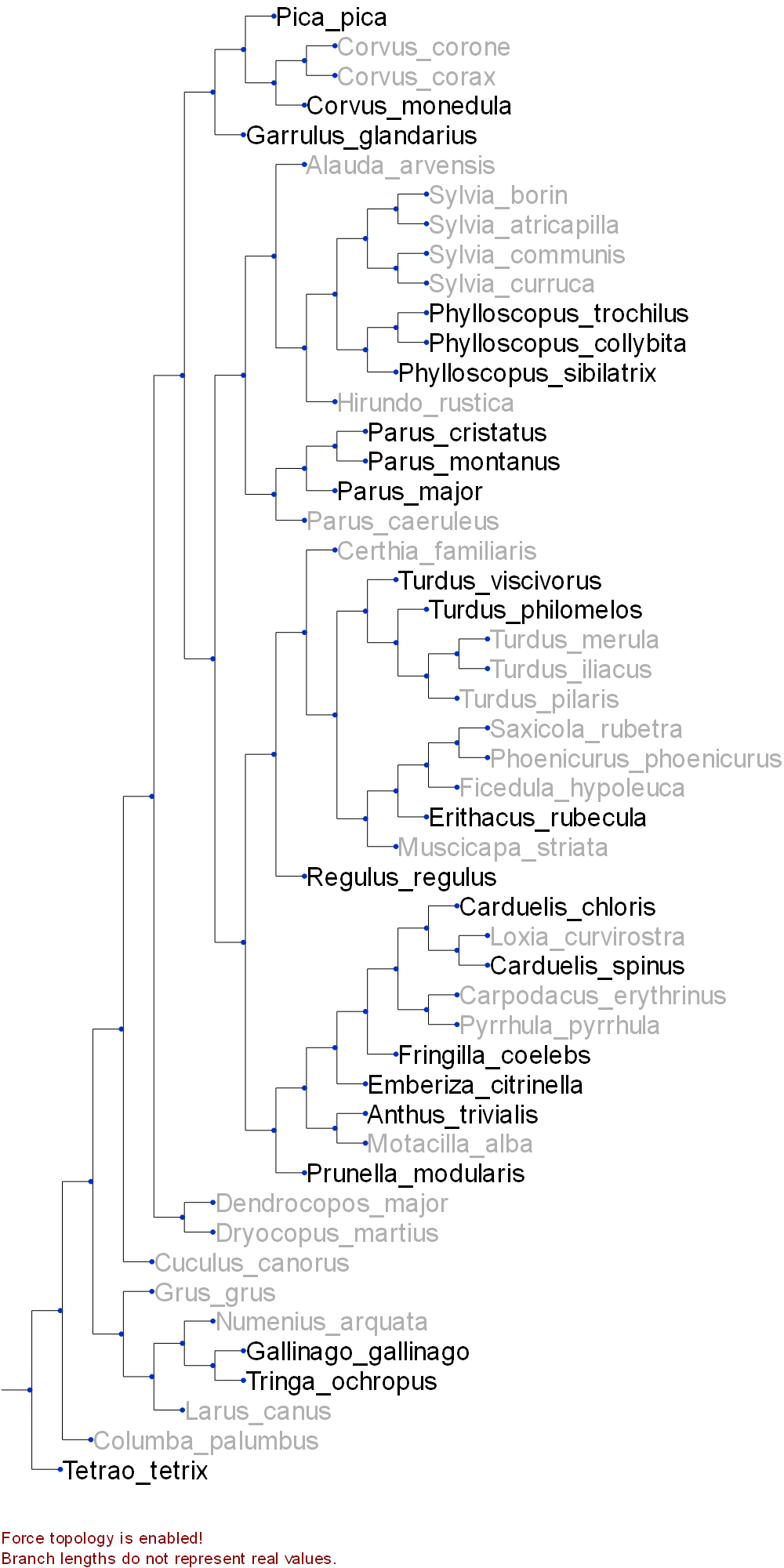}
\caption{Species dependent on $w_5$.}
\label{sub_fig:phylo5}
\end{subfigure}
\caption{
Phylogenetic trees of bird species, where species names in black indicate dependence on latent attribute $w_j$ (i.e. species $i$ with $g_{i, j} = 1$), while names in gray indicate no dependence, for $j = 2, 4, 5$.
For latent attributes $w_1, w_3$, refer to the supplementary material Figure \ref{fig:phylo_full}.
}
\label{fig:phylo}
\end{figure}

We follow the same procedures for model selection and posterior inference as in the simulation studies, with details provided in the Supplementary Material Section \ref{supp_sec:app}.
Using the Watanabe–Akaike information criterion \citep[\texttt{WAIC},][]{watanabe2010asymptotic}, we select the optimal model structure with $q = 5$ latent attributes and $d = 4$ deep latent classes.
In the following, we provide an overview of model interpretation, with additional details, including further figures and tables, also deferred to the Supplementary Material Section \ref{supp_sec:app}.

The deep latent class assignments of in-sample observations are visualized in Figure \ref{fig:app_z_map}.
The inferred deep latent classes align well with Finland's geography, effectively capturing distinct ecological zones.
Class $z = 1$ represents a borderline arctic tundra habitat, characterized by high mountains, the absence of trees, and a unique assemblage of bird species found nowhere else in the country.
Class $z = 2$ corresponds to the northern boreal forest, which harbors Lapland forest species along with other old-growth forest species.
Class $z = 3$ spans southern Finland, where a greater diversity of southern bird species is present.
Finally, class $z = 4$ is concentrated along a small southwestern coastal region, which supports southern broadleaf forest species and potentially includes habitats for mire or peatland species.

The inferred classes have clear environmental interpretations.
Class $z = 3$ encompasses almost all urban habitat sampling locations, while almost all wetland locations belong to classes $z = 2$.
Additionally, locations in class $z = 1$ tend to have lower temperatures, those in $z = 3$ have higher temperatures, and those in $z \in \{2, 4\}$ experience moderate temperatures.
Posterior inference on $\Ab$ further reveals distinct relationships between deep latent classes and latent attributes.
Observations in class $z = 1$ tend to have zero entries in their latent attributes $\wb$, while those in class $z = 3$ predominantly have ones.
In contrast, observations in class $z = 2$ are more likely to have $\wb = (1, 0, 0, 1, 0)$, while observations in class $z = 4$ are more likely to have $\wb = (1, 0, 0, 0, 1)$ or $(1, 0, 1, 0, 1)$.

We visualize three representative attributes over sampling locations in Figure \ref{fig:app_w_map}.
Each attribute $w_j$ exhibits a spatial pattern, with distinct geographical separations for different $j$, particularly along the latitude axis. Our model offers an alternative perspective to the Region of Common Profiles (\texttt{RCP}) framework, which delineates spatial regions sharing similar environmental, ecological, or statistical properties \citep{foster2013modelling, foster2017ecological}.
In Figures \ref{fig:app_z_map} and \ref{fig:app_w_map}, the clusters of sampling locations formed by the latent cluster $z$ and attributes $\wb$ align with well-defined environmental and ecological regions in Finland, consistent with results from the \texttt{RCP} literature \citep{scherting2024inferring}.
Furthermore, the combination of latent attributes introduces a hierarchical \texttt{RCP} framework, progressively partitioning the sampling locations into 2, 4, 8, 16, 32 subregions, offering a more granular representation of environmental and ecological variation.

The binary matrix $\Gb$, which governs the dependencies between species presence/absence and latent attributes, provides additional insights into species associations.
Through the posterior inference of $\Gb$, we identify species related to each latent attribute, as visualized on the phylogenetic tree in Figure \ref{fig:phylo}.
Despite the phylogenetic tree not being used by our model, we observe a clear pattern that species in the same terminal branches tend to share dependencies on the same latent attributes.
Furthermore, inferred latent attributes correlate with species' biological traits.
All resident species are not dependent on $w_4$. Among migrant species, those dependent on $w_1$, $w_2$, or $w_5$ tend to be short distance migrants, whereas those not dependent on them are typically long distance migrants.
Species dependent on $w_3$ and species not dependent on $w_2$ generally have higher body weights.

\subsection{Prediction}\label{ssec:pred}

As there is a trade-off between model interpretability and expressivity, we expect to pay a penalty in predictive performance relative to hierarchical modeling of species communities \citep[\texttt{HMSC},][]{ovaskainen2020joint}, which is a Bayesian generalized linear latent variable model that 
outperformed 33 competing methods in predictive performance in a benchmark study \citep{norberg2019comprehensive}.
\texttt{HMSC} is useful for inferring covariate effects and species correlations, but does not produce insights into latent structure.
Thus, predictive performance comparable to \texttt{HMSC} indicates that \texttt{BLIP} achieves a favorable trade-off between interpretability and accuracy.

We compare \texttt{BLIP} predictions with those of \texttt{HMSC} for in-sample and out-of-sample observations.
To assess the impact of the binary attribute layer, we also compare  with a traditional latent class regression (\texttt{LCR}) model.
\texttt{HMSC} and \texttt{LCR} are implemented using the R packages \citep{tikhonov2020joint} and \citep{linzer2011polca}, respectively.
The \texttt{LCR} model is fitted with $d = 6$ latent classes, selected based on \texttt{AIC}.
Compared to our \texttt{BLIP}, the sampling locations within each class of \texttt{LCR} exhibit less distinct and interpretable geographical patterns.
We provide a visualization of the inferred \texttt{LCR} classes on the map of Finland and further discussions of the \texttt{LCR} model in the Supplementary Material Section \ref{supp_sec:app}.

\begin{table}[h!]
\centering
\begin{tabular}{cc||ccc}
\hline\hline
model & dataset & RMSE & AUC & co-occurrence score \\
\hline\hline
\multirow{2}{*}{\texttt{HMSC}} & in-sample & 0.386 & 0.864 & 0.671 \\
& out-of-sample & 0.377 & 0.874 & 0.689 \\
\hline
\multirow{2}{*}{\texttt{BLIP}} & in-sample & 0.398 & 0.848 & 0.641 \\
& out-of-sample & 0.390 & 0.855 & 0.649 \\
\hline
\multirow{2}{*}{\texttt{LCR}} & in-sample & 0.476 & 0.721 & 0.495 \\
& out-of-sample & 0.469 & 0.735 & 0.511 \\
\hline\hline
\end{tabular}
\caption{
Prediction metrics for \texttt{HMSC}, our model \texttt{BLIP}, and \texttt{LCR}, computed over all in-sample and out-of-sample observations in the Finnish birds dataset.
}
\label{tab:pred_metrics}
\end{table}

To evaluate the predictive accuracy for the presence or absence of species, we follow \citet{norberg2019comprehensive}, using the root mean squared error (RMSE) between each observed $y_i^{(n)}$ and its posterior predictive mean $\PP(y_i^{(n)} ~|~ \cdot)$.
We used the area under the curve (AUC) to assess the discrimination power of the posterior predictive mean.
Furthermore, to analyze species interactions at each location, we further assess the prediction of species co-occurrence, which is the product $y_i^{(n)} y_{i'}^{(n)}$ for each pair of species $i$ and $i'$.
To address class imbalance, we introduce the co-occurrence score, which measures predictive accuracy using $F_1$ score for co-occurring species pairs.
Specifically, it is defined as $\frac{2TP}{2TP + FP + FN}$, where $TP, FP, FN$ denote the number of true positives, false positives, and false negatives, respectively, for $y_i^{(n)} y_{i'}^{(n)}$.

The overall predictive metrics for \texttt{HMSC}, our model \texttt{BLIP}, and \texttt{LCR} are evaluated in-sample and out-of-sample and reported in Table \ref{tab:pred_metrics}.
Focusing on the out-of-sample performance using \texttt{HMSC} as a reference, our \texttt{BLIP} is within 2\%, 2\% and 4\% for the three metrics, while the \texttt{LCR} performance decreases to 11\%, 14\% and 17\%, respectively.
Additional results in the Supplementary Material Section \ref{supp_sec:app} further illustrate that \texttt{BLIP} is competitive with \texttt{HMSC} in predictive accuracy for all species and sampling locations.
These results suggest that \texttt{BLIP} achieves an excellent trade-off between model accuracy and interpretability.

\section{Discussion}\label{sec:discuss}

There are multiple directions for future work. First, we currently select the model structure using the \texttt{WAIC} criterion.
An alternative approach could involve Bayesian nonparametric methods, where a Dirichlet process prior is placed on the deep latent classes, and an Indian buffet process prior is used for the latent attributes.
Notably, the Bayes oracle clustering property established in this paper remains valid for the model with a Dirichlet process prior.
Second, we have extended our model to handle mixed discrete observations using multinomial logistic regression and Poisson regression.
Other generalized linear models could also be explored for modeling the conditional distribution of observed data given the latent attributes, potentially enhancing flexibility and robustness in different data settings.
Third, our current model employs a weakly informative prior that incorporates the meta covariates through a simple hierarchical structure.
Exploring more structured and informative priors may further improve model performance and parameter estimation.
Finally, while we have applied our model to joint species distribution modeling with multivariate binary observations, it would be of significant interest to explore its application to high-dimensional mixed discrete observations in domains such as genetics and sociology.

\medskip
\textbf{Code and data availability.}
Code and data are available at \url{https://github.com/BayesianModels01/Bayesian-Latent-Class-Regression-with-Interpretable-Binary-Profiles}.
The repository includes implementations of the methods used in the paper, scripts for the simulation studies and real-data applications, as well as the datasets and visualization code used in the numerical experiments.

\medskip
\textbf{Acknowledgments.}
This research was partially supported by the National Institutes of Health (R01ES035625), by the European Research Council under the European Union’s Horizon 2020 research and innovation programme (856506), by the National Science Foundation (NSF IIS-2426762, DMS-2210796), and by the Office of Naval Research (N00014-24-1-2626).

\bibliography{ref}
\bibliographystyle{apalike}

\clearpage
\newpage

\appendix

\tableofcontents

\textbf{Notations.}
We use bold capital letters (e.g. $\Gb, \Bb$) to denote matrices, bold letters (e.g. $\xb, \yb$) to denote vectors, and non-bold letters (e.g. $z$) to denote scalars.
Following \cite{kolda2009tensor}, we use $\Xb \circ \Yb$, $\Xb \otimes \Yb$, and $\Xb \odot \Yb$ to denote the Hadamard product, Kronecker product, and Khatri-Rao product, of two matrices $\Xb, \Yb$, and use $\Vec(\cdot)$ to denote the vectorization of tensors.

For integer-valued variables and integers $k_1, k_2 \in \ZZ$, we let $[k_1, k_2]$ denote the collection of integers $\{k_1, k_1 + 1, \ldots, k_2\}$ and further abbreviate $[k_2] := [1, k_2]$.
For a sequence of random variables $\{Z^{(n)}\}_{n = 1}^\infty$, we let $Z^{(N_1:N_2)}$ denote the collection $\{Z^{(n)}\}_{n = N_1}^{N_2}$ and let $Z^{(1:\infty)}$ denote the whole sequence.
The cardinality of a set $\cB$ is denoted by $\card(\cB)$.

The logit function is defined as $\logit(p) := \log\frac{p}{1 - p}$ for $p \in (0, 1)$.
The logistic function is defined as $\logistic(x) := \frac{\exp(x)}{1 + \exp(x)}$.
For any dimension $d$, we define the softmax function from $\RR^d$ to $\RR^d$ as $\softmax(x_1, \ldots, x_d) := \big( \frac{\exp(x_1)}{\sum_{h = 1}^d \exp(x_h)}, \ldots, \frac{\exp(x_d)}{\sum_{h = 1}^d \exp(x_h)} \big)$.

We let $\eb_i$ denote the canonical basis vector with the $i$th entry equal to one and all other entries equal to zero.
For a vector $\vb$, we let $\diag(\vb)$ denote the diagonal matrix with diagonal entries being entries of $\vb$.
For a matrix $\Ab$ and a subset of row indices $\cQ$, we let $\Ab_{\cQ}$ denote the submatrix of $\Ab$ formed by the rows in $\cQ$.

For a $p \times q$ matrix $\Gb$, we let $(\one ~ \Gb)$ denote the $p \times (q + 1)$ matrix obtained by left appending an all-ones column to $\Gb$.
For a $q \times q$ matrix $\Pb$, we let $\diag(1, \Pb)$ denote the $(q + 1) \times (q + 1)$ block diagonal matrix with its upper-left entry being 1 and its $q \times q$ lower-right block being $\Pb$.
For $\beta_0 \in \RR$ and a $q$-dimensional vector $\bbeta$, we let $(\beta_0, \bbeta)$ denote the $(q + 1)$-dimensional vector obtained by left appending an entry $\beta_0$ to $\bbeta$.

\section{More Intuition behind the Theorems}\label{supp_sec:intuit}

In this section, we provide additional intuition and proof overviews for the theorems presented in the main paper, which are deferred here due to space constraints.

\subsection{Theorem \ref{theo:strict} on Strict Identifiability}\label{supp_ssec:intuit_strict}

To further illustrate the intuition behind the conditions in Theorem \ref{theo:strict}, we briefly outline its proof technique.
For any covariate $\xb \in \supp(\cP_x)$, the conditional distribution of $\yb ~|~ \xb$ is represented as a vector in the $(2^p - 1)$-dimensional probability simplex, which can be transformed into a three-way tensor.
The conditional distributions $\yb ~|~ \wb$ and $\wb ~|~ \xb$ serve as its CP decomposition \citep{kolda2009tensor}.
The seminal work of \citet{kruskal1977three} established a sufficient condition for the uniqueness of three-way tensor CP decomposition, now known as Kruskal's condition.
In Theorem \ref{theo:strict}, both the partition in condition (i) and the blockwise structure in condition (ii) are designed such that Kruskal's condition holds in our settings.
These structural conditions enable the unique recovery of the conditional distributions $z ~|~ \xb$, $\wb ~|~ z$, and $\yb ~|~ \wb$ from $\yb ~|~ \xb$, up to permutations.
The full proof of Theorem \ref{theo:strict} is provided in Section \ref{supp_sec:iden}.

\subsection{Theorem \ref{theo:generic} on Generic Identifiability}\label{supp_ssec:intuit_generic}

The proof of Theorem \ref{theo:generic} relies on showing that the set of non-identifiable parameters in $\sS_2$ is contained within the zero set of a non-constant holomorphic function, a technique for establishing generic identifiability developed in \citet{zhou2024bayesian}.
Details of the proof is deferred to Section \ref{supp_sec:iden}.

\subsection{Theorem \ref{theo:post} on Posterior Consistency}\label{supp_ssec:intuit_post}

Theorem \ref{theo:post} establishes the posterior consistency in multiple scenarios.
In scenario (i), where $\sS$ is compact, any identifiable true parameter within the prior support $\supp(\pi)$ has a consistent posterior distribution.
Its proof follows Schwartz' theory of posterior consistency \citep{schwartz1965bayes}, which derives the posterior consistency of model likelihood within the Wasserstein space of probability measures.
The compactness of $\sS$ then allows us to convert this result into the posterior consistency of model parameters, as defined in Definition \ref{defi:post}.

In scenario (ii), where $\sS$ is a subset of $\sS_1$, any true parameter within $\supp(\pi)$ is automatically identifiable by Theorem \ref{theo:strict}.
As in scenario (i), this ensures a consistent posterior distribution, following a similar proof strategy.
However, unlike in scenario (i), the compactness of $\sS$ is no longer required.
This is due to the topological properties of the space $\sS_1$, which guarantee any $\sS \subset \sS_1$ to satisfy a weaker form of compactness sufficient for the proof.

In scenario (iii), where $\sS$ is a subset of $\sS_2$, the posterior consistency is established for almost every parameter in $\sS$ under the prior measure $\uppi$.
Its proof adapts Doob's theory of posterior consistency \citep{doob1949application} to our model by introducing the quotient parameter space induced by the equivalence class $\perm_{\sS}(\Ab^*, \Bb^*, \bGamma^*, \Gb^*)$.
We defer the detailed proof of Theorem \ref{theo:post} to Section \ref{supp_sec:post_cons}.

\subsection{Proposition \ref{prop:curse_example} on Bayesian Inference of Mixture Models}\label{supp_ssec:intuit_mix}

The proof of Proposition \ref{prop:curse_example} relies on an intermediate result, which states that for any partition of observations $[N]$ where at least one cluster contains multiple observations, there always exists a finer partition that divides this cluster into two and has a dominantly higher posterior probability as $p \to \infty$.
As a consequence, as $p \to \infty$, the posterior distribution increasingly favors partitions with smaller clusters, ultimately leading to each observation being assigned to its own individual cluster.
We defer the detailed proof of Proposition \ref{prop:curse_example} to Section \ref{supp_sec:mix}.

\section{Identifiability Theory}\label{supp_sec:iden}

In this section, we develop the identifiability theory for our model.
Section \ref{supp_ssec:prelim} provides an overview of our proof techniques, which are based on tensor decomposition, complex analysis, and geometric measure theory.
We present the proof of Theorem \ref{theo:strict} on strict identifiability in Section \ref{supp_ssec:proof_strict}, deferring the proofs of its auxiliary lemmas to Section \ref{supp_ssec:proof_lemmas_strict}.
Similarly, the proof of Theorem \ref{theo:generic} on generic identifiability is given in Section \ref{supp_ssec:proof_generic}, with auxiliary lemmas deferred to Section \ref{supp_ssec:proof_lemmas_generic}.

To clearly distinguish between conditional distributions and conditional probabilities in this section and the following Section \ref{supp_sec:post_cons}, we adopt the notation $\law$ for distributions and $\PP$ for probabilities or densities.
For instance, we use $\law(\yb ~|~ \wb, \Gb, \Bb)$ to denote the categorical distribution of $\yb ~|~ \wb, \Gb, \Bb$ over the space of binary vectors $\{0, 1\}^p$, while $\PP(\yb ~|~ \wb, \Gb, \Bb)$ represents the probability of a particular $\yb \in \{0, 1\}^p$ under this distribution.

\subsection{Preliminaries and Core Theorems}\label{supp_ssec:prelim}

For the proof of strict identifiability, we need Kruskal's Theorem \citep{kruskal1977three} on the uniqueness of three-way tensor CP decomposition.
We provide a brief overview and refer interested readers to \cite{kolda2009tensor, kruskal1977three} for more details.

For a matrix $\Ab$, we let its \textit{Kruskal rank} $\rank_K(\Ab)$ be the maximal number $r$ such that any $r$ columns of $\Ab$ are linearly independent.
For three matrices $\Ab \in \RR^{a \times d}, \Bb \in \RR^{b \times d}, \Cb \in \RR^{c \times d}$, we let $\cX := \llbracket \Ab, \Bb, \Cb \rrbracket \in \RR^{a \times b \times c}$ denote the three-way tensor defined entrywise by
$$
X_{i, j, k}
:=
\sum_{\ell = 1}^d A_{i, \ell} B_{i, \ell} C_{i, \ell}
,\quad
\forall i \in [a], j \in [b], k \in [c]
.
$$
Letting $\ab_\ell, \bbb_\ell, \cbb_\ell$ denote columns of $\Ab, \Bb, \Cb$, then we have $\llbracket \Ab, \Bb, \Cb \rrbracket = \sum_{\ell = 1}^d \llbracket \ab_\ell, \bbb_\ell, \cbb_\ell \rrbracket$, which is know as the CP decomposition of $\cX$ of rank $d$.
For a matrix $\Ab$, we let its \textit{Kruskal rank} $\rank_K(\Ab)$ be the maximal number $r$ such that any $r$ columns of $\Ab$ are linearly independent.
The following theorem from \citet{kruskal1977three} establishes, up to column rescaling and permutation, a sufficient condition for the uniqueness of CP decomposition of rank $d$, or equivalently, the uniqueness of $\Ab, \Bb, \Cb$ from $\llbracket \Ab, \Bb, \Cb \rrbracket$.

\begin{theorem}[Kruskal's Theorem]\label{theo:kruskal}
Let $\Ab \in \RR^{a \times d}$, $\Bb \in \RR^{b \times d}$, $\Cb \in \RR^{c \times d}$ be matrices satisfying the condition
$$
\rank_K(\Ab) + \rank_K(\Bb) + \rank_K(\Cb)
\ge
2d + 2
.
$$
Then for any matrices $\tilde\Ab \in \RR^{a \times d}$, $\tilde\Bb \in \RR^{b \times d}$, $\tilde\Cb \in \RR^{c \times d}$ such that
$$
\llbracket \Ab, \Bb, \Cb \rrbracket
=
\llbracket \tilde\Ab, \tilde\Bb, \tilde\Cb \rrbracket
,
$$
there exist diagonal matrices $\Sb_A, \Sb_B, \Sb_C \in \RR^{d \times d}$ satisfying $\Sb_A \Sb_B \Sb_C = \Ib$ and a permutation matrix $\Qb \in \RR^{d \times d}$ such that
$$
\tilde\Ab
=
\Ab \Sb_A \Qb
,\quad
\tilde\Bb
=
\Bb \Sb_B \Qb
,\quad
\tilde\Cb
=
\Cb \Sb_C \Qb
.
$$
\end{theorem}

For the proof of generic identifiability, we need some results based on complex analysis. The following theorems were originally introduced and proven in \citet{zhou2024bayesian}.
We will only present the necessary background for these theorems and refer interested readers to Sections B.4 and B.5 in \citet{zhou2024bayesian} for more details and proofs.

A function of several complex variables is \textit{holomorphic} if it is complex-differentiable within its domain, or equivalently, if it satisfies Cauchy-Riemann equations in each variable within its domain.
The following theorem states that a holomorphic function of $n$ complex variables with zero set having positive Lebesgue measure in $\RR^n$ is constant zero.

\begin{theorem}[Theorem B.6 in \citet{zhou2024bayesian}]\label{theo:holo_zero}
Let $\Omega_1, \ldots, \Omega_n \subset \CC$ be connected regions and $f(z_1, \ldots, z_n)$ be a holomorphic function in the region $\Omega = \Omega_1 \times \cdots \Omega_n \subset \CC^n$.
If its zero set in $\RR^n \cap \Omega$ has Lebesgue measure
$$
\uplambda\left( \left\{
x \in \RR^n \cap \Omega:~ f(x) = 0
\right\} \right)
>
0
,
$$
then $f \equiv 0$ in $\Omega$.
\end{theorem}

A function on $\RR^n$ is \textit{analytic} if it has convergent Taylor series at every point in its domain, or equivalently, if it is smooth and can be locally extended to a holomorphic function of $n$ complex variables.
Recall that $\uplambda$ denotes the Lebesgue measure and $\upmu$ denotes the counting measure.
The following theorem is a technical result based on the area formula from geometric measure theory and is useful in dealing with equivalence classes of parameters that have finite cardinalities.

\begin{theorem}[Theorem B.7 in \citet{zhou2024bayesian}]\label{theo:geom}
Let $f, g:~ \RR^n \to \RR^m$ $(m \ge n)$ be analytic functions.
Let $B \subset \RR^n$ be a set satisfying
$$
\sup_{y \in \RR^m} \upmu\left(
f^{-1}(y) \cap B
\right)
<
\infty
.
$$
Then for any set $A \subset \RR^n$ with measure $\uplambda(A) = 0$, we have
$$
\uplambda\left(
f^{-1}(g(A)) \cap B
\right)
=
0
.
$$
\end{theorem}

\subsection{Proof of Theorem \ref{theo:strict}}\label{supp_ssec:proof_strict}

Recall that $\cP_x$ denotes the data generating distribution of the covariates.
Unlike discrete latent variable models without covariates (e.g. \citet{gu2023bayesian, zhou2024bayesian}), the presence of covariates $\xb$ requires the proof to hold for $\cP_x$-almost every $\xb$.
This is established by proving the result for all $\xb \in \supp(\cP_x)$.
In the following proof, we consider an arbitrary $\xb \in \supp(\cP_x)$ and note that the proof remains valid regardless of the specific choice of $\xb$ within this support.

Define $\sY: \{0, 1\}^p \to [2^p]$ as an arbitrary bijective mapping from the space of $p$-dimensional binary vectors $\{0, 1\}^p$ to the integers $1, \ldots, 2^p$.
Similarly, let $\sW: \{0, 1\}^q \to [2^q]$ be an arbitrary bijective mapping from the space of $q$-dimensional binary vectors $\{0, 1\}^q$ to the integers $1, \ldots, 2^q$.

For any given parameters $\Gb$ and $\Bb$, we characterize the conditional distribution of $\yb ~|~ \wb, \Gb, \Bb$ using the $2^p \times 2^q$ probability matrix $\bLambda$, defined entrywise as:
\begin{equation}\label{eq:Lambda_st}
\Lambda_{s, t}
:=
\PP(\yb = \sY^{-1}(s) ~|~ \wb = \sW^{-1}(t), \Gb, \Bb)
,\quad
\forall s \in [2^p], t \in [2^q]
.
\end{equation}
Similarly, for each $i \in [p]$, the conditional distribution of entry $y_i ~|~ \wb, \gb_i, \beta_{i, 0}, \bbeta_i$ is characterized by the $2 \times 2^q$ probability matrix $\bLambda_{(i)}$, defined entrywise as:
\begin{equation}\label{eq:Lambda_ist}
\Lambda_{(i), s, t}
:=
\PP(y_i = 1_{s = 1} ~|~ \wb = \sW^{-1}(t), \gb_i, \beta_{i, 0}, \bbeta_i)
,\quad
\forall s \in [2], t \in [2^q]
.
\end{equation}
Since the entries $y_i$ are conditionally independent given $\wb$, it follows that
$$
\PP(\yb ~|~ \wb, \Gb, \Bb)
=
\prod_{i = 1}^p \PP(y_i ~|~ \wb, \gb_i, \beta_{i, 0}, \bbeta_i)
.
$$
Consequently, under appropriate row permutations, the Khatri-Rao product $\odot_{i = 1}^p \bLambda_{(i)}$ coincides with $\bLambda$.

Next, let $\bnu \in \RR^{2^q}$ be the probability vector characterizing the conditional distribution of $\wb ~|~ \xb, \Ab, \bGamma$, defined entrywise as:
\begin{equation}\label{eq:nu_t}
\nu_t
:=
\PP(\wb = \sW^{-1}(t) ~|~ \xb, \Ab, \bGamma)
,\quad
\forall t \in [2^q]
.
\end{equation}
This allows us to characterize the conditional distribution of $\yb ~|~ \xb, \Ab, \Bb, \bGamma, \Gb$ through the probability vector $\bLambda \bnu$, i.e.
\begin{equation}\label{eq:Lambdanu_s}
(\bLambda \bnu)_s
=
\PP(\yb = \sY^{-1}(s) ~|~ \xb, \Ab, \Bb, \bGamma, \Gb)
,\quad
\forall s \in [2^p]
.
\end{equation}

The following lemma establishes that if $\Gb, \Bb$ satisfy conditions (ii) and (iii) of Theorem \ref{theo:strict}, then the probability matrix $\bLambda$ and the probability vector $\bnu$ can be uniquely identified, up to column permutations, from their product $\bLambda \bnu$.
Its proof is provided in Section \ref{supp_ssec:proof_lemmas_strict}.

\begin{lemma}\label{lemm:y_to_w}
Let $\bLambda$ be the probability matrix defined in \eqref{eq:Lambda_st}, and let $\bnu$ be the probability vector defined in \eqref{eq:nu_t}.
Suppose the parameters $\Gb$ and $\Bb$ satisfy the following conditions:
\begin{enumerate}[(i)]
\item
The matrix $\Gb$ contains three distinct identity blocks $\Ib_q$, i.e.
$$
\Pb \Gb
=
\left(\begin{matrix}
\Ib_q \\ \Ib_q \\ \Ib_q \\ \Gb_{\cI}
\end{matrix}\right)
$$
for some $p \times p$ permutation matrix $\Pb$ and subset $\cI \subset [p]$ with $\card(\cI) = p - 3q$;
\item
$\forall i \in [p], j \in [q]$, if $g_{i, j} = 1$ then $\beta_{i, j} \ne 0$.
\end{enumerate}
Then, for any other $2^p \times 2^q$ matrix $\tilde\bLambda$ and $2^q$-dimensional vector $\tilde\bnu$ satisfying
$$
\bLambda \bnu
=
\tilde\bLambda \tilde \bnu
,\quad
\tilde \bLambda^\top \one
=
\one
,\quad
\one^\top \tilde\bnu
=
1
,
$$
there exists a $2^q \times 2^q$ permutation matrix $\Qb$ such that
$$
\tilde\bLambda
=
\bLambda \Qb
,\quad
\tilde\bnu
=
\Qb^\top \bnu
.
$$
\end{lemma}

Similar to the definition of $\bLambda$ in \eqref{eq:Lambda_st}, for any given parameter $\Ab$, we characterize the conditional distribution of $\wb ~|~ z, \Ab$ using the $2^q \times d$ probability matrix $\bSigma$, defined entrywise as:
\begin{equation}\label{eq:Sigma_th}
\Sigma_{t, h}
:=
\PP(\wb = \sW^{-1}(t) ~|~ z = h, \Ab)
,\quad
\forall t \in [2^q], h \in [d]
.
\end{equation}
For each $j \in [q]$, analogous to \eqref{eq:Lambda_ist}, we define the $2 \times d$ probability matrix $\bSigma_{(j)}$ to characterize the conditional distribution of entry $w_j ~|~ z, \Ab$:
\begin{equation}\label{eq:Sigma_jth}
\Sigma_{(j), 1, h}
:=
\PP(w_j = 1 ~|~ z = h, \Ab)
,\quad
\Sigma_{(j), 2, h}
:=
\PP(w_j = 0 ~|~ z = h, \Ab)
,\quad
\forall h \in [d]
.
\end{equation}
Since the entries $w_j$ are conditionally independent given $z$, it follows that $\odot_{j = 1}^q \bSigma_{(j)}$ coincides with $\bSigma$ up to row permutations.

Similar to $\bnu$ in \eqref{eq:nu_t}, we define $\bmu \in \RR^d$ as the probability vector characterizing the conditional distribution of $\zb ~|~ \xb, \bGamma$, defined entrywise as:
\begin{equation}\label{eq:mu_h}
\mu_h
:=
\PP(z = h ~|~ \xb, \bGamma)
,\quad
h \in [d]
.
\end{equation}
Since
$$
\PP(\wb ~|~ \xb, \Ab, \bGamma)
=
\sum_{z \in [d]} \PP(\wb ~|~ z, \Ab) \PP(z ~|~ \xb, \bGamma)
,\quad
\forall \wb \in \{0, 1\}^q
,
$$
we have $\bnu = \bSigma \bmu$.

The following lemma establishes that if $\Ab$ satisfies condition (i) of Theorem \ref{theo:strict}, then the probability matrix $\bSigma$ and probability vector $\bmu$ can be uniquely identified, up to column permutations, from their product $\bSigma \bmu$.
We defer its proof to Section \ref{supp_ssec:proof_lemmas_strict}.

\begin{lemma}\label{lemm:w_to_z}
Let $\bSigma$ be the probability matrix defined in \eqref{eq:Sigma_th} and $\bmu$ be the probability vector defined in \eqref{eq:mu_h}.
Suppose there exists a partition of $[q]$ as $\cQ_1 \cup \cQ_2 \cup \cQ_3$ with $\card(\cQ_1), \card(\cQ_2) \ge \log_2 d$, such that the two sets of vectors $\bigcirc_{j \in \cQ_1} \{\one, \balpha_j\}$, $\bigcirc_{j \in \cQ_2} \{\one, \balpha_j\}$ have full ranks and the submatrix $\Ab_{\cQ_3}$ has distinct columns.
Then, for any other $2^q \times d$ matrix $\tilde\bSigma$ and $d$-dimensional vector $\tilde\bmu$ satisfying
$$
\bSigma \bmu
=
\tilde\bSigma \tilde\bmu
,\quad
\tilde\bSigma^\top \one
=
\one
,\quad
\one^\top \tilde\bmu
=
1
,
$$
there exists a $d \times d$ permutation matrix $\Rb$ such that
$$
\tilde\bSigma
=
\bSigma \Rb
,\quad
\tilde\bmu
=
\Rb^\top \bmu
.
$$
\end{lemma}

Lemmas \ref{lemm:y_to_w} and \ref{lemm:w_to_z} collectively establish that the conditional distributions $\law(\yb ~|~ \wb, \Gb, \Bb)$, $\law(\wb ~|~ z, \Ab)$, and $\law(z ~|~ \xb, \bGamma)$ can be uniquely recovered from the marginal distribution $\law(\yb ~|~ \xb, \Ab, \Bb, \bGamma, \Gb)$, up to permutations of $\wb$ over $\{0, 1\}^q$ and permutations of $z$ over $[d]$.
The following lemma further demonstrates that the parameters $\Ab, \Bb, \bGamma, \Gb$ can be recovered from these conditional distributions, with proof included in Section \ref{supp_ssec:proof_lemmas_strict}.

\begin{lemma}\label{lemm:law_to_params}
Suppose the data generating distribution $\cP_x$ of $\xb$ has full rank.
Let $(\Ab, \Bb, \bGamma, \Gb)$ and $(\tilde\Ab, \tilde\Bb, \tilde\bGamma, \tilde\Gb)$ be two sets of model parameters satisfying the following condition:
$$
\forall i \in [p], j \in [q],
\text{ if } g_{i, j} = 1,
\text{ then } \beta_{i, j} \ne 0,
\text{ and if } \tilde{g}_{i, j} = 1,
\text{ then } \tilde\beta_{i, j} \ne 0
.
$$
Then, the following statements hold:
\begin{enumerate}[(i)]
\item
If the conditional distributions satisfy
$$
\law(\yb ~|~ \wb, \Gb, \Bb)
=
\law(\yb ~|~ \wb, \tilde\Gb, \tilde\Bb)
,\quad
\forall \wb \in \{0, 1\}^q
,
$$
then $\Gb = \tilde\Gb$ and $(1 ~ \Gb) \circ \Bb = (1 ~ \tilde\Gb) \circ \tilde\Bb$.
\item
If the conditional distributions satisfy
$$
\law(\wb ~|~ z, \Ab)
=
\law(\wb ~|~ z, \tilde\Ab)
,\quad
\forall z \in [d]
,
$$
then $\Ab = \tilde\Ab$.
\item
If the conditional distributions satisfy
$$
\law(z ~|~ \xb, \bGamma)
=
\law(z ~|~ \xb, \tilde\bGamma)
,\quad
\cP_x\text{-a.e.}
$$
then $\bGamma = \tilde\bGamma$.
\end{enumerate}
\end{lemma}

By combining Lemmas \ref{lemm:y_to_w}, \ref{lemm:w_to_z}, \ref{lemm:law_to_params}, we now present the formal proof of Theorem \ref{theo:strict}.
Lemma \ref{lemm:y_to_w} establishes the uniqueness of $\bLambda$ and $\bnu$ up to permutations of $\wb$ over the full space $\{0, 1\}^q$, resulting in $(2^q)!$ possible permutations.
However, Theorem \ref{theo:strict} only allows for permutations of the $q$ dimensions of $\wb$, which corresponds to $q!$ possible permutations.
The following proof bridges this gap and ensures the desired level of uniqueness.

\begin{proof}[Proof of Theorem \ref{theo:strict}]
Recall that $\sS_1$ denotes the space of parameters $(\Ab, \Bb, \bGamma, \Gb)$ satisfying conditions (i), (ii), and (iii) in Theorem \ref{theo:strict}.
Since permutations of the entries in $\wb$ and the classes of $z$ do not affect the conditional distribution of $\yb ~|~ \xb, \Ab, \Bb, \bGamma, \Gb$, it follows that
$$
\perm_{\sS_1}(\Ab, \Bb, \bGamma, \Gb)
\subseteq
\marg_{\sS_1}(\Ab, \Bb, \bGamma, \Gb)
.
$$
Thus, to establish the desired identifiability result, it suffices to prove the reverse inclusion:
$$
\perm_{\sS_1}(\Ab, \Bb, \bGamma, \Gb)
\supseteq
\marg_{\sS_1}(\Ab, \Bb, \bGamma, \Gb)
.
$$

Let $(\tilde\Ab, \tilde\Bb, \tilde\bGamma, \tilde\Gb)$ be an arbitrary set of model parameters belonging to $\marg_{\sS_1}(\Ab, \Bb, \bGamma, \Gb)$.
By the definition of $\marg$ in \eqref{eq:marg}, we have
\begin{equation}\label{eq:tilde_from_marg_laws}
\law(\yb ~|~ \xb, \Ab, \Bb, \bGamma, \Gb)
=
\law(\yb ~|~ \xb, \tilde\Ab, \tilde\Bb, \tilde\bGamma, \tilde\Gb)
.
\end{equation}
Recall from \eqref{eq:Lambda_st} and \eqref{eq:nu_t} that the probability matrix $\bLambda$ and probability vector $\bnu$ characterize the conditional distributions $\law(\yb ~|~ \wb, \Gb, \Bb)$ and $\law(\wb ~|~ \xb, \Ab, \bGamma)$.
Similarly, we define the probability matrix $\tilde\bLambda$ and probability vector $\tilde\bnu$ to characterize the conditional distributions $\law(\yb ~|~ \wb, \tilde\Gb, \tilde\Bb)$ and $\law(\wb ~|~ \xb, \tilde\Ab, \tilde\bGamma)$.
From \eqref{eq:tilde_from_marg_laws}, it follows that $\bLambda \bnu = \tilde\bLambda \tilde\bnu$.
Applying Lemma \ref{lemm:y_to_w}, we conclude that there exists a $2^q \times 2^q$ permutation matrix $\Qb$ such that
$$
\tilde\bLambda
=
\bLambda \Qb
,\quad
\tilde\bnu
=
\Qb^\top \bnu
.
$$

We let $\sQ$ denote the permutation map on $\{0, 1\}^q$ corresponding to $\Qb$, such that
$$
\law(\yb ~|~ \wb, \tilde\Gb, \tilde\Bb)
=
\law(\yb ~|~ \sQ(\wb), \Gb, \Bb)
,\quad
\law(\wb ~|~ \xb, \tilde\Ab, \tilde\bGamma)
=
\law(\sQ(\wb) ~|~ \xb, \Ab, \bGamma)
.
$$
This further implies that for each $i \in [p]$,
$$
\law(y_i ~|~ \wb, \tilde\gb_i, \tilde\beta_{i, 0}, \tilde\bbeta_i)
=
\law(y_i ~|~ \sQ(\wb), \gb_i, \beta_{i, 0}, \bbeta_i)
.
$$
For each $i \in [p]$, we define the sets
\begin{align*}
\cC_i
&:=
\Big\{
\PP(y_i = 1 ~|~ \wb, \gb_i, \beta_{i, 0}, \bbeta_i)
:~
\wb \in \{0, 1\}^q
\Big\}
\\&=
\left\{
\frac{\exp(\beta_{i, 0} + \sum_{j = 1}^q w_j g_{i, j} \beta_{i, j})}{1 + \exp(\beta_{i, 0} + \sum_{j = 1}^q w_j g_{i, j} \beta_{i, j})}
:~
\wb \in \{0, 1\}^q
\right\}
,
\end{align*}
and
\begin{align*}
\tilde\cC_i
&:=
\Big\{
\PP(y_i = 1 ~|~ \sQ(\wb), \tilde\gb_i, \tilde\beta_{i, 0}, \tilde\bbeta_i)
:~
\wb \in \{0, 1\}^q
\Big\}
\\&=
\Big\{
\PP(y_i = 1 ~|~ \wb, \tilde\gb_i, \tilde\beta_{i, 0}, \tilde\bbeta_i)
:~
\wb \in \{0, 1\}^q
\Big\}
\\&=
\left\{
\frac{\exp(\tilde\beta_{i, 0} + \sum_{j = 1}^q w_j \tilde{g}_{i, j} \tilde\beta_{i, j})}{1 + \exp(\tilde\beta_{i, 0} + \sum_{j = 1}^q w_j \tilde{g}_{i, j} \tilde\beta_{i, j})}
:~
\wb \in \{0, 1\}^q
\right\}
,
\end{align*}
then $\cC_i = \tilde\cC_i$.

Recall that $\Gb$ and $\Bb$ satisfy the condition that $\forall i \in [p], j \in [q]$, if $g_{i, j} = 1$ then $\beta_{i, j} \ne 0$.
Consequently, if $\gb_i = \zero$, then the set
$$
\cC_i
=
\left\{
\frac{\exp(\beta_{i, 0})}{1 + \exp(\beta_{i, 0})}
\right\}
$$
has cardinality $\card(\cC_i) = 1$.
If $\gb_i = \eb_j$ for some $j \in [q]$, then the set
$$
\cC_i
=
\left\{
\frac{\exp(\beta_{i, 0})}{1 + \exp(\beta_{i, 0})}
,~
\frac{\exp(\beta_{i, 0} + \beta_{i, j})}{1 + \exp(\beta_{i, 0} + \beta_{i, j})}
\right\}
$$
has cardinality $\card(\cC_i) = 2$.
If $\one^\top \gb_i \ge 2$, then there exist distinct $j, j' \in [q]$ such that $g_{i, j} = g_{i, j'} = 1$, and the set
$$
\cC_i
\supset
\left\{
\frac{\exp(\beta_{i, 0})}{1 + \exp(\beta_{i, 0})}
,~
\frac{\exp(\beta_{i, 0} + \beta_{i, j})}{1 + \exp(\beta_{i, 0} + \beta_{i, j})}
,~
\frac{\exp(\beta_{i, 0} + \beta_{i, j'})}{1 + \exp(\beta_{i, 0} + \beta_{i, j'})}
,~
\frac{\exp(\beta_{i, 0} + \beta_{i, j} + \beta_{i, j'})}{1 + \exp(\beta_{i, 0} + \beta_{i, j} + \beta_{i, j'})}
\right\}
.
$$
We notice that either $\beta_{i, 0}, \beta_{i, 0} + \beta_{i, j}, \beta_{i, 0} + \beta_{i, j'}$ are three distinct values, or $\beta_{i, 0}, \beta_{i, 0} + \beta_{i, j}, \beta_{i, 0} + \beta_{i, j} + \beta_{i, j'}$ are three distinct values, therefore suggesting that $\card(\cC_i) \ge 3$ if $\one^\top \gb_i \ge 2$.

As required by the parameter space $\sS_1$, $\tilde\Gb$ and $\tilde\Bb$ also satisfy that $\forall i \in [p], j \in [q]$, if $\tilde{g}_{i, j} = 1$ then $\tilde\beta_{i, j} \ne 0$.
Similarly, we can derive that if $\one^\top \gb_i = 0$, then $\card(\tilde\cC_i) = 1$; if $\one^\top \gb_i = 1$, then $\card(\tilde\cC_i) = 2$; and if $\one^\top \gb_i \ge 2$, then $\card(\tilde\cC_i) \ge 3$.
Consequently, for each $i \in [p]$, by $\cC_i = \tilde\cC_i$ we have
\begin{equation}\label{eq:gap_step1}
\gb_i = \zero
~\iff~
\tilde\gb_i = \zero
,\quad
\one^\top \gb_i = 1
~\iff~
\one^\top \tilde\gb_i = 1
.
\end{equation}

Let $i, i' \in [p]$ be any two indices with $\one^\top \gb_i = \one^\top \gb_{i'} = 1$.
Since $\card(\cC_i) = 2$, we can partition the $2^q$ columns of $\blambda_{(i)}$ into two groups $\cW_{i, 1}$ and $\cW_{i, 2}$, both having cardinality $2^{q - 1}$, such that $\PP(y_i = 1 ~|~ \wb = \sW^{-1}(j), \gb_i, \beta_{i, 0}, \bbeta_i)$ agrees within $\cW_{i, 1}$ and $\cW_{i, 2}$.
Similarly, we also obtain the partition of $[2^q]$ into $\cW_{i', 1} \cup \cW_{i', 2}$ by grouping on the values of $\PP(y_{i'} = 1 ~|~ \wb, \gb_i, \bbeta_{i, 0}, \bbeta_i)$.
We notice that $\gb_i = \gb_{i'}$ if and only if the partition $\cW_{i, 1} \cup \cW_{i, 2}$ is the same as the partition $\cW_{i', 1} \cup \cW_{i', 2}$.
Analogous results hold for $\tilde\Gb$.
Therefore, for any $i, i' \in [p]$ with $\one^\top \gb_i = \one^\top \gb_{i'} = 1$, we have
\begin{equation}\label{eq:gap_step2}
\gb_i = \gb_{i'}
\quad\iff\quad
\tilde\gb_i = \tilde\gb_{i'}
.
\end{equation}

For any two vectors $\vb, \vb' \in \RR^q$, we denote $\vb \succeq \vb'$ if entrywise we have $v_j \ge v_j'$ for all $j \in [q]$.
Now let $i \in [p]$ be any index with $\one^\top \gb_i = 1$ and $i' \in [p]$ any index with $\one^\top \gb_{i'} \ge 2$.
We notice that $\gb_{i'} \succeq \gb_i$ if and only if there exists a bijective mapping $\sW$ from $\cW_{i, 1}$ to $\cW_{i, 2}$ such that
$$
\logit \Lambda_{(i'), 1, j} - \logit \Lambda_{(i'), 1, \sW(j)}
=
\logit \Lambda_{(i), 1, j} - \logit \Lambda_{(i), 1, \sW(j)}
,\quad
\forall j \in \cW_{i, 1}.
$$
Analogous results again hold for $\tilde\Gb$.
This suggests that for any $i, i' \in [p]$ with $\one^\top \gb_i = 1$ and $\one^\top \gb_{i'} \ge 2$, we have
\begin{equation}\label{eq:gap_step3}
\gb_{i'} \succeq \gb_i
\quad\iff\quad
\tilde\gb_{i'} \succeq \tilde\gb_i
.
\end{equation}

Since $\Gb, \tilde\Gb$ contain identity blocks, by combining \eqref{eq:gap_step1}, \eqref{eq:gap_step2}, and \eqref{eq:gap_step3}, there exists a permutation map $\sP_w$ on $[q]$, denoted by the $q \times q$ permutation matrix $\Pb_w$, such that $\tilde\Gb = \Gb \Pb_w^\top$.
We further notice that for $i \in [p]$ with $\one^\top \gb_i = 1$, we have $\wb \succeq \gb_i$ if and only if
$$
\PP(y_i = 1 ~|~ \wb, \gb_i, \beta_{i, 0}, \bbeta_i)
-
\PP(y_i = 0 ~|~ \wb, \gb_i, \beta_{i, 0}, \bbeta_i)
\ne
0
,
$$
and $\wb \succeq \tilde\gb_i$ if and only if
$$
\PP(y_i = 1 ~|~ \wb, \tilde\gb_i, \tilde\beta_{i, 0}, \tilde\bbeta_i)
-
\PP(y_i = 0 ~|~ \wb, \tilde\gb_i, \tilde\beta_{i, 0}, \tilde\bbeta_i)
\ne
0
.
$$
This suggests
\begin{equation}\label{eq:Q_to_P}
\sQ(\wb)
=
\Pb_w^\top \wb
,\quad
\forall \wb \in \{0, 1\}^q
,
\end{equation}
which gives
$$
\law\left(
\yb
~\big|~
\wb, \tilde\Gb, \tilde\Bb
\right)
=
\law\left(
\yb
~\big|~
\Pb_w^\top \wb, \Gb, \Bb
\right)
=
\law\left(
\yb
~\big|~
\wb, \Gb\Pb_w^\top, \Bb \diag(1, \Pb_w)^\top
\right)
.
$$
Along with Lemma \ref{lemm:law_to_params}, we obtain
\begin{equation}\label{eq:G_and_1GB_perm}
\tilde\Gb
=
\Gb \Pb_w^\top
,\quad
(\one ~ \tilde\Gb) \circ \tilde\Bb
=
((1 ~ \Gb) \circ \Bb) \diag(1, \Pb_w)
.
\end{equation}
This bridges the gap from the permutation $\sQ$ over $\{0, 1\}^q$ to the permutation matrix $\Pb_w$ on the $q$ entries of $\wb$.

\eqref{eq:Q_to_P} further suggests
$$
\law(\wb ~|~ \xb, \tilde\Ab, \tilde\bGamma)
=
\law(\Pb_w^\top \wb ~|~ \xb, \Ab, \bGamma)
=
\law(\wb ~|~ \xb, \Pb_w \Ab, \bGamma)
.
$$
Similar to \eqref{eq:Sigma_th} and \eqref{eq:mu_h}, we let $\bSigma'$ and $\bmu$ characterize the distributions $\law(\wb ~|~ z, \Pb_w \Ab)$ and $\law(\wb ~|~ \xb, \bGamma)$, and let $\tilde\bSigma$ and $\tilde\bmu$ characterize the distributions $\law(\wb ~|~ z, \tilde\Ab)$ and $\law(\wb ~|~ \xb, \tilde\bGamma)$.
Applying Lemma \ref{lemm:w_to_z} suggests the existence of some $d \times d$ permutation matrix $\Pb_z$ such that
$$
\tilde\bSigma
=
\bSigma' \Pb_z
,\quad
\tilde\bmu
=
\Pb_z^\top \bmu
,
$$
which equivalently gives
$$
\law(\wb ~|~ z, \tilde\Ab)
=
\law(\wb ~|~ z, \Pb_w \Ab \Pb_z^\top)
,\quad
\law(z ~|~ \xb, \tilde\bGamma)
=
\law(z ~|~ \xb, \Pb_z\bGamma)
.
$$
Applying Lemma \ref{lemm:law_to_params}, we further obtain
\begin{equation}\label{eq:A_and_Gamma_perm}
\tilde\Ab
=
\Pb_w \Ab \Pb_z^\top
,\quad
\tilde\bGamma
=
\Pb_z \bGamma
.
\end{equation}

Recall the definition of $\perm$ from \eqref{eq:perm}.
By combining \eqref{eq:G_and_1GB_perm} and \eqref{eq:A_and_Gamma_perm}, it follows that $(\tilde\Ab, \tilde\Bb, \tilde\bGamma, \tilde\Gb) \in \perm_{\sS_1}(\Ab, \Bb, \bGamma, \Gb)$.
This shows that $\perm_{\sS_1}(\Ab, \Bb, \bGamma, \Gb) \supseteq \marg_{\sS_1}(\Ab, \Bb, \bGamma, \Gb)$, which completes the proof.
\end{proof}

\subsection{Proofs of Auxiliary Lemmas for Theorem \ref{theo:strict}}\label{supp_ssec:proof_lemmas_strict}

We begin by presenting the proof of Lemma \ref{lemm:y_to_w}, which leverages Kruskal's Theorem \ref{theo:kruskal} on the uniqueness of three-way tensor CP decomposition.

\begin{proof}[Proof of Lemma \ref{lemm:y_to_w}]
Without loss of generality, we assume that the $p \times p$ permutation $\Pb$ in condition (i) of Lemma \ref{lemm:y_to_w} is the identity matrix $\Ib$, with the row index subset $\cI = [3q + 1, p]$.

We recall the definition of $\bLambda_{[i]}$ from \eqref{eq:Lambda_ist}.
Based on the partition $[p] = [q] \cup [q + 1, 2q] \cup [2q + 1, 3q] \cup [3q + 1, p]$ corresponding to the blockwise structure of $\Gb$, we define the probability matrices
\begin{equation}\label{eq:Lambda[1234]}
\bLambda_{[1]}
:=
\odot_{i \in [q]} \bLambda_{(i)}
,\quad
\bLambda_{[2]}
:=
\odot_{i \in [q + 1, 2q]} \bLambda_{(i)}
,\quad
\bLambda_{[3]}
:=
\odot_{i \in [2q + 1, 3q]} \bLambda_{(i)}
,\quad
\bLambda_{[4]}
:=
\odot_{i \in [3q + 1, p]} \bLambda_{(i)}
.
\end{equation}
The Khatri-Rao product $\bLambda_{[1]} \odot \bLambda_{[2]} \odot \bLambda_{[3]} \odot \bLambda_{[4]}$ then coincides with $\bLambda$, up to row permutations.
Notably, each row of $\bLambda_{(i)}$ can be obtained by summing specific rows of $\bLambda$, formally expressed as
$$
\Lambda_{(i), s, t}
=
\sum_{r \in \sY(\{\yb \in \{0, 1\}^p:~ y_i = 1_{s = 1}\})} \Lambda_{r, t}
,\quad
\forall s \in [2], t \in [2^q]
.
$$

Starting from the probability matrix $\tilde\bLambda$, we similarly define $\tilde\bLambda_{(i)}$ for each $i \in [p]$ entrywise as
$$
\tilde\Lambda_{(i), s, t}
:=
\sum_{r \in \sY(\{\yb \in \{0, 1\}^p:~ y_i = 1_{s = 1}\})} \tilde\Lambda_{r, t}
,\quad
\forall s \in [2], t \in [2^q]
$$
and introduce the probability matrices
$$
\tilde\bLambda_{[1]}
:=
\odot_{i \in [q]} \tilde\bLambda_{(i)}
,\quad
\tilde\bLambda_{[2]}
:=
\odot_{i \in [q + 1, 2q]} \tilde\bLambda_{(i)}
,\quad
\tilde\bLambda_{[3]}
:=
\odot_{i \in [2q + 1, 3q]} \tilde\bLambda_{(i)}
,\quad
\tilde\bLambda_{[4]}
:=
\odot_{i \in [3q + 1, p]} \tilde\bLambda_{(i)}
.
$$
Since $\bLambda \bnu = \tilde\bLambda \tilde\bnu$, we equivalently rewrite this as the equality of the three-way tensors
$$
\llbracket \bLambda_{[1]}, \bLambda_{[2]}, (\bLambda_{[3]} \odot \bLambda_{[4]}) \diag(\bnu) \rrbracket
=
\llbracket \tilde\bLambda_{[1]}, \tilde\bLambda_{[2]}, (\tilde\bLambda_{[3]} \odot \tilde\bLambda_{[4]}) \diag(\tilde\bnu) \rrbracket
.
$$
This reformulation transforms the problem of uniqueness of $\bLambda$ and $\bnu$ from their product $\bLambda \bnu$ into the uniqueness of a three-way tensor CP decomposition.

We first analyze the Kruskal ranks of $\bLambda_{[1]}$.
To do so, we define the $2^q \times 2^q$ matrix
\begin{align*}
\bLambda_{[1]}'
&:=
\bigodot_{i \in [q]} \left(\begin{matrix}
\Lambda_{(i), 1, 1} - \frac{\exp(\beta_{i, 0})}{1 + \exp(\beta_{i, 0})} & \cdots & \Lambda_{(i), 1, 2^q} - \frac{\exp(\beta_{i, 0})}{1 + \exp(\beta_{i, 0})} \\
1 & \cdots & 1
\end{matrix}\right)
\\&=
\bigodot_{i \in [q]} \left(
\left(\begin{matrix}
1 & -\frac{\exp(\beta_{i, 0})}{1 + \exp(\beta_{i, 0})} \\
0 & 1
\end{matrix}\right)
\cdot
\left(\begin{matrix}
1 & 0 \\
1 & 1
\end{matrix}\right)
\cdot
\bLambda_{(i)}
\right)
\\&=
\underbrace{
\bigotimes_{i \in [q]} \left(\begin{matrix}
1 & -\frac{\exp(\beta_{i, 0})}{1 + \exp(\beta_{i, 0})} \\
0 & 1
\end{matrix}\right)
}_{:= \Mb_1}
\cdot
\underbrace{
\bigotimes_{i \in [q]} \left(\begin{matrix}
1 & 0 \\
1 & 1
\end{matrix}\right)
}_{:= \Mb_2}
\cdot
\bigodot_{i \in [q]} \bLambda_{(i)}
,
\end{align*}
where the matrices $\Mb_1, \Mb_2$ are defined via Kronecker products $\otimes$, leading to the relation
$$
\bLambda_{[1]}'
=
\Mb_1 \Mb_2 \bLambda_{[1]}
.
$$
Observing the structure of these matrices, we note that $\Mb_1$ is an upper-triangular $2^q \times 2^q$ matrix with an all-ones diagonal, while $\Mb_2$ is a lower-triangular $2^q \times 2^q$ matrix $\Mb_2$, also with an all-ones diagonal.
This ensures that both $\Mb_1$ and $\Mb_2$ are invertible.
Therefore, if we can establish that $\bLambda_{[1]}'$ is invertible, it follows that $\bLambda_{[1]}$ must also be invertible, which directly implies that the Kruskal rank satisfies $\rank_K(\bLambda_{[1]}) = 2^q$.

We define the lexicographical order over the space of binary vectors $\{0, 1\}^q$ by
\begin{equation}\label{eq:lex}
\ub \lex \vb
\quad\iff\quad
\sum_{j = 1}^q u_j 2^j < \sum_{j = 1}^q v_j 2^j
.
\end{equation}
Since $\sW$ is an arbitrary bijective map from $\{0, 1\}^q$ to $[2^q]$, we assume without loss of generality that it is the unique bijective map satisfying $\sW^{-1}(1) \lex \cdots \lex \sW^{-1}(2^q)$.

Next, we define the $2^q \times 2^q$ matrix $\bLambda_{[1]}''$ entrywise by
$$
\Lambda_{[1], s, t}''
:=
\prod_{i = 1}^q \left(
\left(
\Lambda_{(i), 1, t} - \frac{\exp(\beta_{i, 0})}{1 + \exp(\beta_{i, 0})}
\right) 1_{\sW^{-1}(s)_i = 1} + 1_{\sW^{-1}(s)_i = 0}
\right)
,\quad
\forall s, t \in [2^q]
.
$$
The matrix represents a row and column permutation of $\bLambda_{[1]}'$.
For any $1 \le t < s \le 2^q$, since $\sW^{-1}(t) \lex \sW^{-1}(s)$, it follows from the definition in \eqref{eq:lex} that there exists an index $\iota \in [q]$ such that the $\iota$th entries of $\sW^{-1}(t)$ and $\sW^{-1}(s)$ satisfy $\sW^{-1}(t)_\iota = 0$ while $\sW^{-1}(s)_\iota = 1$.
Given that the $\iota$th row of $\Gb$ is $\gb_{\iota} = \eb_\iota$, we obtain
\begin{align*}
&\qquad
\left(
\Lambda_{(\iota), 1, t} - \frac{\exp(\beta_{\iota, 0})}{1 + \exp(\beta_{\iota, 0})}
\right) 1_{\sW^{-1}(s)_\iota = 1} + 1_{\sW^{-1}(s)_\iota = 0}
\\&=
\Lambda_{(\iota), 1, t} - \frac{\exp(\beta_{\iota, 0})}{1 + \exp(\beta_{\iota, 0})}
\\&=
\frac{
\exp(\beta_{\iota, 0} + (\gb_\iota \circ \sW^{-1}(t))^\top \bbeta_\iota)
}{
1 + \exp(\beta_{\iota, 0} + (\gb_\iota \circ \sW^{-1}(t))^\top \bbeta_\iota)
}
-
\frac{\exp(\beta_{\iota, 0})}{1 + \exp(\beta_{\iota, 0})}
\\&=
\frac{
\exp(\beta_{\iota, 0} + \sW^{-1}(t)_\iota \beta_{\iota, \iota})
}{
1 + \exp(\beta_{\iota, 0} + \sW^{-1}(t)_\iota \beta_{\iota, \iota})
}
-
\frac{\exp(\beta_{\iota, 0})}{1 + \exp(\beta_{\iota, 0})}
\\&=
\frac{\exp(\beta_{\iota, 0})}{1 + \exp(\beta_{\iota, 0})}
-
\frac{\exp(\beta_{\iota, 0})}{1 + \exp(\beta_{\iota, 0})}
\\&=
0
.
\end{align*}
Thus, $\bLambda_{[1], s, t}'' = 0$ for $t < s$, confirming that $\bLambda_{[1]}''$ is upper triangular.
For any $t = s \in [2^q]$ and each $i \in [q]$, we compute
\begin{align*}
&\qquad
\left(
\Lambda_{(i), 1, t} - \frac{\exp(\beta_{i, 0})}{1 + \exp(\beta_{i, 0})}
\right) 1_{\sW^{-1}(t)_i = 1} + 1_{\sW^{-1}(t)_i = 0}
\\&=
\left(
\frac{
\exp(\beta_{i, 0} + (\gb_i \circ \sW^{-1}(t))^\top \bbeta_i)
}{
1 + \exp(\beta_{i, 0} + (\gb_i \circ \sW^{-1}(t))^\top \bbeta_i)
}
-
\frac{\exp(\beta_{i, 0})}{1 + \exp(\beta_{i, 0})}
\right) 1_{\sW^{-1}(t)_i = 1}
+
1_{\sW^{-1}(t)_i = 0}
\\&=
\left(
\frac{
\exp(\beta_{i, 0} + \sW^{-1}(t)_i \beta_{i, i})
}{
1 + \exp(\beta_{i, 0} + \sW^{-1}(t)_i \beta_{i, i})
}
-
\frac{\exp(\beta_{i, 0})}{1 + \exp(\beta_{i, 0})}
\right) 1_{\sW^{-1}(t)_i = 1}
+
1_{\sW^{-1}(t)_i = 0}
\\&\ne
0
,
\end{align*}
regardless of the value of $\sW^{-1}(t)_i$.
This confirms that $\bLambda_{[1], t, t}'' \ne 0$ for all $t \in [2^q]$.
Since $\bLambda_{[1]}''$ is upper triangular with non-zero diagonal entries, it is invertible, implying that $\bLambda_{[1]}'$ is also invertible.
This establishes the Kruskal $\rank_K(\bLambda_{[1]}) = 2^q$.
By applying the same argument to $\bLambda_{[2]}$ and $\bLambda_{[3]}$, we also conclude that $\rank_K(\bLambda_{[2]}) = 2^q$ and $\rank_K(\bLambda_{[3]}) = 2^q$.

We now analyze the Kruskal rank of $(\bLambda_{[3]} \odot \bLambda_{[4]}) \diag(\bnu)$.
Since each $\alpha_{j, h} \in (0, 1)$, the probability vector $\bnu$ has no zero entry, implying that $$
\rank_K((\bLambda_{[3]} \odot \bLambda_{[4]}) \diag(\bnu))
=
\rank_K(\bLambda_{[3]} \odot \bLambda_{[4]})
.
$$
Furthermore, we observe that each row of $\bLambda_{[3]}$ can be obtained as a summation of rows in $\bLambda_{[3]} \odot \bLambda_{[4]}$, which implies that a subset of columns in $\bLambda_{[3]} \odot \bLambda_{[4]}$ is linearly dependent only if the corresponding columns in $\bLambda_{[3]}$ are also linearly dependent.
Since we have already established that $\rank_K(\bLambda_{[3]}) = 2^q$, it follows that $\rank_K((\bLambda_{[3]} \odot \bLambda_{[4]}) \diag(\bnu)) = 2^q$.

Putting everything together, we obtain
$$
\rank_K(\bLambda_{[1]}) + \rank_K(\bLambda_{[2]}) + \rank_K(\bLambda_{[3]} \diag(\bnu))
=
3 \cdot 2^q
\ge
2 \cdot 2^q + 2
.
$$
By applying Theorem \ref{theo:kruskal}, we conclude that there exist $2^q \times 2^q$ diagonal matrices $\Sb_1, \Sb_2, \Sb_3$ satisfying $\Sb_1 \Sb_2 \Sb_3 = \Ib$ and a $2^q \times 2^q$ permutation matrix $\Qb$ such that
$$
\tilde\bLambda_{[1]}
=
\bLambda_{[1]} \Sb_1 \Qb
,\quad
\tilde\bLambda_{[2]}
=
\bLambda_{[2]} \Sb_2 \Qb
,\quad
(\tilde\bLambda_{[3]} \odot \tilde\bLambda_{[4]}) \diag(\tilde\bnu)
=
(\bLambda_{[3]} \odot \bLambda_{[4]}) \diag(\bnu) \Sb_3 \Qb
.
$$
Since each column of $\bLambda_{[i]}$ and $\tilde\bLambda_{[i]}$ $(i \in \{1, 2, 3, 4\})$ sums to one, we must have $\Sb_1 = \Sb_2 = \Sb_3 = \Ib$.
Therefore,
$$
\tilde\bLambda_{[1]}
=
\bLambda_{[1]} \Qb
,\quad
\tilde\bLambda_{[2]}
=
\bLambda_{[2]} \Qb
.
$$
For $\tilde\bnu$, we compute
$$
\tilde\bnu
=
\left(
(\tilde\bLambda_{[3]} \odot \tilde\bLambda_{[4]}) \diag(\tilde\bnu)
\right)^\top \one
=
\left(
(\bLambda_{[3]} \odot \bLambda_{[4]}) \diag(\bnu) \Qb
\right)^\top \one
=
\Qb^\top \bnu
.
$$
For $\tilde\bLambda_{[3]} \odot \tilde\bLambda_{[4]}$, we have
\begin{align*}
\tilde\bLambda_{[3]} \odot \tilde\bLambda_{[4]}
&=
(\bLambda_{[3]} \odot \bLambda_{[4]}) \diag(\bnu) \Qb \diag(\tilde\bnu)^{-1}
\\&=
(\bLambda_{[3]} \odot \bLambda_{[4]}) \diag(\bnu) \Qb (\Qb^\top \diag(\bnu)^{-1} \Qb)
\\&=
(\bLambda_{[3]} \odot \bLambda_{[4]}) \Qb
.
\end{align*}
Finally, we obtain
\begin{align*}
\tilde\bLambda
&=
\tilde\bLambda_{[1]} \odot \tilde\bLambda_{[2]} \odot \tilde\bLambda_{[3]} \odot \tilde\bLambda_{[4]}
\\&=
(\bLambda_{[1]} \Qb) \odot (\bLambda_{[2]} \Qb) \odot ((\bLambda_{[3]} \odot \bLambda_{[4]}) \Qb)
\\&=
(\bLambda_{[1]} \odot \bLambda_{[2]} \odot \bLambda_{[3]} \odot \bLambda_{[4]}) \Qb
\\&=
\bLambda \Qb
.
\end{align*}
\end{proof}

We next present the proof of Lemma \ref{lemm:w_to_z}, employing the same technique based on Kruskal's Theorem \ref{theo:kruskal} as used in the proof of Lemma \ref{lemm:y_to_w}.

\begin{proof}[Proof of Lemma \ref{lemm:w_to_z}]
We recall the definition of $\bSigma_{[j]}$ from \eqref{eq:Sigma_jth}.
Based on the partition $[q] = \cQ_1 \cup \cQ_2 \cup \cQ_3$, we define the probability matrices
$$
\bSigma_{[1]}
:=
\odot_{j \in \cQ_1} \bSigma_{(j)}
,\quad
\bSigma_{[2]}
:=
\odot_{j \in \cQ_2} \bSigma_{(j)}
,\quad
\bSigma_{[3]}
:=
\odot_{j \in \cQ_3} \bSigma_{(j)}
.
$$
Then, the Khatri-Rao product $\bSigma_{[1]} \odot \bSigma_{[2]} \odot \bSigma_{[3]}$ coincides with $\bSigma$, up to row permutations.
Observe that each row of $\bSigma_{(j)}$ can be obtained as the summation of rows in $\bSigma$, i.e.
$$
\Sigma_{(j), t, h}
=
\sum_{r \in \sW(\{\wb \in \{0, 1\}^q:~ w_j = 1_{t = 1}\})} \Sigma_{r, h}
,\quad
\forall t \in [2], h \in [d]
.
$$
Accordingly, we define the probability matrices $\tilde\bSigma_{(j)}$ for each $j \in [q]$ entrywise as
$$
\tilde\Sigma_{(j), t, h}
:=
\sum_{r \in \sW(\{\wb \in \{0, 1\}^q:~ w_j = 1_{t = 1}\})} \tilde\Sigma_{r, h}
,\quad
\forall t \in [2], h \in [d]
.
$$
Further, we define the probability matrices
$$
\tilde\bSigma_{[1]}
:=
\odot_{j \in \cQ_1} \tilde\bSigma_{(j)}
,\quad
\tilde\bSigma_{[2]}
:=
\odot_{j \in \cQ_2} \tilde\bSigma_{(j)}
,\quad
\tilde\bSigma_{[3]}
:=
\odot_{j \in \cQ_3} \tilde\bSigma_{(j)}
.
$$
Since $\bSigma \bmu = \tilde\bSigma \tilde\bmu$, we obtain its equivalent representation in terms of three-way tensors as
$$
\llbracket \bSigma_{[1]}, \bSigma_{[2]}, \bSigma_{[3]} \diag(\bmu) \rrbracket
=
\llbracket \tilde\bSigma_{[1]}, \tilde\bSigma_{[2]}, \tilde\bSigma_{[3]} \diag(\tilde\bmu) \rrbracket
.
$$
To establish the uniqueness of $\bSigma, \bmu$ up to column permutations given their product $\bSigma \bmu$, it suffices to demonstrate the uniqueness of the above three-way tensor CP decomposition.

Recall that the vector $\balpha_j = (\alpha_{j, 1}, \ldots, \alpha_{j, d}) \in \RR^d$ represents the $j$th row of $\Ab$.
For each $j \in [q]$, we observe that
$$
\bSigma_{(j)}
=
\left(\begin{matrix}
\balpha_j^\top \\ \one^\top - \balpha_j^\top
\end{matrix}\right)
=
\left(\begin{matrix}
1 & 0 \\
-1 & 1
\end{matrix}\right)
\left(\begin{matrix}
\balpha_j^\top \\ \one^\top
\end{matrix}\right)
.
$$
This leads to the expression
$$
\bSigma_{[1]}
=
\odot_{j \in \cQ_1} \bSigma_{(j)}
=
\underbrace{
\bigotimes_{j \in \cQ_1} \left(\begin{matrix}
1 & 0 \\
-1 & 1
\end{matrix}\right)
}_{:= \Mb_3}
\cdot
\underbrace{
\bigodot_{j \in \cQ_1} \left(\begin{matrix}
\balpha_j^\top \\ \one^\top
\end{matrix}\right)
}_{:= \Mb_4}
.
$$
The matrix $\Mb_3$ is of dimension $2^{\card(\cQ_1)} \times 2^{\card(\cQ_1)}$ and is invertible since it is lower triangular with an all-ones diagonal.
The row vectors of $\Mb_4$ span the set $\bigcirc_{j \in \cQ_1} \{\one, \balpha_j\} \subset \RR^d$, which has full rank $d$ as assumed in Lemma \ref{lemm:w_to_z}.
Therefore, we conclude that $\rank(\bSigma_{[1]}) = \rank(\Mb_4) = d$, which implies Kruskal $\rank_K(\bSigma_{[1]}) = d$.
Using the same argument, we also establish $\rank_K(\bSigma_{[2]}) = d$.

By definition, each column of the probability matrix $\bSigma_{[3]}$ sums to one.
Furthermore, our model formulation \eqref{eq:model_z_x} ensures that every entry of the probability vector $\bmu$ is positive.
Consequently, two arbitrary columns of $\bSigma_{[3]} \diag(\bmu)$ are linearly dependent if and only if the corresponding two columns of $\bSigma_{[3]}$ are linearly dependent, which in turn holds if and only if these two columns of $\bSigma_{[3]}$ are identical.
For $\bSigma_{[3]} = \odot_{j \in \cQ_3} \bSigma_{(j)}$, two columns $h, h' \in [d]$ are identical if and only if the entries of $\bSigma_{(j)}$ satisfy $\Sigma_{(j), 1, h} = \Sigma_{(j), 1, h'}$ for all $j \in \cQ_3$.
By definition \eqref{eq:Sigma_jth}, we have $\Sigma_{(j), 1, h} = \alpha_{j, h}$ and $\Sigma_{(j), 1, h'} = \alpha_{j, h'}$.
Since the submatrix $\Ab_{\cQ_3}$ is assumed to have distinct columns, it follows that any two columns of $\bSigma_{[3]} \diag(\bmu)$ are linearly independent.
This implies that the Kruskal rank $\rank_K(\bSigma_{[3]} \diag(\bmu)) \ge 2$.

The results above together yield
$$
\rank_K(\bSigma_{[1]}) + \rank_K(\bSigma_{[2]}) + \rank_K(\bSigma_{[3]} \diag(\bmu)) 
\ge
2d + 2
.
$$
The remainder of the proof follows the same approach as the proof of Lemma \ref{lemm:y_to_w}, which we include here for completeness.

By applying Theorem \ref{theo:kruskal}, there exist $d \times d$ diagonal matrices $\Sb_1, \Sb_2, \Sb_3$ satisfying $\Sb_1 \Sb_2 \Sb_3 = \Ib$ and a $d \times d$ permutation matrix $\Rb$ such that
$$
\tilde\bSigma_{[1]}
=
\bSigma_{[1]} \Sb_1 \Rb
,\quad
\tilde\bSigma_{[2]}
=
\bSigma_{[2]} \Sb_2 \Rb
,\quad
\tilde\bSigma_{[3]} \diag(\tilde\bmu)
=
\bSigma_{[3]} \diag(\bmu) \Sb_3 \Rb
.
$$
Since each column of $\bSigma_{[i]}$ and $\tilde\bSigma_{[i]}$ $(i \in \{1, 2, 3, 4\})$ sums to one, we must have $\Sb_1 = \Sb_2 = \Sb_3 = \Ib$.
Therefore,
$$
\tilde\bSigma_{[1]}
=
\bSigma_{[1]} \Rb
,\quad
\tilde\bSigma_{[2]}
=
\bSigma_{[2]} \Rb
.
$$
For $\tilde\bmu$, we have
$$
\tilde\bmu
=
\left(
\tilde\bSigma_{[3]} \diag(\tilde\bmu)
\right)^\top \one
=
\left(
\bSigma_{[3]} \diag(\bmu) \Rb
\right)^\top \one
=
\Rb^\top \bmu
.
$$
For $\tilde\bSigma_{[3]}$, we have
\begin{align*}
\tilde\bSigma_{[3]}
&=
\bSigma_{[3]} \diag(\bmu) \Rb \diag(\tilde\bmu)^{-1}
\\&=
\bSigma_{[3]} \diag(\bmu) \Rb (\Rb^\top \diag(\bmu)^{-1} \Rb)
\\&=
\bSigma_{[3]} \Rb
.
\end{align*}
Finally, we conclude
\begin{align*}
\tilde\bSigma
&=
\tilde\bSigma_{[1]} \odot \tilde\bSigma_{[2]} \odot \tilde\bSigma_{[3]}
\\&=
(\bSigma_{[1]} \Rb) \odot (\bSigma_{[2]} \Rb) \odot (\bSigma_{[3]} \Rb)
\\&=
(\bSigma_{[1]} \odot \bSigma_{[2]} \odot \bSigma_{[3]}) \Rb
\\&=
\bSigma \Rb
.
\end{align*}
\end{proof}

\begin{proof}[Proof of Lemma \ref{lemm:law_to_params}]
Let the conditional distributions satisfy
$$
\law(\yb ~|~ \wb, \Gb, \Bb)
=
\law(\yb ~|~ \wb, \tilde\Gb, \tilde\Bb)
,\quad
\forall \wb \in \{0, 1\}^q
.
$$
By the conditional independence of entries $y_i$ given $\wb$, for each $i \in [p]$ and $\wb \in \{0, 1\}^q$, we have
\begin{align*}
\frac{\exp(\beta_{i, 0} + (\gb_i \circ \bbeta_i)^\top \wb)}{1 + \exp(\beta_{i, 0} + (\gb_i \circ \bbeta_i)^\top \wb)}
&=
\PP(y_i = 1 ~|~ \wb, \gb_i, \beta_{i, 0}, \bbeta_i)
\\&=
\sum_{\{\yb \in \{0, 1\}^p:~ y_i = 1\}} \PP(\yb ~|~ \wb, \Gb, \Bb)
\\&=
\sum_{\{\yb \in \{0, 1\}^p:~ y_i = 1\}} \PP(\yb ~|~ \wb, \tilde\Gb, \tilde\Bb)
\\&=
\PP(y_i = 1 ~|~ \wb, \tilde\gb_i, \tilde\beta_{i, 0}, \tilde\bbeta_i)
\\&=
\frac{\exp(\tilde\beta_{i, 0} + (\tilde\gb_i \circ \tilde\bbeta_i)^\top \wb)}{1 + \exp(\tilde\beta_{i, 0} + (\tilde\gb_i \circ \tilde\bbeta_i)^\top \wb)}
.
\end{align*}
Since the logistic function is monotonically increasing, it follows that
$$
\beta_{i, 0} + (\gb_i \circ \bbeta_i)^\top \wb
=
\tilde\beta_{i, 0} + (\tilde\gb_i \circ \tilde\bbeta_i)^\top \wb
,\quad
\forall i \in [p]
,\quad
\forall \wb \in \{0, 1\}^q
.
$$
This yields the equalities
$$
\beta_{i, 0} = \tilde\beta_{i, 0}
,\quad
\gb_i \circ \bbeta_i = \tilde\gb_i \circ \tilde\bbeta_i
,\quad
\forall i \in [p]
,
$$
which can be rewritten in matrix form as $(1 ~ \Gb) \circ \Bb = (1 ~ \tilde\Gb) \circ \Bb$.
Moreover, since each $\beta_{i, j} \ne 0$ if $g_{i, j} = 1$ and each $\tilde\beta_{i, j} \ne 0$ if $\tilde{g}_{i, j} = 1$, it follows that $\Gb = \tilde\Gb$.

Let the conditional distributions satisfy
$$
\law(\wb ~|~ z, \Ab)
=
\law(\wb ~|~ z, \tilde\Ab)
,\quad
\forall z \in [d]
.
$$
By the conditional independence of entries $w_j$ given $z$, for each $j \in [q]$ and $z \in [d]$ we have
\begin{align*}
\alpha_{j, z}
&=
\PP(w_j = 1 ~|~ z, \Ab)
\\&=
\sum_{\{\wb \in \{0, 1\}^q:~ w_j = 1\}} \PP(\wb ~|~ z, \Ab)
\\&=
\sum_{\{\wb \in \{0, 1\}^q:~ w_j = 1\}} \PP(\wb ~|~ z, \tilde\Ab)
\\&=
\PP(w_j = 1 ~|~ z, \tilde\Ab)
\\&=
\tilde\alpha_{j, z}
.
\end{align*}
Therefore, we have $\Ab = \tilde\Ab$.

Let the conditional distributions satisfy
$$
\law(z ~|~ \xb, \bGamma)
=
\law(z ~|~ \xb, \tilde\bGamma)
,\quad
\cP_x\text{-almost everywhere}
.
$$
Recall that the parameters $\gamma_{d, 0}, \bgamma_d$ and $\tilde\gamma_{d, 0}, \tilde\bgamma_d$ associated with the baseline class $d$ are fixed as zero to ensure identifiability.
Therefore, for each $h \in [d]$, we have
$$
\gamma_{h, 0} + \bgamma_h^\top \xb
=
\log\frac{
\PP(z = h ~|~ \xb, \bGamma)
}{
\PP(z = d ~|~ \xb, \bGamma)
}
=
\log\frac{
\PP(z = h ~|~ \xb, \tilde\bGamma)
}{
\PP(z = d ~|~ \xb, \tilde\bGamma)
}
=
\tilde\gamma_{h, 0} + \tilde\bgamma_h^\top \xb
,
$$
holding for $\cP_x$-almost every $\xb$.
We define the continuous function
$$
\psi(\xb)
:=
\sum_{h = 1}^d \left(
(\gamma_{h, 0} - \tilde\gamma_{h, 0})
+
(\bgamma_h - \tilde\bgamma_h)^\top \xb
\right)^2
.
$$
The result above implies that $\psi(\xb) = 0$ for all $\xb \in \supp(\cP_x)$.
Since the data generating distribution $\cP_x$ is assumed to have full rank, for any $\vb \in \RR^{p_x + 1}$ with $\vb \ne \zero$, there exists $\xb \in \supp(\cP_x)$ such that $\vb^\top (1, \xb) \ne 0$.
This ensures that $\gamma_{h, 0} - \tilde\gamma_{h, 0} = 0$ and $\bgamma_h - \tilde\bgamma_h = 0$ for all $h \in [d]$, which implies $\bGamma = \tilde\bGamma$.
\end{proof}

\subsection{Proof of Theorem \ref{theo:generic}}\label{supp_ssec:proof_generic}

Recall that $\sS_2$ is the space of parameters $(\Ab, \Bb, \bGamma, \Gb)$ satisfying condition ($*$) of Theorem \ref{theo:generic}, which states that $\Gb$ has distinct columns and contains two distinct submatrices with all-one diagonals, while the remaining submatrix does not have all-zeros columns.
We adopt the same notation of $\bLambda, \bnu, \bSigma, \bmu$ as defined in \eqref{eq:Lambda_st}, \eqref{eq:nu_t}, \eqref{eq:Sigma_th}, \eqref{eq:mu_h} in Section \ref{supp_ssec:proof_strict}.

For any subset $\sS_* \subset \sS_2$, we say that $\sS_*$ is \textit{permutation invariant} if
\begin{equation}\label{eq:perm_inv}
\perm_{\sS_2}(\Ab, \Bb, \bGamma, \Gb)
\subset
\sS_*
,\quad
\forall (\Ab, \Bb, \bGamma, \Gb) \in \sS_*
,
\end{equation}
meaning that all permutations of parameters $(\Ab, \Bb, \bGamma, \Gb)$ over the $q$ dimensions of $\wb$ and the $d$ classes of $z$ are either entirely included in $\sS_*$ or entirely excluded.

Compared to Theorem \ref{theo:strict}, Theorem \ref{theo:generic} assumes only the weaker condition ($*$).
However, even under condition ($*$), we can still establish a weaker version of Lemma \ref{lemm:y_to_w} by excluding a measure-zero subset from the parameter space $\sS_2$, as stated in the following lemma.
Its proof is deferred to Section \ref{supp_ssec:proof_lemmas_generic}.
Recall that the base measure on $\sS_2$ is the product $\uplambda \times \upmu$ of the Lebesgue measure for the continuous parameters $\Ab, \Bb, \bGamma$ and the counting measure $\upmu$ for the discrete parameter $\Gb$.

\begin{lemma}\label{lemm:y_to_w_meas0}
There exists a permutation invariant subset $\sS_{2, 1} \subset \sS_2$ with measure $(\uplambda \times \upmu)(\sS_{2, 1}) = 0$ such that the following holds:
For any parameter $(\Ab, \Bb, \bGamma, \Gb) \in \cS_2 \cap \cS_{2, 1}^c$, let $\bLambda$ be the probability matrix defined in \eqref{eq:Lambda_st} and let $\bnu$ be the probability vector defined in \eqref{eq:nu_t}.
Then, for any other $2^p \times 2^q$ matrix $\tilde\bLambda$ and $2^q$-dimensional vector $\tilde\bnu$ satisfying
$$
\bLambda \bnu
=
\tilde\bLambda \tilde\bnu
,\quad
\tilde\bLambda^\top \one
=
\one
,\quad
\one^\top \tilde\bnu
=
1
,
$$
there exists a $2^q \times 2^q$ permutation matrix $\Qb$ such that
$$
\tilde\bLambda
=
\bLambda \Qb
,\quad
\tilde\bnu
=
\Qb^\top \bnu
.
$$
\end{lemma}

By excluding a zero measure subset $\sS_{2, 1}$ from the parameter space $\sS_2$, Lemma \ref{lemm:y_to_w} allows recovery of $\bLambda$ and $\bnu$ from their product $\bLambda \bnu$, up to column permutations, i.e. permutations on the $2^q$ values of $\wb$ within $\{0, 1\}^q$.
Similarly, by excluding another zero-measure subset $\sS_{2, 2}$ from the parameter space $\sS_2$, we can recover $\bSigma$ and $\bmu$ from their product $\bSigma \bmu$, up to column permutations, i.e. permutations on the $d$ classes of $z$.
This result is formally stated in the following lemma, with the proof provided in Section \ref{supp_ssec:proof_lemmas_generic}.

\begin{lemma}\label{lemm:w_to_z_meas0}
There exists a permutation invariant subset $\sS_{2, 2} \subset \sS_2$ with measure $(\uplambda \times \upmu)(\sS_{2, 2}) = 0$ such that the following holds:
For any parameter $(\Ab, \Bb, \bGamma, \Gb) \in \sS_2 \cap \sS_{2, 2}^c$, let $\bSigma$ be the probability matrix defined in \eqref{eq:Sigma_th} and $\bmu$ be the probability vector defined in \eqref{eq:mu_h}.
Then, for any other $2^q \times d$ matrix $\tilde\bSigma$ and $d$-dimensional vector $\tilde\bmu$ satisfying
$$
\bSigma \bmu
=
\tilde\bSigma \tilde\bmu
,\quad
\tilde\bSigma^\top \one
=
\one
,\quad
\one^\top \tilde\bmu
=
1
,
$$
there exists a $d \times d$ permutation matrix $\Rb$ such that
$$
\tilde\bSigma
=
\bSigma \Rb
,\quad
\tilde\bmu
=
\Rb^\top \bmu
.
$$
\end{lemma}

Intuitively, when restricting parameters within the space $\sS_2 \cap (\sS_{2, 1} \cup \sS_{2, 2})^c$, Lemmas \ref{lemm:y_to_w_meas0} and \ref{lemm:w_to_z_meas0} allow us to recover the conditional distributions $\law(\yb ~|~ \wb, \Gb, \Bb)$, $\law(\wb ~|~ z, \Ab)$, and $\law(z ~|~ \xb, \bGamma)$ from the distribution $\law(\yb ~|~ \xb, \Ab, \Bb, \bGamma, \Gb)$.
To further identify the parameters $\Ab, \Bb, \bGamma, \Gb$ from these recovered conditional distributions, the following lemma rephrases Lemma \ref{lemm:law_to_params} and ensures this by additionally excluding a zero measure set from $\sS_{2, 3}$.
We include the proof of the lemma in Section \ref{supp_ssec:proof_lemmas_generic}.

\begin{lemma}\label{lemm:law_to_params_meas0}
Suppose the data generating distribution $\cP_x$ of $\xb$ has full rank.
There exists a permutation invariant subset $\sS_{2, 3} \subset \sS_2$ with measure $(\uplambda \times \upmu)(\sS_{2, 3}) = 0$ such that, for any two sets of parameters $(\Ab, \Bb, \bGamma, \Gb), (\tilde\Ab, \tilde\Bb, \tilde\bGamma, \tilde\Gb) \in \sS_2 \cap \sS_{2, 3}^c$, the following statements hold:
\begin{enumerate}[(i)]
\item
If the conditional distributions satisfy
$$
\law(\yb ~|~ \wb, \Gb, \Bb)
=
\law(\yb ~|~ \wb, \tilde\Gb, \tilde\Bb)
,\quad
\forall \wb \in \{0, 1\}^q
,
$$
then $\Gb = \tilde\Gb$ and $(1 ~ \Gb) \circ \Bb = (1 ~ \tilde\Gb) \circ \tilde\Bb$.
\item
If the conditional distributions satisfy
$$
\law(\wb ~|~ z, \Ab)
=
\law(\wb ~|~ z, \tilde\Ab)
,\quad
\forall z \in [d]
,
$$
then $\Ab = \tilde\Ab$.
\item
If the conditional distributions satisfy
$$
\law(z ~|~ \xb, \bGamma)
=
\law(z ~|~ \xb, \tilde\bGamma)
,\quad
\cP_x-\text{a.e.}
$$
then $\bGamma = \tilde\bGamma$.
\end{enumerate}
\end{lemma}

Combining Lemmas \ref{lemm:y_to_w_meas0}, \ref{lemm:w_to_z_meas0}, and \ref{lemm:law_to_params_meas0}, we can now establish Theorem \ref{theo:generic} in a manner similar to the proof of Theorem \ref{theo:strict}.
However, due to the weaker condition on $\Gb$, bridging the gap between permutations of $\wb$ over $\{0, 1\}^q$ in Lemma \ref{lemm:y_to_w_meas0} and the permutations of the $q$ dimensions of $\wb$ in Theorem \ref{theo:generic} is more challenging compared to the proof of Theorem \ref{theo:generic}.
Moreover, in these lemmas, we have excluded the set $\sS_{2, 1} \cup \sS_{2, 2} \cup \sS_{2, 3}$ from the parameter space $\sS_2$.
Despite its zero measure, we need to demonstrate that incorporating $\sS_{2, 1} \cup \sS_{2, 2} \cup \sS_{2, 3}$ back into the parameter space does not lead to a positive measure set of parameters becoming non-identifiable. This is not immediately obvious, as a measurable mapping can potentially transform a zero-measure set into an image of positive measure.

\begin{proof}[Proof of Theorem \ref{theo:generic}]
For the subsets $\sS_{2, 1}, \sS_{2, 2}, \sS_{2, 3} \subset \sS_2$, while they are not explicitly defined in the previous lemmas, the definitions are given in \eqref{eq:sS21}, \eqref{eq:sS22}, and \eqref{eq:sS23} in the proofs of these lemmas.
In addition to these subsets, we further exclude a permutation invariant subset $\sS_{2, 4}$ from $\sS_2$, defined as
\begin{equation}\label{eq:sS24}
\sS_{2, 4}
:=
\left\{
(\Ab, \Bb, \bGamma, \Gb) \in \sS_2:~
\exists i \in [p],~
\exists \cbb \in \ZZ^q \cap \{\zero\}^c
\text{ s.t. }
\cbb^\top \bbeta_i = 0
\right\}
,
\end{equation}
where $\cbb \in \ZZ^q \cap \{\zero\}^c$ denotes an arbitrary non-zero $q$-dimensional integer vector.
Intuitively, by excluding $\sS_{2, 4}$ from $\sS_2$, we are restricting attention to parameters $\Bb$ where, for each $i \in [p]$, the collection of its entries in the $i$th row, i.e. $\{\beta_{i, j}:~ j \in [q]\}$, are linearly independent under integer coefficients.
Furthermore, we observe that for any given $\cbb \in \ZZ^q \cap \{\zero\}^c$, the set $\{\cbb^\top \bbeta_i = 0\}$ forms a linear subspace in $\RR^q$, which has measure zero. Taking a union bound over the finitely many possible values of $i$ and the countably many possible values of $\cbb$, we obtain the measure $(\uplambda \times \upmu)(\sS_{2, 4}) = 0$.

We begin by studying identifiability within the smaller parameter space
\begin{equation}\label{eq:sS3}
\sS_3
:=
\sS_2 \cap (\sS_{2, 1} \cup \sS_{2, 2} \cup \sS_{2, 3} \cup \sS_{2, 4})^c
.
\end{equation}
For any parameter $(\Ab, \Bb, \bGamma, \Gb) \subset \sS_3$, let $(\tilde\Ab, \tilde\Bb, \tilde\bGamma, \tilde\Gb) \subset \marg_{\sS_3}(\Ab, \Bb, \bGamma, \Gb)$ be an arbitrary parameter that shares the same conditional distribution of $\yb ~|~ \xb$.
By applying Lemma \ref{lemm:y_to_w_meas0}, we know that there exists a permutation map $\sQ$ over the space $\{0, 1\}^q$ of $\wb$ such that
$$
\law(\yb ~|~ \wb, \tilde\Gb, \tilde\Bb)
=
\law(\yb ~|~ \sQ(\wb), \Gb, \Bb)
,\quad
\law(\wb ~|~ \xb, \tilde\Ab, \tilde\bGamma)
=
\law(\sQ(\wb) ~|~ \xb, \Ab, \bGamma)
.
$$
This further suggests that, for each $i \in [p]$,
$$
\law(y_i ~|~ \wb, \tilde\Gb, \tilde\Bb)
=
\law(y_i ~|~ \sQ(\wb), \Gb, \Bb)
.
$$

For each $\wb \in \{0, 1\}^q$ and $i \in [p]$, we define
$$
\xi_{\wb, i}
:=
\log\frac{
\PP(y_i = 1 ~|~ \wb, \Gb, \Bb)
}{
\PP(y_i = 0 ~|~ \wb, \Gb, \Bb)
}
=
\beta_{i, 0} + (\gb_i \circ \wb)^\top \bbeta_i
=
\beta_{i, 0} + \sum_{j = 1}^q g_{i, j} w_j \tilde\beta_{i, j}
,
$$
and similarly,
$$
\tilde\xi_{\wb, i}
:=
\log\frac{
\PP(y_i = 1 ~|~ \wb, \tilde\Gb, \tilde\Bb)
}{
\PP(y_i = 0 ~|~ \wb, \tilde\Gb, \tilde\Bb)
}
=
\tilde\beta_{i, 0} + (\tilde\gb_i \circ \wb)^\top \tilde\bbeta_i
=
\tilde\beta_{i, 0} + \sum_{j = 1}^q \tilde{g}_{i, j} w_j \tilde\beta_{i, j}
.
$$
By their definitions, we have
\begin{equation}\label{eq:xi_sQ}
\tilde\xi_{\wb, i}
=
\xi_{\sQ(\wb), i}
.
\end{equation}
The definition of $\sS_{2, 3}$ in \eqref{eq:sS23}, specified in the proof of Lemma \ref{lemm:law_to_params_meas0}, suggests that for all $i \in [p], j \in [q]$, if $g_{i, j} = 1$ then $\beta_{i, j} \ne 0$.
Recall that the parameter space $\sS_2$ ensures $\Gb$ to have distinct columns and contain two distinct submatrices with all-one diagonals.
Consequently, for $\wb \in \{0, 1\}^q$, we have
\begin{equation}\label{eq:card_xi_w}\begin{aligned}
&
\card(\{\xi_{\wb, i}:~ i \in [p]\}) = 1
\quad\iff\quad
\wb = \zero
,\\&
\card(\{\xi_{\wb, i}:~ i \in [p]\}) = 2
\quad\iff\quad
\one^\top \wb = 1
,
\end{aligned}\end{equation}
with reasoning similar to what we used in the proof of Theorem \ref{theo:strict}.
Similar results hold for $\tilde\Gb$ and $\tilde\xi$, i.e.
\begin{equation}\label{eq:card_txi_tw}\begin{aligned}
&
\card(\{\tilde\xi_{\wb, i}:~ i \in [p]\}) = 1
\quad\iff\quad
\wb = \zero
,\\&
\card(\{\tilde\xi_{\wb, i}:~ i \in [p]\}) = 2
\quad\iff\quad
\one^\top \wb = 1
.
\end{aligned}\end{equation}
Note that $\one^\top \wb = 1$ is equivalent to $\wb = \eb_j$ for some $j \in [q]$.
Combining with \eqref{eq:xi_sQ}, it follows that
$$
\sQ(\zero)
=
\zero
,\quad
\sQ\big(
\{\eb_j:~ j \in [q]\}
\big)
=
\{\eb_j:~ j \in [q]\}
.
$$
This suggests the existence of some permutation map $\sN$ on $[q]$ such that
\begin{equation}\label{eq:Q_0_1}
\sQ(\eb_j)
=
\eb_{\sN(j)}
,\quad
\forall j \in [q]
.
\end{equation}
Moreover, for each $i \in [p]$, we have
\begin{equation}\label{eq:beta_tbeta_0}
\beta_{i, 0}
=
\xi_{\zero, i}
=
\tilde\xi_{\zero, i}
=
\tilde\beta_{i, 0}
,
\end{equation}
and for each $j \in [q]$,
\begin{equation}\label{eq:beta_tbeta_Nj}
\beta_{i, \sN(j)}
=
\xi_{\eb_{\sN(j)}, i} - \xi_{\zero, i}
=
\xi_{\sQ(\eb_j), i} - \xi_{\zero, i}
=
\tilde\xi_{\eb_j, i} - \tilde\xi_{\zero, i}
=
\tilde\beta_{i, j}
.
\end{equation}
Importantly, by \eqref{eq:card_xi_w} and \eqref{eq:card_txi_tw}, we can examine the cardinality of $\{\xi_{\wb, i}:~ i \in [p]\}$ to identify the columns in $\bLambda, \tilde\bLambda$ corresponding to $\wb = \zero$ and the collection of columns corresponding to $\wb$ satisfying $\one^\top \wb = 1$.
This allows us to recover $\beta_{i, 0}$ from $\bLambda$, as well as the set of parameters $\{\beta_{i, j}:~ j \in [q]\}$, although we can not match each value in this set to the corresponding index $j$ yet.

By the definition of $\sS_{2, 4}$ in \eqref{eq:sS24}, the parameters $\beta_{i, 1}, \ldots, \beta_{i, q}$ are linearly independent with integer coefficients. We know that for each $i \in [p]$ and $\wb \in \{0, 1\}^q$, the value $\xi_{\wb, i} - \xi_{\zero, i}$ can be uniquely decomposed into the sum of a subset of $\{\beta_{i, j}:~ j \in [q]\}$ through
\begin{equation}\label{eq:xi_int_sum}
\xi_{\wb, i} - \xi_{\zero, i}
=
\sum_{j = 1}^q w_j g_{i, j} \beta_{i, j}
=
\sum_{j \in \{j \in [q]:~ g_{i, j} = 1, w_j = 1\}} \beta_{i, j}
.
\end{equation}
Similar results hold for $\tilde\Gb, \tilde\Bb$ and $\tilde\xi$, which along with $\tilde\xi_{\wb, i} = \xi_{\sQ(\wb), i}$ gives
\begin{equation}\label{eq:txi_int_sum}
\xi_{\sQ(\wb), i} - \xi_{\zero, i}
=
\tilde\xi_{\wb, i} - \tilde\xi_{\zero, i}
=
\sum_{j \in \{j \in [q];~ \tilde{g}_{i, j} = 1, w_j = 1\}} \tilde\beta_{i, j}
=
\sum_{j \in \{j \in [q];~ \tilde{g}_{i, j} = 1, w_j = 1\}} \tilde\beta_{i, \sN(j)}
.
\end{equation}
By combining \eqref{eq:beta_tbeta_0}, \eqref{eq:beta_tbeta_Nj}, \eqref{eq:xi_int_sum}, and \eqref{eq:txi_int_sum}, it follows that for each $i \in [p]$ and $\wb \in \{0, 1\}^q$, the sets
$$
\Big\{
\sN(j):~
j \in [q]
,~
\tilde{g}_{i, j} = 1
,~
w_j = 1
\Big\}
=
\Big\{
j \in [q]:~
g_{i, j} = 1
,~
\sQ(\wb)_j = 1
\Big\}
.
$$
Since $\Gb$ has no all-zeros column, taking union over $i \in [p]$, we obtain
\begin{align*}
\Big\{
\sN(j):~
j \in [q]
,~
w_j = 1
\Big\}
&=
\bigcup_{i \in [p]} \Big\{
\sN(j):~
j \in [q]
,~
\tilde{g}_{i, j} = 1
,~
w_j = 1
\Big\}
\\&=
\bigcup_{i \in [p]} \Big\{
j \in [q]:~
g_{i, j} = 1
,~
\sQ(\wb)_j = 1
\Big\}
\\&=
\Big\{
j \in [q]:~
\sQ(\wb)_j = 1
\Big\}
.
\end{align*}
That is, for any $\wb \in \{0, 1\}^q$, the entries
\begin{equation}\label{eq:sQ_to_sN}
\sQ(\wb)_j = 1
\quad\iff\quad
\wb_{\sN(j)} = 1
.
\end{equation}

Notably, \eqref{eq:sQ_to_sN} establishes that the permutation map $\sQ$ over the space $\{0, 1\}^q$ of all possible values of $\wb$ is structured by the permutation map $\sN$ over $[q]$ on the $q$ dimensions of $\wb$.
This successfully bridges the gap between Lemma \ref{lemm:y_to_w_meas0} and Theorem \ref{theo:generic}.
Now with the exact same arguments as in the proof of Theorem \ref{theo:strict}, we combine Lemmas \ref{lemm:y_to_w_meas0}, \ref{lemm:w_to_z_meas0}, and \ref{lemm:law_to_params_meas0} and conclude for parameters within $\sS_3$ that
$$
\perm_{\sS_3}(\Ab, \Bb, \bGamma, \Gb)
\supseteq
\marg_{\sS_3}(\Ab, \Bb, \bGamma, \Gb)
,\quad
\forall (\Ab, \Bb, \bGamma, \Gb) \in \cS_3
.
$$
Since each subset $\sS_{2, 1}, \sS_{2, 2}, \sS_{2, 3}, \sS_{2, 4}$ is defined to be permutation invariant, so is $\sS_3$ by its definition \eqref{eq:sS3}. Thus, permutations on the $q$ dimensions of $\wb$ and on the $D$ classes of $z$ do not affect the distribution of $\yb ~|~ \xb$, giving $\perm_{\sS_3}(\balpha, \bbeta, \bgamma, \Gb) \subseteq \marg_{\sS_3}(\balpha, \bbeta, \bgamma, \Gb)$.
The above results collectively yield
$$
\perm_{\sS_3}(\Ab, \Bb, \bGamma, \Gb)
=
\marg_{\sS_3}(\Ab, \Bb, \bGamma, \Gb)
,\quad
\forall (\Ab, \Bb, \bGamma, \Gb) \in \cS_3
.
$$

We now extend this result to the larger parameter space $\sS_2$.
By permutation invariance of $\sS_3$, we have
$$
\perm_{\sS_3}(\Ab, \Bb, \bGamma, \Gb)
=
\perm_{\sS_2}(\Ab, \Bb, \bGamma, \Gb)
,\quad
\forall (\Ab, \Bb, \bGamma, \Gb) \in \cS_3
.
$$
However, the same relation does not necessarily extend to $\marg_{\sS_3}$ and $\marg_{\sS_2}$, because we could potentially have $\marg_{\sS_3}(\Ab, \Bb, \bGamma, \Gb) \subsetneq \marg_{\sS_2}(\Ab, \Bb, \bGamma, \Gb)$.

For each $\Gb$ satisfying the condition ($*$) required by $\sS_2$, we denote the parameter space for $(\Ab, \Bb, \bGamma)$ corresponding to $\sS_3$ and $\sS_2$ by
$$
\sS_{(\Gb), 3}
:=
\Big\{
(\Ab, \Bb, \bGamma):~
(\Ab, \Bb, \bGamma, \Gb) \in \sS_3
\Big\}
,\quad
\sS_{(\Gb), 2}
:=
\Big\{
(\Ab, \Bb, \bGamma):~ (\Ab, \Bb, \bGamma, \Gb) \in \sS_2
\Big\}
.
$$
By definition, we have $\uplambda(\sS_{(\Gb), 2} \cap \sS_{(\Gb), 3}^c) = 0$.
We further let $\sB_{(\Gb)}$ denote the collection of parameters with all inactive entries $\beta_{i, j}$ (i.e. corresponding to $g_{i, j} = 0$) constrained to zero, that is,
\begin{equation}\label{eq:sB_Gb}
\sB_{(\Gb)}
:=
\Big\{
\forall i \in [p], j \in [q],~
\text{if } g_{i, j} = 0
\text{ then } \beta_{i, j} = 0
\Big\}
.
\end{equation}

Since $\yb \in \{0, 1\}^p$, the distribution $\law(\yb ~|~ \xb, \Ab, \Bb, \bGamma, \Gb)$ can be viewed as a $2^p$-dimensional probability vector from the probability simplex.
For a given $\Gb$, we let $\law_{\Gb}$ denote the mapping from the space $\sS_{(\Gb), 2}$ of continuous parameters $\Ab, \Bb, \bGamma$ to the probability simplex of the distribution $\law(\yb ~|~ \xb, \Ab, \Bb, \bGamma, \Gb)$.
Then by definition, we have
$$
\Big\{
(\Ab', \Bb', \bGamma'):~
(\Ab', \Bb', \bGamma', \Gb) \in \marg_{\sS_2}(\Ab, \Bb, \bGamma, \Gb)
\big\}
=
\law_{\Gb}^{-1}\left(
\{\law_{\Gb}(\Ab, \Bb, \bGamma)\}
\right) \cap \sS_{(\Gb), 2}
,
$$
$$
\Big\{
(\Ab', \Bb', \bGamma'):~
(\Ab', \Bb', \bGamma', \Gb) \in \marg_{\sS_3}(\Ab, \Bb, \bGamma, \Gb)
\Big\}
=
\law_{\Gb}^{-1}\left(
\{\law_{\Gb}(\Ab, \Bb, \bGamma)\}
\right) \cap \sS_{(\Gb), 3}
.
$$

Our model formulations in \eqref{eq:model_y_w}, \eqref{eq:model_w_z}, and \eqref{eq:model_z_x} ensure that for all $\Gb$, the mapping $\law_{\Gb}$ is an analytic function.
We notice that there are at most $(q!)(d!)$ permutations of the $q$ dimensions of $\wb$ and the $d$ classes of $z$.
Within $\sB_{(\Gb)}$ where all inactive parameter entries are constrained to zero, for each $\Gb$ we have
\begin{align*}
\sup_{\cP} \upmu\left(
\law_{\Gb}^{-1}(\sP) \cap \sS_{(\Gb), 3} \cap \sB_{(\Gb)}
\right)
&\le
\sup_{(\Ab, \Bb, \bGamma, \Gb) \in \sS_3} \upmu\left(
\marg_{\sS_3}(\Ab, \Bb, \bGamma, \Gb) \cap \sB_{(\Gb)}
\right)
\\&=
\sup_{(\Ab, \Bb, \bGamma, \Gb) \in \sS_3} \upmu\left(
\perm_{\sS_3}(\Ab, \Bb, \bGamma, \Gb) \cap \sB_{(\Gb)}
\right)
\\&\le
(q!)(d!)
<
\infty
,
\end{align*}
where $\sP$ denotes an arbitrary $2^p$-dimensional probability vector in the probability simplex.
Since $\uplambda(\cS_{(\Gb), 2} \cap \cS_{(\Gb), 3}^c) = 0$ for each $\Gb$, we can apply Theorem \ref{theo:geom} and obtain
$$
\uplambda\left(
\law_{\Gb}^{-1}\left( \law_{\Gb}\left(
\cS_{(\Gb), 2} \cap \cS_{(\Gb), 3}^c
\right) \right) \cap \sS_{(\Gb), 3} \cap \sB_{(\Gb)}
\right)
=
0
,
$$
with $\uplambda$ here being the Lebesgue measure over the unconstrained parameter entries.
Relaxing the constraint $\sB_{(\Gb)}$ on the inactive parameter entries, it directly follows that
$$
\uplambda\left(
\law_{\Gb}^{-1}\left( \law_{\Gb}\left(
\cS_{(\Gb), 2} \cap \cS_{(\Gb), 3}^c
\right) \right) \cap \sS_{(\Gb), 3}
\right)
=
0
.
$$

Taking a union bound, we obtain
\begin{align*}
&\qquad
(\uplambda \times \upmu) \left(\left\{
(\Ab, \Bb, \bGamma, \Gb) \in \sS_2:~
\perm_{\sS_2}(\Ab, \Bb, \bGamma, \Gb) \subsetneq \marg_{\sS_2}(\Ab, \Bb, \bGamma, \Gb)
\right\}\right)
\\&\le
(\uplambda \times \upmu) \left(\left\{
(\Ab, \Bb, \bGamma, \Gb) \in \sS_3:~
\perm_{\sS_2}(\Ab, \Bb, \bGamma, \Gb) \subsetneq \marg_{\sS_2}(\Ab, \Bb, \bGamma, \Gb)
\right\}\right)
+
(\uplambda \times \upmu)(\sS_2 \cap \sS_3^c)
\\&\le
\sum_{\Gb} \uplambda\left(\left\{
(\Ab, \Bb, \bGamma) \in \sS_{(\Gb), 3}:~
\marg_{\sS_3}(\Ab, \Bb, \bGamma, \Gb) \subsetneq \marg_{\sS_2}(\Ab, \Bb, \bGamma, \Gb)
\right\}\right)
+
0
\\&\le
\sum_{\Gb} \uplambda\left(\left\{
(\Ab, \Bb, \bGamma) \in \sS_{(\Gb), 3}:~
\law_{\Gb}^{-1}\left(
\{\law_{\Gb}(\Ab, \Bb, \bGamma)\}
\right) \cap (\sS_{(\Gb), 2} \cap \sS_{(\Gb), 3}^c)
\ne
\varnothing
\right\}\right)
\\&\le
\sum_{\Gb} \uplambda\left(
\law_{\Gb}^{-1}\left( \law_{\Gb}\left(
\cS_{(\Gb), 2} \cap \cS_{(\Gb), 3}^c
\right) \right) \cap \sS_{(\Gb), 3}
\right)
\\&=
0
.
\end{align*}
Therefore, our model is generically identifiable within parameter space $\sS_2$.
\end{proof}

\subsection{Proofs of Auxiliary Lemmas for Theorem \ref{theo:generic}}\label{supp_ssec:proof_lemmas_generic}

\begin{proof}[Proof of Lemma \ref{lemm:y_to_w_meas0}]
Let $\Gb$ be an arbitrary $p \times q$ binary matrix satisfying the condition ($*$) of Theorem \ref{theo:generic}.
That is, $\Gb$ contains two distinct submatrices with all-one diagonals, while the remaining submatrix does not have an all-zeros column.
We denote the collections of rows corresponding to the two distinct submatrices with all-one diagonals by $\cI_1$ and $\cI_2$, and denote the collection of remaining rows by $\cI_3$, such that $\cI_1 \cup \cI_2 \cup \cI_3$ forms a partition of $[p]$.
We recall the definition of the probability matrix $\bLambda$ from \eqref{eq:Lambda_st} and the definition of the probability vector $\bnu$ from \eqref{eq:nu_t}.
Similar to \eqref{eq:Lambda[1234]}, we define the probability matrices
\begin{equation}\label{eq:Lambda[123]}
\bLambda_{[1]}
:=
\odot_{i \in \cI_1} \bLambda_{(i)}
,\quad
\bLambda_{[2]}
:=
\odot_{i \in \cI_2} \bLambda_{(i)}
,\quad
\bLambda_{[3]}
:=
\odot_{i \in \cI_3} \bLambda_{(i)}
.
\end{equation}
Using the same proof technique as in Lemma \ref{lemm:y_to_w}, Kruskal's Theorem \ref{theo:kruskal} implies that the probability matrix $\bLambda$ and the probability vector $\bnu$ can be uniquely recovered, up to column permutations, if the sum of the Kruskal ranks of $\bLambda_{[1]}, \bLambda_{[2]}, \bLambda_{[3]}$ satisfies
$$
\rank_K(\bLambda_{[1]}) + \rank_K(\bLambda_{[2]}) + \rank_K(\bLambda_{[3]} \diag(\bnu))
\ge
2 \cdot 2^q + 2
.
$$
Since there are only finitely many possible values of $\Gb$, to prove Lemma \ref{lemm:y_to_w_meas0}, it suffices to show that for each $\Gb$ satisfying condition ($*$) in Theorem \ref{theo:generic}, the Lebesgue measure
$$
\uplambda\left(\left\{
(\Ab, \Bb, \bGamma):~
\rank_K(\bLambda_{[1]}) + \rank_K(\bLambda_{[2]}) + \rank_K(\bLambda_{[3]} \diag(\bnu))
<
2 \cdot 2^q + 2
\right\}\right)
=
0
.
$$

Without loss of generality, we can permute the rows of $\Gb$ and restrict our attention to the case where $\Gb$ satisfies $g_{j, j} = g_{q + j, j} = 1$ for all $j \in [q]$.
This means that the first $q$ rows and the rows from $q + 1$ to $2q$ of $\Gb$ form two submatrices that have all-one diagonals, while the remaining rows belong to $\cI_3 = [2q + 1, p]$.
Recall from \eqref{eq:sB_Gb} that $\sB_{(\Gb)}$ denotes the collection of parameters that satisfy the property that $\forall i \in [p], j \in [q]$, if $g_{i, j} = 1$ then $\beta_{i, j} \ne 0$, and its measure $\uplambda(\sB_{(\Gb)}) = 0$.
Given $\Gb$, we let $\Bb \in \sB_{(\Gb)}$ be an arbitrary such parameter.

Based on $\Gb$, we define a new binary square matrix $\Gb'$ by setting the submatrices of $\Gb$ with all-one diagonals to be identity matrices $\Ib_q$, i.e.
\begin{equation}\label{eq:Gb'}
\Gb'
:=
\left(\begin{matrix}
\Ib_q \\ \Ib_q \\ \Gb_{\cI_3}
\end{matrix}\right)
.
\end{equation}
Furthermore, we define a new parameter $\Bb'$ entrywise as
\begin{equation}\label{eq:Bb'}
\beta_{i, j}'
=
\left\{\begin{matrix}
\beta_{i, j} & \text{if } g_{i, j}' = 1 \\
0 & \text{if } g_{i, j}' = 0
\end{matrix}\right.
\quad
\forall i \in [p], j \in [q]
.
\end{equation}
By definition, we have $\Bb' = (\one ~ \Gb') \circ \Bb$ and $\Bb' \in \sB_{(\Gb')}$.
Importantly, we observe that the conditional distributions satisfy
$$
\law( \yb ~|~ \wb, \Gb, \Bb')
=
\law( \yb ~|~ \wb, \Gb', \Bb')
.
$$
Similar to \eqref{eq:Lambda_st} and \eqref{eq:Lambda[123]}, let the probability matrix $\bLambda'$ characterize the conditional distribution $\law(\yb ~|~ \wb, \Gb', \Bb')$, and let the probability matrices $\bLambda_{[1]}', \bLambda_{[2]}', \bLambda_{[3]}'$ characterize the conditional distribution of the $\yb$ entries in $\cI_1, \cI_2, \cI_3$, respectively, such that $\bLambda'$ coincides with $\bLambda_{[1]}' \odot \bLambda_{[2]}' \odot \bLambda_{[3]}'$ under row permutations.
In the proof of Lemma \ref{lemm:y_to_w}, we have established that for $\Gb_{\cI_1}' = \Gb_{\cI_2}' = \Ib$ and $\Bb' \in \sB_{(\Gb')}$,
$$
\rank_K(\bLambda_{[1]}')
=
\rank_K(\bLambda_{[2]}')
=
2^q
.
$$

We further observe that $\rank_K(\bLambda_{[1]}) = 2^q$ if and only if $\det(\bLambda_{[1]}) \ne 0$.
From our model formulations in \eqref{eq:model_y_w}, \eqref{eq:model_w_z}, and \eqref{eq:model_z_x}, $\det(\bLambda_{[1]})$ is a composition of rational, exponential, and polynomial functions of the continuous parameters $\Ab, \Bb, \bGamma$ and hence a holomorphic function.
We denote this holomorphic function as $\psi_{\Gb, [1]}(\Ab, \Bb, \bGamma)$, with the subscript $\Gb$ emphasizing its dependency on $\Gb$.
Similarly, we denote the holomorphic function $\det(\bLambda_{[2]})$ by $\psi_{\Gb, [2]}(\Ab, \Bb, \bGamma)$.
Since we have established that $\rank_K(\bLambda_{[1]}') = \rank_K(\bLambda_{[2]})' = 2^q$ for $\Gb', \Bb'$ defined in \eqref{eq:Gb'} and \eqref{eq:Bb'}, it follows that
$$
\psi_{\Gb, [1]}(\Ab, \Bb', \bGamma)
=
\psi_{\Gb', [1]}(\Ab, \Bb', \bGamma)
\ne
0
,
$$
$$
\psi_{\Gb, [2]}(\Ab, \Bb', \bGamma)
=
\psi_{\Gb', [2]}(\Ab, \Bb', \bGamma)
\ne
0
.
$$
Thus, the holomorphic functions $\psi_{\Gb, 1}$ and $\psi_{\Gb, 2}$ are not identically zero.
Applying Theorem \ref{theo:holo_zero}, we conclude that the zero sets of $\psi_{\Gb, 1}$ and $\psi_{\Gb, 2}$ have Lebesgue measure zero, which implies
$$
\uplambda\left(\left\{
(\Ab, \Bb, \bGamma):~
\rank_K(\bLambda_{[1]})
<
2^q
\right\}\right)
\le
\uplambda\left(
\psi_{\Gb, [1]}^{-1}(\{0\})
\right)
+
\uplambda\left(
\sB_{(\Gb)}
\right)
=
0
,
$$
$$
\uplambda\left(\left\{
(\Ab, \Bb, \bGamma):~
\rank_K(\bLambda_{[2]})
<
2^q
\right\}\right)
\le
\uplambda\left(
\psi_{\Gb, [2]}^{-1}(\{0\})
\right)
+
\uplambda\left(
\sB_{(\Gb)}
\right)
=
0
.
$$

Since the submatrix $\Gb_{\cI_3}$ has no all-zeros column, for each $j \in [q]$ there exists a row $i \in \cI_3$ such that $g_{i, j} = 1$.
We denote this row $i \in \cI_3$ corresponding to $j$ by $i_j$ to emphasize the dependency.
Recall from the proof of Lemma \ref{lemm:y_to_w} that $\rank_K(\bLambda_{[3]} \diag(\bnu)) \ge 2$ as long as $\bLambda_{[3]}$ has distinct columns.
This, in turn, holds if for any $\wb \ne \wb' \in \{0, 1\}^q$, there exists $i \in \cI_3$ such that
$$
\PP(y_i = 1 ~|~ \wb, \gb_i, \beta_{i, 0}, \bbeta_i)
\ne
\PP(y_i = 1 ~|~ \wb', \gb_i, \beta_{i, 0}, \bbeta_i)
,
$$
or equivalently,
$$
(\wb - \wb')^\top (\gb_i \circ \bbeta_i)
\ne
0
.
$$
We define a collection of parameters as
$$
\sC_{(\Gb)}
:=
\left\{
\Bb:~
\exists \Cb \in \{-1, 0, 1\}^{p \times q},~
\Cb \ne \zero,~
\text{s.t. }
\tr\left(
(\zero ~ \Gb)^\top \Bb
\right)
=
\sum_{i = 1}^p \sum_{j = 1}^q C_{i, j} \beta_{i, j}
=
0
\right\}
.
$$
Intuitively, $\sC_{(\Gb)}$ includes all parameters $\Bb$ satisfying the existence of two disjoint subsets $\cJ_1, \cJ_2$ of index pairs $(i, j) \in [p] \times [q]$ such that $\sum_{(i, j) \in \cJ_1} \beta_{i, j} = \sum_{(i, j) \in \cJ_2} \beta_{i, j}$.
For any $\wb \ne \wb' \in \{0, 1\}^q$, there exists $j \in [q]$ such that $w_j - w_j' \ne 0$, which implies
$$
(\wb - \wb')^\top (\gb_{i_j} \circ \bbeta_{i_j})
\ne
0
,\quad
\forall \Bb \notin \sC_{(\Gb)}
.
$$
Since there are only finitely many possible values of $\Cb \in \{-1, 0, 1\}^{p \times q}$ and each set $\{\tr((\zero ~ \Gb)^\top \Bb) = 0\}$ forms a linear subspace, we conclude that the Lebesgue measure $\uplambda(\sC_{(\Gb)}) = 0$.
This further implies
\begin{align*}
&\qquad
\uplambda\left(\left\{
(\Ab, \Bb, \bGamma):~
\rank_K(\bLambda_{[3]} \diag(\bnu))
<
2
\right\}\right)
\\&\le
\uplambda\left(\left\{
(\Ab, \Bb, \bGamma):~
\exists \wb \ne \wb' \in \{0, 1\}^q,~
\text{s.t. }
\forall i \in \cI_3,~
(\wb - \wb')^\top (\gb_i \circ \bbeta_{i, j}) = 0
\right\}\right)
\\&\le
\uplambda\left(\left\{
(\Ab, \Bb, \bGamma):~
\Bb \in \sC_{(\Gb)}
\right\}\right)
\\&=
0
.
\end{align*}

Putting everything together, we obtain
\begin{align*}
&\qquad
\uplambda\left(\left\{
(\Ab, \Bb, \bGamma):~
\rank_K(\bLambda_{[1]}) + \rank_K(\bLambda_{[2]}) + \rank_K(\bLambda_{[3]} \diag(\bnu))
<
2 \cdot 2^q + 2
\right\}\right)
\\&\le
\uplambda\left(\left\{
(\Ab, \Bb, \bGamma):~
\rank_K(\bLambda_{[1]})
<
2^q
\right\}\right)
+
\uplambda\left(\left\{
(\Ab, \Bb, \bGamma):~
\rank_K(\bLambda_{[2]})
<
2^q
\right\}\right)
\\&\qquad+
\uplambda\left(\left\{
(\Ab, \Bb, \bGamma):~
\rank_K(\bLambda_{[3]} \diag(\bnu))
<
2
\right\}\right)
\\&=
0
.
\end{align*}
The above result holds for an arbitrary choice of $\Gb$ satisfying condition (*) in Theorem \ref{theo:generic}.
Taking a union bound over all such choices of $\Gb$ completes the proof of the lemma.

We can now construct the subset $\sS_{2, 1} \subset \sS_2$ of parameters to be excluded as
\begin{equation}\label{eq:sS21}
\sS_{2, 1}
:=
\left\{
(\Ab, \Bb, \bGamma, \Gb) \in \sS_2:~
\rank_K(\bLambda_{[1]}) + \rank_K(\bLambda_{[2]}) + \rank_K(\bLambda_{[3]} \diag(\bnu))
<
2 \cdot 2^q + 2
\right\}
,
\end{equation}
which has measure $\uplambda \times \upmu(\sS_{2, 1}) = 0$.
Furthermore, the permutation invariance of $\sS_{2, 1}$ follows directly from the fact that permuting columns of $\bLambda_{[1]}$, $\bLambda_{[2]}$, and $\bLambda_{[3]}$ does not affect their Kruskal ranks.
\end{proof}

\begin{proof}[Proof of Lemma \ref{lemm:w_to_z_meas0}]
We define a subset of the parameter space $\sS_2$ as
\begin{equation}\label{eq:sS22}\begin{aligned}
\sS_{2, 2}
:=
\Big\{
(\Ab, \Bb, \bGamma, \Gb) \subset \sS_2:~
&
\forall \text{ partition } [q] = \cQ_1 \cup \cQ_2 \cup \cQ_3 \text{ with } |\cQ_1|, |\cQ_2| \ge \log_2 D
,\\&
\rank\left(\bigcirc_{j \in \cQ_1} \{\one, \balpha_j\}\right) < d
,\text{ or }
\rank\left(\bigcirc_{j \in \cQ_2} \{\one, \balpha_j\}\right) < d
,\\&\text{or }
\exists h_1 \ne h_2 \in [d] \text{ s.t. } \forall j \in \cQ_3,~ \alpha_{j, h_1} = \alpha_{j, h_2}
\Big\}
.
\end{aligned}\end{equation}
Since $\card(\cQ_1), \card(\cQ_2) \ge \log_2 d$, the conditions
$$
\rank\left(\bigcirc_{j \in \cQ_1} \{\one, \balpha_j\}\right)
<
d
,\quad
\rank\left(\bigcirc_{j \in \cQ_2} \{\one, \balpha_j\}\right)
<
d
$$
imply that either the set of vectors $\bigcirc_{j \in \cQ_1} \{\one, \balpha_j\}$ or the set of vectors $\bigcirc_{j \in \cQ_2} \{\one, \balpha_j\}$ does not have full rank.
Additionally, the condition
\begin{equation}\label{eq:cQ3_distinct_cols}
\exists h_1 \ne h_2 \in [d]
\text{ s.t. }
\forall j \in \cQ_3,~ \alpha_{j, h_1} = \alpha_{j, h_2}
\end{equation}
implies that the $h_1, h_2$ columns of the submatrix $\Ab_{\cQ_3}$ are identical.
Thus, the subset $\sS_{2, 2}$ includes all corner cases of parameters that the condition of Lemma \ref{lemm:w_to_z} excludes.
By applying Lemma \ref{lemm:w_to_z}, we conclude that the statement in Lemma \ref{lemm:w_to_z_meas0} holds within the parameter space $\sS_2 \cap \sS_{2, 2}^c$.
By its definition \eqref{eq:sS22}, $\sS_{2, 2}$ is permutation invariant, so it only remains to prove that it has measure zero.

We consider an arbitrary partition $[q] = \cQ_1 \cup \cQ_2 \cup \cQ_3$ with cardinalities $|\cQ_1|, |\cQ_2| \ge \log_2 d$.
Let $\Cb$ denote the $2^{\card(\cQ_1)} \times d$ matrix with rows consisting of the vectors in $\bigcirc_{j \in \cQ_1} \{\one, \balpha_j\}$.
We define a function of the parameter $\Ab$ as
\begin{equation}\label{eq:phi_Ab}
\phi(\Ab)
:=
\sum_{\cJ \subset \cQ_1, \card(\cJ) = D} \det(\Cb_{\cJ})^2
,
\end{equation}
where $\cJ$ is summed over all subsets of $\cQ_1$ with cardinality $d$, and $\Cb_{\cJ}$ denotes the $d \times d$ submatrix of $\Cb$ corresponding to the rows in $\cJ$.
From linear algebra, we have
\begin{equation}\label{eq:rank_to_phi}
\rank\big(\bigcirc_{j \in \cQ_1} \{\one, \balpha_j\}\big) < d
\quad\iff\quad
\phi(\Ab) = 0
.
\end{equation}

By its definition \eqref{eq:phi_Ab}, $\phi(\Ab)$ is a holomorphic function of $\Ab$.
We now verify that $\phi(\Ab)$ is not identically zero by explicitly constructing an example.
Let $\{p_j\}_{j \in \cQ_1}$ be $\card(\cQ_1)$ distinct prime numbers.
For each $j \in \cQ_1$, we specify the $j$th row of $\Ab$ as
$$
\balpha_j
:=
\left(\begin{matrix}
p_j^{-1} & p_j^{-2} & \cdots & p_j^{-d}
\end{matrix}\right)
.
$$
Then, for each $r \in [2^{\card(\cQ_1)}]$, there exists a distinct binary vector $\vb \in \{0, 1\}^{\card(\cQ_1)}$ such that for $a_r := \prod_{j \in \cQ_1} p_j^{-v_j}$, the $r$th row of $\Cb$ is given by
$$
\cbb_r
=
\left(\begin{matrix}
a_r & a_r^2 & \cdots & a_r^d
\end{matrix}\right)
.
$$
By the nature of prime numbers, we observe that $a_r \ne a_{r'}$ for any two rows $r \ne r'$.
Therefore, for any subset $\cJ \subset \cQ_1$ with $\card(\cJ) = d$, the submatrix $\Cb_{\cJ}$ forms a Vandermonde matrix, which has determinant
$$
\det(\Cb_{\cJ})
=
\prod_{r, r' \in \cJ, r < r'} (a_r - a_{r'})
\ne
0
.
$$
This suggests that $\phi(\Ab) \ne 0$ for this specific construction of $\Ab$.

Since $\phi(\Ab)$ is a holomorphic function that is not identically zero, we can apply Theorem \ref{theo:holo_zero} to conclude that its zero set has Lebesgue measure zero.
Using \eqref{eq:rank_to_phi}, this implies
$$
\uplambda\left(\left\{
\Ab:~
\rank\left(
\bigcirc_{j \in \cQ_1} \{\one, \balpha_j\}
\right)
<
d
\right\}\right)
=
\uplambda\left(
\phi^{-1}(0)
\right)
=
0
.
$$
Similarly, for the rows in $\cQ_2$, we obtain
$$
\uplambda\left(\left\{
\Ab:~
\rank\left(
\bigcirc_{j \in \cQ_2} \{\one, \balpha_j\}
\right)
<
d
\right\}\right)
=
0
.
$$
Furthermore, we observe that the parameters satisfying the condition \eqref{eq:cQ3_distinct_cols} form a finite union of linear subspaces and hence have measure zero.
Since the number of possible partitions $[q] = \cQ_1 \cup \cQ_2 \cup \cQ_3$ is finite, applying a union bound yields $(\uplambda \times \upmu)(\sS_{2, 2}) = 0$.
\end{proof}

\begin{proof}[Proof of Lemma \ref{lemm:law_to_params_meas0}]
We define a subset of parameter space $\sS_2$ as
\begin{equation}\label{eq:sS23}
\sS_{2, 3}
:=
\left\{
(\Ab, \Bb, \bGamma, \Gb) \in \sS_2:~
\exists i \in [p], j \in [q]
\text{ s.t. }
g_{i, j} = 1
\text{ and }
\beta_{i, j} = 0
\right\}
,
\end{equation}
which, by definition, is permutation invariant.
We observe that $\sS_{2, 3}$ contains all corner cases of parameters that we are excluding in Lemma \ref{lemm:law_to_params}.
Therefore, by applying Lemma \ref{lemm:law_to_params}, we conclude that the statements in Lemma \ref{lemm:law_to_params_meas0} hold within parameter space $\sS_2 \cap \sS_{2, 3}^c$.
It remains to show that $\sS_{2, 3}$ has measure zero.

For any fixed $\Gb \in \{0, 1\}^q$, the collection of parameters $\Bb$ where an entry $\beta_{i, j} = 0$ with $g_{i, j} = 1$ forms a finite union of linear subspaces, each of which has zero Lebesgue measure.
Taking a union bound over the finitely many possible values of $\Gb$, we obtain 
$(\uplambda \times \upmu)(\sS_{2, 3}) = 0$.
\end{proof}

\section{Posterior Consistency}\label{supp_sec:post_cons}

In this section, we develop the posterior consistency theory for our model, drawing on both Schwartz' theorem \citep{schwartz1965bayes} and Doob's theorem \citep{doob1949application}.
In Section \ref{supp_ssec:proof_post}, we provide the proof of Theorem \ref{theo:post}, systematically addressing the three scenarios considered in Theorem \ref{theo:post}.
Proofs of the auxiliary lemmas introduced in Section \ref{supp_ssec:proof_post} are presented in Section \ref{supp_ssec:proof_lemmas_post}.

\subsection{Proof of Theorem \ref{theo:post}}\label{supp_ssec:proof_post}

In the following, we establish parts (i) and (ii) of Theorem \ref{theo:post} using Schwartz' theory of posterior consistency, while part (iii) is proven using Doob's theory of posterior consistency.

Let $\cS^{2^p - 1} \subset \RR^{2^p}$ denote the probability simplex to which the distribution of $\yb$ belongs.
Recall that $\xb$ follows the data-generating distribution $\cP_x$.
For a parameter space $\sS$, we define the mapping $\cL:~ \sS \to \cS^{2^p - 1}$ that associates each set of parameters $(\Ab, \Bb, \bGamma, \Gb)$ with the corresponding marginal distribution $\law(\yb ~|~ \Ab, \Bb, \bGamma, \Gb)$.
By our model formulation in \eqref{eq:model_y_w}, \eqref{eq:model_w_z}, and \eqref{eq:model_z_x}, the mapping from $(\Ab, \Bb, \bGamma, \Gb)$ to the conditional distribution $\law(\yb ~|~ \xb, \Ab, \Bb, \bGamma, \Gb)$ is continuous, and hence $\cL$ is also continuous.
Given a prior measure $\pi$ on $\sS$, we denote by $\cL_\# \pi$ the pushforward measure induced by $\cL$ on $\cS^{2^p - 1}$.
For a distribution $P \in \cS^{2^p - 1}$ and a probability measure $m$ over $\cS^{2^p - 1}$, we say that $\cP$ is in the \textit{Kullback-Leibler (KL) support} of $m$ if every KL neighborhood of $P$ has positive measure under $m$, i.e.
$$
m\left\{
P' \in \cS^{2^p - 1}:~
D_{KL}(P, P')
:=
P\left[
\log \frac{\ud P}{\ud P'}
\right] < \epsilon
\right\}
>
0
,\quad
\forall \epsilon > 0
.
$$
We define the $\epsilon$-neighborhoods of $P$ in $\cS^{2^p - 1}$ under $L_1$ distance and $KL$ divergence as follows:
\begin{equation}\label{eq:eps_nbhd}
\cO_\epsilon(P)
:=
\left\{
P' \in \cS^{2^p - 1}:~
\|P - P'\|_1
<
\epsilon
\right\}
,\quad
\cO_\epsilon^{KL}(P)
:=
\left\{
P' \in \cS^{2^p - 1}:~
P\left[
\log \frac{P}{P'}
\right]
<
\epsilon
\right\}
.
\end{equation}
As a key step in applying Schwartz' theory of posterior consistency, the following lemma establishes that the distribution $\cL(\Ab, \Bb, \bGamma, \Gb)$ belongs to the KL support of $\cL_\# \pi$ whenever $(\Ab, \Bb, \bGamma, \Gb)$ lies in the support of $\pi$.

\begin{lemma}\label{lemm:prior_supp}
Suppose that the data-generating distribution $\cP_x$ on $\RR^{p_x}$ has finite first moment, that is, $\int_{\RR^{p_x}} \|\xb\|_\infty \cP_x(\ud \xb) < \infty$.
Then, for any $(\Ab, \Bb, \bGamma, \Gb) \in \supp(\pi) \subset \sS$, the distribution $\cL(\Ab, \Bb, \bGamma, \Gb)$ belongs to the KL support of $\cL_\# \pi$.
\end{lemma}

By applying Schwartz' theorem on posterior consistency \citep{schwartz1965bayes, ghosal2017fundamentals}, we establish the consistency of $\cL(\Ab, \Bb, \bGamma, \Gb)$ within the probability simplex $\cS^{2^p - 1}$.
The remaining task is to recover the consistency of the parameters $(\Ab, \Bb, \bGamma, \Gb)$ from this result.
Since $\cL$ is continuous, this follows naturally if the parameter space $\sS$ is compact.
We now proceed with the proof of (i) in Theorem \ref{theo:post}.

\begin{proof}[Proof of Theorem \ref{theo:post} (i)]
We select an arbitrary true parameter $(\Ab^*, \Bb^*, \bGamma^*, \Gb^*) \in \supp(\pi)$ that satisfies the identifiability condition $\perm_{\sS}(\Ab^*, \Bb^*, \bGamma^*, \Gb^*) = \marg_{\sS}(\Ab^*, \Bb^*, \bGamma^*, \Gb^*)$ within the parameter space $\sS$.
By Lemma \ref{lemm:prior_supp}, the corresponding marginal distribution $\law(\yb ~|~ \Ab^*, \Bb^*, \bGamma^*, \Gb^*)$, characterized by the probability vector $\cL(\Ab^*, \Bb^*, \bGamma^*, \Gb^*) \in \cS^{2^p - 1}$, lies in the KL support of the pushforward prior measure $\cL_\# \pi$.

Since the weak topology on the space of probability measures over a countable sample space, e.g. the probability simplex $\cS^{2^p - 1}$, can be equivalently induced by the $L_1$ metric, we apply Schwartz' theorem (see Theorem 6.16 in \citet{ghosal2017fundamentals}) to conclude that the posterior distribution over $\cS^{2^p - 1}$ is strongly consistent at $\cL(\Ab^*, \Bb^*, \bGamma^*, \Gb^*)$.
That is, for any $\delta > 0$, we have
$$
\lim_{N \to \infty} \PP\Big(
\cO_\delta(\cL(\Ab^*, \Bb^*, \bGamma^*, \Gb^*))^c
~\Big|~
\xb^{(1:N)}, \yb^{(1:N)}
\Big)
=
0
,
$$
almost surely with respect to $\PP(\xb^{(1:\infty)}, \yb^{(1:\infty)} ~|~ \Ab^*, \Bb^*, \bGamma^*, \Gb^*)$.

Let $\epsilon > 0$ be arbitrary.
Since the parameter space $\sS$ is compact, the set $\cO_\epsilon(\Ab^*, \Bb^*, \bGamma^*, \Gb^*)^c$ is compact in $\sS$.
Moreover, since $\cL$ is continuous, its image $\cL(\cO_\epsilon(\Ab^*, \Bb^*, \bGamma^*, \Gb^*)^c)$ is also compact in $\cS^{2^p - 1}$.
Given that $\perm_{\sS}(\Ab^*, \Bb^*, \bGamma^*, \Gb^*) = \marg_{\sS}(\Ab^*, \Bb^*, \bGamma^*, \Gb^*)$, we have the disjointness
$$
\cL(\cO_\epsilon(\Ab^*, \Bb^*, \bGamma^*, \Gb^*)^c)
\cap
\cL(\Ab^*, \Bb^*, \bGamma^*, \Gb^*)
=
\varnothing
.
$$
This implies that the compact set $\cL(\cO_\epsilon(\Ab^*, \Bb^*, \bGamma^*, \Gb^*)^c)$ is separated from $\cL(\Ab^*, \Bb^*, \bGamma^*, \Gb^*)$ by some positive distance $\delta > 0$, i.e.
$$
\cL(\cO_\epsilon(\Ab^*, \Bb^*, \bGamma^*, \Gb^*)^c)
\subset
\cO_\delta(\cL(\Ab^*, \Bb^*, \bGamma^*, \Gb^*))^c
.
$$
Consequently, we obtain
\begin{align*}
&\qquad
\lim_{N \to \infty} \PP\Big(
\cO_\epsilon(\Ab^*, \Bb^*, \bGamma^*, \Gb^*)^c
~\Big|~
\xb^{(1:N)}, \yb^{(1:N)}
\Big)
\\&\le
\lim_{N \to \infty} \PP\Big(
\cO_\delta(\cL(\Ab^*, \Bb^*, \bGamma^*, \Gb^*))^c
~\Big|~
\xb^{(1:N)}, \yb^{(1:N)}
\Big)
=
0
,
\end{align*}
almost surely under $\PP(\xb^{(1:\infty)}, \yb^{(1:\infty)} ~|~ \Ab^*, \Bb^*, \bGamma^*, \Gb^*)$.
Thus, the posterior distribution is strongly consistent at $(\Ab^*, \Bb^*, \bGamma^*, \Gb^*)$.
\end{proof}

Compared to (i) in Theorem \ref{theo:post}, part (ii) removes the requirement that the true parameter $(\Ab^*, \Bb^*, \bGamma^*, \Gb^*)$ be identifiable within the parameter space $\sS \subset \sS_1$.
This relaxation is justified by Theorem \ref{theo:strict}, which already ensures that all parameters in $\sS_1$ are identifiable.
Furthermore, although $\sS_1$ is not a compact set, the compactness of $\sS \subset \sS_1$ is no longer necessary for establishing posterior consistency.
The following lemma clarifies this by demonstrating that $\sS_1$ retains properties that make it close to being compact, thus enabling a similar argument to that used in the proof of part (i).

A set $\sB$ of parameters $(\Ab, \Bb, \bGamma, \Gb)$ is said to be \textit{actively compact} if it constrains all active parameters within a compact set.
Formally, this means that the set
\begin{equation}\label{eq:actively_compact}
\Big\{
(\Ab, \bGamma, (\one ~ \Gb) \circ \Bb):~
(\Ab, \Bb, \bGamma, \Gb) \in \sB
\Big\}
\end{equation}
is compact.
Intuitively, this notion of active compactness allows us to focus on the topological properties of the parameters that influence the model's behavior while ignoring the potential unboundedness introduced by inactive parameters.
Specifically, if some $\beta_{i, j}$ values become arbitrarily large while the corresponding $g_{i, j} = 0$, they do not contribute to the model's effective structure.

\begin{lemma}\label{lemm:S1_almost_compact}
Suppose the data generating distribution $\cP_x$ has full rank, i.e. $\forall \vb \in \RR^{p_x + 1}$, $\vb \ne \zero$, we have $\{\xb:~ \vb^\top (1, \xb) \ne 0\} \cap \supp(\cP_x) \ne \varnothing$.
Then, for any parameter $(\Ab, \Bb, \bGamma, \Gb) \in \sS_1$, there exists $\delta > 0$ and an actively compact set $\sB \subset \sS_1$ such that
$$
\cL(\sS_1 \cap \sB^c) \cap \cO_\delta\left(
\cL(\Ab, \Bb, \bGamma, \Gb)
\right)
=
\varnothing
.
$$
\end{lemma}

At a high level, Lemma \ref{lemm:S1_almost_compact} implies that the non-compact portion of $\sS_1$, namely $\sS_1 \cap \sB^c$, remains sufficiently distant from the marginal distribution $\cL(\Ab, \Bb, \bGamma, \Gb)$ within the probability simplex $\cS^{2^p - 1}$.
This separation ensures that, when analyzing $\cO_\delta$ neighborhoods in $\cS^{2^p - 1}$ to establish posterior consistency, the parameter space $\sS_1$ is effectively equivalent to the actively compact set $\sB$ under the pullback mapping $\cL^{-1}$.

With the topological properties of $\sS_1$ established in Lemma \ref{lemm:S1_almost_compact}, we now proceed with the proof of Theorem \ref{theo:post} (ii), in a manner similar to the proof of Theorem \ref{theo:post} (i).

\begin{proof}[Proof of Theorem \ref{theo:post} (ii)]
For an arbitrary true parameter $(\Ab^*, \Bb^*, \bGamma^*, \Gb^*) \in \supp(\pi)$, Lemma \ref{lemm:prior_supp} and Schwartz' theorem \citep{ghosal2017fundamentals} together ensure that the posterior distribution over the probability simplex $\cS^{2^p - 1}$ is strongly consistent at $\cL(\Ab^*, \Bb^*, \bGamma^*, \Gb^*)$.
That is, for any $\delta > 0$, we have
$$
\lim_{N \to \infty} \PP\Big(
\cO_\delta(\cL(\Ab^*, \Bb^*, \bGamma^*, \Gb^*))^c
~\Big|~
\xb^{(1:N)}, \yb^{(1:N)}
\Big)
=
0
,
$$
almost surely under $\PP(\xb^{(1:\infty)}, \yb^{(1:\infty)} ~|~ \Ab^*, \Bb^*, \bGamma^*, \Gb^*)$.

Applying Lemma \ref{lemm:S1_almost_compact}, we establish the existence of $\delta^* > 0$ and a compact set $\sB \subset \sS_1$ such that
$$
\cL(\sS_1 \cap \sB^c) \cap \cO_{\delta^*}(\cL(\Ab^*, \Bb^*, \bGamma^*, \Gb^*))
=
\varnothing
.
$$
Let $\epsilon > 0$ be arbitrary.
Since $\sB$ is actively compact, the set $\cO_\epsilon(\Ab^*, \Bb^*, \bGamma^*, \Gb^*)^c \cap \sB$ is actively compact in $\sS_1$.
Moreover, since $\cL$ is a continuous function of $(\Ab, \Bb, \bGamma, \Gb)$ that depends only on the active parameters $\Ab, \bGamma, (\one ~ \Gb) \circ \Bb$, it follows that $\cL(\cO_\epsilon(\Ab^*, \Bb^*, \bGamma^*, \Gb^*)^c \cap \sB)$ is compact in $\cS^{2^p - 1}$.
Given that $\perm_{\sS}(\Ab^*, \Bb^*, \bGamma^*, \Gb^*) = \marg_{\sS}(\Ab^*, \Bb^*, \bGamma^*, \Gb^*)$, we obtain
$$
\cL(\cO_\epsilon(\Ab^*, \Bb^*, \bGamma^*, \Gb^*)^c \cap \sB) \cap \cL(\Ab^*, \Bb^*, \bGamma^*, \Gb^*)
=
\varnothing
.
$$
This separation ensures the existence of $\delta > 0$ such that
$$
\cL(\cO_\epsilon(\Ab^*, \Bb^*, \bGamma^*, \Gb^*)^c \cap \sB)
\subset
\cO_\delta(\cL(\Ab^*, \Bb^*, \bGamma^*, \Gb^*))^c
.
$$
Thus, we obtain
\begin{align*}
&\qquad
\lim_{N \to \infty} \PP\Big(
\cO_\epsilon(\Ab^*, \Bb^*, \bGamma^*, \Gb^*)^c
~\Big|~
\xb^{(1:N)}, \yb^{(1:N)}
\Big)
\\&\le
\lim_{N \to \infty} \PP\Big(
\cO_\delta(\cL(\Ab^*, \Bb^*, \bGamma^*, \Gb^*))^c \cup (\sS_1 \cap \sB^c)
~\Big|~
\xb^{(1:N)}, \yb^{(1:N)}
\Big)
\\&\le
\lim_{N \to \infty} \PP\Big(
\cO_{\min\{\delta, \delta^*\}}(\cL(\Ab^*, \Bb^*, \bGamma^*, \Gb^*))^c
~\Big|~
\xb^{(1:N)}, \yb^{(1:N)}
\Big)
\\&=
0
.
\end{align*}
This result establishes that the posterior distribution is strongly consistent at $(\Ab^*, \Bb^*, \bGamma^*, \Gb^*)$.
\end{proof}

To prove (iii) of Theorem \ref{theo:post}, we apply Doob's theory of posterior consistency \citep{doob1949application}.
For a modern treatment of Doob's theory, see Section 6.2 of \citet{ghosal2017fundamentals}.
A key step in Doob's theory is to establish the measurability of the parameters $(\Ab, \Bb, \bGamma, \Gb)$ under the $\sigma$-algebra $\sigma(\{\xb^{(n)}, \yb^{(n)}\}_{n = 1}^\infty)$, which is typically proven using the identifiability of the model.
However, in our setting, the parameters are inherently non-identifiable due to permutation invariances.

To address this issue, we introduce the quotient parameter space, which consists of equivalence classes of parameters under these permutations.
By considering the posterior distribution over this quotient space, we eliminate the ambiguity introduced by permutations, thereby allowing us to establish the required measurability conditions necessary for applying Doob's theorem.
We formally define the quotient parameter space as follows.

Within the parameter space $\sS$, we observe that permutation of parameters induces an equivalence relation.
Specifically, for all $(\Ab, \Bb, \bGamma, \Gb), (\Ab', \Bb', \bGamma', \Gb') \in \sS$, we define
\begin{equation}\label{eq:perm_equiv}
(\Ab, \Bb, \bGamma, \Gb)
\sim
(\Ab', \Bb', \bGamma', \Gb')
\quad
\iff
\quad
\perm_{\sS}(\Ab, \Bb, \bGamma, \Gb)
=
\perm_{\sS}(\Ab', \Bb', \bGamma', \Gb')
.
\end{equation}
For simplicity, we denote the equivalence class of $(\Ab, \Bb, \bGamma, \Gb)$ by $\perm_{\sS}(\Ab, \Bb, \bGamma, \Gb)$ as well.
Letting $\perm(\sS) := \{\perm_{\sS}(\Ab, \Bb, \bGamma, \Gb):~ (\Ab, \Bb, \bGamma, \Gb) \in \sS\}$ represent the space of all parameter equivalence classes, the mapping
$$
\perm_{\sS}:~
(\Ab, \Bb, \bGamma, \Gb)
\to
\perm_{\sS}(\Ab, \Bb, \bGamma, \Gb)
$$
is a surjective mapping from $\sS$ onto $\perm(\sS)$, thereby defining a quotient topology on $\perm(\sS)$.

When the parameter space $\sS \subset \sS_2$, we recall that $\sS$ is said to be permutation invariant if $\perm_{\sS_2}(\Ab, \Bb, \bGamma, \Gb) \subset \sS$ for all $(\Ab, \Bb, \bGamma, \Gb) \in \sS$.
That is, if a set of parameters $(\Ab, \Bb, \bGamma, \Gb)$ belongs to $\sS$, then all of its permutations (over the dimensions of $\wb$ and the classes of $z$) must also belong to $\sS$.
The following lemma shows that for a permutation invariant $\sS$, the quotient topology on $\perm(\sS)$ is equivalent to the topology induced by the natural metric
\begin{equation}\label{eq:perm_equiv_metric}\begin{aligned}
&\qquad
\left\|
\perm_{\sS}(\Ab, \Bb, \bGamma, \Gb) - \perm_{\sS}(\Ab', \Bb', \bGamma', \Gb')
\right\|_1
\\&:=
\inf_{(\tilde\Ab, \tilde\Bb, \tilde\bGamma, \tilde\Gb) \in \perm_{\sS}(\Ab, \Bb, \bGamma, \Gb)}
\inf_{(\tilde\Ab', \tilde\Bb', \tilde\bGamma', \tilde\Gb') \in \perm_{\sS}(\Ab', \Bb', \bGamma', \Gb')}
\left\|
(\tilde\Ab, \tilde\Bb, \tilde\bGamma, \tilde\Gb) - (\tilde\Ab', \tilde\Bb', \tilde\bGamma', \tilde\Gb')
\right\|_1
.
\end{aligned}\end{equation}

\begin{lemma}\label{lemm:quotient_metric}
Let $\sS \subset \sS_2$ be a permutation invariant parameter space.
Then, $\|\cdot\|_1$ in \eqref{eq:perm_equiv_metric} defines a metric on $\perm(\sS)$ that induces its quotient topology.
\end{lemma}

As a consequence of Lemma \ref{lemm:quotient_metric}, we conclude that $\perm(\sS)$ is a Polish space.
This allows us to proceed with the proof of Theorem \ref{theo:post} (iii) using Doob's theory of posterior consistency.

\begin{proof}[Proof of Theorem \ref{theo:post} (iii)]
By Theorem \ref{theo:generic}, the subset of parameters $(\Ab, \Bb, \bGamma, \Gb)$ within $\sS_2$ for which $\perm_{\sS_2}(\Ab, \Bb, \bGamma, \Gb) \subsetneq \marg_{\sS_2}(\Ab, \Bb, \bGamma, \Gb)$ constitutes a zero-measure set.
We recall the definitions of the subsets $\sS_{2, 1}, \sS_{2, 2}, \sS_{2, 3}, \sS_{2, 4}$ from \eqref{eq:sS21}, \eqref{eq:sS22}, \eqref{eq:sS23}, and \eqref{eq:sS24} in the proof of Theorem \ref{theo:generic}.
Their union, $\sS_{2, 1} \cup \sS_{2, 2} \cup \sS_{2, 3} \cup \sS_{2, 4}$, is Borel measurable with measure zero, and ensures that for the restricted parameter space $\sS_3 = \sS_2 \cap (\sS_{2, 1} \cup \sS_{2, 2} \cup \sS_{2, 3} \cup \sS_{2, 4})^c$, every equivalence class $\perm_{\sS_3}(\balpha, \bbeta, \bgamma, \Gb)$ in the quotient space $\perm(\sS_3)$ is identifiable.

From our model formulation, the mapping from $(\Ab, \Bb, \bGamma, \Gb)$ to the marginal distribution $\law(\yb ~|~ \Ab, \Bb, \bGamma, \Gb)$ is Borel measurable.
By Lemma \ref{lemm:quotient_metric} and the lifting property of quotient topology, it follows that the mapping from the equivalence class $\perm_{\sS_3}(\Ab, \Bb, \bGamma, \Gb)$ to the marginal distribution $\law(\yb ~|~ \Ab, \Bb, \bGamma, \Gb)$ is also Borel measurable.
Since the prior $\pi$ on $\sS_2$ is absolutely continuous with respect to the base measure $\uplambda \times \upmu$, we can, without loss of generality, restrict it to $\sS_3$.
This restriction corresponds to a pushforward measure on $\perm(\sS_3)$ denoted by $\perm_\# \pi$.
Applying Proposition 6.10 in \citet{ghosal2017fundamentals}, we conclude that $\perm_{\sS_3}(\Ab, \Bb, \bGamma, \Gb)$ is measurable with respect to the $\sigma$-algebra $\sigma(\{\xb^{(n)}, \yb^{(n)}\}_{n = 1}^\infty)$.

By applying Doob's theorem (see Theorem 6.9 in \citet{ghosal2017fundamentals}), we obtain that the pushforward posterior distribution $\law(\perm_{\sS_3}(\Ab, \Bb, \bGamma, \Gb) ~|~ \xb^{(1:N)}, \yb^{(1:N)})$ is strongly consistent at the equivalence class of the true parameter $\perm_{\sS_3}(\Ab^*, \Bb^*, \bGamma^*, \Gb^*)$ for $\pi$-almost every $(\Ab^*, \Bb^*, \bGamma^*, \Gb^*) \in \sS_3$.
To link this result to posterior consistency in the original parameter space, we note that
\begin{align*}
\tilde\cO_\epsilon(\Ab^*, \Bb^*, \bGamma^*, \Gb^*)
&=
\bigcup_{(\Ab', \Bb', \bGamma', \Gb') \in \perm_{\sS_3}(\Ab^*, \Bb^*, \bGamma^*, \Gb^*)} \cO_\epsilon(\Ab', \Bb', \bGamma', \Gb')
\\&=
\perm_{\sS_3}^{-1}\left( \cO_\epsilon\left(
\perm_{\sS_3}(\Ab^*, \Bb^*, \bGamma^*, \Gb^*)
\right) \right)
,
\end{align*}
which shows that the preimage of an $\epsilon$-neighborhood in the quotient space corresponds to an $\tilde\cO_\epsilon$-neighborhood in the original space.
Thus, the strong consistency of the pushforward posterior distribution $\law(\perm_{\sS_3}(\Ab, \Bb, \bGamma, \Gb) ~|~ \xb^{(1:N)}, \yb^{(1:N)})$ at $\perm_{\sS_3}(\Ab^*, \Bb^*, \bGamma^*, \Gb^*)$ implies that
\begin{align*}
&\qquad
\lim_{N \to \infty}
\PP\left(
\tilde\cO_\epsilon(\Ab^*, \Bb^*, \bGamma^*, \Gb^*)^c
~\Big|~
\xb^{(1:N)}, \yb^{(1:N)}
\right)
\\&=
\lim_{N \to \infty}
\PP\left(
\cO_\epsilon(\perm_{\sS_3}(\Ab^*, \Bb^*, \bGamma^*, \Gb^*))^c
~\Big|~
\xb^{(1:N)}, \yb^{(1:N)}
\right)
=
0
\end{align*}
for any $\epsilon > 0$.
That is, the posterior distribution $\law((\Ab, \Bb, \bGamma, \Gb) ~|~ \xb^{(1:N)}, \yb^{(1:N)})$ is also strongly consistent under Definition \ref{defi:post} at $(\Ab^*, \Bb^*, \bGamma^*, \Gb^*)$.
Finally, since $\pi(\sS_2 \cap \sS_3^c) = 0$, we conclude that the posterior is strongly consistent for $\pi$-almost every parameter in $\sS_2$.
\end{proof}

\subsection{Proofs of Auxiliary Lemmas for Theorem \ref{theo:post}}\label{supp_ssec:proof_lemmas_post}

We begin by proving Lemma \ref{lemm:prior_supp}.

\begin{proof}[Proof of Lemma \ref{lemm:prior_supp}]
Let the parameters $(\Ab, \Bb, \bGamma, \Gb) \in \sS$ and the observation $\yb \in \{0, 1\}^p$ be arbitrary.
For any other parameter $\Ab' \in (0, 1)^{q \times d}$, we have
\begin{align*}
&\qquad
\left|
\PP(\yb ~|~ \Ab, \Bb, \bGamma, \Gb) - \PP(\yb ~|~ \Ab', \Bb, \bGamma, \Gb)
\right|
\\&\le
\int_{\RR^{p_x}} \sum_{\wb \in \{0, 1\}^q} \sum_{z \in [d]} \PP(\yb ~|~ \wb, \Gb, \Bb) ~ \big|
\PP(\wb ~|~ z, \Ab) - \PP(\wb ~|~ z, \Ab')
\big| ~ \PP(z ~|~ \xb, \bGamma) \cP_x(\ud \xb)
\\&\le
2^q d \sup_{\wb \in \{0, 1\}^q, z \in [d]} \left|
\prod_{j = 1}^q \alpha_{j, z}^{w_j} (1 - \alpha_{j, z})^{1 - w_j}
-
\prod_{j = 1}^q (\alpha_{j, z}')^{w_j} (1 - \alpha_{j, z}')^{1 - w_j}
\right|
\\&\le
2^q d \sup_{\wb \in \{0, 1\}^q, z \in [d]} \left(
\sum_{j = 1}^q |\alpha_{j, z} - \alpha_{j, z}'|
\right)
\\&=
2^q d \|\Ab - \Ab'\|_{1, 1}
,
\end{align*}
where, for the $q \times d$ matrix $\Ab$, the matrix norm $\|\Ab\|_{1, 1}$ is defined as the sum of the absolute value of all its entries, i.e. $\|\Ab\|_{1, 1} := \sum_{j = 1}^q \sum_{h = 1}^d |\alpha_{j, h}|$.

For any other parameter $\bGamma' \in \RR^{d \times (p_x + 1)}$, with its $d$th row fixed as $\gamma_{d, 0} = 0$ and $\bgamma_d = \zero$ to ensure identifiability, we have
\begin{align*}
&\qquad
\left|
\PP(\yb ~|~ \Ab, \Bb, \bGamma, \Gb) - \PP(\yb ~|~ \Ab, \Bb, \bGamma', \Gb)
\right|
\\&\le
\int \sum_{\wb \in \{0, 1\}^q} \sum_{z \in [d]} \PP(\yb ~|~ \wb, \Gb, \Bb) \PP(\wb ~|~ z, \Ab) ~ \big|
\PP(z ~|~ \xb, \bGamma) - \PP(z ~|~ \xb, \bGamma')
\big| ~ \cP_x(\ud \xb)
\\&\le
2^q d \sup_{z \in [d]} \int \left|
\frac{\exp(\gamma_{z, 0} + \bgamma_z^\top \xb)}{\sum_{h = 1}^d \exp(\gamma_{h, 0} + \bgamma_h^\top \xb)} - \frac{\exp(\gamma_{z, 0}' + (\bgamma_z')^\top \xb)}{\sum_{h = 1}^d \exp(\gamma_{h, 0}' + (\bgamma_h')^\top \xb)}
\right| \cP_x(\ud \xb)
.
\end{align*}
Notice that the functions $f(\gamma) := \frac{\exp(a \gamma + b)}{c + \exp(a \gamma + b)}$ and $g(\gamma) := \frac{c}{c + \exp(a \gamma + b)}$ have derivatives satisfying
\begin{equation}\label{eq:logistic_derivative}
|f'(\gamma)|
=
|g'(\gamma)|
=
\frac{a c \exp(a \gamma + b)}{(\exp(a \gamma + b) + c)^2}
\le
\frac{a}{4}
,\quad
\forall \gamma \in \RR
,~
\forall b, c \in \RR
.
\end{equation}
Thus, it follows that
\begin{align*}
&\qquad
\left|
\PP(\yb ~|~ \Ab, \Bb, \bGamma, \Gb) - \PP(\yb ~|~ \Ab, \Bb, \bGamma', \Gb)
\right|
\\&\le
2^q d \sup_{z \in [d]} \int \left|
\frac{\exp(\gamma_{z, 0} + \bgamma_z^\top \xb)}{\sum_{h = 1}^d \exp(\gamma_{h, 0} + \bgamma_h^\top \xb)} - \frac{\exp(\gamma_{z, 0}' + (\bgamma_z')^\top \xb)}{\sum_{h = 1}^d \exp(\gamma_{h, 0}' + (\bgamma_h')^\top \xb)}
\right| \cP_x(\ud \xb)
\\&\le
2^q d \cdot \frac{\max\left\{1,~ \int \|\xb\|_\infty \cP_x(\ud \xb)\right\}}{4} \|\bGamma - \bGamma'\|_{1, 1}
.
\end{align*}

For any other parameter $\Bb' \in \RR^{p \times (q + 1)}$, we have
\begin{align*}
&\qquad
\left|
\PP(\yb ~|~ \Ab, \Bb, \bGamma, \Gb) - \PP(\yb ~|~ \Ab, \Bb', \bGamma, \Gb)
\right|
\\&\le
\int \sum_{\wb \in \{0, 1\}^q} \sum_{z \in [d]} ~\big|
\PP(\yb ~|~ \wb, \Gb, \Bb) - \PP(\yb ~|~ \wb, \Gb, \Bb')
\big|~ \PP(\wb ~|~ z, \Ab) \PP(z ~|~ \xb, \bGamma) \cP_x(\ud \xb)
\\&\le
2^q d \sup_{\wb \in \{0, 1\}^q} \left|
\prod_{i = 1}^p \frac{\exp(\beta_{i, 0} + (\gb_i \circ \bbeta_i)^\top \wb)}{1 + \exp(\beta_{i, 0} + (\gb_i \circ \bbeta_i)^\top \wb)}
-
\prod_{i = 1}^p \frac{\exp(\beta_{i, 0}' + (\gb_i \circ \bbeta_i')^\top \wb)}{1 + \exp(\beta_{i, 0}' + (\gb_i \circ \bbeta_i')^\top \wb)}
\right|
\\&\le
2^q d \sup_{\wb \in \{0, 1\}^q} \sum_{i = 1}^p \left|
\frac{\exp(\beta_{i, 0} + (\gb_i \circ \bbeta_i)^\top \wb)}{1 + \exp(\beta_{i, 0} + (\gb_i \circ \bbeta_i)^\top \wb)}
-
\frac{\exp(\beta_{i, 0}' + (\gb_i \circ \bbeta_i')^\top \wb)}{1 + \exp(\beta_{i, 0}' + (\gb_i \circ \bbeta_i')^\top \wb)}
\right|
\\&\le
2^q d \frac{\|\Bb - \Bb'\|_{1, 1}}{4}
,
\end{align*}
where again we have applied \eqref{eq:logistic_derivative}.

Since $\Gb$ has finitely many possible values and given all the results established above, we conclude that $\cL$ is a Lipschitz function of $(\Ab, \Bb, \bGamma, \Gb)$, with a Lipschitz constant $C$ that depends on the first moment of $\cP_x$.

Let $(\Ab, \Bb, \bGamma, \Gb) \in \supp(\pi)$ be arbitrary.
For simplicity, we denote the marginal distribution given $(\Ab, \Bb, \bGamma, \Gb)$ by $P := \cL(\Ab, \Bb, \bGamma, \Gb)$.
Since $P$ always lies within the interior of the probability simplex $\cS^{2^p - 1}$, there exists a positive constant $\kappa_P > 0$, dependent on $P$, such that the closed neighborhood $\overline\cO_{\kappa_P}(P)$ is also contained within the interior of $\cS^{2^p - 1}$.

On the compact set $\overline\cO_{\kappa_P}(P)$, we can identify positive constants $m_1, m_2 > 0$ that serve as lower and upper bounds for the Radon-Nikodym derivative $\frac{\ud P'}{\ud P}$, which is a continuous function of $P'$, i.e.
$$
0 < m_1 \le \frac{\ud P'}{\ud P} \le m_2
,\quad
\forall P' \in \overline\cO_{\kappa_P}(P)
.
$$
Applying the reverse Pinsker's inequality \citep{binette2019note}, we obtain
$$
D_{KL}(P, P')
\le
\frac12 \left(
\frac{\log m_1}{1 - m_1} + \frac{\log m_2}{m_2 - 1}
\right) \|P - P'\|_1
,\quad
\forall P' \in \overline\cO_{\kappa_P}(P)
.
$$
Defining the positive constant $m_P := \frac12 \left( \frac{\log m_1}{1 - m_1} + \frac{\log m_2}{m_2 - 1} \right)$, this result implies
$$
\cO_\epsilon^{KL}(P)
\supset
\cO_{\min\{\kappa_P, \frac{\epsilon}{m_P}\}}(P)
,\quad
\forall \epsilon > 0
.
$$
By the Lipschitz continuity of $\cL$, we further obtain
$$
\cO_{\min\{\kappa_P, \frac{\epsilon}{m_P}\}}(P)
\supset
\cL\left(
\cO_{C^{-1} \min\{\kappa_P, \frac{\epsilon}{m_P}\}}(\Ab, \Bb, \bGamma, \Gb)
\right)
.
$$
Since $(\Ab, \Bb, \bGamma, \Gb) \in \supp(\pi)$, for any $\epsilon > 0$, it follows that
$$
(\cL_\# \pi)\left(
\cO_\epsilon^{KL}(P)
\right)
\ge
(\cL_\# \pi)\left(
\cO_{\min\{\kappa_P, \frac{\epsilon}{m_P}\}}(P)
\right)
\ge
\pi\left(
\cO_{C^{-1} \min\{\kappa_P, \frac{\epsilon}{m_P}\}}(\Ab, \Bb, \bGamma, \Gb)
\right)
>
0
.
$$
Thus, $P = \cL(\Ab, \Bb, \bGamma, \Gb)$ is in the KL support of the pushforward prior measure $\cL_\# \pi$ on $\cS^{2^p - 1}$.
\end{proof}

Next, we proceed with the proof of Lemma \ref{lemm:S1_almost_compact}.

\begin{proof}[Proof of Lemma \ref{lemm:S1_almost_compact}]
We prove by contradiction and suppose that for all $\delta > 0$ and any actively compact set $\sB \subset \sS_1$, there exists a parameter $(\Ab^*, \Bb^*, \bGamma^*, \Gb^*) \in \sS_1$ such that
\begin{equation}\label{eq:assu_contra}
\cL(\sS_1 \cap \sB^c) \cap \cO_\delta(\cL(\Ab^*, \Bb^*, \bGamma^*, \Gb^*))
\ne
\varnothing
.
\end{equation}

We consider the sequence of positive values $\{\delta_m := \frac{1}{m}\}_{m = 1}^\infty$ and the sequence of actively compact sets $\{\sB_m\}_{m = 1}^\infty$ defined as
\begin{equation}\label{eq:sBm_seq}\begin{aligned}
\sB_m
:=
\Big\{
(\Ab, \Bb, \bGamma, \Gb) \in \sS_1:~
\forall i \in [p],~
\forall j \in [q],~
\forall h \in [d],~
\forall k \in [p_x],~
|\gamma_{j, 0}| \le k,~
|\gamma_{h, k}| \le k,
&\\
\frac{1}{k} \le \alpha_{j, h} \le 1 - \frac{1}{k},~
|\beta_{i, 0}| \le k,~
|g_{i, j} \beta_{i, j}| \le k
&
\Big\}
.
\end{aligned}\end{equation}
By the assumption made in \eqref{eq:assu_contra}, corresponding to the sequence $\{\sB_m\}_{m = 1}^\infty$, there further exists a sequence of parameters $\{(\Ab^{(m)}, \Bb^{(m)}, \bGamma^{(m)}, \Gb^{(m)})\}_{m = 1}^\infty \subset \sS_1$ such that, for each $m \in \NN$,
$$
(\Ab^{(m)}, \Bb^{(m)}, \bGamma^{(m)}, \Gb^{(m)}) \notin \sB_m
,\quad
\cL(\Ab^{(m)}, \Bb^{(m)}, \bGamma^{(m)}, \Gb^{(m)}) \in \cO_{\delta_m}(\cL(\Ab^*, \Bb^*, \bGamma^*, \Gb^*))
.
$$
Utilizing the concept of the extended real line $\RR \cup \{\pm \infty\}$, we observe that for each active parameter dimension, including each $\gamma_{h, k}$, $\alpha_{j, h}$, $\beta_{i, 0}$, or $g_{i, j} \beta_{i, j}$, we can extract a subsequence from $\{(\Ab^{(m)}, \Bb^{(m)}, \bGamma^{(m)}, \Gb^{(m)})\}_{m = 1}^\infty$ that converges in that specific dimension, potentially to $\infty$ or $-\infty$.
By sequentially applying this procedure across all parameter dimensions, we obtain a convergent subsequence of $\{(\Ab^{(m)}, \Bb^{(m)}, \bGamma^{(m)}, \Gb^{(m)})\}_{m = 1}^\infty$.

For notational simplicity, we redefine $\{(\Ab^{(m)}, \Bb^{(m)}, \bGamma^{(m)}, \Gb^{(m)})\}_{m = 1}^\infty$ to represent this convergent subsequence, which must satisfy at least one of the following conditions:
\begin{enumerate}
\item[$(C_1)$]
$\exists i \in [p]$ s.t. $\lim_{m \to \infty} \beta_{i, 0}^{(m)} = \infty$;
\item[$(C_2)$]
$\exists i \in [p]$ s.t. $\lim_{m \to \infty} \beta_{i, 0}^{(m)} = -\infty$;
\item[$(C_3)$]
$\exists i \in [p], j \in [q]$ s.t. $\lim_{m \to \infty} g_{i, j}^{(m)} = 1$ and $\lim_{m \to \infty} \beta_{i, j}^{(m)} = \infty$;
\item[$(C_4)$]
$\exists i \in [p], j \in [q]$ s.t. $\lim_{m \to \infty} g_{i, j}^{(m)} = 1$ and $\lim_{m \to \infty} \beta_{i, j}^{(m)} = -\infty$;
\item[$(C_5)$]
$\exists j \in [q], h \in [d]$ s.t. $\lim_{m \to \infty} \alpha_{j, h}^{(m)} = 0$;
\item[$(C_6)$]
$\exists j \in [q], h \in [d]$ s.t. $\lim_{m \to \infty} \alpha_{j, h}^{(m)} = 1$;
\item[$(C_7)$]
$\exists h \in [d], k \in [p_k]$ s.t. $\lim_{m \to \infty} \gamma_{h, k}^{(m)} \in \{\pm\infty\}$.
\item[$(C_8)$]
$\exists h \in [d]$ s.t. $\lim_{m \to \infty} \gamma_{h, 0}^{(m)} = \infty$;
\item[$(C_9)$]
$\exists h \in [d]$ s.t. $\lim_{m \to \infty} \gamma_{h, 0}^{(m)} = -\infty$.
\end{enumerate}

We recall the definitions from \eqref{eq:Lambda_st} and \eqref{eq:Sigma_th} in the proof of Theorem \ref{theo:strict}.
The $2^p \times 2^q$ probability matrix $\bLambda$ characterizes the conditional distribution $\law(\yb ~|~ \wb, \Gb, \Bb)$, while the $2^q \times d$ probability matrix $\bSigma$ characterizes the conditional distribution $\law(\wb ~|~ z, \Ab)$.
In contrast to the definitions in \eqref{eq:nu_t} and \eqref{eq:mu_h}, we now redefine the probability vectors $\bnu$ and $\bmu$ by integrating out $\xb$ with respect to its distribution $\cP_x$.
Specifically, the $2^q$-dimensional probability vector $\bnu$ now characterizes the marginal distribution $\law(\wb ~|~ \Ab, \bGamma)$, while the $d$-dimensional probability vector $\bmu$ characterizes the marginal distribution $\law(z ~|~ \Gamma)$.

Importantly, despite this slight modification in definition, where $\xb$ has been integrated out, Lemmas \ref{lemm:y_to_w} and \ref{lemm:w_to_z} remain valid for the unique identification of $(\bLambda, \bnu)$ from $\bLambda \bnu$ and $(\bSigma, \bmu)$ from $\bSigma \bmu$, up to column permutations.
This can be verified by examining their proofs, which do not rely on the explicit presence of $\xb$ but rather on the structure of the probability matrices.

We put superscripts $(m)$ or $*$ on the probability matrices $\bLambda, \bSigma$ and probability vectors $\bnu, \bmu$ corresponding to the parameters $(\Ab^{(m)}, \Bb^{(m)}, \bGamma^{(m)}, \Gb^{(m)})$ or $(\Ab^*, \Bb^*, \bGamma^*, \Gb^*)$.
By the continuity of the mapping from the parameters $(\Ab, \Bb, \bGamma, \Gb)$ to the probability matrices and vectors $\bLambda, \bnu, \bSigma, \bmu$, it follows that the sequence $\{(\bLambda^{(m)}, \bnu^{(m)}, \bSigma^{(m)}, \bmu^{(m)})\}_{m = 1}^\infty$ also converges.

To deal with the simultaneous shuffling of columns in $\bLambda$ and entries in $\bnu$, we define $\shuff(\bLambda, \bnu)$ as the set of all possible shuffled versions of $(\bLambda, \bnu)$, given by
$$
\shuff(\bLambda, \bnu)
:=
\Big\{
(\bLambda', \bnu'):~
\exists \text{ permutation matrix } \Qb \in \RR^{2^q \times 2^q}
\text{ s.t. }
\bLambda' = \bLambda \Qb
,~
\bnu' = \Qb^\top \bnu
\Big\}
.
$$
Similarly, we define $\shuff(\bSigma, \bmu)$ as the set of all possible shuffled versions of $(\bSigma, \bmu)$, given by
$$
\shuff(\bSigma, \bmu)
:=
\Big\{
(\bSigma', \bmu'):~
\exists \text{ permutation matrix } \Rb \in \RR^{d \times d}
\text{ s.t. }
\bSigma' = \bSigma \Rb
,~
\bmu' = \Rb^\top \bmu
\Big\}
.
$$
We equip the space of $(\bLambda, \bnu, \bSigma, \bmu)$ with the $L_1$ distance.
For any $\kappa > 0$, we define the union of open neighborhoods around each shuffled version as
$$
\tilde\cO_\kappa(\bLambda, \bnu)
:=
\bigcup_{(\bLambda', \bnu') \in \shuff(\bLambda, \bnu)} \cO_\kappa(\bLambda', \bnu')
,\quad
\tilde\cO_\kappa(\bSigma, \bmu)
:=
\bigcup_{(\bSigma', \bmu') \in \shuff(\bSigma, \bmu)} \cO_\kappa(\bSigma', \bmu')
.
$$

By Lemmas \ref{lemm:y_to_w} and \ref{lemm:w_to_z}, we establish the following equivalence relations:
$$
\bLambda^{(m)} \bnu^{(m)} = \bLambda' \bnu'
\iff
(\bLambda^{(m)}, \bnu^{(m)}) \in \shuff(\bLambda', \bnu')
,
$$
$$
\bSigma^{(m)} \bmu^{(m)} = \bSigma' \bmu'
\iff
(\bSigma^{(m)}, \bmu^{(m)}) \in \shuff(\bSigma', \bmu')
.
$$
Notably, the mapping from $(\bLambda, \bnu)$ to their product $\bLambda \bnu$, which defines $\cL(\Ab, \Bb, \bGamma, \Gb)$, is continuous, as is the mapping from $(\bSigma, \bmu)$ to their product $\bnu = \bSigma \bmu$.
Since all entries of $\bLambda, \bnu, \bSigma, \bmu$ are constrained within $[0, 1]$, the domains of both mappings are compact.
Consequently, if
$$
\lim_{m \to \infty} (\bLambda^{(m)}, \bnu^{(m)})
\notin
\shuff(\bLambda^*, \bnu^*)
,
$$
then the sequence $\{(\bLambda^{(m)}, \bnu^{(m)})\}_{m = 1}^\infty$ must eventually remain within some compact set $\sC$ that is disjoint from $\shuff(\bLambda^*, \bnu^*)$, i.e. $\sC \cap \shuff(\bLambda^*, \bnu^*) = \varnothing$.
The image of $\sC$ under the mapping $(\bLambda, \bnu) \to \bLambda \bnu$ within the probability simplex $\cS^{2^p - 1}$ must then also be compact and at a positive distance from $\cL(\Ab^*, \Bb^*, \bGamma^*, \Gb^*)$.
However, this contradicts the assumption that $\cL(\Ab^{(m)}, \Bb^{(m)}, \bGamma^{(m)}, \Gb^{(m)}) \in \cO_{\delta_m}(\cL(\Ab^*, \Bb^*, \bGamma^*, \Gb^*))$ for each $m \in \NN$.

This contradiction implies that $\lim_{m \to \infty} (\bLambda^{(m)}, \bnu^{(m)}) \in \shuff(\bLambda^*, \bnu^*)$.
A similar argument applies to $(\bSigma, \bmu)$, yielding $\lim_{m \to \infty} (\bSigma^{(m)}, \bmu^{(m)}) \in \shuff(\bSigma^*, \bmu^*)$.
Since the shuffling does not affect the subsequent analysis, we can assume, without loss of generality, that
\begin{equation}\label{eq:prob_mat_vec_convergent}
\lim_{m \to \infty} (\bLambda^{(m)}, \bnu^{(m)}, \bSigma^{(m)}, \bmu^{(m)})
=
(\bLambda^*, \bnu^*, \bSigma^*, \bmu^*)
.
\end{equation}

We now analyze conditions $C_1$ through $C_9$ individually and demonstrate that each leads to a contradiction with \eqref{eq:prob_mat_vec_convergent}.

If condition $C_1$ holds for some $i \in [p]$, then for $\wb = \zero$, we have
\begin{equation}\label{eq:C1}\begin{aligned}
\lim_{m \to \infty} \PP(y_i = 0 ~|~ \wb = \zero, \gb_i^{(m)}, \beta_{i, 0}^{(m)}, \bbeta_i^{(m)})
&=
\lim_{m \to \infty} \frac{1}{1 + \exp(\beta_{i, 0}^{(m)} + (\gb_i^{(m)} \circ \zero)^\top \bbeta_i^{(m)})}
\\&=
\lim_{m \to \infty} \frac{1}{1 + \exp(\beta_{i, 0}^{(m)})}
=
0
.
\end{aligned}\end{equation}

If condition $C_2$ holds for some $i \in [p]$, then for $\wb = \zero$, we have
\begin{equation}\label{eq:C2}\begin{aligned}
\lim_{m \to \infty} \PP(y_i = 1 ~|~ \wb = \zero, \gb_i^{(m)}, \beta_{i, 0}^{(m)}, \bbeta_i^{(m)})
&=
\lim_{m \to \infty} \frac{\exp(\beta_{i, 0}^{(m)} + (\gb_i^{(m)} \circ \zero)^\top \bbeta_i^{(m)})}{1 + \exp(\beta_{i, 0}^{(m)} + (\gb_i^{(m)} \circ \zero)^\top \bbeta_i^{(m)})}
\\&=
\lim_{m \to \infty} \frac{\exp(\beta_{i, 0}^{(m)})}{1 + \exp(\beta_{i, 0}^{(m)})}
=
0
.
\end{aligned}\end{equation}

If condition $C_3 \wedge \neg(C_1 \vee C_2)$ holds for some $i \in [p], j \in [q]$, then for $\wb = \eb_j$, we have
\begin{equation}\label{eq:C3}\begin{aligned}
\lim_{m \to \infty} \PP(y_i = 0 ~|~ \wb = \eb_j, \gb_i^{(m)}, \beta_{i, 0}^{(m)}, \bbeta_i^{(m)})
&=
\lim_{m \to \infty} \frac{1}{
1 + \exp(\beta_{i, 0}^{(m)} + (\gb_i^{(m)} \circ \eb_j)^\top \bbeta_i^{(m)})
}
\\&=
\lim_{m \to \infty} \frac{1}{
1 + \exp(\beta_{i, 0}^{(m)} + \beta_{i, j}^{(m)})
}
=
0
.
\end{aligned}\end{equation}

If condition $C_4 \wedge \neg(C_1 \vee C_2)$ holds for some $i \in [p], j \in [q]$, then for $\wb = \eb_j$, we have
\begin{equation}\label{eq:C4}\begin{aligned}
\lim_{m \to \infty} \PP(y_i = 1 ~|~ \wb = \eb_j, \gb_i^{(m)}, \beta_{i, 0}^{(m)}, \bbeta_i^{(m)})
&=
\lim_{m \to \infty} \frac{
\exp(\beta_{i, 0}^{(m)} + (\gb_i^{(m)} \circ \eb_j)^\top \bbeta_i^{(m)})
}{
1 + \exp(\beta_{i, 0}^{(m)} + (\gb_i^{(m)} \circ \eb_j)^\top \bbeta_i^{(m)})
}
\\&=
\lim_{m \to \infty} \frac{
\exp(\beta_{i, 0}^{(m)} + \beta_{i, j}^{(m)})
}{
1 + \exp(\beta_{i, 0}^{(m)} + \beta_{i, j}^{(m)})
}
=
0
.
\end{aligned}\end{equation}

Therefore, if condition $C_1 \vee C_2 \vee C_3 \vee C_4$ holds, then at least one of the equations \eqref{eq:C1}, \eqref{eq:C2}, \eqref{eq:C3}, or \eqref{eq:C4} must be satisfied.
This implies that at least one entry in the matrix $\lim_{m \to \infty} \bLambda^{(m)}$ must be zero.
However, for any valid true parameters $\Gb^*, \Bb^*$, all entries of $\bLambda^*$ must be strictly positive, as implied by our model formulation \eqref{eq:model_y_w}.
This directly contradicts with $\lim_{m \to \infty} \bLambda^{(m)} = \bLambda^*$ in \eqref{eq:prob_mat_vec_convergent}.

If condition $C_5$ holds for some $j \in [q]$ and $h \in [d]$, then for $z = h$, we have
\begin{equation}\label{eq:C5}
\lim_{m \to \infty} \PP(w_j = 1 ~|~ z = h, \Ab^{(m)})
=
\lim_{m \to \infty} \alpha_{j, h}
=
0
.
\end{equation}

If condition $C_6$ holds for some $j \in [q]$ and $h \in [D]$, then we have
\begin{equation}\label{eq:C6}
\lim_{m \to \infty} \PP(w_j = 0 ~|~ z = h, \Ab^{(m)})
=
\lim_{m \to \infty} (1 - \alpha_{j, h})
=
0
.
\end{equation}

Thus, if condition $C_5 \vee C_6$ holds, then at least one of the equations \eqref{eq:C5} or \eqref{eq:C6} must be satisfied.
This implies that at least one entry in the matrix $\lim_{m \to \infty} \bSigma^{(m)}$ must be zero.
However, according to our model formulation \eqref{eq:model_w_z}, all entries of $\bSigma^*$ must be strictly positive for any valid true parameter $\Ab^*$, causing a contradiction with $\lim_{m \to \infty} \bSigma^{(m)} = \bSigma^*$ in \eqref{eq:prob_mat_vec_convergent}.

If condition $C_7$ holds for some $h \in [d]$ and $k \in [p_k]$, then since $\lim_{m \to \infty} \bmu^{(m)}$ exists and equals $\bmu^*$ by \eqref{eq:prob_mat_vec_convergent}, we conclude that for any $\xb \in \supp(\cP_x)$, the limit $\lim_{m \to \infty} \xb^\top \bgamma_h^{(m)}$ must also exist in the extended real line $\RR \cup \{\pm \infty\}$.
Furthermore, since $\cP_x$ has full rank, we can find some $\hat\xb \in \supp(\cP_x)$ such that 
$$
\lim_{m \to \infty} \hat\xb^\top \bgamma_h^{(m)}
\in
\{\pm \infty\}
.
$$
Recall that the parameters for the baseline class $d$ are fixed as $\gamma_{d, 0} = 0$ and $\bgamma_d = \zero$ to ensure identifiability, which implies $h \ne d$.
Now, if $\lim_{m \to \infty} \hat\xb^\top \bgamma_h^{(m)} = \infty$, we obtain $\lim_{m \to \infty} \frac{\exp(\hat\xb^\top \bgamma_d^{(m)})}{\exp(\hat\xb^\top \bgamma_h^{(m)})} = 0$.
Conversely, if $\lim_{m \to \infty} \hat\xb^\top \bgamma_h^{(m)} = -\infty$, we have $\lim_{m \to \infty} \frac{\exp(\hat\xb^\top \bgamma_h^{(m)})}{\exp(\hat\xb^\top \bgamma_d^{(m)})} = 0$.
In either case, we can find classes $h_1, h_2 \in [d]$, some $\hat\xb \in \supp(\cP_x)$, and a small positive constant $\kappa > 0$ such that
$$
\lim_{m \to \infty} \frac{\exp(\xb^\top \bgamma_{h_1}^{(m)})}{\exp(\xb^\top \bgamma_{h_2}^{(m)})}
=
0
,\quad
\forall \xb \in \cO_\kappa(\hat\xb)
,
$$
where $\cO_\kappa(\hat\xb)$ denotes the open neighborhood in $\cX$ around $\hat\xb$ with $L_1$-radius $\kappa$.
This further implies that for any $\xb \in \cO_\kappa(\hat\xb)$,
\begin{equation}\label{eq:C7_mid}
\lim_{m \to \infty} \PP(z = h_1 ~|~ \xb, \bGamma^{(m)})
=
0
.
\end{equation}

We define $\cP_x 1_{\cO_\kappa(\hat\xb)}$ as the probability measure obtained by restricting $\cP_x$ to $\cO_\kappa(\xb)$ and normalizing, i.e. with Radon-Nikodym derivative
$$
\frac{\ud \cP_x 1_{\cO_\kappa(\hat\xb)}}{\ud \cP_x}
=
\cP_x(\cO_\kappa(\hat\xb))^{-1} 1_{\cO_\kappa(\hat\xb)}
,
$$
which is well defined since $\xb \in \supp(\cP_x)$.
Now, let $\tilde\PP$ denote the marginal distribution of $z$ obtained by integrating out $\xb$ using this restricted measure $\cP_x 1_{\cO_\kappa(\hat\xb)}$.
By \eqref{eq:C7_mid}, we obtain
\begin{equation}\label{eq:C7_lim}
\lim_{m \to \infty} \tilde\PP(z = h_1 ~|~ \bGamma^{(m)})
=
\cP_x(\cO_k(\hat\xb))^{-1} \int_{\cO_\kappa(\hat\xb)} \PP(z = h_1 ~|~ \xb, \bGamma^{(m)}) ~ \cP_x(\ud x)
=
0
.
\end{equation}
At the same time, since our model formulation \eqref{eq:model_z_x} implies that $\PP(z ~|~ \xb, \bGamma^*)$ for any $\xb \in \RR^{p_x}$ and any valid true parameter $\bGamma^*$, we have
\begin{equation}\label{eq:C7_true}
\tilde\PP(z = h_1 ~|~ \bGamma^*)
=
\cP_x(\cO_k(\hat\xb))^{-1} \int_{\cO_\kappa(\hat\xb)} \PP(z = h_1 ~|~ \xb, \bGamma^*) ~ \cP_x(\ud x)
>
0
.
\end{equation}

We observe that Lemmas \ref{lemm:y_to_w} and \ref{lemm:w_to_z} continue to hold when considering $\cP_x 1_{\cO_\kappa(\hat\xb)}$ as the data generating distribution for $\xb$.
By employing arguments similar to those used in establishing \eqref{eq:prob_mat_vec_convergent}, we find that the marginal distributions of $z$ satisfy
$$
\lim_{m \to \infty} \tilde\PP(z = h_1 ~|~ \bGamma^{(m)})
=
\tilde\PP(z = h_1 ~|~ \bGamma^*)
,
$$
which directly contradicts the results in \eqref{eq:C7_lim} and \eqref{eq:C7_true}.
Thus, condition $C_7$ is also incompatible with \eqref{eq:prob_mat_vec_convergent}.

If condition $C_8 \wedge \neg C_7$ holds for some $h \in [d]$, recalling that $z = d$ is the baseline class with $\gamma_{d, 0} = 0$ and $\bgamma_d = \zero$, we have for any $\xb \in \supp(\cP_x)$ that
\begin{equation}\label{eq:C8}
\lim_{m \to \infty} \PP(z = d ~|~ \xb, \bGamma^{(m)})
=
\frac{1}{\sum_{h = 1}^d \exp(\gamma_{h, 0}^{(m)} + \xb^\top \bgamma_h^{(m)})}
=
0
.
\end{equation}

If condition $C_9 \wedge \neg C_7$ holds for some $h \in [d]$, similarly, we have for any $\xb \in \supp(\cP_x)$ that
\begin{equation}\label{eq:C9}
\lim_{m \to \infty} \PP(z = h ~|~ \xb, \bGamma^{(m)})
=
\frac{\exp(\gamma_{h, 0}^{(m)} + \xb^\top \bgamma_h^{(m)})}{\sum_{h' = 1}^d \exp(\gamma_{h', 0}^{(m)} + \xb^\top \bgamma_{h'}^{(m)})}
=
0
.
\end{equation}

Therefore, if condition $(C_8 \vee C_9) \wedge \neg C_7$ holds, then at least one of the equations \eqref{eq:C8} or \eqref{eq:C9} must be satisfied.
This implies that at least one entry in the vector $\lim_{m \to \infty} \bmu^{(m)}$ must be zero.
In contrast, our model formulation \eqref{eq:model_z_x} implies that all entries of $\bmu^*$ must be strictly positive for any true parameter $\bGamma^*$.
This raises a contradiction with $\lim_{m \to \infty} \bmu^{(m)} = \bmu^*$ in \eqref{eq:prob_mat_vec_convergent}.

Putting everything together, we conclude that condition $C_1 \vee C_2 \vee C_3 \vee C_4 \vee C_5 \vee C_6 \vee C_7 \vee C_8 \vee C_9$ is mutually exclusive with \eqref{eq:prob_mat_vec_convergent}.
This implies that none of the conditions $C_1, \ldots, C_9$ can hold for the chosen sequence of parameters $(\Ab^{(m)}, \Bb^{(m)}, \bGamma^{(m)}, \Gb^{(m)})$, thus completing the proof.
\end{proof}

Finally, we provide the proof of Lemma \ref{lemm:quotient_metric} below.

\begin{proof}[Proof of Lemma \ref{lemm:quotient_metric}]
Let the subset $\sS \subset \sS_2$ be permutation invariant.

We begin by showing that $\|\cdot\|_1$ defined in \eqref{eq:perm_equiv_metric} is a well-defined metric on the quotient space $\perm(\sS)$.
Although its positivity and symmetry are straightforward, its triangle inequality property is less immediate and is based on the permutation invariance of $\sS$.
The key observation is that for a permutation invariant parameter space $\sS$, we obtain the simplification
\begin{align*}
&\qquad
\left\|
\perm_{\sS}(\Ab, \Bb, \bGamma, \Gb) - \perm_{\sS}(\Ab', \Bb', \bGamma', \Gb')
\right\|_1
\\&=
\inf_{(\tilde\Ab, \tilde\Bb, \tilde\bGamma, \tilde\Gb) \in \perm_{\sS}(\Ab, \Bb, \bGamma, \Gb)}
\left\|
(\tilde\Ab, \tilde\Bb, \tilde\bGamma, \tilde\Gb) - (\Ab', \Bb', \bGamma', \Gb')
\right\|_1
.
\end{align*}
Therefore, for any equivalence classes of parameters $\perm_{\sS}(\Ab, \Bb, \bGamma, \Gb)$, $\perm_{\sS}(\Ab', \Bb', \bGamma', \Gb')$, and $\perm_{\sS}(\Ab'', \Bb'', \bGamma'', \Gb'') \in \perm(\sS)$, we have
\begin{align*}
&\qquad
\Big\|
\perm_{\sS}(\Ab, \Bb, \bGamma, \Gb) - \perm_{\sS}(\Ab'', \Bb'', \bGamma'', \Gb'')
\Big\|_1
\\&=
\inf_{(\tilde\Ab, \tilde\Bb, \tilde\bGamma, \tilde\Gb) \in \perm_{\sS}(\Ab, \Bb, \bGamma, \Gb)}
\inf_{(\tilde\Ab'', \tilde\Bb'', \tilde\bGamma'', \tilde\Gb'') \in \perm_{\sS}(\Ab'', \Bb'', \bGamma'', \Gb'')}
\left\|
(\tilde\Ab, \tilde\Bb, \tilde\bGamma, \tilde\Gb) - (\tilde\Ab'', \tilde\Bb'', \tilde\bGamma'', \tilde\Gb'')
\right\|_1
\\&\le
\inf_{(\tilde\Ab, \tilde\Bb, \tilde\bGamma, \tilde\Gb) \in \perm_{\sS}(\Ab, \Bb, \bGamma, \Gb)}
\left\|
(\tilde\Ab, \tilde\Bb, \tilde\bGamma, \tilde\Gb) - (\Ab', \Bb', \bGamma', \Gb')
\right\|_1
\\&\qquad+
\inf_{(\tilde\Ab'', \tilde\Bb'', \tilde\bGamma'', \tilde\Gb'') \in \perm_{\sS}(\Ab'', \Bb'', \bGamma'', \Gb'')}
\left\|
(\tilde\Ab'', \tilde\Bb'', \tilde\bGamma'', \tilde\Gb'') - (\Ab', \Bb', \bGamma', \Gb')
\right\|_1
\\&=
\Big\|
\perm_{\sS}(\Ab, \Bb, \bGamma, \Gb) - \perm_{\sS}(\Ab', \Bb', \bGamma', \Gb')
\Big\|_1
\\&\qquad+
\Big\|
\perm_{\sS}(\Ab'', \Bb'', \bGamma'', \Gb'') - \perm_{\sS}(\Ab', \Bb', \bGamma', \Gb')
\Big\|_1
.
\end{align*}
Hence $\|\cdot\|_1$ is a metric on $\sS$.

We next show that the metric $\|\cdot\|_1$ induces the quotient topology on $\sS$.

Let $\cU$ be any open set in $\perm(\sS)$ under the quotient topology.
By definition, its preimage in $\sS$, denoted by $\perm_{\sS}^{-1}(\cU)$, is open.
Consider an arbitrary equivalence class of parameters $\perm_{\sS}(\Ab, \Bb, \bGamma, \Gb) \in \cU$. 
Then, for any $(\tilde\Ab, \tilde\Bb, \tilde\bGamma, \tilde\Gb) \in \perm_{\sS}(\Ab, \Bb, \bGamma, \Gb)$, there exists an open ball  centered at $(\tilde\Ab, \tilde\Bb, \tilde\bGamma, \tilde\Gb)$ with positive radius, fully contained in $\perm_{\sS}^{-1}(\cU)$.
Defining the distance between a parameter $s \in \sS$ and a set $E \subset \sS$ as $d_{\sS}(s, E) := \inf_{s' \in E} \|s - s'\|_1$, then we obtain
$$
d_{\sS}\big(
(\tilde\Ab, \tilde\Bb, \tilde\bGamma, \tilde\Gb)
,~
\perm_{\sS}^{-1}(\cU)^c
\big)
>
0
,\quad
\forall (\tilde\Ab, \tilde\Bb, \tilde\bGamma, \tilde\Gb) \in \perm_{\sS}(\Ab, \Bb, \bGamma, \Gb)
.
$$
Importantly, for any two parameter sets $(\tilde\Ab, \tilde\Bb, \tilde\bGamma, \tilde\Gb), (\tilde\Ab, \tilde\Bb', \tilde\bGamma, \tilde\Gb) \in \perm_{\sS}(\Ab, \Bb, \bGamma, \Gb)$ that satisfy $(\one ~ \tilde\Gb) \circ \tilde\Bb = (\one ~ \tilde\Gb) \circ \tilde\Bb'$, meaning $\beta_{i, j} \ne \beta_{i, j}'$ only when $g_{i, j} = 0$ (i.e. only on inactive entries of $\Bb$), we have
\begin{align*}
&\qquad
d_{\sS}\left(
(\tilde\Ab, \tilde\Bb, \tilde\bGamma, \tilde\Gb)
,~
\perm_{\sS}^{-1}(\cU)^c
\right)
\\&=
\inf_{(\Ab', \Bb', \bGamma', \Gb') \in \perm_{\sS}(\cU)^c} \left(
\|\tilde\Ab - \Ab'\|_1
+
\|\tilde\Gb \circ \tilde\Bb - \Gb' \circ \Bb'\|_1
+
\|\tilde\bGamma - \bGamma'\|_1
+
\|\tilde\Gb - \Gb'\|_1
\right)
\\&=
\inf_{(\Ab', \Bb', \bGamma', \Gb') \in \perm_{\sS}(\cU)^c} \left(
\|\tilde\Ab - \Ab'\|_1
+
\|\tilde\Gb \circ \tilde\Bb' - \Gb' \circ \Bb'\|_1
+
\|\tilde\bGamma - \bGamma'\|_1
+
\|\tilde\Gb - \Gb'\|_1
\right)
\\&=
d_{\sS}\left(
(\tilde\Ab, \tilde\Bb', \tilde\bGamma, \tilde\Gb)
,~
\perm_{\sS}^{-1}(\cU)^c
\right)
.
\end{align*}
Since there are finitely many possible combinations of $\tilde\Ab, (\one ~ \tilde\Gb) \circ \tilde\Bb, \tilde\bGamma, \tilde\Gb$ within $\perm_{\sS}(\Ab, \Bb, \bGamma, \Gb)$, the distance
$$
\epsilon
:=
\inf_{(\tilde\Ab, \tilde\Bb, \tilde\bGamma, \tilde\Gb) \in \perm_{\sS}(\Ab, \Bb, \bGamma, \Gb)} d_{\sS}\big(
(\tilde\Ab, \tilde\Bb, \tilde\bGamma, \tilde\Gb)
,~
\perm_{\sS}^{-1}(\cU)
\big)
>
0
.
$$
Thus, the open ball $\cO_\epsilon(\perm_{\sS}(\Ab, \Bb, \bGamma, \Gb))$ is contained within $\cU$, implying that the metric topology induced by $\|\cdot\|_1$ is finer than the quotient topology on $\perm(\sS)$.

Now, for any $\epsilon > 0$, the preimage of an arbitrary open ball $\cO_\epsilon(\perm_{\sS}(\Ab, \Bb, \bGamma, \Gb)) \subset \perm(\sS)$ is given by
\begin{align*}
&\qquad
\perm_{\sS}^{-1}\left(
\cO_\epsilon(\perm_{\sS}(\Ab, \Bb, \bGamma, \Gb))
\right)
\\&=
\left\{
(\Ab', \Bb', \bGamma', \Gb') \in \sS:~
\inf_{(\tilde\Ab, \tilde\Bb, \tilde\bGamma, \tilde\Gb) \in \perm_{\sS}(\Ab, \Bb, \bGamma, \Gb)}
\left\|
(\tilde\Ab, \tilde\Bb, \tilde\bGamma, \tilde\Gb) - (\Ab', \Bb', \bGamma', \Gb')
\right\|_1
<
\epsilon
\right\}
\\&=
\bigcap_{m = 1}^\infty \bigcup_{(\tilde\Ab, \tilde\Bb, \tilde\bGamma, \tilde\Gb) \in \perm_{\sS}(\Ab, \Bb, \bGamma, \Gb)} \cO_{\epsilon + \frac{1}{m}}\big(
(\tilde\Ab, \tilde\Bb, \tilde\bGamma, \tilde\Gb)
\big)
,
\end{align*}
which is an open set in $\sS$.
This further implies that $\cO_\epsilon(\perm_{\sS}(\Ab, \Bb, \bGamma, \Gb))$ is open in $\perm(\sS)$ under the quotient topology.
Therefore, the metric topology induced by $\|\cdot\|_1$ is also coarser than the quotient topology on $\perm(\sS)$.

Combining these results, we conclude that the two topologies are equivalent, meaning that the metric $\|\cdot\|_1$ on $\perm(\sS)$ indeed induces its quotient topology.
\end{proof}

\section{Curse of Dimensionality in Mixture Models}\label{supp_sec:mix}

In this section, we provide the proof of Proposition \ref{prop:curse_example}, which demonstrates the curse of dimensionality faced by Bayesian Bernoulli mixture models in clustering high-dimensional binary observations.

\subsection{Proof of Proposition \ref{prop:curse_example}}

Recall that $\PP_{\mix}$ denotes the distributions specified by the Bernoulli mixture model \eqref{eq:mix_model}, and $\PP_{\mix}^*$ denotes the true data generating distribution \eqref{eq:mix_true}.
For the observations $[N] := \{1, \ldots, N\}$, we let $z^{(1:N)} = \{[N]\}$ denote the partition of $[N]$ that forms all $N$ observations into a single shared cluster, and let $z^{(1:N)} = \{\{n\}, [N] \cap \{n\}^c\}$ denote the partition of $[N]$ into two clusters, with one cluster containing only observation $n$ and the other cluster containing all remaining $N - 1$ observations.

The following lemma shows that for any number of observations $N \ge 2$, there almost surely exists an observation $n \in [N]$ such that asymptotically, the likelihood of forming $[N]$ into two clusters $\{n\}$ and $[N] \cap \{n\}^c$ dominates the likelihood of forming $[N]$ into a single cluster.

\begin{lemma}\label{lemm:take_one_out}
For any $N \ge 2$, we have
$$
\liminf_{p \to \infty} \frac{
\sum_{n = 1}^N \PP_{\mix}( \yb^{(1:N)} ~|~ z^{(1:N)} = \{\{n\}, [N] \cap \{n\}^c\})
}{
\PP_{\mix}(\yb^{(1:N)} ~|~ z^{(1:N)} = \{[N]\})
}
=
\infty
,
$$
$\PP_{\mix}^*$-almost surely.
\end{lemma}

The proof of Lemma \ref{lemm:take_one_out} is provided in Section \ref{supp_ssec:curse_lemmas}.
In the following, we prove Proposition \ref{prop:curse_example} using Lemma \ref{lemm:take_one_out}.

\begin{proof}[Proof of Proposition \ref{prop:curse_example}]
For a partition $z^{(1:N)}$ of the observations $[N]$, we let $\card(z^{(1:N)})$ denote the number of clusters formed in the partition.
Let $z^{(1:N)}$ be an arbitrary partition of [N] with $\card(z^{(1:N)}) < N$, then there exists a cluster in $z^{(1:N)}$ with number of observations $M \ge 2$.
Without loss of generality, we let the $M$ observations in this cluster be the observations $1, \ldots, M$ and denote this partition by $z^{(1:N)} = \{[M], \ldots\}$.

For $m \in [M]$, we let $z^{(1:N)} = \{\{m\}, [M] \cap \{m\}^c, \ldots\}$ denote the partition of $[N]$ obtained by further splitting the cluster of observations $1, \ldots, M$ in the partition $z^{(1:N)} = \{[M], \ldots\}$ into a cluster containing only observation $m$ and a cluster containing the remaining $M - 1$ observations $1, \ldots, m - 1, m + 1, \ldots, M$.
By definition, $\card(\{\{m\}, [M] \cap \{m\}^c, \ldots\}) = \card(\{[M], \ldots\}) + 1$.

In the mixture model \eqref{eq:mix_model}, using the conditional independence of $y_i$ given $z$, we have
\begin{equation}\label{eq:mix_N_to_M}
\frac{
\PP_{\mix}(\yb^{(1:N)} ~|~ z^{(1:N)} = \{\{m\}, [M] \cap \{m\}^c, \ldots\})
}{
\PP_{\mix}(\yb^{(1:N)} ~|~ z^{(1:N)} = \{[M], \ldots\})
}
=
\frac{
\PP_{\mix}(\yb^{(1:M)} ~|~ z^{(1:M)} = \{\{m\}, [M] \cap \{m\}^c\})
}{
\PP_{\mix}(\yb^{(1:M)} ~|~ z^{(1:M)} = \{[M]\})
}
.
\end{equation}
By applying Lemma \ref{lemm:take_one_out} to the $M$ observations $1, \ldots, M$ and combining with \eqref{eq:mix_N_to_M}, we obtain
\begin{equation}\label{eq:mix_N_take_one_out}\begin{aligned}
&\qquad
\liminf_{p \to \infty} \frac{
\sum_{m = 1}^M \PP_{\mix}(\yb^{(1:N)} ~|~ z^{(1:N)} = \{\{m\}, [M] \cap \{m\}^c, \ldots\})
}{
\PP_{\mix}(\yb^{(1:N)} ~|~ z^{(1:N)} = \{[M], \ldots\})
}
\\&=
\liminf_{p \to \infty} \frac{
\sum_{m = 1}^M \PP_{\mix}(\yb^{(1:M)} ~|~ z^{(1:M)} = \{\{m\}, [M] \cap \{m\}^c\})
}{
\PP_{\mix}(\yb^{(1:M)} ~|~ z^{(1:M)} = \{[M]\})
}
=
\infty
,
\end{aligned}\end{equation}
$\PP_{\mix}^*$-almost surely.

Let $P(N, K)$ denote the number of ways of partitioning $[N]$ into $K$ clusters.
Then for any $1 \le K \le N - 1$, we have
\begin{equation}\label{eq:mix_K_K_1}\begin{aligned}
&\qquad
\liminf_{p \to \infty} \frac{
\sum_{z^{(1:N)}:~\card(z^{(1:N)}) = K + 1} \PP_{\mix}(\yb^{(1:N)} ~|~ z^{(1:N)})
}{
\sum_{z^{(1:N)}:~\card(z^{(1:N)}) = K} \PP_{\mix}(\yb^{(1:N)} ~|~ z^{(1:N)})
}
\\&\ge
P(N, K)^{-1} \inf_{z^{(1:N)}:~\card(z^{(1:N)}) = K} \liminf_{p \to \infty} \frac{
\sum_{\tz^{(1:N)}:~\card(\tz^{(1:N)}) = K + 1} \PP_{\mix}(\yb^{(1:N)} ~|~ \tz^{(1:N)})
}{
\PP_{\mix}(\yb^{(1:N)} ~|~ z^{(1:N)})
}
\ge
\infty
\end{aligned}\end{equation}
$\PP_{\mix}^*$-almost surely, where the last inequality is due to \eqref{eq:mix_N_take_one_out} and the fact that the collection of partitions with $\card(z^{(1:N)}) + 1$ clusters include all the partitions obtained by splitting a cluster of $z^{(1:N)}$ into two.

Let $z^{(1:N)} = \{\{1\}, \{2\}, \ldots, \{N\}\}$ denote the partition of observations $[N]$ with each observation in a separate cluster.
Telescoping \eqref{eq:mix_K_K_1} over $K = 1, \ldots, N - 1$, we obtain for any $K \in [N - 1]$ that
$$
\liminf_{p \to \infty} \frac{
\PP_{\mix}(\yb^{(1:N)} ~|~ z^{(1:N)} = \{\{1\}, \{2\}, \ldots, \{N\}\})
}{
\sum_{z^{(1:N)}:~\card(z^{(1:N)}) = K} \PP_{\mix}(\yb^{(1:N)} ~|~ z^{(1:N)})
}
=
\infty
,
$$
$\PP_{\mix}^*$-almost surely.
Recall that the prior on $z^{(1:N)}$ assigns positive probability to any possible partition of $[N]$.
Therefore, for $\sZ$ being an arbitrary partition of $[N]$ other than the partition $z^{(1:N)} = \{\{1\}, \{2\}, \ldots, \{N\}\}$, we have the posterior probability ratio
\begin{align*}
&\qquad
\liminf_{p \to \infty} \frac{
\PP_{\mix}(z^{(1:N)} = \{\{1\}, \{2\}, \ldots, \{N\}\} ~|~ \yb^{(1:N)})
}{
\PP_{\mix}(z^{(1:N)} = \sZ ~|~ \yb^{(1:N)})
}
\\&=
\frac{
\PP_{\mix}(z^{(1:N)} = \{\{1\}, \{2\}, \ldots, \{N\}\})
}{
\PP_{\mix}(z^{(1:N)} = \sZ)
}
\liminf_{p \to \infty} \frac{
\PP_{\mix}(\yb^{(1:N)} ~|~ z^{(1:N)} = \{\{1\}, \{2\}, \ldots, \{N\}\})
}{
\PP_{\mix}(\yb^{(1:N)} ~|~ z^{(1:N)} = \sZ)
}
\\&\ge
\frac{
\PP_{\mix}(z^{(1:N)} = \{\{1\}, \{2\}, \ldots, \{N\}\})
}{
\PP_{\mix}(z^{(1:N)} = \sZ)
}
\liminf_{p \to \infty} \frac{
\PP_{\mix}(\yb^{(1:N)} ~|~ z^{(1:N)} = \{\{1\}, \{2\}, \ldots, \{N\}\})
}{
\sum_{z^{(1:N)}:~ \card(z^{(1:N)}) = \card(\sZ)} \PP_{\mix}(\yb^{(1:N)} ~|~ z^{(1:N)})
}
\\&=
\infty
,
\end{align*}
$\PP_{\mix}^*$-almost surely.
\end{proof}

\subsection{Proofs of Auxiliary Lemmas for Proposition \ref{prop:curse_example}}\label{supp_ssec:curse_lemmas}

For notational simplicity, we denote $N_i := \sum_{n = 1}^N y_i^{(n)}$ and $\chi_i := \frac{2N_i - N}{\sqrt{N}}$, such that $N_i = \frac{N}{2} + \frac{\sqrt{N}}{2} \chi_i$.
Following \citet{wainwright2019high}, a random variable $X$ with mean $\EE[X]$ is sub-exponential if there are non-negative parameters $(v, a)$ such that
$$
\EE[e^{\lambda (X - \EE[X])}]
\le
e^{\frac{v^2 \lambda^2}{2}}
,\quad
\forall |\lambda| < \frac{1}{a}
.
$$
The next lemma shows that each $\chi_i$ follows a sub-exponential distribution with parameters uniform over $N$.

\begin{lemma}\label{lemm:chi_subexp}
There exists constants $v, a > 0$, such that for any $N \ge 1$ and $i \in [p]$, $\chi_i$ is sub-exponential with parameters $(v, a)$.
\end{lemma}

The proof of Lemma \ref{lemm:chi_subexp} is as follows.

\begin{proof}[Proof of Lemma \ref{lemm:chi_subexp}]
Recalling the data generating distribution $\PP_{\mix}^*$ in \eqref{eq:mix_true}, we have $N_i \sim \mathrm{Binomial}(N, 0.5)$, which gives $\EE[\chi_i] = 0$ and $\EE[\chi_i^2] = 1$.
Moreover, by the central limit theorem, as $N \to \infty$ $\chi_i^2$ converges in distribution to the $\chi^2$-distribution with 1 degree of freedom, which is sub-exponential with parameters $(2, 4)$.
This suggests that
$$
\lim_{N \to \infty} \EE[e^{\lambda \chi_i}]
\le
e^{2 \lambda^2}
,\quad
\forall |\lambda| < \frac14
.
$$
Notice that the function $\lambda \to \frac{\log \EE[e^{\lambda \chi_i}]}{\lambda^2}$ is continuous.
Therefore, over the compact set of $|\lambda| \le \frac18$, the pointwise convergence is uniform in $\lambda$, which implies that there exists $N_0 \in \NN$ such that
$$
\frac{\log \EE[e^{\lambda \chi_i}]}{\lambda^2}
\le
4
,\quad
\forall |\lambda| \le \frac18
,\quad
\forall N \ge N_0
.
$$
By definition, this shows that for any $N \ge N_0$, each $\chi_i$ is sub-exponential with parameters $(2\sqrt{2}, 8)$.

Now for $N \le N_0$, we notice that $|\chi_i| \le \sqrt{N_0}$ is bounded, which implies that each $\chi_i$ is also sub-exponential with some parameter $(v_0, a_0)$.
Therefore, letting $v := \max\{2\sqrt{2}, v_0\}$ and $a = \max\{8, a_0\}$, we have for any $N \ge 1$ and $i \in [p]$ that $\chi_i$ is sub-exponential with parameters $(v, a)$.
\end{proof}

We now prove Lemma \ref{lemm:take_one_out} utilizing the result of Lemma \ref{lemm:chi_subexp}.

\begin{proof}[Proof of Lemma \ref{lemm:take_one_out}]
Recall the specification of independent prior distribution $a_{i, z} \sim \mathrm{Uniform}(0, 1)$.
Integrating out the parameters $\ba$ from the Bernoulli mixture model \eqref{eq:mix_model} by conjugacy, we obtain
\begin{align*}
&\qquad
\PP_{\mix}(\yb^{(1:N)} ~|~ z^{(1:N)})
\\&=
\int \PP_{\mix}(\yb^{(1:N)} ~|~ z^{(1:N)}, \ba) \PP(\ba) \ud\ba
\\&=
\prod_{i = 1}^p \prod_{h = 1}^N \int_0^1 a_{i, h}^{\sum_{n = 1}^N y_i^{(n)} 1_{z^{(n)} = h}} (1 - a_{i, h})^{\sum_{n = 1}^N (1 - y_i^{(n)}) 1_{z^{(n) = h}}} \ud a_{i, h}
\\&=
\prod_{i = 1}^p \prod_{h = 1}^N \frac{\Gamma\left( \sum_{n = 1}^N y_i^{(n)} 1_{z^{(n)} = h} + 1 \right) \Gamma\left( \sum_{n = 1}^N (1 - y_i^{(n)}) 1_{z^{(n)} = h} + 1 \right)}{\Gamma\left( \sum_{n = 1}^N 1_{z^{(n)} = h} + 2 \right)}
,
\end{align*}
where $\Gamma(t) := \int_0^\infty x^{t - 1} e^{-x} \ud x$ denotes the Gamma function. We are taking the product over the clusters $h = 1, \ldots, N$ since there are at most $N$ clusters for $N$ observations.

Recall the notation of $N_i := \sum_{n = 1}^N y_i^{(n)}$ and $\chi_i := \frac{2N_i - N}{\sqrt{N}}$, with $N_i = \frac{N}{2} + \frac{\sqrt{N}}{2} \chi_i$.
For an arbitrary observation $n \in [N]$, we study the likelihood ratio of forming $[N]$ into a cluster containing only the observation $n$ and a cluster containing all remaining $N - 1$ observations, i.e. $z^{(1:N)} = \{\{n\}, [N] \cap \{n\}^c\}$, to forming $[N]$ into a single shared cluster, i.e. $z^{(1:n)} = \{[N]\}$:
\begin{align*}
&\qquad
\frac{
\PP_{\mix}(\yb^{(1:N)} ~|~ z^{(1:N)} = \{\{n\}, [N] \cap \{n\}^c\})
}{
\PP_{\mix}(\yb^{(1:N)} ~|~ z^{(1:N)} = \{[N]\})
}
\\&=
\prod_{i = 1}^p \frac{
\frac{\Gamma(y_i^{(n)} + 1) \Gamma(1 - y_i^{(n)} + 1)}{\Gamma(1 + 2)}
\cdot
\frac{\Gamma(N_i - y_i^{(n)} + 1) \Gamma(N - N_i - (1 - y_i^{(n)}) + 1)}{\Gamma(N - 1 + 2)}
}{
\frac{\Gamma(N_i + 1) \Gamma(N - N_i + 1)}{\Gamma(N + 2)}
}
\\&=
\prod_{i = 1}^p \frac{N + 1}{2 N_i^{y_i^{(n)}} (N - N_i)^{1 - y_i^{(n)}}}
.
\end{align*}
Taking logarithm and summing over $n = 1, \ldots, N$, we obtain
{\allowdisplaybreaks\begin{align*}
&\qquad
\sum_{n = 1}^N \log\frac{
\PP_{\mix}( \yb^{(1:N)} ~|~ z^{(1:N)} = \{\{n\}, [N] \cap \{n\}^c\})
}{
\PP_{\mix}( \yb^{(1:N)} ~|~ z^{(1:N)} = \{[N]\})
}
\\&=
\sum_{i = 1}^p \left(
N \log\frac{N + 1}{2}
-
N_i \log N_i
-
(N - N_i) \log (N - N_i)
\right)
\\&=
\sum_{i = 1}^p \Bigg(
\left(
N \log \frac{N}{2}
+
N \log \left( 1 + \frac{1}{N} \right)
\right)
\\&\hspace{15mm}-
\left(
\frac{N}{2} \log \frac{N}{2}
+
\frac{N}{2} \log \left( 1 + \frac{\chi_i}{\sqrt{N}} \right)
+
\frac{\sqrt{N}}{2} \chi_i \log \left( \frac{N}{2} + \frac{\sqrt{N}}{2} \chi_i \right)
\right)
\\&\hspace{15mm}-
\left(
\frac{N}{2} \log \frac{N}{2}
+
\frac{N}{2} \log \left( 1 - \frac{\chi_i}{\sqrt{N}} \right)
-
\frac{\sqrt{N}}{2} \chi_i \log \left( \frac{N}{2} - \frac{\sqrt{N}}{2} \chi_i \right)
\right)
\Bigg)
\\&=
\sum_{i = 1}^p \left(
N \log \left( 1 + \frac{1}{N} \right)
-
\frac{N}{2} \left(
\log \left( 1 - \frac{\chi_i^2}{N} \right)
+
\frac{\chi_i}{\sqrt{N}} \log \left( 1 + \frac{\chi_i}{\sqrt{N}} \right)
-
\frac{\chi_i}{\sqrt{N}} \log \left( 1 - \frac{\chi_i}{\sqrt{N}} \right)
\right)
\right)
.
\end{align*}}

For the function $f(x) := \log(1 - x^2) + x \log (1 + x) - x \log (1 - x)$ with $-1 < x < 1$, using Taylor series $\log(1 + x) = \sum_{k = 1}^\infty \frac{(-1)^{k + 1} x^k}{k}$, we have
$$
f(x)
=
\sum_{k = 1}^\infty - \frac{x^{2k}}{k}
+
\sum_{k = 1}^\infty \frac{(-1)^{k + 1} x^{k + 1}}{k}
-
\sum_{k = 1}^\infty - \frac{x^{k + 1}}{k}
=
\sum_{k = 1}^\infty \frac{x^{2k}}{k (2k - 1)}
.
$$
Since $\sum_{k = 1}^\infty \frac{1}{k (2k - 1)} = \log 4 < \infty$, we can extend the domain of $f(x)$ to $[-1, 1]$ and obtain the elementary inequality $f(x) \le (\log 4) x^2$.
Noticing for $N \ge 2$ that $N \log\big(1 + \frac{1}{N}\big) \ge 2 \log(1 + \frac12) = \log\frac94$ and $\big|\frac{\chi_i}{\sqrt{N}}\big| \le 1$, we further obtain
\begin{align*}
&\qquad
\sum_{n = 1}^N \log\frac{
\PP_{\mix}(\yb^{(1:N)} ~|~ z^{(1:N)} = \{\{n\}, [N] \cap \{n\}^c\})
}{
\PP_{\mix}(\yb^{(1:N)} ~|~ z^{(1:N)} = \{[N]\})
}
\\&=
\sum_{i = 1}^p \left(
N \log \left( 1 + \frac{1}{N} \right)
-
\frac{N}{2} f\left( \frac{\chi_i}{\sqrt{N}} \right)
\right)
\\&\ge
\sum_{i = 1}^p \left(
\log\frac94 - \frac{N}{2} (\log 4) \frac{\chi_i^2}{N}
\right)
\\&=
(\log 2) p \left(
2(\log_2 3 - 1)
-
\frac{1}{p} \sum_{i = 1}^p \chi_i^2
\right)
.
\end{align*}

Lemma \ref{lemm:chi_subexp} shows that for some constants $v, a > 0$, $\chi_i^2$ is sub-exponential with parameters $(v, a)$ for all $N \ge 1$ and $i \in [p]$.
Applying Bernstein-type concentration inequality (e.g. Proposition 2.9 in \citet{wainwright2019high}), we obtain
$$
\PP_{\mix}^*\left(
\frac{1}{p} \sum_{i = 1}^p \chi_i^2
>
1.1
\right)
\le
\exp\left(
- \min\left\{
\frac{1.21}{2v^2}
,~
\frac{1.1}{2 a}
\right\} p
\right)
.
$$
Then we have
\begin{align*}
&\qquad
\PP_{\mix}^*\left(
\sum_{n = 1}^N \log\frac{
\PP_{\mix}(\yb^{(1:N)} ~|~ z^{(1:N)} = \{\{n\}, [N] \cap \{n\}^c\})
}{
\PP_{\mix}(\yb^{(1:N)} ~|~ z^{(1:N)} = \{[N]\})
}
<
(\log 2) (2 \log_2 3 - 3.1) p
\right)
\\&\le
\PP_{\mix}^*\left(
\frac{1}{p} \sum_{i = 1}^p \chi_i^2
>
1.1
\right)
\le
\kappa^{-p}
\end{align*}
for $\kappa := \exp\big(\min\big\{ \frac{1.21}{2v^2},~ \frac{1.1}{2 a} \big\}\big) > 1$.
Since $\sum_{p = 1}^\infty \kappa^{-p} = \frac{1}{\kappa - 1} < \infty$, by Borel-Cantelli lemma we obtain
$$
\liminf_{p \to \infty}
\sum_{n = 1}^N \log\frac{
\PP_{\mix}(\yb^{(1:N)} ~|~ z^{(1:N)} = \{\{n\}, [N] \cap \{n\}^c\})
}{
\PP_{\mix}(\yb^{(1:N)} ~|~ z^{(1:N)} = \{[N]\})
}
\ge
\liminf_{p \to \infty} (\log 2) (2 \log_2 3 - 3.1) p
=
\infty
,
$$
$\PP_{\mix}^*$-almost surely.

By concavity of function $\log x$, applying Jensen's inequality we further have
\begin{align*}
&\qquad
\liminf_{p \to \infty} \frac{
\sum_{n = 1}^N \PP_{\mix}(\yb^{(1:N)} ~|~ z^{(1:N)} = \{\{n\}, [N] \cap \{n\}^c\})
}{
\PP_{\mix}(\yb^{(1:N)} ~|~ z^{(1:N)} = \{[N]\})
}
\\&=
N \exp\left(
\liminf_{p \to \infty} \log\left(
\frac{1}{N} \sum_{n = 1}^N \frac{
\PP_{\mix}(\yb^{(1:N)} ~|~ z^{(1:N)} = \{\{n\}, [N] \cap \{n\}^c\})
}{
\PP_{\mix}(\yb^{(1:N)} ~|~ z^{(1:N)} = \{[N]\})
}
\right)
\right)
\\&\ge
N \exp\left(
\liminf_{p \to \infty} \frac{1}{N} \sum_{n = 1}^N \log\frac{
\PP_{\mix}(\yb^{(1:N)} ~|~ z^{(1:N)} = \{\{n\}, [N] \cap \{n\}^c\})
}{
\PP_{\mix}(\yb^{(1:N)} ~|~ z^{(1:N)} = \{[N]\})
}
\right)
\\&=
\infty
,
\end{align*}
$\PP_{\mix}^*$-almost surely.
\end{proof}

\section{Bayes Oracle Clustering Property}\label{supp_sec:oracle}

In this section, we focus on the Bayes oracle clustering property of our model.
As promised in Section \ref{ssec:escape}, we present
simpler forms of the sufficient conditions of Assumption \ref{assu:oracle_linear} in Section \ref{supp_ssec:oracle_assu}.
We prove Theorem \ref{theo:oracle} in Section \ref{supp_ssec:oracle_proof}.

\subsection{Sufficient Conditions for Assumption \ref{assu:oracle_linear}}\label{supp_ssec:oracle_assu}

Although Assumption \ref{assu:oracle_linear} has a complex form, simpler sufficient conditions can be obtained. For example, Assumption \ref{assu:oracle_linear} holds if, asymptotically, the average of $(\beta_{i, 0}^*, \gb_i^* \circ \bbeta_i^*)$ remains within a bounded distance of $(\upsilon_0^*, \bupsilon^*)$, and the entries of $\Bb$ are uniformly bounded.
Here, $(\beta_{i, 0}^*, \gb_i^* \circ \bbeta_i^*)$ and $(\upsilon_0^*, \bupsilon^*)$ represent $(q + 1)$-dimensional vectors, obtained by left-appending $\beta_{i, 0}^*$ and $\upsilon_0^*$ to $\gb_i^* \circ \bbeta_i^*$ and $\bupsilon^*$, respectively.
This is formally stated in the following lemma.

\begin{lemma}\label{lemm:oracle_linear}
Suppose
$$
\limsup_{p \to \infty} \left\|
\frac{1}{p} \sum_{i = 1}^p (\beta_{i, 0}^*, \gb_i^* \circ \bbeta_i^*) - (\upsilon_0^*, \bupsilon^*)
\right\|
<
\delta'
$$
and
$$
|\beta_{i, j}^*|
\le
B
\quad
\forall i \in \NN, j \in [0, q]
,
$$
then Assumption \ref{assu:oracle_linear} holds for some $\delta > 0$.
\end{lemma}

The proof of Lemma \ref{lemm:oracle_linear} is as follows.

\begin{proof}[Proof of Lemma \ref{lemm:oracle_linear}]
For the function $\logistic(x) = \frac{e^x}{1 + e^x}$ with domain $[-(q + 1)B, (q + 1)B]$, we have linearization
$$
\left|
\logistic(x) - \frac12 - \frac{1}{2(q + 1)B} x
\right|
\le
c
$$
for some $c > 0$.
The function $f(x) := \logistic(x) - \frac12 - \frac{1}{2(q + 1)B}x$ with $x \ge 0$ has derivative $f'(x) = \frac{e^x}{(1 + e^x)^2} - \frac{1}{2(q + 1)B}$ and achieves its maximum at $x^* := \log\big( (q + 1)B - 1 + \sqrt{(q + 1)^2 B^2 - 2(q + 1)B} \big)$.
By symmetry of $f(x)$ with respect to the origin, we can take $c := \max\{f(x^*), -f((q + 1)B)\}$.

As the inverse function of $\logistic(x)$, we also obtain the linearization of $\logit(x)$ as
$$
\left|
\logit(x) - 2(q + 1)B x + (q + 1)B
\right|
\le
2(q + 1)Bc
.
$$

Then for each $n \in [N]$, we have
\begin{align*}
&\qquad
\logit\left(
\frac{1}{p} \sum_{i = 1}^p \logistic\big(
\beta_{i, 0}^* + (\gb_i^* \circ \bbeta_i^*)^\top \wb_*^{(n)}
\big)
\right)
-
\left(
\upsilon_0^*
+
(\bupsilon^*)^\top \wb_*^{(n)}
\right)
\\&\le
2(q + 1)B \left(
\frac{1}{p} \sum_{i = 1}^p \logistic\big(
\beta_{i, 0}^* + (\gb_i^* \circ \bbeta_i^*)^\top \wb_*^{(n)}
\big)
-
\logistic\big(
\upsilon_0^* + (\bupsilon^*)^\top \wb_*^{(n)}
\big)
\right)
+
4(q + 1)Bc
\\&=
\frac{2(q + 1)B}{p} \sum_{i = 1}^p  \left(
\logistic\big(
\beta_{i, 0}^* + (\gb_i^* \circ \bbeta_i^*)^\top \wb_*^{(n)}
\big)
-
\logistic\big(
\upsilon_0^* + (\bupsilon^*)^\top \wb_*^{(n)}
\big)
\right)
+
4(q + 1)Bc
\\&\le
\frac{2(q + 1)B}{p} \sum_{i = 1}^p  \left(
\frac{1}{2(q + 1)B} \left(
(\beta_{i, 0}^* - \upsilon_0^*)
+
(\gb_i^* \circ \bbeta_i^* - \bupsilon^*)^\top \wb_*^{(n)}
\right)
+
2c
\right)
+
4(q + 1)Bc
\\&=
\left(
\frac{1}{p} \sum_{i = 1}^p (\beta_{i, 0}^*, \gb_i^* \circ \bbeta_i^*)
-
(\upsilon_0^*, \bupsilon^*)
\right)^\top (1, \wb_*^{(n)})
+
8(q + 1)Bc
.
\end{align*}
The other side of this inequality holds similarly, i.e.
\begin{align*}
&\qquad
\logit\left(
\frac{1}{p} \sum_{i = 1}^p \logistic\big(
\beta_{i, 0}^* + (\gb_i^* \circ \bbeta_i^*)^\top \wb_*^{(n)}
\big)
\right)
-
\left(
\upsilon_0^*
+
(\bupsilon^*)^\top \wb_*^{(n)}
\right)
\\&\ge
\left(
\frac{1}{p} \sum_{i = 1}^p (\beta_{i, 0}^*, \gb_i^* \circ \bbeta_i^*)
-
(\upsilon_0^*, \bupsilon^*)
\right)^\top (1, \wb_*^{(n)})
-
8(q + 1)Bc
.
\end{align*}

Combining both sides of the above inequality and noticing $\|(1, \wb_*^{(n)})\| \le \sqrt{q + 1}$, we obtain for all $n \in [N]$
\begin{align*}
&\qquad
\limsup_{p \to \infty} \left| \logit\left(
\frac{1}{p} \sum_{i = 1}^p \frac{\exp\big( \beta_{i, 0}^* + (\gb_i^* \circ \bbeta_i^*)^\top \wb_*^{(n)} \big)}{1 + \exp\big( \beta_{i, 0}^* + (\gb_i^* \circ \bbeta_i^*)^\top \wb_*^{(n)} \big)}
\right)
-
\upsilon_0^*
-
(\bupsilon^*)^\top \wb_*^{(n)}
\right|
\\&\le
\limsup_{p \to \infty} \left|
\left(
\frac{1}{p} \sum_{i = 1}^p (\beta_{i, 0}^*, \gb_i^* \circ \bbeta_i^*)^\top - (\upsilon_0^*, \bupsilon^*)
\right)^\top (1, \wb_*^{(n)})
\right|
+
8(q + 1)Bc
\\&\le
\sqrt{q + 1} \delta' + 8(q + 1)Bc
.
\end{align*}
This suggests that Assumption \ref{assu:oracle_linear} holds with any $\delta > \sqrt{q + 1} \delta' + 8(q + 1)Bc$.
\end{proof}

\subsection{Proof of Theorem \ref{theo:oracle}}\label{supp_ssec:oracle_proof}

For true attribute vectors $\wb_*^{(1:N)}$ satisfying Assumption \ref{assu:oracle_w}, we define the collection of all $\wb^{(1:N)}$ obtained by permuting the $q$ attributes and interchanging the zero/one of each attribute as
\begin{equation}\label{eq:cW}
\cW(\wb_*^{(1:N)})
:=
\left\{
\wb^{(1:N)} \in \{0, 1\}^{N \times q}:~
\exists \sP, \sW_1, \ldots, \sW_q
\text{ s.t. }
w_j^{(n)} = \sW_j(w_{*, \sP(j)}^{(n)})
,~
\forall n \in [N], j \in [q]
\right\}
,
\end{equation}
where $\sP$ is an arbitrary permutation map on $[q]$, and each $\sW_j$ $(j \in [q])$ is an arbitrary permutation map on $\{0, 1\}$.
For instance, for $\wb_*^{(1:4)} = \{(0, 0), (0, 1), (1, 0), (1, 1)\}$ with $q = 2$, we have
\begin{align*}
&\qquad
\cW(\wb_*^{(1:4)})
\\&=
\Bigg\{
\left\{
\wb^{(1)} = (a, b)
,~
\wb^{(2)} = (a, 1 - b)
,~
\wb^{(3)} = (1 - a, b)
,~
\wb^{(4)} = (1 - a, 1 - b)
\right\}
:~
a, b \in \{0, 1\}
\Bigg\}
\\&\quad\bigcup
\Bigg\{
\left\{
\wb^{(1)} = (b, a)
,~
\wb^{(2)} = (1 - b, a)
,~
\wb^{(3)} = (b, 1 - a)
,~
\wb^{(4)} = (1 - b, 1 - a)
\right\}
:~
a, b \in \{0, 1\}
\Bigg\}
,
\end{align*}
which includes a total of 8 possible combinations of $\wb^{(1:4)}$.

Importantly, for a prior that is i.i.d. and symmetric on $\Ab$, attribute vectors $\wb^{(1:N)} \in \cW(\wb_*^{(1:N)})$ generate the same conditional probabilities of $z^{(1:N)} ~|~ \xb^{(1:N)}, \wb^{(1:N)}$ as the oracle probabilities defined in Definition \ref{defi:oracle_prob}.
This is stated formally in the following lemma, whose proof is deferred to Section \ref{supp_ssec:oracle_lemmas}.

\begin{lemma}\label{lemm:cW_same_prob}
Let the true attribute vectors $\wb_*^{(1:N)}$ satisfy Assumption \ref{assu:oracle_w}.
If the prior distribution is i.i.d. and symmetric on $\Ab$, then for any attribute vectors $\wb^{(1:N)} \in \cW(\wb_*^{(1:N)})$, we have
$$
\PP(z^{(1:N)} ~|~ \xb^{(1:N)}, \wb^{(1:N)})
=
\PP(z^{(1:N)} ~|~ \xb^{(1:N)}, \wb_*^{(1:N)})
,\quad
\forall z^{(1:N)} \in [d]^N
.
$$
\end{lemma}

For any data dimensionality $p \ge 1$ and each observation $n \in [N]$, we define the event
\begin{equation}
\cH_{p, n}
:=
\left\{
\left|
\logit\left(
\frac{1}{p} \sum_{i = 1}^p y_i^{(n)}
\right)
-
\upsilon_0^*
-
(\bupsilon^*)^\top \wb_*^{(n)}
\right|
<
\delta
\right\}
.
\end{equation}
We find that, on the intersection of events $\bigcap_{n = 1}^N \cH_{p, n}$, the posterior distribution of $\wb^{(1:N)} ~|~ \xb^{(1:N)}, \yb^{(1:N)}$ concentrates on $\cW(\wb_*^{(1:N)})$, as presented in the next lemma.

\begin{lemma}\label{lemm:w_in_cW}
Let Assumptions \ref{assu:oracle_w}, \ref{assu:oracle_linear}, \ref{assu:oracle_inequality} hold.
On the event $\bigcap_{n = 1}^N \cH_{p, n}$, we have
$$
\PP(\wb^{(1:N)} \in \cW(\wb_*^{(1:N)}) ~|~ \xb^{(1:N)}, \yb^{(1:N)})
=
1
.
$$
\end{lemma}

Moreover, under the data generating distribution $\PP(\xb^{(1:N)}, \yb^{(1:N)} ~|~ \Ab^*, \Bb^*, \bGamma^*, \Gb^*)$, the complement of each event $\cH_{p, n}$ has probability decaying exponentially fast in $p$ as $p \to \infty$, as stated in the following lemma.

\begin{lemma}\label{lemm:event_cH_prob}
Let Assumptions \ref{assu:oracle_w}, \ref{assu:oracle_linear}, \ref{assu:oracle_inequality} hold.
Then there exists constants $p_0 \in \NN$ and $c > 0$, such that for all $p \ge p_0$ and $n \in [N]$, we have
$$
\PP(\cH_{p, n}^c ~|~ \wb_*^{(n)}, \Gb^*, \Bb^*)
\le
2e^{- c p}
.
$$
\end{lemma}

The proofs of Lemmas \ref{lemm:w_in_cW} and \ref{lemm:event_cH_prob} are also deferred to Section \ref{supp_ssec:oracle_lemmas}.

With the introduction of Lemmas \ref{lemm:cW_same_prob}, \ref{lemm:w_in_cW}, \ref{lemm:event_cH_prob}, we are now ready to prove Theorem \ref{theo:oracle}.

\begin{proof}[Proof of Theorem \ref{theo:oracle}]
For notational simplicity, for each $p \ge 1$ we define the event $\cH_p := \bigcap_{n = 1}^N \cH_{p, n}$.
By Lemma \ref{lemm:event_cH_prob}, taking the union bound we have
$$
\PP(\cH_p^c ~|~ \wb_*^{(1:N)}, \Gb^*, \Bb^*)
\le
\sum_{n = 1}^N \PP(\cH_{p, n}^c ~|~ \wb_*^{(n)}, \Gb^*, \Bb^*)
\le
2N e^{- cp}
,
$$
which gives $\sum_{p = 1}^\infty \PP(\cH_p^c ~|~ \wb_*^{(1:N)}, \Gb^*, \Bb^*) = \frac{2N}{e^c - 1} < \infty$.
By Borel-Cantelli lemma, we obtain that
$$
\PP\left(
\limsup_{p \to \infty} \cH_p^c
~\Big|~
\wb_*^{(1:N)}, \Gb^*, \Bb^*
\right)
=
0
$$
i.e. the event $\liminf_{p \to \infty} \cH_p$ holds, $\PP(\yb^{(1:N)} ~|~ \wb_*^{(1:N)}, \Gb^*, \Bb^*)$-almost surely.

Applying Lemmas \ref{lemm:cW_same_prob} and \ref{lemm:w_in_cW}, on the event $\cH_p$ we have
\begin{align*}
\PP(z^{(1:N)} ~|~ \xb^{(1:N)}, \yb^{(1:N)})
&=
\sum_{\wb^{(1:N)} \in \{0, 1\}^{N \times q}} \PP(z^{(1:N)} ~|~ \xb^{(1:N)}, \wb^{(1:N)}) \PP(\wb^{(1:N)} ~|~ \xb^{(1:N)}, \yb^{(1:N)})
\\&=
\sum_{\wb^{(1:N)} \in \cW(\wb_*^{(1:N)})} \PP(z^{(1:N)} ~|~ \xb^{(1:N)}, \wb^{(1:N)}) \PP(\wb^{(1:N)} ~|~ \xb^{(1:N)}, \yb^{(1:N)})
\\&=
\PP(z^{(1:N)} ~|~ \xb^{(1:N)}, \wb_*^{(1:N)}) \sum_{\wb^{(1:N)} \in \cW(\wb_*^{(1:N)})} \PP(\wb^{(1:N)} ~|~ \xb^{(1:N)}, \yb^{(1:N)})
\\&=
\PP(z^{(1:N)} ~|~ \xb^{(1:N)}, \wb_*^{(1:N)})
.
\end{align*}

By combining all above results above, we obtain
$$
\lim_{p \to \infty} \PP(z^{(1:N)} ~|~ \xb^{(1:N)}, \yb^{(1:N)})
=
\PP(z^{(1:N)} ~|~ \xb^{(1:N)}, \wb_*^{(1:N)})
,
$$
$\PP(\yb^{(1:N)} ~|~ \wb_*^{(1:N)}, \Gb^*, \Bb^*)$-almost surely.
\end{proof}

\subsection{Proofs of Auxiliary Lemmas for Theorem \ref{theo:oracle}}\label{supp_ssec:oracle_lemmas}

In the following, we present the proof of Lemma \ref{lemm:cW_same_prob}.

\begin{proof}[Proof of Lemma \ref{lemm:cW_same_prob}]
Let the deep latent classes $z^{(1:N)} \in [d]^N$ be arbitrary.
For any $\wb^{(1:N)} \in \cW(\wb_*^{(1:N)})$, by its definition in \eqref{eq:cW}, there exist a permutation map $\sP$ on $[q]$ and permutation maps $\sW_1, \ldots, \sW_q$ on $\{0, 1\}$ such that
$$
w_j^{(n)}
=
\sW_j(w_{*, \sP(j)}^{(n)})
,\quad
\forall n \in [N], j \in [q]
.
$$
Since the prior distribution is independent across entries of $\Ab$, we have
$$
\PP(\wb^{(1:N)} ~|~ z^{(1:N)})
=
\prod_{j = 1}^q \prod_{h = 1}^d \int_0^1 \PP(\alpha_{j, h}) \prod_{n \in [N]:~ z^{(n)} = h} \PP(w_j^{(n)} ~|~ z^{(n)}, \alpha_{j, h}) ~\ud\alpha_{j, h}
.
$$

For each $j \in [q]$ and $h \in [d]$, since the prior is symmetric on $\Ab$ and identical for $\alpha_{j, h}$ and $\alpha_{\sP(j), h}$, we have
{\allowdisplaybreaks\begin{align*}
&\qquad
\int_0^1 \PP(\alpha_{j, h}) \prod_{n \in [N]:~ z^{(n)} = h} \PP(w_j^{(n)} ~|~ z^{(n)}, \alpha_{j, h}) ~\ud\alpha_{j, h}
\\&=
\int_0^1 \PP(\alpha_{j, h}) \prod_{n \in [N]:~ z^{(n)} = h} \PP(\sW_j(w_{*, \sP(j)}^{(n)}) ~|~ z^{(n)}, \alpha_{j, h}) ~\ud\alpha_{j, h}
\\&=
\int_0^1 \PP(\alpha_{\sP(j), h}) \prod_{n \in [N]:~ z^{(n)} = h} \PP(\sW_j(w_{*, \sP(j)}^{(n)}) ~|~ z^{(n)}, \alpha_{\sP(j), h}) ~\ud\alpha_{\sP(j), h}
\\&=
1_{\sW_j(0) = 0} \int_0^1 \PP(\alpha_{\sP(j), h})~ \alpha_{\sP(j), h}^{\sum_{n = 1}^N w_{*, \sP(j)}^{(n)} 1_{z^{(n)} = h}} (1 - \alpha_{\sP(j), h})^{\sum_{n = 1}^N (1 - w_{*, \sP(j)}^{(n)}) 1_{z^{(n)} = h}} ~\ud\alpha_{\sP(j), h}
\\&\qquad+
1_{\sW_j(0) = 1} \int_0^1 \PP(\alpha_{\sP(j), h})~ \alpha_{\sP(j), h}^{\sum_{n = 1}^N (1 - w_{*, \sP(j)}^{(n)}) 1_{z^{(n)} = h}} (1 - \alpha_{\sP(j), h})^{\sum_{n = 1}^N w_{*, \sP(j)}^{(n)} 1_{z^{(n)} = h}} ~\ud\alpha_{j, h}
\\&=
1_{\sW_j(0) = 0} \int_0^1 \PP(\alpha_{\sP(j), h})~ \alpha_{\sP(j), h}^{\sum_{n = 1}^N w_{*, \sP(j)}^{(n)} 1_{z^{(n)} = h}} (1 - \alpha_{\sP(j), h})^{\sum_{n = 1}^N (1 - w_{*, \sP(j)}^{(n)}) 1_{z^{(n)} = h}} ~\ud\alpha_{\sP(j), h}
\\&\qquad+
1_{\sW_j(0) = 1} \int_0^1 \PP(1 - \alpha_{\sP(j), h})~ (1 - \alpha_{\sP(j), h})^{\sum_{n = 1}^N (1 - w_{*, \sP(j)}^{(n)}) 1_{z^{(n)} = h}} \alpha_{\sP(j), h}^{\sum_{n = 1}^N w_{*, \sP(j)}^{(n)} 1_{z^{(n)} = h}} ~\ud\alpha_{\sP(j), h}
\\&=
\int_0^1 \PP(\alpha_{\sP(j), h})~ \alpha_{\sP(j), h}^{\sum_{n = 1}^N w_{*, \sP(j)}^{(n)} 1_{z^{(n)} = h}} (1 - \alpha_{\sP(j), h})^{\sum_{n = 1}^N (1 - w_{*, \sP(j)}^{(n)}) 1_{z^{(n)} = h}} ~\ud\alpha_{\sP(j), h}
\\&=
\int_0^1 \PP(\alpha_{\sP(j), h}) \prod_{n \in [N]:~ z^{(n)} = h} \PP(w_{*, \sP(j)}^{(n)} ~|~ z^{(n)}, \alpha_{\sP(j), h}) ~\ud\alpha_{\sP(j), h}
.
\end{align*}
Since $\sP$ permutes the index set $[q]$, this implies
\begin{align*}
&
\prod_{j = 1}^q \int_0^1 \PP(\alpha_{j, h}) \prod_{n \in [N]:~ z^{(n)} = h} \PP(w_j^{(n)} ~|~ z^{(n)}, \alpha_{j, h}) ~\ud\alpha_{j, h}
\\=&
\prod_{j = 1}^q \int_0^1 \PP(\alpha_{j, h}) \prod_{n \in [N]:~ z^{(n)} = h} \PP(w_{*, j}^{(n)} ~|~ z^{(n)}, \alpha_{j, h}) ~\ud\alpha_{j, h}
.
\end{align*}}

Putting everything above together, we obtain for any $\wb^{(1:N)} \in \cW(\wb_*^{(1:N)})$ and any $z^{(1:N)} \in [d]^N$ that
\begin{align*}
\PP(\wb^{(1:N)} ~|~ z^{(1:N)})
&=
\prod_{j = 1}^q \prod_{h = 1}^d \int_0^1 \PP(\alpha_{j, h}) \prod_{n \in [N]:~ z^{(n)} = h} \PP(w_j^{(n)} ~|~ z^{(n)}, \alpha_{j, h}) ~\ud\alpha_{j, h}
\\&=
\prod_{j = 1}^q \prod_{h = 1}^d \int_0^1 \PP(\alpha_{j, h}) \prod_{n \in [N]:~ z^{(n)} = h} \PP(w_{*, j}^{(n)} ~|~ z^{(n)}, \alpha_{j, h}) ~\ud\alpha_{j, h}
\\&=
\PP(\wb_*^{(1:N)} ~|~ z^{(1:N)})
.
\end{align*}
This further gives
\begin{align*}
\PP(z^{(1:N)} ~|~ \xb^{(1:N)}, \wb^{(1:N)})
&=
\frac{\PP(\wb^{(1:N)} ~|~ z^{(1:N)} \PP(z^{(1:N)} ~|~ \xb^{(1:N)}))}{\sum_{\tz^{(1:N)} \in [d]^N} \PP(\wb^{(1:N)} ~|~ \tz^{(1:N)} \PP(\tz^{(1:N)} ~|~ \xb^{(1:N)}))}
\\&=
\frac{\PP(\wb_*^{(1:N)} ~|~ z^{(1:N)} \PP(z^{(1:N)} ~|~ \xb^{(1:N)}))}{\sum_{\tz^{(1:N)} \in [d]^N} \PP(\wb_*^{(1:N)} ~|~ \tz^{(1:N)} \PP(\tz^{(1:N)} ~|~ \xb^{(1:N)}))}
\\&=
\PP(z^{(1:N)} ~|~ \xb^{(1:N)}, \wb_*^{(1:N)})
,
\end{align*}
for any $z^{(1:N)} \in [d]^N$.
\end{proof}

We next give the proof of Lemma \ref{lemm:w_in_cW}.

\begin{proof}[Proof of Lemma \ref{lemm:w_in_cW}]
For simplicity, we denote $L_p^{(n)} := \logit\left( \frac{1}{p} \sum_{i = 1}^p y_i^{(n)} \right)$ for each $n \in [N]$, and define the event $\cH_p := \bigcap_{n = 1}^N \cH_{p, n}$.
On the event $\cH_p$, for any $n, m \in [N]$ we have
\begin{equation}\label{eq:L_m_n}
\left|
(L_p^{(n)} - L_p^{(m)})
-
(\bupsilon^*)^\top (\wb_*^{(n)} - \wb_*^{(m)})
\right|
<
2\delta
.
\end{equation}

For the set of values $\{L^{(n)}:~ n \in [N]\}$, we start by assigning each $L^{(n)}$ to a separate cluster and then merge the clusters through the single linkage agglomeration algorithm, so that any two clusters with minimum distance less than $\frac{\tau^*}{4}$ are merged. By Assumptions \ref{assu:oracle_w} and \ref{assu:oracle_inequality}, this procedure gives a total of $2^q$ clusters, each cluster consisting of $L^{(n)}$'s having equivalent $\wb_*^{(n)}$'s.
We denote the center, i.e., the mean of the cluster corresponding to $\wb_*^{(n)}$ by $\sL(\wb_*^{(n)})$.
While the map $\sL:~ \{0, 1\}^q \to \RR$ is so far unknown, the set of $2^q$ values $\cV := \{\sL(\wb_*^{(n)}):~ n \in [N]\}$ is deterministic given the set of values $\{L^{(n)}:~ n \in [N]\}$.
We make inference on the map $\sL$ in the following.

On the event $\cH_p$, the cluster centers also satisfy
$$
\left|
\sL(\wb_*^{(n)}) - (\bupsilon^*)^\top \wb_*^{(n)}
\right|
<
\delta
,\quad
\forall n \in [N]
.
$$
For any $n_1, n_2, n_3, n_4 \in [N]$, by \eqref{eq:L_m_n} we have
\begin{align*}
&
(\bupsilon^*)^\top (\wb_*^{(n_1)} - \wb_*^{(n_2)} - \wb_*^{(n_3)} + \wb^{(n_4)})
-
4\delta
\\\le&
(\sL(\wb_*^{(n_1)}) - \sL(\wb_*^{(n_2)}))
-
(\sL(\wb_*^{(n_3)}) - \sL(\wb_*^{(n_4)}))
\\\le&
(\bupsilon^*)^\top (\wb_*^{(n_1)} - \wb_*^{(n_2)} - \wb_*^{(n_3)} + \wb_*^{(n_4)})
+
4\delta
.
\end{align*}
By Assumption \ref{assu:oracle_inequality}, we have
\begin{equation}\label{eq:w1234}\begin{aligned}
&\qquad
\left|
(\sL(\wb_*^{(n_1)}) - \sL(\wb_*^{(n_2)}))
-
(\sL(\wb_*^{(n_3)}) - \sL(\wb_*^{(n_4)}))
\right|
<
\frac{\tau^*}{2}
\\&
\text{if and only if}
\quad
\wb_*^{(n_1)} - \wb_*^{(n_2)} - \wb_*^{(n_3)} + \wb_*^{(n_4)}
=
0
.
\end{aligned}\end{equation}

Let $\cV = \cV_1 \cup \cV_2$ be any partition of the set $\cV$ into two subsets $\cV_1, \cV_2$, both of size $2^{q - 1}$, and let $\sV$ be any bijective map from $\cV_1$ to $\cV_2$.
From \eqref{eq:w1234}, we find that $\sV$ satisfies
\begin{equation}\label{eq:ww'}
\left|
(\sL(\wb) - \sL(\wb'))
-
(\sL(\sV(\wb)) - \sL(\sV(\wb')))
\right|
<
\frac{\tau^*}{2}
\quad
\forall \wb, \wb' \in \cV_1
,
\end{equation}
if and only if $\exists \iota \in [q]$ such that all the $2^{q - 1}$ $\wb$'s in $\sL^{-1}(\cV_1)$ have the same $\iota$th entry $w_\iota$, all the $2^{q - 1}$ $\wb$'s in $\sL^{-2}(\cV_2)$ also have the same $\iota$th entry $w_\iota$, and the map $\sV:~ \wb \to \sV(\wb)$ only changes its $\iota$th entry $w_\iota$.
That is, \eqref{eq:ww'} holds for some partition $\cV = \cV_1 \cup \cV_2$ and bijective map $\sV:~ \cV_1 \to \cV_2$ if and only if there exists $\iota \in [q]$ and $k \in \{0, 1\}$ such that
\begin{equation}\label{eq:V12}
\cV_1
=
\left\{
\sL(\wb_*^{(n)}):~ n \in [N],~ w_{*, \iota}^{(n)} = k
\right\}
,\quad
\cV_2
=
\left\{
\sL(\wb_*^{(n)}):~ n \in [N],~ w_{*, \iota}^{(n)} = 1 - k
\right\}
,
\end{equation}
and
\begin{equation}\label{eq:sV}
\sV(\wb)
=
(w_1, \ldots, w_{\iota - 1}, 1 - w_\iota, w_{\iota + 1}, \ldots, w_q)
.
\end{equation}
To further understand the conditions \eqref{eq:V12} and \eqref{eq:sV}, we note that for any partition $\cV = \cV_1 \cup \cV_2$ not satisfying \eqref{eq:V12} with some $\iota \in [q]$ and $k \in \{0, 1\}$, for any bijective map $\sV:~ \cV_1 \to \cV_2$, we can always find $\wb, \wb' \in \sL^{-1}(\cV_1), \wb \ne \wb'$ such that
$$
(\wb - \wb') - (\sV(\wb) - \sV(\wb'))
\ne
\zero
,
$$
which by \eqref{eq:w1234} suggests that \eqref{eq:ww'} can not hold.

Importantly, since the constant $\tau^*$ in Assumption \ref{assu:oracle_inequality} is known, the above derivation extends to any viable choice of $\wb^{(1:N)}$, that is, any $\wb^{(1:N)}$ with posterior probability $\PP(\wb^{(1:N)} ~|~ \xb^{(1:N)}, \yb^{(1:N)}) > 0$.
Let $\wb^{(1:N)}$ be such an arbitrary choice of attribute vectors with positive posterior probability.
Then for any partition $\cV = \cV_1 \cup \cV_2$ and bijective map $\sV:~ \cV_1 \to \cV_2$ satisfying \eqref{eq:ww'}, similar to \eqref{eq:V12} and \eqref{eq:sV}, there must also exist $\iota' \in [q]$ and $k' \in \{0, 1\}$ such that
$$
\cV_1
=
\left\{
\sL(\wb_*^{(n)}):~ n \in [N],~ w_{*, \iota'}^{(n)} = k'
\right\}
,\quad
\cV_2
=
\left\{
\sL(\wb_*^{(n)}):~ n \in [N],~ w_{*, \iota'}^{(n)} = 1 - k'
\right\}
,
$$
and
$$
\sV(\wb)
=
(w_1, \ldots, w_{\iota' - 1}, 1 - w_{\iota'}, w_{\iota' + 1}, \ldots, w_q)
$$
Let $\sP$ denote the permutation maps on $[q]$ that always maps $\iota$ to its corresponding $\iota'$, and let $\sW_{\iota'} (\iota' \in [q])$ be the permutation maps on $\{0, 1\}$ such that $\sW_{\iota'}$ maps $k$ to its corresponding $k'$, then the above implies
$$
w_j^{(n)}
=
\sW_j(w_{*, \sP(j)}^{(n)})
,\quad
\forall j \in [q], n \in [N]
.
$$
Therefore, we conclude that any $\wb^{(1:N)}$ with positive posterior probability must belong to the collection of choices $\cW(\wb_*^{(1:N)})$ defined in \eqref{eq:cW}, i.e.
$$
\PP(\wb^{(1:N)} \in \cW(\wb_*^{(1:N)}) ~|~ \xb^{(1:N)}, \yb^{(1:N)})
=
1
.
$$
\end{proof}

We prove Lemma \ref{lemm:event_cH_prob} as follows.

\begin{proof}[Proof of Lemma \ref{lemm:event_cH_prob}]
By Assumption \ref{assu:oracle_linear}, there exists $p_0 \in \NN$ and $\epsilon \in (0, \delta)$ such that for all $p \in p_0$, we have
$$
\left| \logit\left(
\frac{1}{p} \sum_{i = 1}^p \frac{\exp(\beta_{i, 0}^* + (\gb_i^* \circ \bbeta_i^*)^\top \wb_*^{(n)})}{1 + \exp(\beta_{i, 0}^* + (\gb_i^* \circ \bbeta_i^*)^\top \wb_*^{(n)})}
\right)
-
\upsilon_0^*
-
(\bupsilon^*)^\top \wb_*^{(n)}
\right|
<
\delta - \epsilon
,\quad
\forall n \in [N]
.
$$
We define the set $\cX := \bigcup_{n \in [N]} \left\{ x:~ \left| \logit(x) - \upsilon_0^* - (\bupsilon^*)^\top \wb_*^{(n)} \right| \le \delta - \epsilon \right\}$.
Since the function $\logit(x)$ is continuous and $\cX$ is compact, there exists $\eth > 0$ such that for any $x \in \cX$, $|x' - x| < \eth$ implies $|\logit(x') - \logit(x)| < \epsilon$.
As $\cX$ is invariant to $p$, so is $\eth$.

Under the data generating distribution $\PP(\yb^{(1:N)} ~|~ \wb_*^{(1:N)}, \Gb^*, \Bb^*)$, each data entry $y_i^{(n)}$ is conditionally independent with
$$
\EE\left[
y_i^{(n)}
~\big|~
\wb_*^{(n)}, \Gb^*, \Bb^*
\right]
=
\frac{\exp\left(
\beta_{i, 0}^* + (\gb_i^* \circ \bbeta_i^*)^\top \wb_*^{(n)}
\right)}{1 + \exp\left(
\beta_{i, 0}^* + (\gb_i^* \circ \bbeta_i^*)^\top \wb_*^{(n)}
\right)}
,
$$
and
$$
\left|
y_i^{(n)}
-
\frac{\exp\left(
\beta_{i, 0}^* + (\gb_i^* \circ \bbeta_i^*)^\top \wb_*^{(n)}
\right)}{1 + \exp\left(
\beta_{i, 0}^* + (\gb_i^* \circ \bbeta_i^*)^\top \wb_*^{(n)}
\right)}
\right|
\le
1
.
$$
By applying Azuma-Hoeffding inequality \citep{wainwright2019high}, we obtain for each $n \in [N]$
$$
\PP\left(
\left|
\frac{1}{p} \sum_{i = 1}^p y_i^{(n)}
-
\frac{1}{p} \sum_{i = 1}^p \frac{\exp(\beta_{i, 0}^* + (\gb_i^* \circ \bbeta_i^*)^\top \wb_*^{(n)})}{1 + \exp(\beta_{i, 0}^* + (\gb_i^* \circ \bbeta_i^*)^\top \wb_*^{(n)})}
\right|
\ge
\eth
~\Bigg|~
\wb_*^{(n)}, \Gb^*, \Bb^*
\right)
\le
2\exp\left(
- \frac{\eth^2}{2} p
\right)
.
$$

Let the constant $c := \frac{\eth^2}{2}$.
Combining all above results, we obtain for each $n \in [N]$
\begin{align*}
&\qquad
\PP(\cH_{p, n}^c ~|~ \wb_*^{(n)}, \Gb^*, \Bb^*)
\\&\le
\PP\left(
\left|
\logit\left(
\frac{1}{p} \sum_{i = 1}^p y_i^{(n)}
\right)
-
\logit\left(
\frac{1}{p} \sum_{i = 1}^p \frac{\exp(\beta_{i, 0}^* + (\gb_i^* \circ \bbeta_i^*)^\top \wb_*^{(n)})}{1 + \exp(\beta_{i, 0}^* + (\gb_i^* \circ \bbeta_i^*)^\top \wb_*^{(n)})}
\right)
\right|
>
\epsilon
~\Bigg|~
\wb_*^{(n)}, \Gb^*, \Bb^*
\right)
\\&\le
\PP\left(
\left|
\frac{1}{p} \sum_{i = 1}^p y_i^{(n)}
-
\frac{1}{p} \sum_{i = 1}^p \frac{\exp(\beta_{i, 0}^* + (\gb_i^* \circ \bbeta_i^*)^\top \wb_*^{(n)})}{1 + \exp(\beta_{i, 0}^* + (\gb_i^* \circ \bbeta_i^*)^\top \wb_*^{(n)})}
\right|
>
\eth
~\Bigg|~
\wb_*^{(n)}, \Gb^*, \Bb^*
\right)
\\&\le
2e^{- c p}
.
\end{align*}
\end{proof}

\section{More on Posterior Computation}\label{supp_sec:post}

In this section, we present details on posterior computation that are deferred from Section \ref{sec:post}.
Section \ref{supp_ssec:prior} explores advanced prior formulations inspired by phylogenetic trees in joint species distribution modeling.
In Section \ref{supp_ssec:gibbs}, we derive all the full conditional distributions necessary for our data augmented Gibbs sampler.
Finally, Section \ref{supp_ssec:sub_core} introduces subsampling and coreset methods for approximate posterior computation in ultra-high-dimensional settings.

\subsection{Advanced Prior Ideas}\label{supp_ssec:prior}

For the continuous parameters $\Ab, \Bb, \bGamma$, for simplicity we can impose independent priors
\begin{equation}\label{eq:prior_cont}
\balpha_{j, h} \stackrel{iid}{\sim} \mathrm{Beta}(b, b)
,\quad
(\beta_{i, 0}, \bbeta_i) \stackrel{iid}{\sim} N(\bbm_\beta, \bV_\beta)
,\quad
(\gamma_{h, 0}, \bgamma_h) \stackrel{iid}{\sim} N(\bbm_\gamma, \bV_\gamma)
,
\end{equation}
where $(\beta_{i, 0}, \bbeta_i)$ and $(\gamma_{h, 0}, \bgamma_h)$ are $(q + 1)$- and $(p_x + 1)$-dimensional vectors, respectively, obtained by left-appending $\beta_{i, 0}$ to $\bbeta_i$ and $\gamma_{h, 0}$ to $\bgamma_h$.
This specification ensures that the prior on $\Ab$ is i.i.d. and symmetric, as required by Theorem \ref{theo:oracle} in Section \ref{sec:high_dim}.
By default, we set the hyperparameters to $b = 1, \bbm_\beta = \zero, \bV_\beta = \Ib, \bbm_\gamma = \zero, \bV_\gamma = \Ib$, resulting in a weakly informative prior \citep{gelman2006prior}.
In certain applications, Kronecker product structures may be incorporated into the prior for $\Bb$, such as using a phylogenetic tree to inform its covariance matrix in joint species distribution modeling.
We elaborate on these advanced prior ideas in the following.

As briefly discussed in Section \ref{ssec:prior}, meta covariates can sometimes take forms other than vectors, depending on the application.
For instance, in joint species distribution modeling, in addition to the biological traits of each species $\Tb$, we often also have access to a phylogenetic tree \citep{norberg2019comprehensive}.
In such a tree, each leaf node represents a species, and its path to the root of the tree encodes its evolutionary history.
Species that share a longer path to the root have a more recent common ancestor and tend to exhibit greater biological similarity.
The phylogenetic tree can be equivalently represented by a $p \times p$ affinity matrix $\Fb$, where the $(i, i')$th entry is defined as $F_{i, i'} := \frac{L_{i, i'}}{H}$, with $H$ denoting the tree's height and $L_{i, i'}$ denoting the length of the shared path to the root for species $i$ and $i'$.
In the following, we explore advanced prior ideas that incorporate meta covariates in the form of affinity matrices.

When an affinity matrix $\Fb$ is available for the $p$ dimensions of the observed data, we can incorporate it into the prior for parameters $\Bb$ through the specification
$$
\Vec(\Bb)
\sim
N(\one \otimes \bbm_\beta, \Fb \otimes \Vb_\beta)
,
$$
where $\Vec(\Bb) := (\beta_{1, 0}, \beta_{1, 1}, \ldots, \beta_{1, q}, \beta_{2, 0}, \beta_{2, 1}, \ldots, \ldots, \beta_{p, q})$ vectorizes the $p \times (q + 1)$ matrix $\Bb$ into a $p(q + 1)$-dimensional vector.
The hyperparameter $\Vb_\beta$ and affinity matrix $\Fb$ jointly structure the prior covariance between the parameters $\beta_{i, j}$, such that marginally we have
$$
(\beta_{1, j}, \ldots, \beta_{p, j})
\sim
N(m_{\beta, j} \one, V_{\beta, j, j} \Fb)
,\quad
\forall j \in [0, q]
$$
and
$$
(\beta_{i, 0}, \ldots, \beta_{i, q})
\sim
N(\bbm_\beta, F_{j, j} \Vb_\beta)
,\quad
\forall i \in [p]
.
$$

The meta covariates in vector form $\Tb$ and affinity matrix forms $\Fb$ can be transformed into each other.
Given the meta covariates in vectors $\tb_1, \ldots, \tb_p$, we can construct an affinity matrix $\Fb$ where the $(i, i')$th entry $F_{i, i'}$ captures the pairwise similarity between $\tb_i$ and $\tb_{i'}$:
$$
F_{i, i'}
:=
\exp\left(
- \frac{\kappa}{2} \|\tb_i - \tb_{i'}\|^2
\right)
,
$$
where $\kappa$ is a bandwidth hyperparameter controlling the sensitivity of the Gaussian kernel.
Conversely, given an affinity matrix $\Fb$, we can construct meta covariates in the form of $p_t$-dimensional vectors by extracting the top $p_t$ principal components of $\Fb$ or its entrywise logarithm $\log \Fb$.

\subsection{Data Augmented Gibbs Sampler}\label{supp_ssec:gibbs}

We design a Gibbs sampler using Polya-Gamma data augmentation \citep{polson2013bayesian} to perform posterior inference for our model.
Each data entry $y_i^{(n)}$ is associated with an augmented variable $\omega_i^{(n)}$, each deep latent class $z^{(n)}$ with augmented variables $\omega_h^{(n)}$ for $h \in [d]$, and each $g_{i, j}$ with an augmented variable $\omega_{i, j}$.
This augmentation ensures that the full conditional distributions of all parameters and latent variables become tractable through semi-conjugacy.

Specifically, conditioned on all other parameters and variables, the deep latent class $z$ follows a categorical distribution on $[d]$, each latent attribute $w_j$ follows a Bernoulli distribution, and each $(\gamma_{h, 0}, \bgamma_h)$ and $(\beta_{i, 0}, \bbeta_i)$ follow normal distributions.
Additionally, each $\alpha_{j, h}$ follows a Beta distribution, and each $\gb_i$ follows a categorical distribution on $\{0, 1\}^q$.
We derive the details of this Gibbs sampler in the following.

From Theorem 1 in \citet{polson2013bayesian}, for any $a \in \RR, b, c > 0$, we have the identity
$$
\frac{(e^\psi)^a}{(c + e^\psi)^b}
=
c^{a - b} 2^{-b} e^{\kappa (\psi - \log c)} \int_0^\infty e^{- \omega (\psi - \log c)^2 / 2} p(\omega) \ud \omega
,
$$
where $\kappa := a - b / 2$ and $p(\omega)$ denotes the density of the Polya-Gamma random variable $\PG(b, 0)$.
Moreover, the conditional distribution
$$
\PP(\omega ~|~ \psi)
\propto
e^{- \omega (\psi - \log c)^2 / 2} p(\omega)
$$
arising from the integrand is tractable, suggesting
$$
\omega | \psi
\sim
\PG(b, \psi - \log c)
.
$$

We associate each observed data entry $y_i^{(n)}$ with an augmented variable $\omega_i^{(n)}$, each deep latent class $z^{(n)}$ with augmented variables $\omega_h^{(n)}$ for $h \in [d]$, and each $g_{i, j}$ with an augmented variable $\omega_{i, j}$, with full conditional distributions
$$
\omega_i^{(n)} ~|~ \cdot
\sim
\PG\left(
1
,~
\beta_{i, 0} + (\gb_i \circ \bbeta_i)^\top \wb^{(n)}
\right)
,
$$
$$
\omega_h^{(n)} ~|~ \cdot
\sim
\PG\left(
1
,~
\left(
\gamma_{h, 0} + \bgamma_h^\top \xb^{(n)}
\right)
-
\log\left(
\sum_{h' = 1, h' \ne h}^d \exp\left(
\gamma_{h', 0} + \bgamma_{h'}^\top \xb^{(n)}
\right)
\right)
\right)
,
$$
and
$$
\omega_{i, j} ~|~ \cdot
\sim
\PG\left(
1
,~
\theta_{j, 0} + \tb_i^\top \btheta_j
\right)
.
$$

With the introduction of the augmented variables, the full conditional distributions of all parameters and latent variables are all tractable through semi-conjugacy.
Specifically, the full conditional distribution of the deep latent class $z^{(n)}$ is
\begin{align*}
\PP(z^{(n)} ~|~ \cdot)
&\propto
\PP(z^{(n)} ~|~ \xb^{(n)}, \bGamma) \PP(\wb^{(n)} ~|~ z^{(n)}, \Ab)
\\&\propto
\exp\left(
\gamma_{z^{(n)}, 0} + \bgamma_{z^{(n)}}^\top \xb^{(n)}
\right)
\prod_{j = 1}^q \left(
(\alpha_{j, z^{(n)}})^{w_j^{(n)}} (1 - \alpha_{j, z^{(n)}})^{1 - w_j^{(n)}}
\right)
,
\end{align*}
which is a categorical distribution over $[d]$.

When the number of latent attributes $q$ is small, we can sample $\wb$ as a block through the full conditional distribution
\begin{align*}
\PP(\wb^{(n)} ~|~ \cdot)
&\propto
\PP(\wb^{(n)} ~|~ z^{(n)}, \Ab) \PP(\yb^{(n)} ~|~ \wb^{(n)}, \Gb, \Bb)
\\&\propto
\prod_{j = 1}^q \left(
\frac{\alpha_{j, z^{(n)}}}{1 - \alpha_{j, z^{(n)}}}
\right)^{w_j^{(n)}}
\prod_{i = 1}^p \frac{\exp(y_i^{(n)}(\beta_{i, 0} + (\gb_i \circ \bbeta_i)^\top \wb^{(n)}))}{1 + \exp(\beta_{i, 0} + (\gb_i \circ \bbeta_i)^\top \wb^{(n)})}
,
\end{align*}
which is a categorical distribution over $\{0, 1\}^q$.
When $q$ is large, we can sample $\wb$ entrywise through the full conditional distribution
\begin{align*}
\PP(w_j^{(n)} ~|~ \cdot)
&\propto
\PP(w_j^{(n)} ~|~ z^{(n)}, \Ab) \PP(\yb^{(n)} ~|~ \wb^{(n)}, \Gb, \Bb)
\\&\propto
\left(
\frac{\alpha_{j, z^{(n)}}}{1 - \alpha_{j, z^{(n)}}}
\right)^{w_j^{(n)}}
\prod_{i = 1}^p \frac{\exp(y_i^{(n)}(\beta_{i, 0} + (\gb_i \circ \bbeta_i)^\top \wb^{(n)}))}{1 + \exp(\beta_{i, 0} + (\gb_i \circ \bbeta_i)^\top \wb^{(n)}))}
,
\end{align*}
which is a Bernoulli distribution.

Let the $q \times (p_t + 1)$ matrix $\bTheta$ denote the collection of hyperparameters $\theta_{j, 0}, \btheta_j$ with its $j$th row being $(\theta_{j, 0}, \btheta_j)$.
When $q$ is small, we can sample $\Gb$ in blocks of $\gb_i$ through the full conditional distribution
\begin{align*}
\PP(\gb_i ~|~ \cdot)
&\propto
\PP(\gb_i ~|~ \tb_i, \bTheta)
\prod_{n = 1}^N \PP(\yb^{(n)} ~|~ \wb^{(n)}, \Gb, \Bb)
\\&\propto
\exp\left(
\gb_i^\top \bTheta (1, \tb_i)
\right)
\prod_{n = 1}^N \frac{\exp(y_i^{(n)} (\beta_{i, 0} + (\gb_i \circ \bbeta_i)^\top \wb^{(n)}))}{1 + \exp(\beta_{i, 0} + (\gb_i \circ \bbeta_i)^\top \wb^{(n)})}
,
\end{align*}
which is a categorical distribution over $\{0, 1\}^q$.
When $q$ is large, we can sample $\Gb$ entrywise through the full conditional distribution
\begin{align*}
\PP(g_{i, j} ~|~ \cdot)
&\propto
\PP(g_{i, j} ~|~ \tb_i, \theta_{j, 0}, \btheta_j)
\prod_{n = 1}^N \PP(\yb^{(n)} ~|~ \wb^{(n)}, \Gb, \bbeta)
\\&\propto
\exp\left(
g_{i, j} (\theta_{j, 0} + \tb_i^\top \btheta_j)
\right)
\prod_{n = 1}^N \frac{\exp(y_i^{(n)} (\beta_{i, 0} + (\gb_i \circ \bbeta_i)^\top \wb^{(n)}))}{1 + \exp(\beta_{i, 0} + (\gb_i \circ \bbeta_i)^\top \wb^{(n)})}
,
\end{align*}
which is a Bernoulli distribution.

The parameter $\Ab$ is sampled entrywise from its full conditional distribution
$$
\alpha_{j, h} ~|~ \cdot
\sim
\Beta\left(
b + \sum_{n = 1}^N w_j^{(n)} 1_{z^{(n)} = h}
,~
b + \sum_{n = 1}^N (1 - w_j^{(n)}) 1_{z^{(n)} = h}
\right)
.
$$

The parameter $\Bb$ is sampled in blocks of its rows $(\beta_{i, 0}, \bbeta_i)$ from the full conditional distribution
\begin{align*}
\PP(\beta_{i, 0}, \bbeta_i ~|~ \cdot)
&\propto
\PP(\beta_{i, 0}, \bbeta_i) \prod_{n = 1}^N \PP(y_i^{(n)} ~|~ \wb^{(n)}, \gb_i, \beta_{i, 0}, \bbeta_i)
\\&\propto
\exp\Bigg(
- \frac12 ((\beta_{i, 0}, \bbeta_i) - \bbm_\beta)^\top \bV_\beta^{-1} ((\beta_{i, 0}, \bbeta_i) - \bbm_\beta)
\\&\qquad\qquad-
\sum_{n = 1}^N \frac{\omega_i^{(n)}}{2} (\beta_{i, 0}, \bbeta_i)^\top (1, \gb_i \circ \wb^{(n)}) (1, \gb_i \circ \wb^{(n)})^\top (\beta_{i, 0}, \bbeta_i)
\\&\qquad\qquad+
\sum_{n = 1}^N \big( y_i^{(n)} - \frac12 \big) (1, \gb_i \circ \wb^{(n)})^\top (\beta_{i, 0}, \bbeta_i)
\Bigg)
,
\end{align*}
which corresponds to the multivariate normal distribution $N(\bbm_{\beta, i}, \bV_{\beta, i})$ with parameters
$$
\bV_{\beta, i}
=
\left(
\bV_\beta^{-1}
+
\sum_{n = 1}^N \omega_i^{(n)} (1, \gb_i \circ \wb^{(n)}) (1, \gb_i \circ \wb^{(n)})^\top
\right)^{-1}
$$
and
$$
\bbm_{\beta, i}
=
\bV_{\beta, i} \left(
\bV_\beta^{-1} \bbm_\beta
+
\sum_{n = 1}^N \big( y_i^{(n)} - \frac12 \big) (1, \gb_i \circ \wb^{(n)})
\right)
.
$$

The hyperparameter $\bTheta$ is sampled in blocks of its rows $(\theta_{j, 0}, \btheta_j)$ from the full conditional distribution
\begin{align*}
\PP(\theta_{j, 0}, \btheta_j ~|~ \cdot)
&\propto
\PP(\theta_{j, 0}, \btheta_j) \prod_{i = 1}^p \PP(g_{i, j} ~|~ \tb_i, \theta_{j, 0}, \btheta_j)
\\&\propto
\exp\Bigg(
- \frac12 \|(\theta_{j, 0}, \btheta_j)\|^2)
-
\sum_{i = 1}^p \frac{\omega_{i, j}}{2} (\theta_{j, 0}, \btheta_j)^\top (1, \tb_i) (1, \tb_i)^\top (\theta_{j, 0}, \btheta_j)
\\&\qquad\qquad+
\sum_{i = 1}^p (g_{i, j} - \frac12) (1, \tb_i)^\top (\theta_{j, 0}, \btheta_j)
\Bigg)
,
\end{align*}
which corresponds to the multivariate normal distribution $N(\bbm_{\theta, j}, \bV_{\theta, j})$ with
$$
\Vb_{\theta, j}
=
\left(
\Ib
+
\sum_{i = 1}^p \omega_{i, j} (1, \tb_i) (1, \tb_i)^\top
\right)^{-1}
$$
and
$$
\bbm_{\theta, j}
=
\Vb_{\theta, j} \sum_{i = 1}^p \big( g_{i, j} - \frac12 \big) (1, \tb_i)
.
$$

The parameter $\bGamma$ is sampled in blocks of its rows $(\gamma_{h, 0}, \bgamma_h)$ from the full conditional distribution
\begin{align*}
&\qquad
\PP(\gamma_{h, 0}, \bgamma_h ~|~ \cdot)
\\&\propto
\PP(\gamma_{h, 0}, \bgamma_h) \prod_{n = 1}^N \PP(z^{(n)} ~|~ \xb^{(n)}, \bGamma)
\\&\propto
\exp\Bigg(
- \frac12 ((\gamma_{h, 0}, \bgamma_h) - \bbm_\gamma)^\top \bV_\gamma^{-1} ((\gamma_{h, 0}, \bgamma_h) - \bbm_\gamma)
\\&\qquad\qquad-
\sum_{n = 1}^N \frac{\omega_h^{(n)}}{2} (\gamma_{h, 0}, \bgamma_h)^\top (1, \xb^{(n)}) (1, \xb^{(n)})^\top (\gamma_{h, 0}, \bgamma_h)
\\&\qquad\qquad+
\sum_{n = 1}^N \left(
\log\left(
\sum_{h' = 1, h' \ne h}^d \exp\left(
\gamma_{h', 0} + \bgamma_{h'}^\top \xb^{(n)}
\right) \right) \omega_h^{(n)}
+
1_{z^{(n)} = h} - \frac12
\right) (1, \xb^{(n)})^\top (\gamma_{h, 0}, \bgamma_h)
\Bigg)
,
\end{align*}
which corresponds to the multivariate normal distribution $N(\bbm_{\gamma, h}, \bV_{\gamma, h})$ with
$$
\bV_{\gamma, h}
=
\left(
\bV_\gamma^{-1}
+
\sum_{n = 1}^N \omega_h^{(n)} (1, \xb^{(n)}) (1, \xb^{(n)})^\top
\right)^{-1}
$$
and
$$
\bbm_{\gamma, h}
=
\bV_{\gamma, h} \left(
\bV_\gamma^{-1} \bbm_\gamma
+
\sum_{n = 1}^N \left( \log\left(
\sum_{h' = 1, h' \ne h}^d \exp\left(
\gamma_{h', 0} + \bgamma_{h'}^\top \xb^{(n)}
\right) \right) \omega_h^{(n)} + 1_{z^{(n)} = h} - \frac12
\right) (1, \xb^{(n)})
\right)
.
$$

As a side note, incorporating the intercepts $\beta_{i, 0}, \gamma_{h, 0}, \theta_{j, 0}$ into the Gibbs sampler can be simplified in implementations by implicitly including them.
This can be achieved by augmenting each of the covariates $\xb$, meta covariates $\tb$, latent attributes $\wb$, and binary vectors $\gb_i$ with an additional dimension fixed at one.

\subsection{Subsampling and Coreset for Ultra-high-dimensional Settings}\label{supp_ssec:sub_core}

The data augmented Gibbs sampler designed in Section \ref{supp_ssec:gibbs} has a computational complexity that scales linearly with the data dimensionality $p$.
To enable more efficient posterior computation in ultra-high-dimensional settings, we introduce an approximate variant that leverages ideas of subsampling and coreset techniques.

Subsampling MCMC and coreset methods are widely used approximate sampling methods in large-scale Bayesian inference, when evaluating the full dataset likelihood becomes computationally prohibitive.
Subsampling MCMC \citep{johndrow2015optimal, johndrow2017error, quiroz2018speeding} tackles this challenge by randomly sampling a subset of data at each MCMC iteration and using its likelihood to approximate the full likelihood.
Bayesian Coreset methods \citep{huggins2016coresets, winter2024emerging} offer an alternative approach by constructing a small, weighted subset of data, called coreset, that serves as a compressed representation of the full dataset.
While these methods are typically applied in scenarios where the number of observations $N$ is large, their principles can be effectively extended to our settings where the data dimensionality $p$ is high.

In applications involving ultra-high-dimensional observations, there often exists a subset of data dimensions that exhibit significantly greater variation across observations than the rest.
For example, in joint species distribution modeling, only a small number of species are commonly observed, while the majority are rare \citep{ovaskainen2024common}.
To leverage this structure, we can preselect a coreset $\cS \subset [p]$ of informative data dimensions, e.g. the commonly observed species, with size $\card(\cS) \ll p$.

At the beginning of each Gibbs sampler iteration, we further subsample the remaining data dimensions $\cB \subset [p] \cap \cS^c$ and approximate the log-likelihood function $\PP(\yb^{(1:N)} ~|~ \wb^{(1:N)}, \Gb, \Bb)$ using the surrogate function
$$
\sum_{i \in \cS} \log\PP(y_i^{(1:N)} ~|~ \wb^{(1:N)}, \Gb, \Bb)
+
\frac{p - \card(\cS)}{\card(\cB)} \sum_{i \in \cB} \log\PP(y_i^{(1:N)} ~|~ \wb^{(1:N)}, \Gb, \Bb)
.
$$
The use of this surrogate function introduces approximation error in the stationary distribution of the MCMC sampler.
However, in practice, this trade-off is often acceptable given the significant gains in computational efficiency.
Theoretical bounds on the approximation error can be derived following the general framework of \citet{johndrow2015optimal, johndrow2017error}, though we leave this analysis for future work.

\section{Simulation Studies}\label{supp_sec:sim}

To evaluate the performance of our proposed posterior computation methods and provide numerical validation for the theoretical properties of our model, we conduct a series of simulation studies comprising three key components.
In Section \ref{supp_ssec:sim1}, we examine models with moderate data dimensionality and apply the Watanabe–Akaike information criterion \citep[\texttt{WAIC},][]{watanabe2010asymptotic} for model selection.
The results illustrate that our model selection procedure successfully identifies the true underlying model structure.
In Section \ref{supp_ssec:sim2}, we perform posterior inference using our data-augmented Gibbs sampler from Section \ref{supp_ssec:gibbs} under the selected model structure.
In alignment with the posterior consistency theory in Section \ref{sec:theo}, we show that our Gibbs sampler yields accurate estimates of the true model parameters and provides reliable inference for the latent variables.
In Section \ref{supp_ssec:sim3}, we analyze the effect of increasing the dimensionality of the data within the same model.
We show that as the dimension increases, the posterior probability of the deep latent cluster approaches its oracle probability, providing numerical evidence for the Bayes oracle clustering property established in Section \ref{ssec:escape}.

\subsection{Model Selection}\label{supp_ssec:sim1}

We begin by evaluating our model selection procedures through a simulated example.
Consider $N = 1000$ observations of binary vectors $\yb^{(1:N)}$ with dimensionality $p = 20$, where the data generating distribution follows a well-specified true model with $q = 2$ latent attributes and $d = 2$ deep latent classes.
The observation-specific covariates $\xb^{(1:N)}$ and the meta covariates $\tb^{(1:N)}$ have dimensionalities $p_x = p_t = 4$, with each $\xb^{(n)}$ and $\tb^{(n)}$ drawn independently from a $N(\zero, \Ib)$ distribution.
The true model parameters $(\Ab^*, \Bb^*, \bGamma^*, \Gb^*)$ are randomly chosen within the parameter space $\sS_1$ defined in Theorem \ref{theo:strict}, ensuring strict identifiability.
Specifically, we sample $(\Ab^*, \Bb^*, \bGamma^*, \Gb^*)$ from their prior distribution while enforcing $\Gb^*$ to contain three distinct identity blocks.
Additionally, we strengthen the effects of $\Bb^*$ by imposing a lower bound on each regression coefficient $|\beta_{i, j}^*|$.

\begin{figure}
\centering
\includegraphics[width = \textwidth]{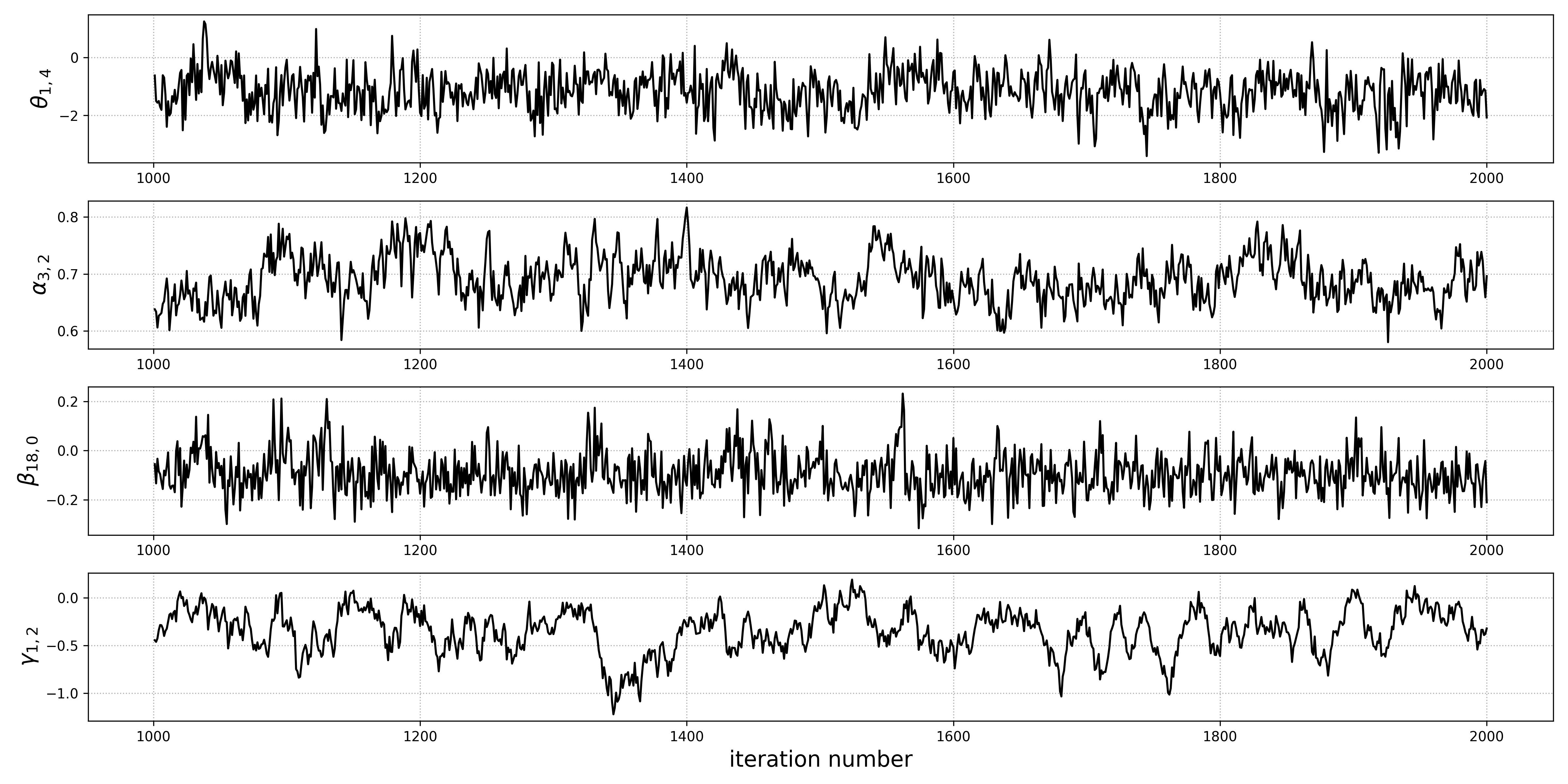}
\caption{
Trace plots for parameters $\theta_{1, 4}, \alpha_{3, 2}, \beta_{18, 0}, \gamma_{1, 2}$ sampled after the burn-in period.
}
\label{fig:sim_trace}
\end{figure}

We run 2000 iterations of our data augmented Gibbs sampler from Section \ref{supp_ssec:gibbs}, initializing the parameters randomly from their prior distribution.
Empirically, we observe that the Gibbs samples rapidly transition from random initialization to stationary distribution within 200 iterations.
Following standard practices \citep{gelman1992inference}, we discard the first half of samples as burn-in and use the remaining half to make posterior inference.
The Gibbs sampler exhibits adequate mixing, as evidenced by the trace plots of several representative parameter dimensions in Figure \ref{fig:sim_trace}.

\begin{figure}
\centering
\includegraphics[width = 0.6\textwidth]{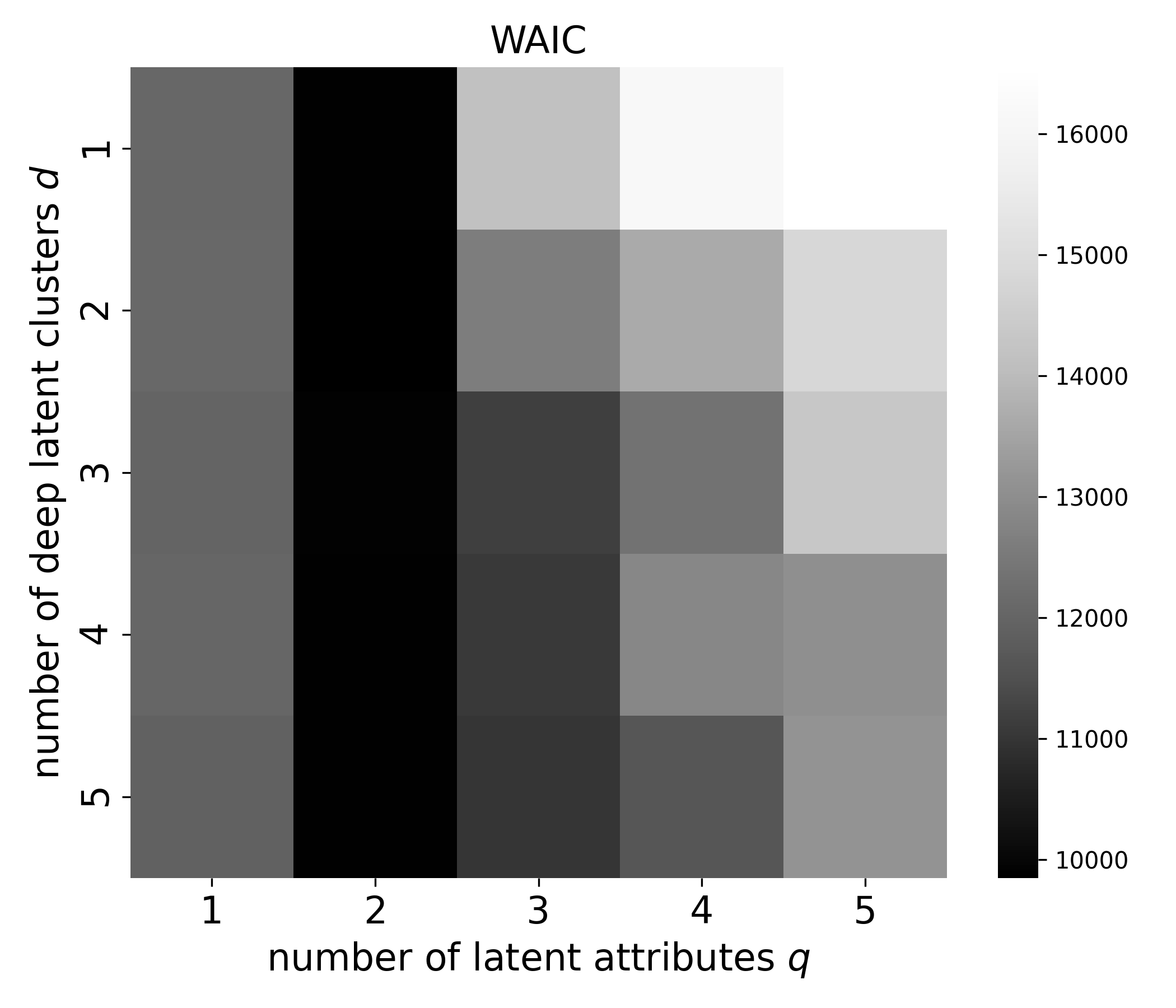}
\caption{
\texttt{WAIC} for each choice of model structure with $q \in \{1, 2, 3, 4, 5\}, d \in \{1, 2, 3, 4, 5\}$.
Darker color represents smaller \texttt{WAIC} and hence a more favorable model structure.
From top to bottom in the column of $q = 2$, the \texttt{WAIC}s are 9881.61, 9850.07, 9909.69, 9882.88, 9894.82.
}
\label{fig:sim_waic}
\end{figure}

We treat the number of latent attributes $q$ and deep latent classes $d$ as unknown and perform model selection using the Watanabe–Akaike information criterion
\citep[\texttt{WAIC},][]{watanabe2010asymptotic}.
Let $\wb^{(1:N; s)}$, $\Gb^{(s)}$, and $\Bb^{(s)}$ denote the posterior samples of the latent attributes $\wb^{(1:N)}$ and the model parameters $\Gb, \Bb$ obtained at the $s$th iteration of the Gibbs sampler.
The log likelihood of observation $\yb^{(n)}$ given the sampled latent attributes and parameters is defined as
\begin{align*}
\ell(s, n)
&:=
\log \PP(\yb^{(n)} ~|~ \wb^{(n; s)}, \Gb^{(s)}, \Bb^{(s)})
\\&=
\sum_{i = 1}^p \left(
y_i^{(n)} \left(
\beta_{i, 0}^{(s)} + (\gb_i^{(s)} \circ \bbeta_i^{(s)})^\top \wb^{(n; s)}
\right)
-
\log\left( 1 + \exp\left(
\beta_{i, 0}^{(s)} + (\gb_i^{(s)} \circ \bbeta_i^{(s)})^\top \wb^{(n; s)}
\right) \right)
\right)
.
\end{align*}
To compute \texttt{WAIC} for latent variable models, we follow \citet{gelman2014understanding, merkle2019bayesian} and use the formula
\begin{equation}\label{eq:waic}
\WAIC
:=
-2 (\lppd - p_{\WAIC})
=
-2 \sum_{n = 1}^N \log\left(
\frac{1}{S} \sum_{s = 1}^S e^{\ell(s, n)}
\right)
+
2 \sum_{n = 1}^N \var_{s = 1}^S(\ell(s, n))
,
\end{equation}
where $s = 1, \ldots, S$ indexes the Gibbs iterations after the burn-in period, $\lppd$ represents the log pointwise predictive density, $p_{\WAIC}$ approximates the effective number of parameters, and $\var_{s = 1}^S(\ell(s, n))$ denotes the sample variance of the log likelihood values $\{\ell(s, n)\}_{s = 1}^S$.
We favor model structures that yield smaller \texttt{WAIC} values.

We vary the number of latent attributes $q$ within $\{1, 2, 3, 4, 5\}$ and the number of deep latent classes $d$ within $\{1, 2, 3, 4, 5\}$.
Each combination of $q$ and $d$ defines a distinct model structure, for which we run the Gibbs sampler and compute WAIC based on the posterior samples obtained.
The \texttt{WAIC} values of all $(q, d)$ combinations are visualized as a heatmap in Figure \ref{fig:sim_waic}.
Our results indicate that \texttt{WAIC} achieves its minimum at $q = 2, d = 2$, with a smooth variation across the $q-d$ plane.
This corresponds to the true model structure.
The \texttt{WAIC} is more definitive about the choice of $q = 2$ than the choice of $d = 2$.

\subsection{Parameter Inference}\label{supp_ssec:sim2}

For the selected model structure with $q = 2$ and $d = 2$, we now evaluate the consistency of posterior inference using our Gibbs sampler.
Due to potential permutations of the $q$ latent attributes and $d$ deep latent classes, the posterior distribution of $\Ab, \Bb, \bGamma, \Gb$ possesses multiple equivalent posterior modes.
Empirically, we observe that after the burn-in period, the Gibbs sampler primarily remains near a single posterior mode and rarely transitions between modes.
Following \citet{stephens2000dealing}, we introduce a post-processing step after sampling to address label switching.
Specifically, we compute the moving average of the Gibbs samples of $\Ab$ and $(\one ~ \Gb) \circ \Bb$, then sort their columns lexicographically to align the latent attributes and the deep latent classes consistently.
For samples where $g_{i, j} \beta_{i, j}$ is close to zero, we note that while $g_{i, j} = 1$ and $g_{i, j} = 0$ have similar conditional probabilities, they imply significantly different interpretations regarding the dependency between data entry $y_i$ and attribute $w_j$.
To ensure clear model interpretation, we refine the posterior samples of $\Gb$ by setting $g_{i, j} = 0$ whenever $|\beta_{i, j}|$ is close to zero.

\begin{figure}
\centering
\includegraphics[width = \textwidth]{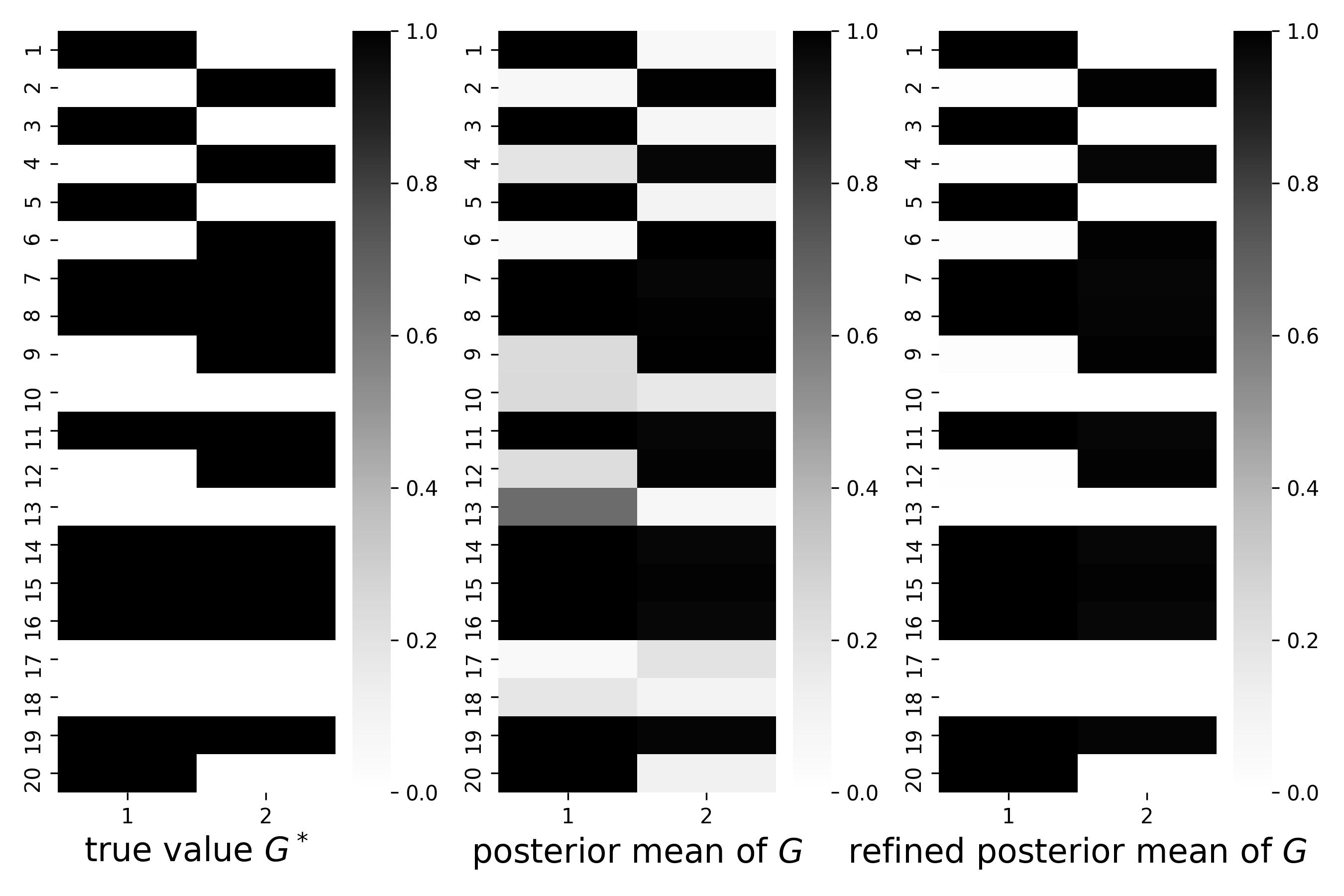}
\caption{
True parameter value $\Gb^*$, empirical posterior mean of $\Gb$, and empirical posterior mean of $\Gb$ after refinement.
}
\label{fig:sim_G}
\end{figure}

\begin{figure}
\centering
\includegraphics[width = 0.6\textwidth]{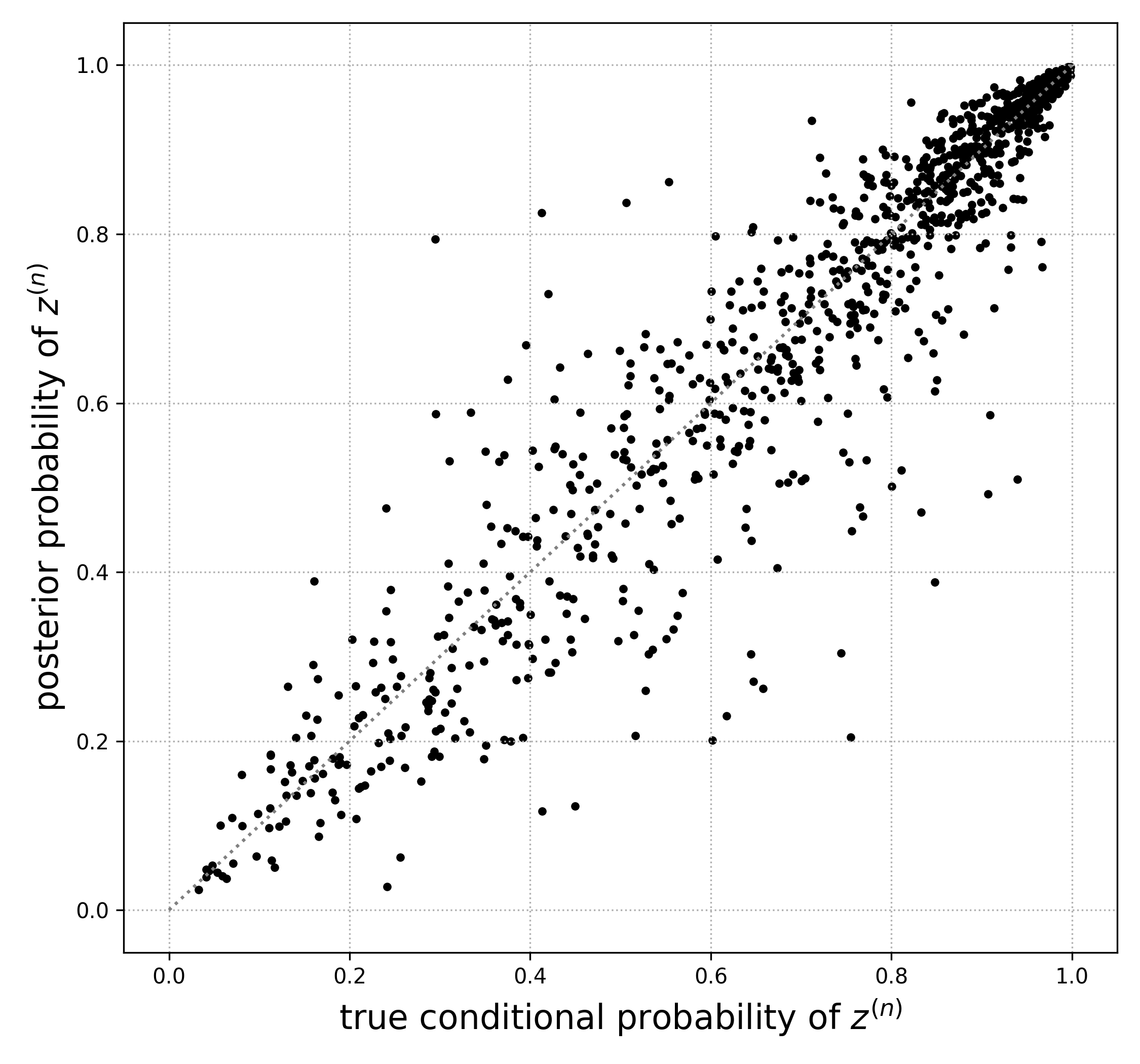}
\caption{
True conditional probability $\PP(z^{(n)} = 2 ~|~ \xb^{(n)}, \wb_*^{(n)}, \bGamma^*, \Ab^*)$ vs. empirical posterior probability $\PP(z^{(n)} = 2 ~|~ \xb^{(1:N)}, \yb^{(1:N)})$.
Each dot represents a single observation $n$.
}
\label{fig:sim_Z}
\end{figure}

We present the true parameter $\Gb^*$ and the empirical posterior mean estimate in Figure \ref{fig:sim_G}.
As observed, the posterior mean of $\Gb$ serves as a consistent estimator of $\Gb^*$, with the refinement step of setting $g_{i, j} = 0$ for small $|\beta_{i, j}|$ further improving the estimation accuracy.
This demonstrates the efficiency and reliability of our Gibbs sampler for posterior inference, while also providing numerical validation of the posterior consistency theory developed in Section \ref{sec:theo}.

We further examine the posterior samples of the latent variables in our model.
As an illustrative example, we compare the empirically estimated posterior probabilities of each deep latent class $z^{(n)}$ to its conditional probability under the true model, i.e.
$$
\PP(z^{(n)} ~|~ \xb^{(n)}, \wb_*^{(n)}, \bGamma^*, \Ab^*)
\propto
\frac{\exp(\gamma_{z^{(n)}, 0} + (\bgamma_{z^{(n)}}^*)^\top \xb^{(n)})}{\sum_{h = 1}^d \exp(\gamma_{h, 0} + (\bgamma_h^*)^\top \xb^{(n)})}
\prod_{j = 1}^q (\alpha_{j, z^{(n)}}^*)^{w_{*, j}^{(n)}} (1 - \alpha_{j, z^{(n)}}^*)^{1 - w_{*, j}^{(n)}}
,
$$
where $\wb_*^{(n)}$ represents the true attribute vector of $\yb^{(n)}$ from its data generating process, with $w_{*, j}^{(n)}$ denoting its $j$th entry.
Since our simulated example has $d = 2$, it suffices to compare the probabilities of $z^{(n)} = 2$, as shown in Figure \ref{fig:sim_Z}.
We observe that the empirical posterior probability $\PP(z^{(n)} = 2 ~|~ \xb^{(1:N)}, \yb^{(1:N)})$ closely approximates the true conditional probability $\PP(z^{(n)} = 2 ~|~ \xb^{(n)}, \wb_*^{(n)}, \bGamma^*, \Ab^*)$.
This result suggests that our posterior inference is not only accurate for model parameters but also reliable for latent variables.

\subsection{Bayes Oracle Clustering}\label{supp_ssec:sim3}

\begin{figure}
\centering
\includegraphics[width = 0.6\textwidth]{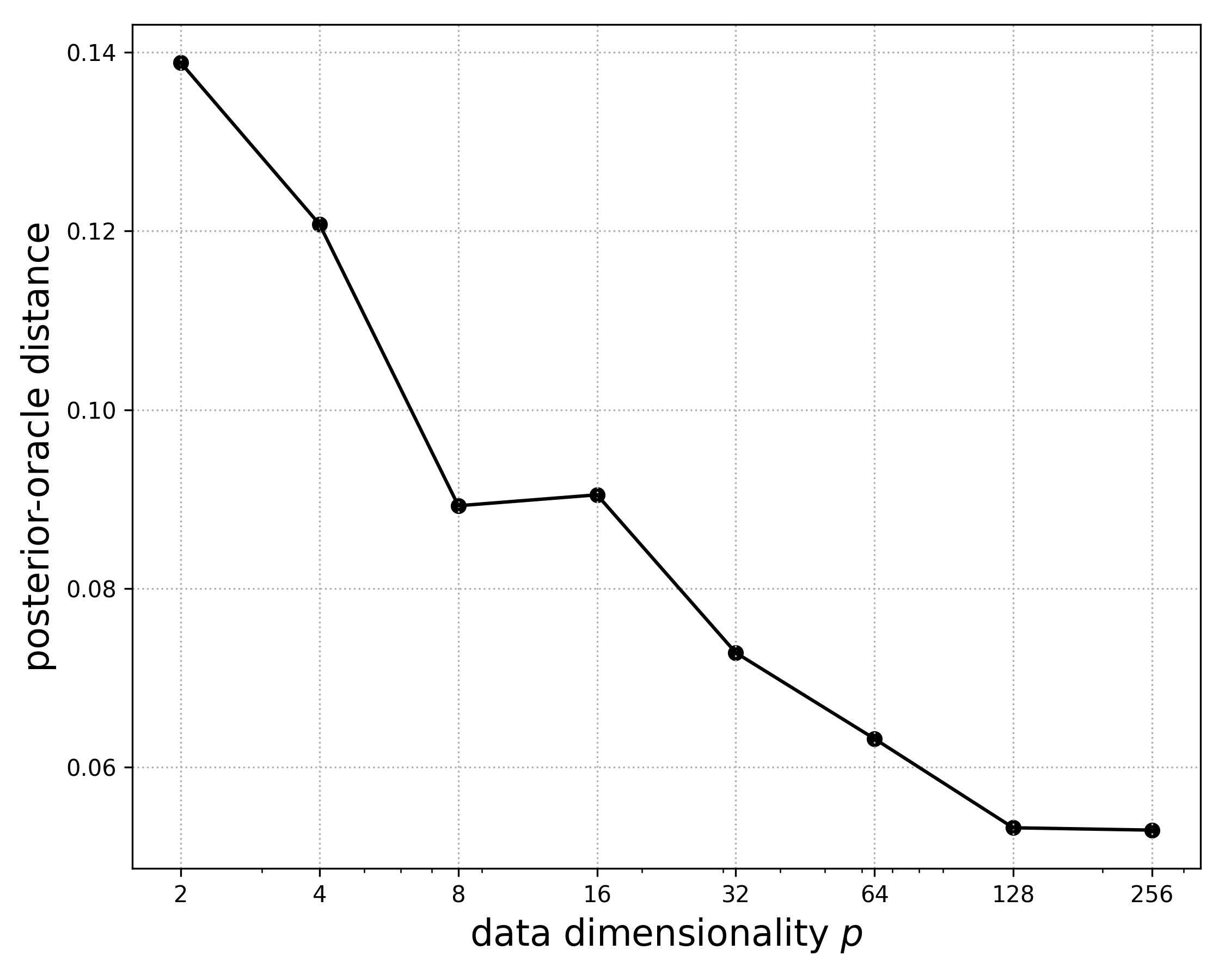}
\caption{
Data dimensionality $p$ vs. distances between empirical posterior and oracle probabilities of deep latent classes.
}
\label{fig:sim_p}
\end{figure}

Having established the consistency of our posterior inference, we now investigate the impact of data dimensionality $p$ on the inferred clustering of observations.
We vary $p$ from $2^1$ to $2^8$ and run our Gibbs sampler for each choice of $p$, keeping all other settings unchanged.
For each dimensionality $p$ and observation $n$, the posterior probability of the deep latent cluster $z^{(n)}$ is empirically estimated by its frequency among the Gibbs samples, denoted as $\bar\PP_p(z^{(n)} ~|~ \xb^{(1:N)}, \yb^{(1:N)})$.

Recalling the definition of the oracle probability $\PP(z^{(n)} ~|~ \xb^{(1:N)}, \wb_*^{(1:N)})$ from Definition \ref{defi:oracle_prob}, we further run the Gibbs sampler with the latent attributes $\wb^{(1:N)}$ fixed at their true values $\wb_*^{(1:N)}$.
This ensures that the samples of $z^{(n)}$ follow a stationary distribution given by its oracle probability.
The oracle probability is then empirically estimated by the frequency of $z^{(n)}$ in these Gibbs samples, denoted as $\bar\PP(z^{(n)} ~|~ \xb^{(1:N)}, \wb_*^{(1:N)})$.

For each $p$, we compute the mean distance between the posterior probability and the oracle probability of the deep latent cluster across all observations $n = 1, \ldots, N$:
$$
\frac{1}{Nd} \sum_{n = 1}^N \sum_{h = 1}^d \left|
\bar\PP_p(z^{(n)} = h ~|~ \xb^{(1:N)}, \yb^{(1:N)})
-
\bar\PP(z^{(n)} = h ~|~ \xb^{(1:N)}, \wb_*^{(1:N)})
\right|
.
$$
We visualize these distances for different $p$ in Figure \ref{fig:sim_p}, which reveals a clear trend: as $p$ increases, the posterior probability of the deep latent cluster approaches its oracle probability.
Note that the distances will not fully vanish in Figure \ref{fig:sim_p}, as empirical estimation using Gibbs samples introduces inherent fixed errors.
This result provides numerical validation of the Bayes oracle clustering property established in Section \ref{sec:high_dim}.

\section{More Data Analyses}\label{supp_sec:app}

In this section, we present the deferred figures and tables from Section \ref{sec:appl} and provide additional discussions.

\subsection{Model Selection}\label{supp_ssec:model_sel}

We begin by selecting the optimal model structure using the \texttt{WAIC} criterion \eqref{eq:waic}, following the approach outlined for the simulation studies in Section \ref{supp_ssec:sim1}.
We explore different model structures by varying the number of latent attributes $q$ in $\{3, 4, 5, 6, 7\}$ and the number of deep latent classes $d$ in $\{2, 3, 4, 5, 6\}$. For each combination of $q$ and $d$, we run our data augmented Gibbs sampler for 2000 iterations, discarding the first half as the burn-in period and computing the WAIC value using the remaining samples.
The \texttt{WAIC} values across model structures are visualized as a heatmap in Figure \ref{fig:app_waic}, revealing a smooth variation over the $q-d$ plane and reaching a minimum at $q = 5, d = 4$.
Based on this result, we analyze the Finnish bird dataset using a model with five latent attributes and four deep latent classes.

\begin{figure}
\centering
\includegraphics[width = 0.6\textwidth]{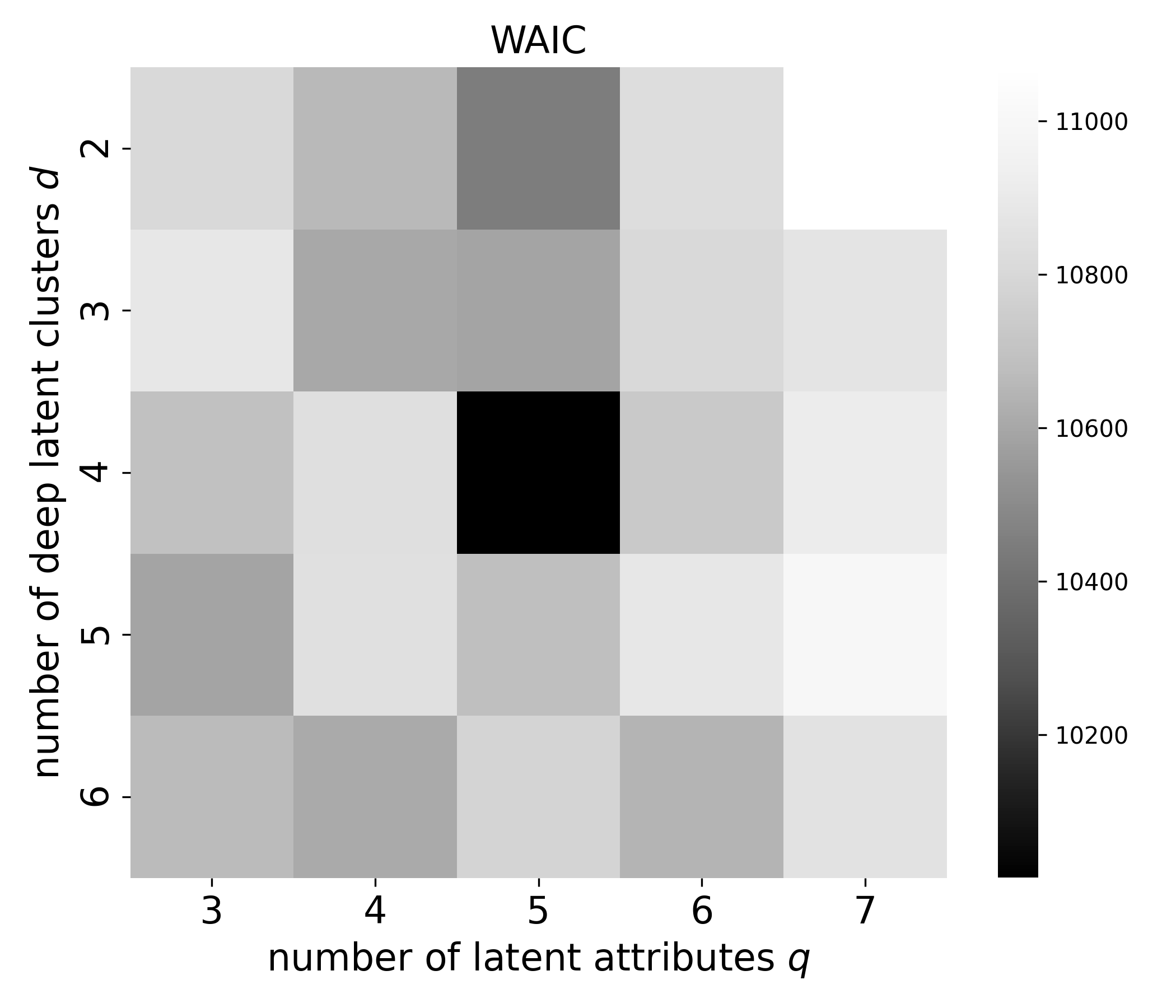}
\caption{
\texttt{WAIC} for each choice of model structure with $q \in \{3, 4, 5, 6, 7\}$ and $d \in \{2, 3, 4, 5, 6\}$.
Darker color represents smaller \texttt{WAIC} and hence a more favorable model structure.
}
\label{fig:app_waic}
\end{figure}

\subsection{Interpretations of Deep Latent Classes}\label{supp_ssec:deep_class}

We estimate the posterior probabilities of belonging to each deep latent class for each sample from the 1000 Gibbs samples after the burn-in period. Figure \ref{fig:app_z_map} provides a spatial map over Finland of the estimates obtained by
setting $z^{(n)}$ to the class with the highest posterior probability, demonstrating clear geographical patterns as discussed in Section \ref{sec:appl}.

To further interpret the inferred deep latent classes, we examine their associations with the environmental covariates. Table \ref{tab:app_z_hab} summarizes the distribution of the posterior modes of $z^{(1)}, \ldots, z^{(N)}$ across the five types of habitats.
We observe that all broadleaved forests belong to classes $z \in \{2, 3\}$, all coniferous forests fall into classes $z \in \{2, 3, 4\}$, most open habitats are assigned to $z = 1$, all urban habitats belong to class $z = 3$, and nearly all wetlands are assigned to class $z = 2$.
Conversely, examining the composition of each deep latent class, we find that class $z = 1$ primarily consists of open habitats, class $z = 4$ is dominated by coniferous forests, while classes $z \in \{2, 3\}$ exhibit more diverse habitat distributions.
These findings provide a clear ecological interpretation of the inferred deep latent classes $z$ in terms of habitat type.

\begin{table}
\centering
\begin{tabular}{c||cccc}
\hline\hline
distribution (count) & $z = 1$ & $z = 2$ & $z = 3$ & $z = 4$ \\
\hline\hline
broadleaved forests & 0 & 15 & 53 & 0 \\
coniferous forests & 0 & 36 & 108 & 31 \\
open habitats & 23 & 6 & 16 & 1 \\
urban habitats & 0 & 0 & 58 & 0 \\
wetlands & 1 & 14 & 0 & 1 \\
\hline\hline
\end{tabular}
\caption{
Distributions (counts) of observations in each deep latent class $z$ across the five types of habitats.
}
\label{tab:app_z_hab}
\end{table}

Furthermore, we analyze the temperature distribution within each deep latent class and present the result as a density plot in Figure \ref{fig:app_z_temp}.
In April and May, sampling locations in class $z = 1$ exhibit the lowest average temperature of $1.44^\circ\text{C}$, while those in class $z = 3$ have the highest average temperature of $7.01^\circ\text{C}$.
The locations assigned to the classes $z = 2$ and $z = 4$ have moderate average temperatures of $3.63^\circ\text{C}$ and $4.49^\circ\text{C}$, respectively. In ecology, ecologically similar locations are referred to as having {\em regions of common profile} \citep{foster2013modelling, foster2017ecological, scherting2024inferring}. Obtaining clear relationships with spatial location, habitat, and temperature provides evidence that our methodology produces meaningful regions of common profile.

\begin{figure}
\centering
\includegraphics[width = 0.8\textwidth]{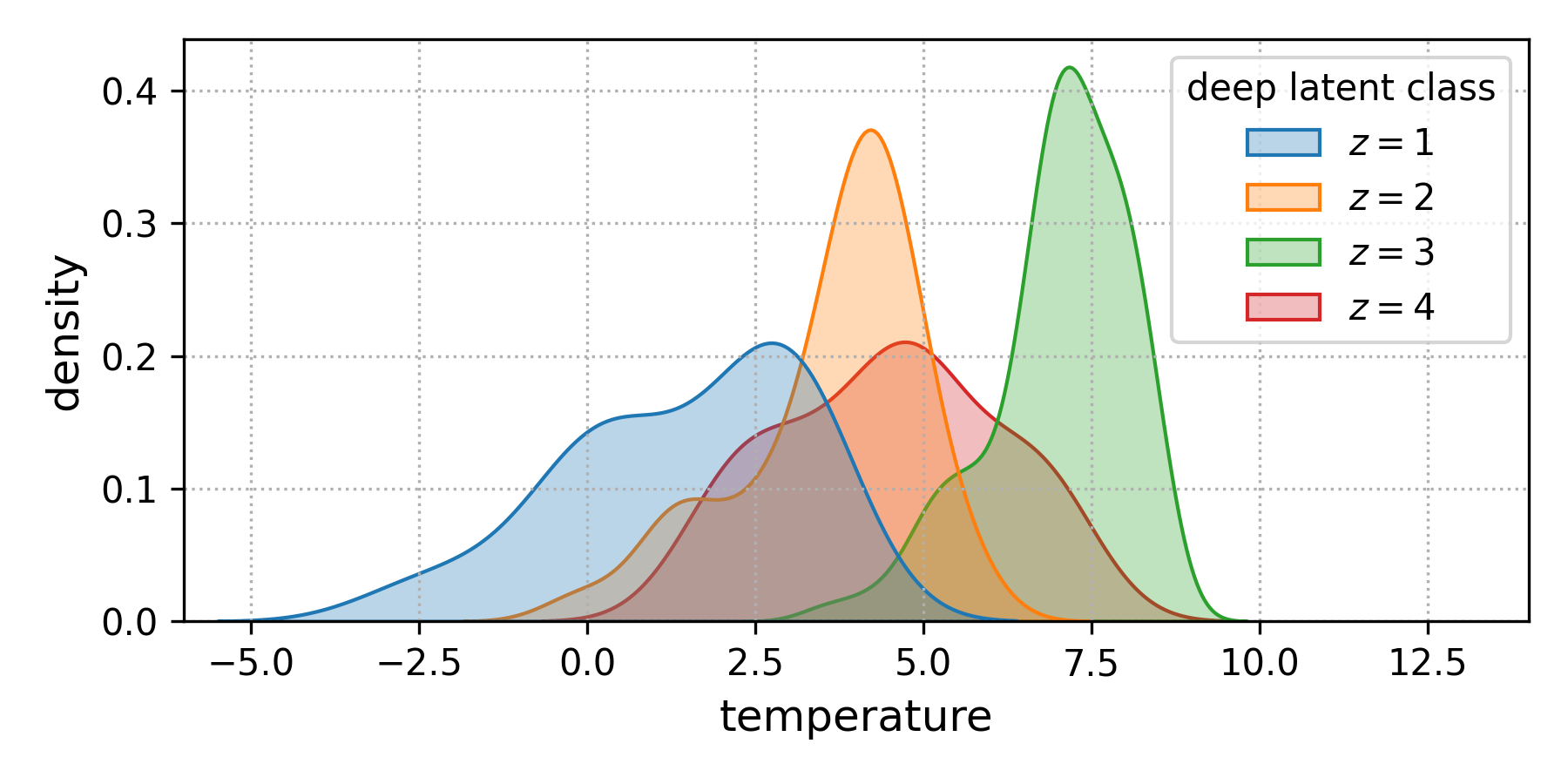}
\caption{
Distributions of temperature covariate in each deep latent class $z$.
}
\label{fig:app_z_temp}
\end{figure}

Posterior inference of the model parameter $\Ab$ further reveals distinct relationships between deep latent classes and latent attributes.
Table \ref{tab:app_A}, presents the posterior mean of $\Ab$.
Specifically, deep latent class $z = 1$ tends to enforce zero entries in the corresponding latent attributes $\wb$, whereas class $z = 3$ predominantly leads to ones in $\wb$.
In contrast, class $z = 2$ is more likely to induce $\wb = (1, 0, 0, 1, 0)$, while the most probable latent attribute configurations in class $z = 4$ are $\wb = (1, 0, 0, 0, 1)$ or $(1, 0, 1, 0, 1)$.
These patterns suggest meaningful structural differences among the deep latent classes in their influence on the latent attributes.

\begin{table}
\centering
\begin{tabular}{c||cccc}
\hline\hline
$\alpha_{j, z}$ & $z = 1$ & $z = 2$ & $z = 3$ & $z = 4$ \\
\hline\hline
$j = 1$ & 0.0511 & 0.6728 & 0.7121 & 0.9069 \\
$j = 2$ & 0.0380 & 0.0156 & 0.9923 & 0.0402 \\
$j = 3$ & 0.0422 & 0.0607 & 0.6262 & 0.4020 \\
$j = 4$ & 0.0586 & 0.9484 & 0.6661 & 0.0498 \\
$j = 5$ & 0.0616 & 0.0549 & 0.8606 & 0.9148 \\
\hline\hline
\end{tabular}
\caption{
Posterior mean of $\Ab = (\alpha_{j, z})$, where $j \in [5]$ indexes the latent attributes and $z \in [4]$ represents the deep latent class.
}
\label{tab:app_A}
\end{table}

\subsection{Interpretations of Latent Attributes}\label{supp_ssec:attribute}

For the latent attributes $w_1^{(n)}, \ldots, w_5^{(n)}$, we similarly determine their posterior modes and visualize them throughout Finland in Figure \ref{fig:app_w_map_full}, which provides a more comprehensive version of Figure \ref{fig:app_w_map}. As discussed in Section \ref{sec:appl}, the latent attributes characterize more nuanced variation than clustering (regions of common profile), leading to additional insight into variation in ecological conditions across sampling locations. 

\begin{figure}
\centering
\includegraphics[width = 0.9\textwidth]{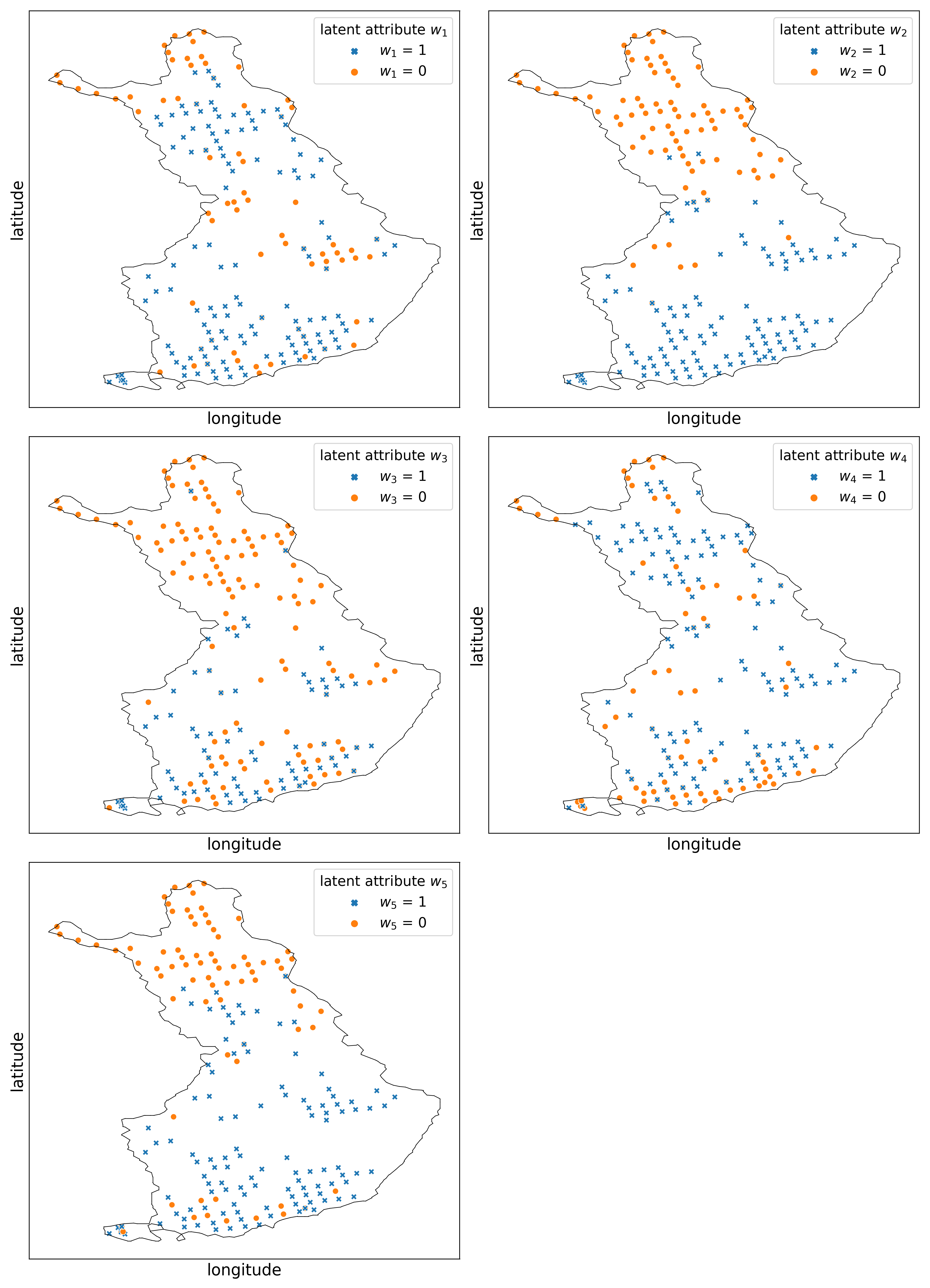}
\caption{
Spatial distribution of bird species sampling locations on the map of Finland, color-coded by their latent attributes $w_1, w_2, w_3, w_4, w_5$.
}
\label{fig:app_w_map_full}
\end{figure}

The binary matrix $\Gb$, which governs the dependencies between species presence/absence and latent attributes, provides additional insights into species associations.
We apply the same post-processing and refining approach as in the simulation studies in Section \ref{supp_ssec:sim2}.
Specifically, after addressing the potential label switching issue, we refine the posterior samples of $\Gb$ by setting $g_{i, j} = 0$ for all $|\beta_{i, j}| < 2$.
The posterior mean and posterior mode of the refined $\Gb$ samples are presented in Figure \ref{fig:app_G}.
Importantly, the inferred posterior mode of $\Gb$ satisfies the condition in Theorem \ref{theo:generic}, ensuring the estimability of model parameters through generic identifiability of the model.

\begin{figure}[ht!]
\centering
\includegraphics[width = 0.85\textwidth]{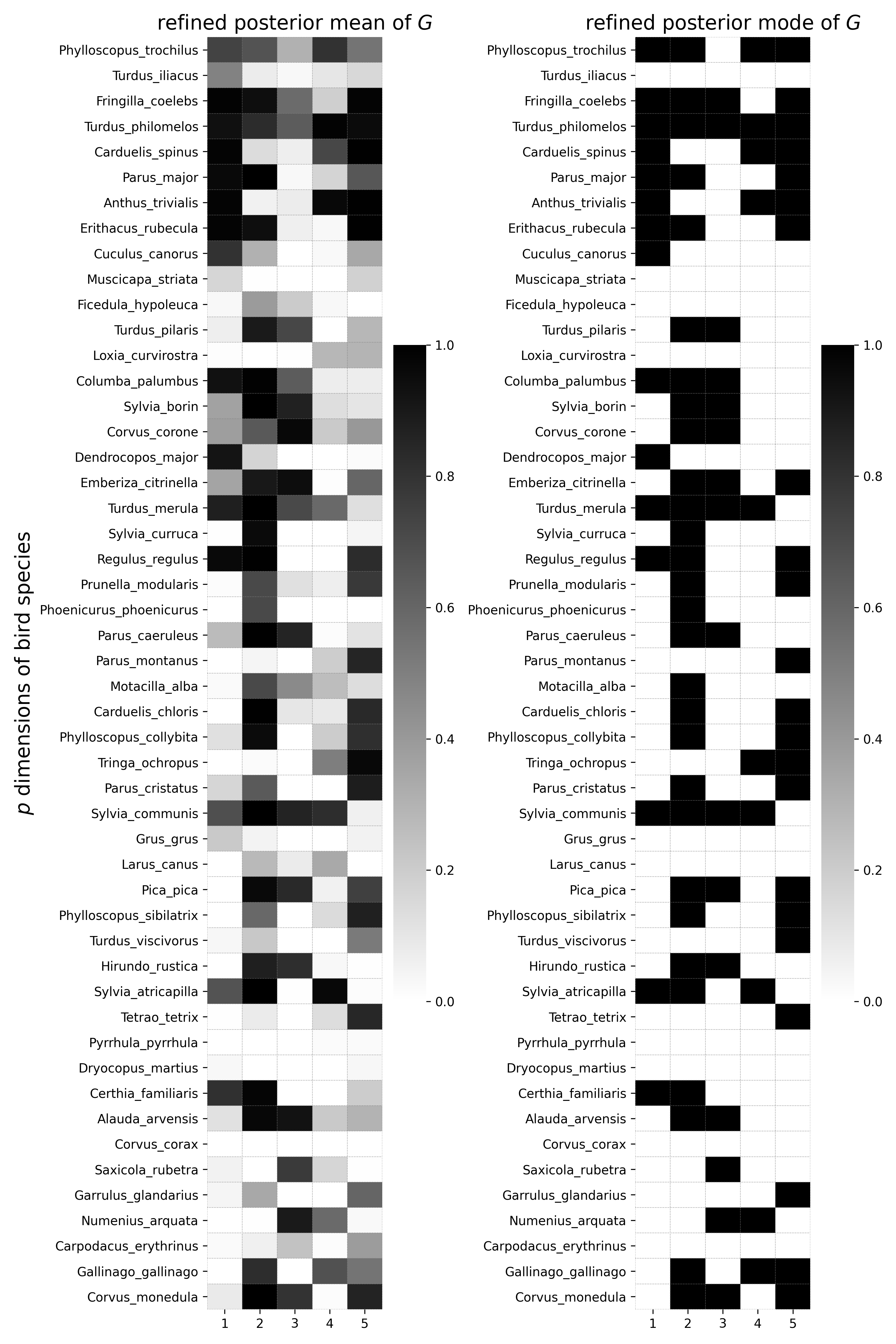}
\caption{
Posterior mean and posterior mode of $\Gb$ after refinement.
}
\label{fig:app_G}
\end{figure}

Posterior inference of the latent attributes $\wb$ and the binary matrix $\Gb$ can be further interpreted in the context of our biological knowledge of the bird species represented in the phylogenetic tree.
For each latent attribute $w_j$, we visualize the bird species $i$ that depends on $w_j$, i.e. those with a posterior mode $g_{i, j} = 1$, on the phylogenetic tree, as shown in Figures \ref{fig:phylo} and \ref{fig:phylo_full}.
Although our model does not directly incorporate information of the phylogenetic tree, we observe that species within the same terminal branches tend to share dependencies on the same latent attributes, also reflected in the structural similarity of their corresponding rows in $\Gb$.

\begin{figure}[ht!]
\centering
\begin{subfigure}[b]{0.49\textwidth}
\centering
\includegraphics[width = \textwidth, trim = 0 15 0 0, clip]{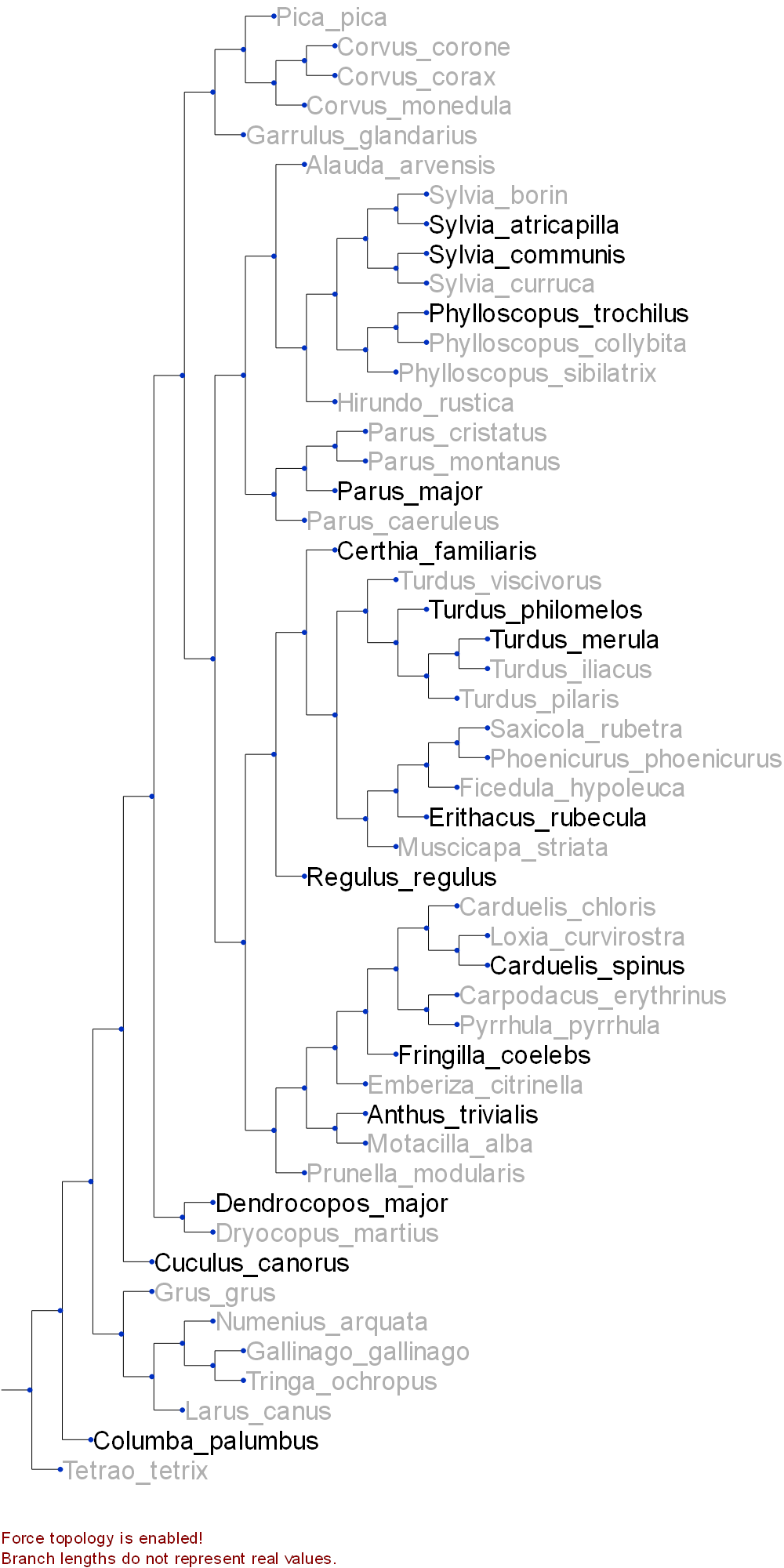}
\caption{Species dependent on $w_1$.}
\label{sub_fig:phylo1}
\end{subfigure}
\begin{subfigure}[b]{0.49\textwidth}
\centering
\includegraphics[width = \textwidth, trim = 0 15 0 0, clip]{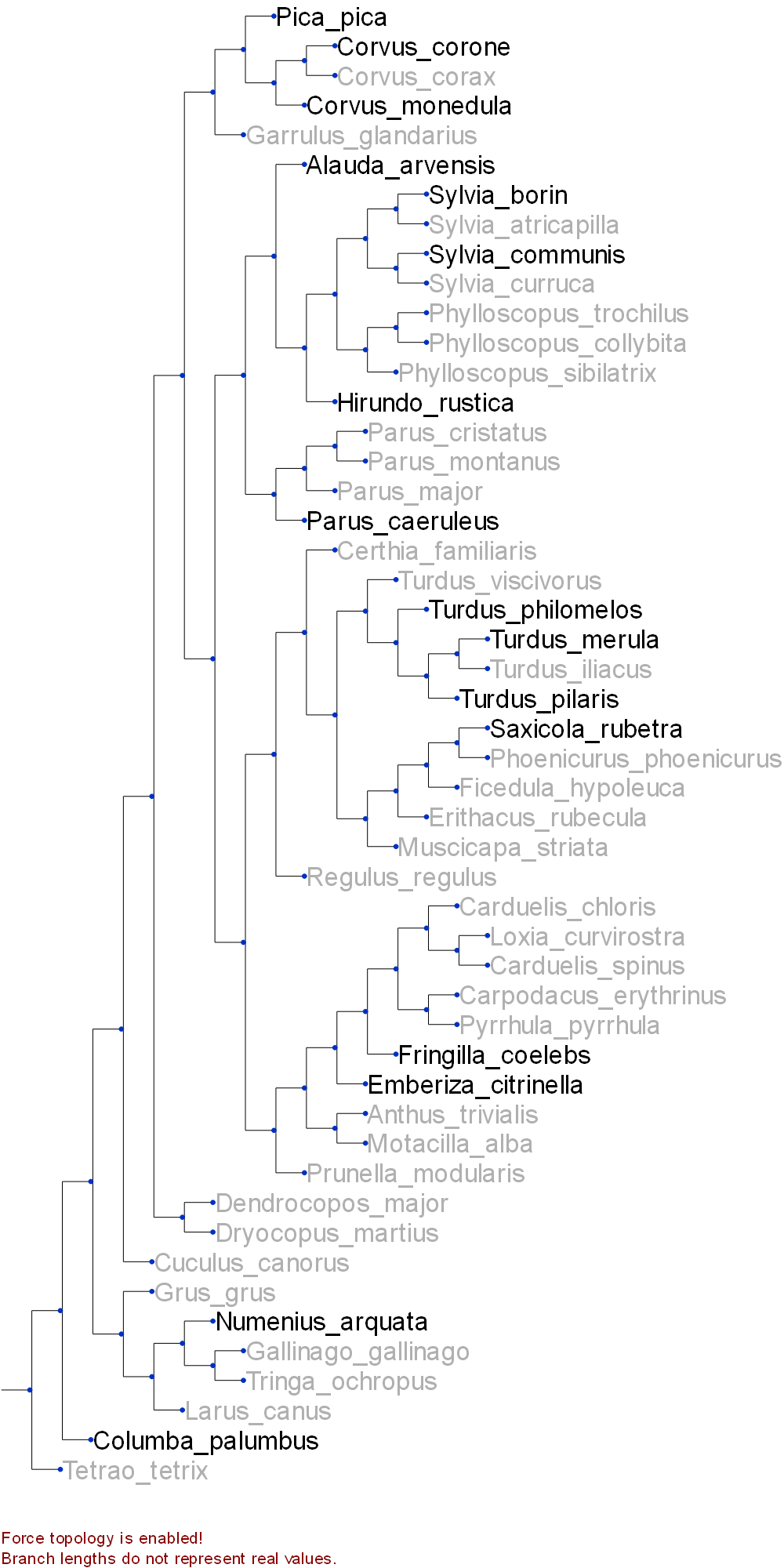}
\caption{Species dependent on $w_3$.}
\label{sub_fig:phylo3}
\end{subfigure}
\caption{
Phylogenetic trees of bird species, where species names in black indicate dependence on latent attribute $w_j$ (i.e. species $i$ with $g_{i, j} = 1$), while names in gray indicate no dependence, for $j = 1, 3$.
For latent attributes $w_2, w_4, w_5$, refer to Figure \ref{fig:phylo}.
}
\label{fig:phylo_full}
\end{figure}

Furthermore, each inferred latent attribute $w_j$ can be interpreted by examining patterns in the biological traits of bird species that depend on $w_j$, i.e. species $i$ with $g_{i, j} = 1$.
The biological traits considered include the migration type and log-transformed body mass, with migration type categorized and represented by three indicators: long distance migrant, short distance migrant, and resident species.
Using the binary matrix $\Gb$, we summarize the average biological traits separately for species dependent on $w_j$ and those not dependent on $w_j$, as presented in Table \ref{tab:app_WGT}.
Notably, all resident species are not dependent on attribute $w_4$.
Among migrant species, those associated with $w_1$, $w_2$, or $w_5$ are more likely to be short distance migrants, whereas those not dependent on these attributes are more often long distance migrants.
Additionally, we observe that species dependent on $w_3$ and those not dependent on $w_2$ tend to have higher body weights.

\begin{table}[ht!]
\centering
\begin{tabular}{c||ccc||c}
\hline\hline
& long distance migrant & short distance migrant & resident species & log mass \\
\hline\hline
$g_{i, 1} = 1$ & 0.3333 & 0.4000 & 0.2667 & 3.3788 \\
$g_{i, 1} = 0$ & 0.3143 & 0.2857 & 0.4000 & 4.0903 \\
\hline
$g_{i, 2} = 1$ & 0.3103 & 0.3448 & 0.3448 & 3.4123 \\
$g_{i, 2} = 0$ & 0.3333 & 0.2857 & 0.3810 & 4.5184 \\
\hline
$g_{i, 3} = 1$ & 0.2500 & 0.4375 & 0.3125 & 4.2148 \\
$g_{i, 3} = 0$ & 0.3529 & 0.2647 & 0.3824 & 3.7178 \\
\hline
$g_{i, 4} = 1$ & 0.5000 & 0.5000 & 0.0000 & 3.8092 \\
$g_{i, 4} = 0$ & 0.2750 & 0.2750 & 0.4500 & 3.8938 \\
\hline
$g_{i, 5} = 1$ & 0.2273 & 0.3182 & 0.4545 & 3.5660 \\
$g_{i, 5} = 0$ & 0.3929 & 0.3214 & 0.2857 & 4.1211 \\
\hline\hline
\end{tabular}
\caption{
Average of meta covariates over bird species dependent or not dependent on each latent attribute $w_j$, i.e. over $\{i \in [p]:~ g_{i, j} = 1\}$ and $\{i \in [p]:~ g_{i, j} = 0\}$.
}
\label{tab:app_WGT}
\end{table}

\subsection{Prediction Task}\label{supp_ssec:pred_task}

In Section \ref{ssec:pred}, Table \ref{tab:pred_metrics} presents predictive metrics, including RMSE, AUC, and co-occurrence score, for \texttt{HMSC}, our model \texttt{BLIP}, and \texttt{LCR}, evaluated on both in-sample and out-of-sample data.
To enable evaluation at a finer scale, we further visualize RMSE and AUC for each species in Figure \ref{fig:app_species} and for each sampling location in Figure \ref{fig:app_location}.
In these figures, the $x$-axes represent the occurrence frequency $\frac{1}{N} \sum_{n = 1}^N y_i^{(n)}$ of each species $i$ and the species richness $\frac{1}{p} \sum_{i = 1}^p y_i^{(n)}$ of each sampling location $n$, respectively.
Our results indicate that our model performs consistently across species and sampling locations, achieving predictive accuracy comparable to the state-of-the-art baseline \texttt{HMSC}.

\begin{figure}[ht!]
\centering
\begin{subfigure}[b]{\textwidth}
\centering
\includegraphics[width = \textwidth]{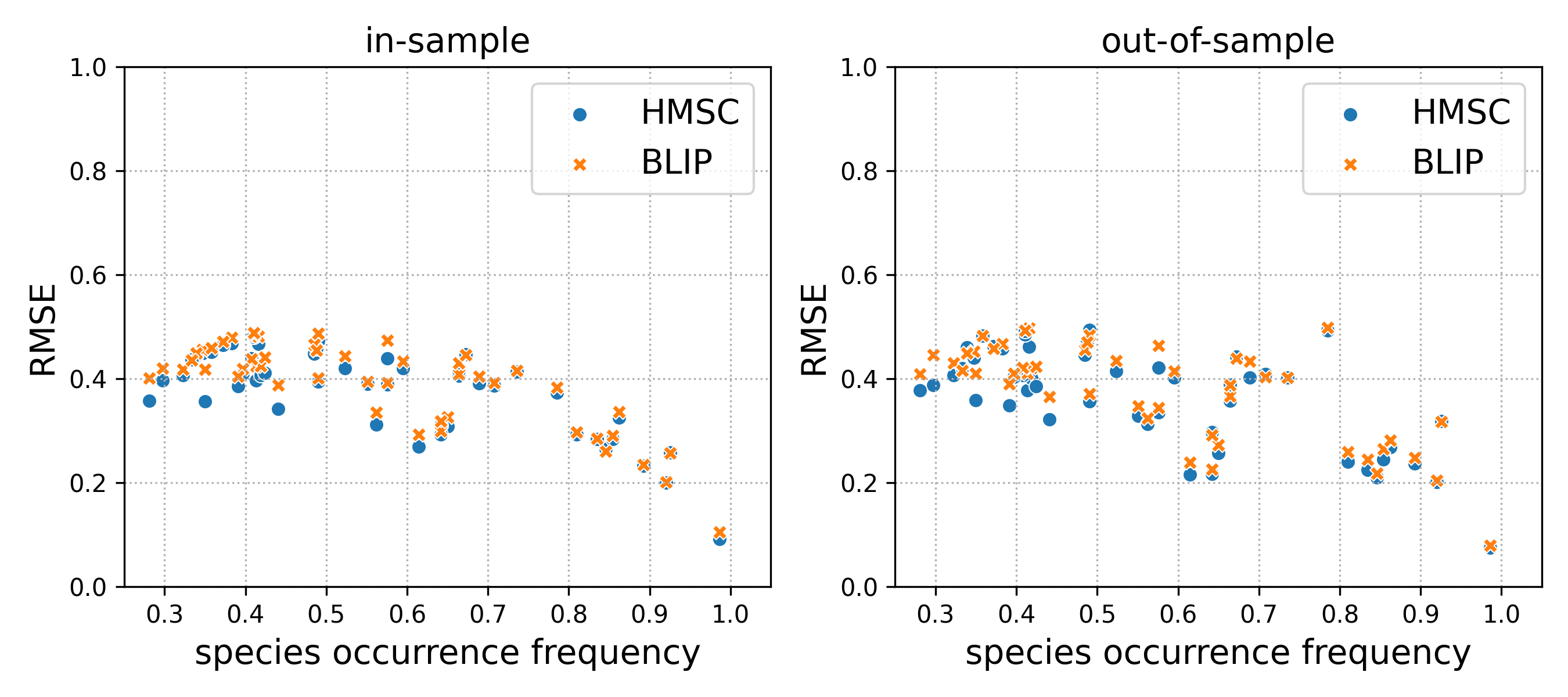}
\caption{RMSE}
\label{sub_fig:app_rmse_species}
\end{subfigure}
\begin{subfigure}[b]{\textwidth}
\centering
\includegraphics[width = \textwidth]{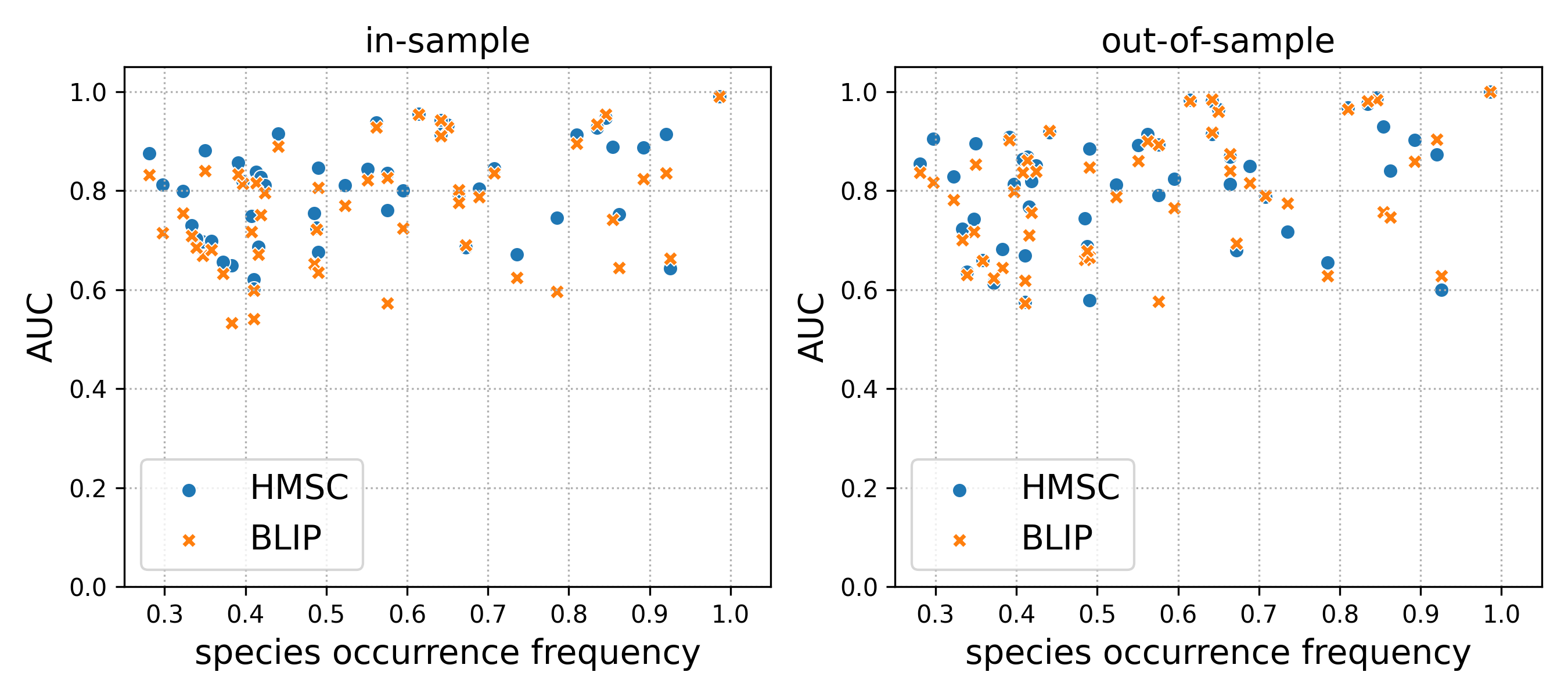}
\caption{AUC}
\label{sub_fig:app_auc_species}
\end{subfigure}
\caption{
Prediction metrics for \texttt{HMSC} and our \texttt{BLIP}, computed for each species over both in-sample and out-of-sample observations.
Each dot represents a species $i \in [p]$, with the $x$-axis indicating the species occurrence frequency.
}
\label{fig:app_species}
\end{figure}

\begin{figure}[ht!]
\centering
\begin{subfigure}[b]{\textwidth}
\centering
\includegraphics[width = \textwidth]{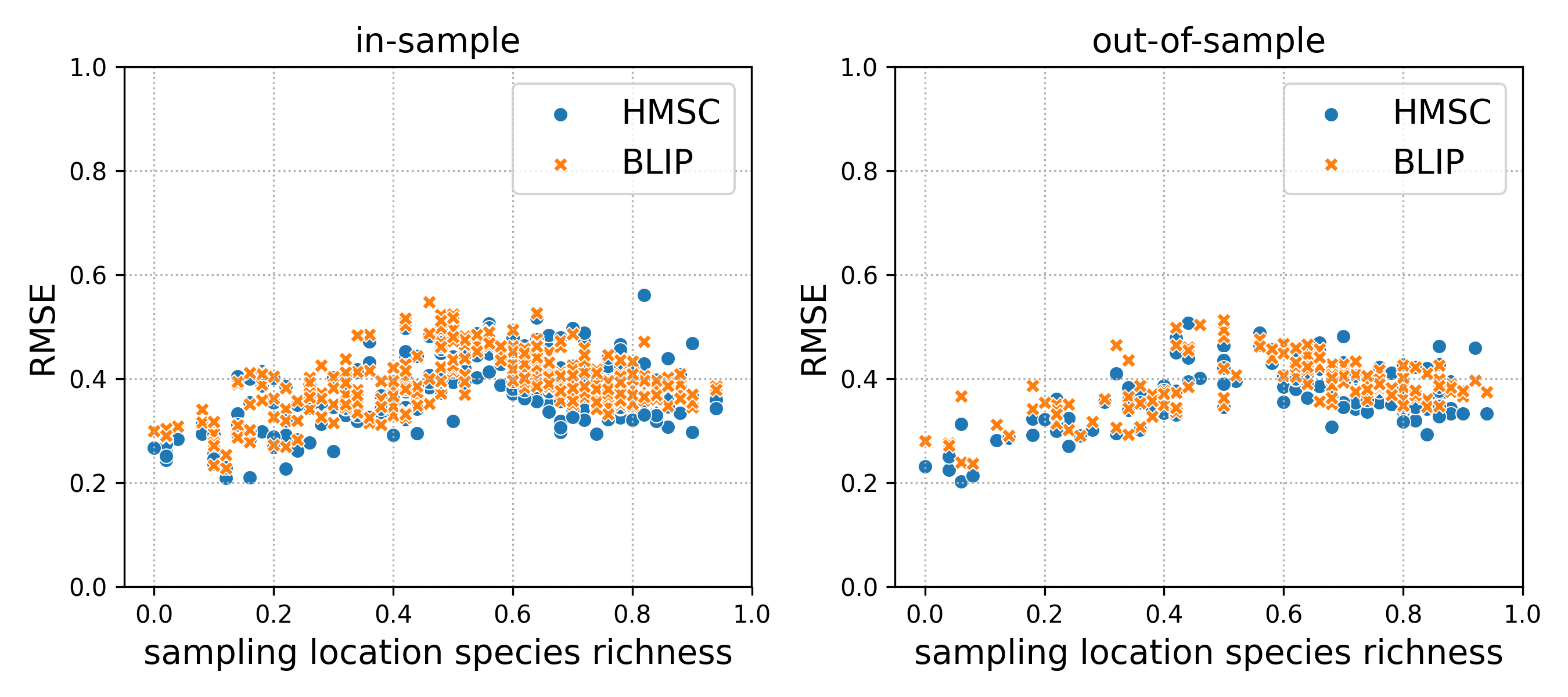}
\caption{RMSE}
\label{sub_fig:app_rmse_location}
\end{subfigure}
\begin{subfigure}[b]{\textwidth}
\centering
\includegraphics[width = \textwidth]{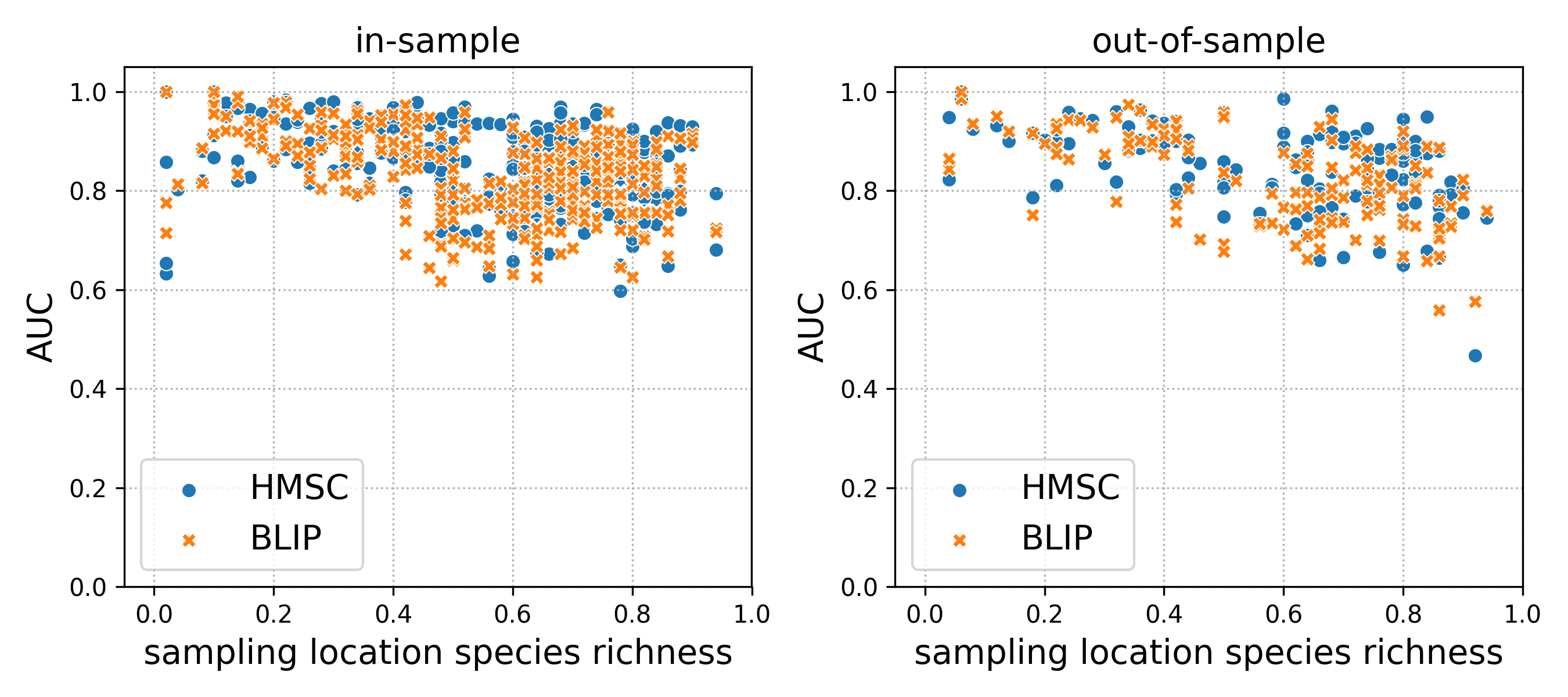}
\caption{AUC}
\label{sub_fig:app_auc_location}
\end{subfigure}
\caption{
Prediction metrics for \texttt{HMSC} and our \texttt{BLIP}, computed for each sampling location over both in-sample and out-of-sample observations.
Each dot represents a sampling location $n \in [N]$, with the $x$-axis indicating the sampling location species richness.
}
\label{fig:app_location}
\end{figure}

We provide additional discussions on the \texttt{LCR} model.
As noted in Section \ref{sec:appl}, we use the \texttt{LCR} model with $d = 6$ latent classes as a baseline, selected according to the \texttt{AIC} (the approach used in the standard package we use for implementation), as shown in Figure \ref{fig:lcr_aic}.
The sampling locations within each inferred \texttt{LCR} class are visualized on the Finland map in Figure \ref{fig:app_lcr_z_map}.
Compared to the inferred class of our \texttt{BLIP} in Figure \ref{fig:app_z_map}, the inferred class of \texttt{LCR} exhibits less distinct and interpretable geographical patterns. 
Furthermore, our model divides the observations into smaller clusters with $d = 4$. 
Interestingly, we find that the \texttt{LCR} model with $d = 4$ outperforms $d = 6$ in prediction tasks, despite having a lower model likelihood and being less favored by \texttt{AIC}.
This aligns with our theoretical analysis in Section \ref{sec:high_dim}, which suggests that \texttt{LCR} is already prone to over clustering even at a data dimensionality of $p = 50$.
Regardless of whether $d = 4$ or $d = 6$, our model consistently outperforms the \texttt{LCR} model in prediction tasks.

\begin{figure}[ht!]
\centering
\includegraphics[width = 0.7\textwidth]{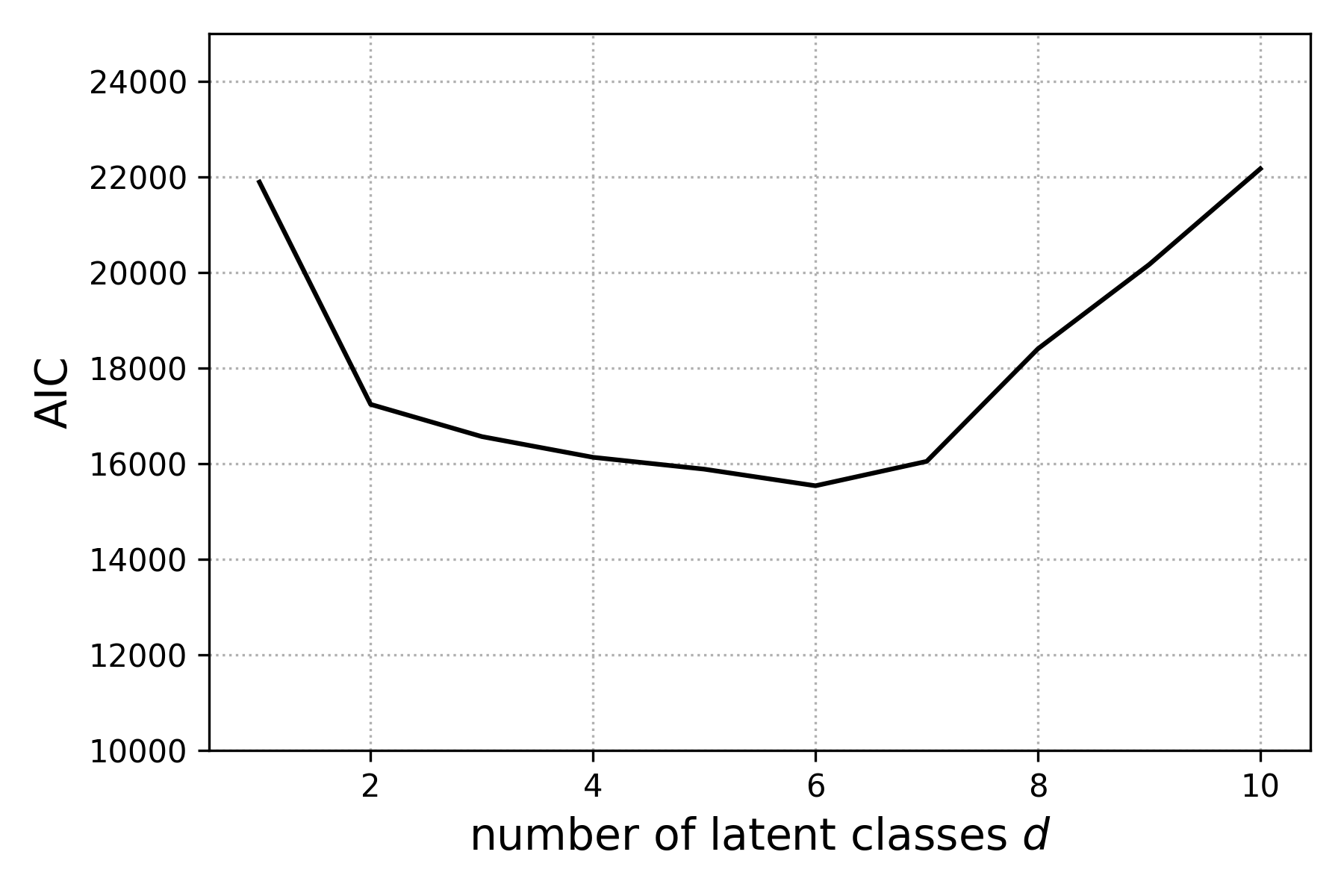}
\caption{
\texttt{AIC} for each \texttt{LCR} model with number of latent classes $d \in \{1, 2, 3, 4, 5, 6, 7, 8, 9, 10\}$.
}
\label{fig:lcr_aic}
\end{figure}

\begin{figure}[ht!]
\centering
\includegraphics[width = 0.5\textwidth]{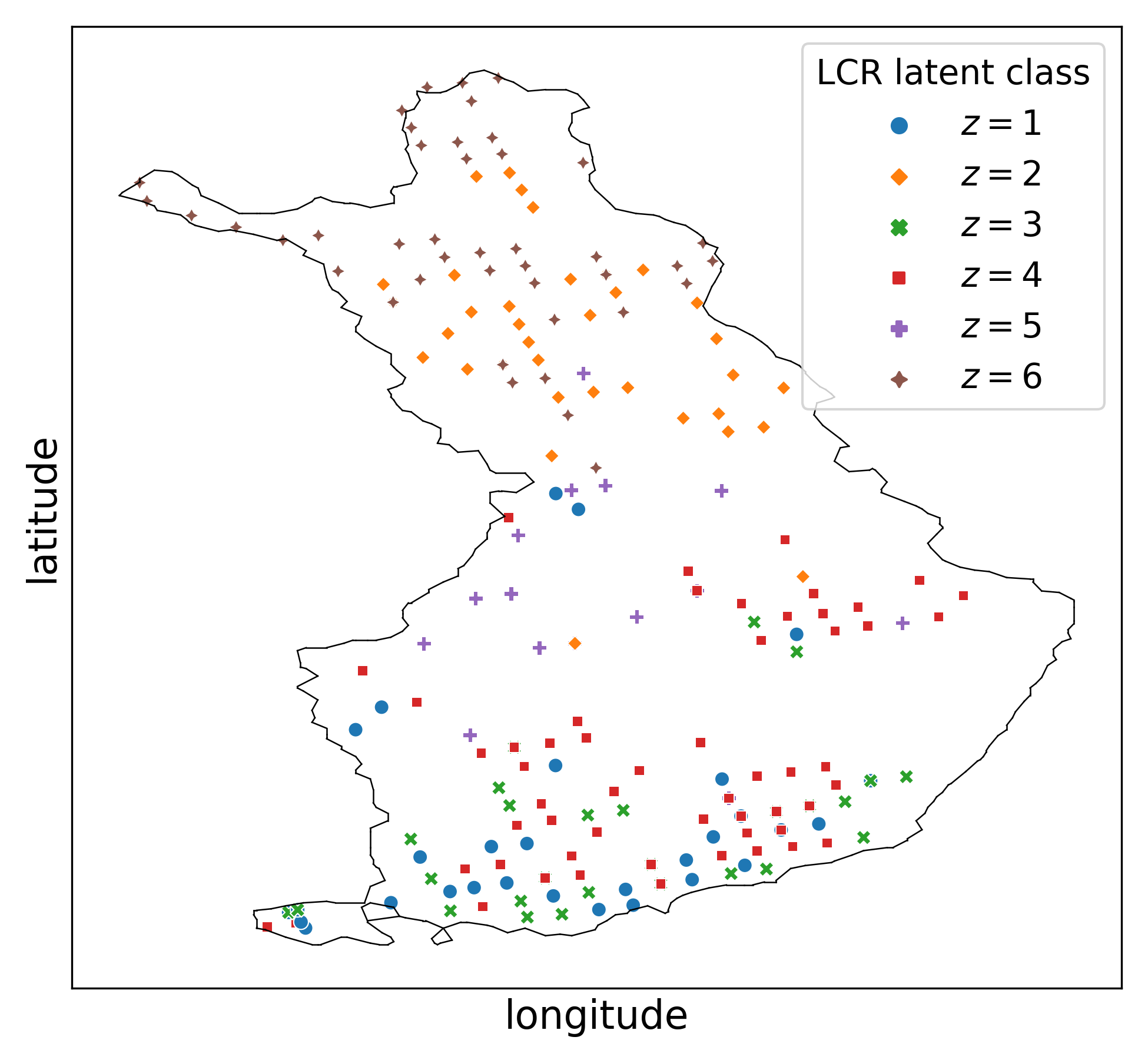}
\caption{
Spatial distribution of bird species sampling locations on the map of Finland, color-coded by their inferred \texttt{LCR} latent class.
}
\label{fig:app_lcr_z_map}
\end{figure}

\section{Model Generalizations}\label{supp_sec:gen}

Multivariate categorical data can appear in various complex forms, including vectors of mixed discrete outcomes \citep{dunson2005bayesian} and adjacency matrices \citep{guha2024covariate, zhou2024bayesian}.
For clarity and simplicity, we have so far introduced our model and presented the theoretical results for observed data in the form of high-dimensional binary vectors.
In this section, we extend our model to accommodate general types of high-dimensional categorical or count data by incorporating generalized linear models \citep{mccullagh2019generalized} into the bottom two layers of the Bayesian network depicted in Figure \ref{fig:dag}.
Theorems on model identifiability and posterior consistency, along with the MCMC sampling algorithms, are also generalized to the extended model.

In Section \ref{supp_ssec:gen_model}, we generalize the formulation of \eqref{eq:model}, incorporating multinomial logistic regression \eqref{eq:yi_cat} for categorical entries $y_i$ and Poisson regression \eqref{eq:yi_count} for count-valued entries.
The theoretical properties of the generalized model are established in Section \ref{supp_ssec:gen_theo}, with the corresponding proofs detailed in Section \ref{supp_ssec:gen_proof}.
Furthermore, in Section \ref{supp_ssec:gen_post_comp}, we extend the hierarchical prior distribution introduced in Section \ref{ssec:prior} and design a Metropolis-within-Gibbs sampler for posterior computation.

\subsection{Formulation for Mixed Categorical and Count Data}\label{supp_ssec:gen_model}

For a categorical data entry $y_i$ taking values in $[D_i]$, we specify its conditional distribution $y_i ~|~ \wb$ using multinomial logistic regression:
\begin{equation}\label{eq:yi_cat}
y_i^{(n)}
\sim
\Categorical\left(
\frac{
\exp(\beta_{i, 1, 0} + (\gb_i \circ \bbeta_{i, 1})^\top \wb^{(n)})
}{
\sum_{\ell = 1}^{D_i} \exp(\beta_{i, \ell, 0} + (\gb_i \circ \bbeta_{i, \ell})^\top \wb^{(n)})
}
,\ldots,
\frac{
\exp(\beta_{i, D_i, 0} + (\gb_i \circ \bbeta_{i, D_i})^\top \wb^{(n)})
}{
\sum_{\ell = 1}^{D_i} \exp(\beta_{i, \ell, 0} + (\gb_i \circ \bbeta_{i, \ell})^\top \wb^{(n)})
}
\right)
,
\end{equation}
where, as before, $\gb_i$ is a $q$-dimensional binary vector characterizing the dependency of $y_i$ on the attributes in $\wb$.
Each category $\ell \in [D_i]$ of $y_i$ is associated with the parameters $\beta_{i, \ell, 0} \in \RR$ and $\bbeta_{i, \ell} \in \RR^q$, representing its regression intercept and coefficients, respectively.
To ensure identifiability, we set category $D_i$ as the baseline with fixed parameters $\beta_{i, D_i, 0} = 0$ and $\bbeta_{i, D_i} = \zero$.

For a count-valued data entry $y_i \in \NN$, we specify its conditional distribution $y_i ~|~ \wb$ using Poisson regression:
\begin{equation}\label{eq:yi_count}
y_i^{(n)}
\sim
\Poisson\left( \exp\left(
\beta_{i, \infty, 0} + (\gb_i \circ \bbeta_{i, \infty})^\top \wb^{(n)}
\right) \right)
,
\end{equation}
where $\gb_i$ again indicates the dependency structure, and $\beta_{i, \infty, 0} \in \RR, \bbeta_{i, \infty} \in \RR^q$ denote the regression intercept and coefficients.

The generalized model retains the same theoretical properties as the multivariate binary model \eqref{eq:model}, including strict identifiability, generic identifiability, and posterior consistency.
Likewise, posterior computation methods naturally extend to the generalized framework.
The Polya-Gamma augmentation technique \citep{polson2013bayesian} applies to multinomial logistic regression, enabling similar sampling steps for categorical entries $y_i$ governed by \eqref{eq:yi_cat}.
For entries with count values $y_i$ following \eqref{eq:yi_count}, the associated parameters $\beta_{i, \infty, 0}, \bbeta_{i, \infty}$ can be efficiently sampled via the Metropolis-Hastings algorithm, leading to an overall Metropolis-within-Gibbs sampler.
We elaborate on these theories and methods in the following subsections.

\subsection{Theoretical Properties}\label{supp_ssec:gen_theo}

For entries in mixed categorical and count observation $\yb$, we define $D_i \in \NN$ as the number classes for $y_i$ when $y_i$ is categorical and set $D_i = \infty$ when $y_i$ is count valued.
Recall that for finite $D_i$, the set $[D_i]$ is defined as $\{1, 2, \ldots, D_i\}$.
For simplicity, we introduce the notation $[\infty] := \{\infty\}$, allowing a unified representation under the model formulation of $\yb ~|~ \wb$ given in \eqref{eq:yi_cat} and \eqref{eq:yi_count}.
The continuous parameters governing the conditional distribution of $y_i ~|~ \wb$ can then be expressed as $\{\beta_{i, \ell, 0}, \bbeta_{i, \ell}:~ \ell \in [D_i]\}$.

We use $\Bb$ to denote the collection of all continuous parameters associated with the conditional distribution of $\yb ~|~ \wb$, i.e.
\begin{equation}\label{eq:gen_Bb}
\Bb
:=
\big\{
\beta_{i, \ell, 0}, \bbeta_{i, \ell}:~
i \in [p], \ell \in [D_i]
\big\}
.
\end{equation}
Furthermore, we extend the entrywise product $(\one ~ \Gb) \circ \Bb$ to denote an operation in which each $\beta_{i, \ell, 0}$ remains unchanged, while each vector $\bbeta_{i, \ell}$ is multiplied entrywise by its corresponding binary vector $\gb_i$, i.e.
\begin{equation}\label{eq:gen_entrywise_prod}
(\one ~ \Gb) \circ \Bb
:=
\big\{
\beta_{i, \ell, 0}, \gb_i \circ \bbeta_{i, \ell}:~
i \in [p], \ell \in [D_i]
\big\}
.
\end{equation}

For the general model of mixed discrete observations, the definitions of $\perm_{\sS}(\Ab, \Bb, \bGamma, \Gb)$ and $\marg_{\sS}(\Ab, \Bb, \bGamma, \Gb)$ in \eqref{eq:perm} and \eqref{eq:marg} remain valid.
Consequently, the notions of strict identifiability, generic identifiability, and posterior consistency for the general model follow the same definition as in Definitions \ref{defi:strict}, \ref{defi:generic}, and \ref{defi:post}.

In the following, we present the theorems establishing strict identifiability, generic identifiability, and posterior consistency for the general model.
Proofs of the theorems are provided in Section \ref{supp_ssec:gen_proof}.

\begin{theorem}[Strict Identifiability]\label{theo:gen_strict}
If the distribution $\cP_x$ has full rank, then our model with the parameter space $\sS_1$ is strictly identifiable, for $\sS_1$ consisting of parameters satisfying the following conditions:
\begin{enumerate}[(i)]
\item
There exists a partition of $[q]$ as $\cQ_1 \cup \cQ_2 \cup \cQ_3$ with $\card(\cQ_1), \card(\cQ_2) \ge \log_2 d$, such that the two sets of vectors $\bigcirc_{j \in \cQ_1} \{\one, \balpha_j\}$, $\bigcirc_{j \in \cQ_2} \{\one, \balpha_j\}$ have full ranks and the submatrix $\Ab_{\cQ_3}$ has distinct columns;
\item
The matrix $\Gb$ contains three distinct identity blocks $\Ib_q$, i.e.
$$
\Pb \Gb
=
\left(\begin{matrix}
\Ib_q \\ \Ib_q \\ \Ib_q \\ \Gb_{\cI}
\end{matrix}\right)
$$
for some $p \times p$ permutation matrix $\Pb$ and subset $\cI \subset [p]$ with $\card(\cI) = p - 3q$;
\item
$\forall i \in [p]$, $\exists \ell \in [D_i]$ such that $\forall j \in [q]$, if $g_{i, j} = 1$ then $\beta_{i, \ell, j} \ne 0$.
\end{enumerate}
\end{theorem}

\begin{theorem}[Generic Identifiability]\label{theo:gen_generic}
If the distribution $\cP_x$ has full rank, then our model with the parameter space $\sS_2$ is generically identifiable, for $\sS_2$ consisting of parameters satisfying the following condition:
\begin{itemize}
\item[($*$)]
$\Gb$ has distinct columns and takes the blockwise form
$$
\Pb \Gb
=
\left(\begin{matrix}
\Gb_{\cI_1} \\ \Gb_{\cI_2} \\ \Gb_{\cI_3}
\end{matrix}\right)
$$
for some $p \times p$ permutation matrix $\Pb$, with two $q \times q$ submatrices $\Gb_{\cI_1}, \Gb_{\cI_2}$ having all-ones diagonals and the remaining $(p - 2q) \times q$ submatrix $\Gb_{\cI_3}$ having no all-zeros column, i.e.
$$
\forall j \in [q]
,\quad
(\Gb_{\cI_1})_{j, j} = (\Gb_{\cI_2})_{j, j} = 1
,\quad
\sum_{i \in \cI_3} g_{i, j} > 0
.
$$
\end{itemize}
\end{theorem}

\begin{theorem}[Posterior Consistency]\label{theo:gen_post}
Let the prior $\uppi$ be absolutely continuous w.r.t. $\uplambda \times \upmu$.
Suppose $\cP_x$ has full rank and a finite first moment, then our model with parameter space $\sS$ is posterior consistent at
\begin{enumerate}[(i)]
\item
any identifiable parameter in $\supp(\uppi)$ if $\sS$ is compact;
\item
any parameter in $\supp(\uppi)$ if $\sS \subset \sS_1$;
\item
$\uppi$-almost every parameter in $\sS$ if $\sS \subset \sS_2$.
\end{enumerate}
\end{theorem}

Compared to Theorems \ref{theo:strict}, \ref{theo:generic}, and \ref{theo:post} for the model of multivariate binary observations, the only modification in the general model pertains to the condition on $\Bb$ for strict identifiability.
In Theorem \ref{theo:strict}, we require each parameter $\beta_{i, j}$ associated with $g_{i, j} = 1$ to be non-zero.
In contrast, for the general model, it suffices for $\beta_{i, \ell, j}$ to be non-zero for at least one $\ell \in [D_i]$.
Intuitively, identifiability becomes easier to establish when discrete observations take more than two possible values, as more information is contained in their marginal distributions.

\subsection{Proof of Theoretical Properties}\label{supp_ssec:gen_proof}

We now present the proofs of Theorems \ref{theo:gen_strict}, \ref{theo:gen_generic}, and \ref{theo:gen_post}.

\begin{proof}[Proof of Theorem \ref{theo:gen_strict}]
The proof largely follows the structure of the proof of Theorem \ref{theo:strict}.
To establish a connection between them, we construct a multivariate binary observation $\yb'$ from the mixed discrete observation $\yb$.
As assumed in Theorem \ref{theo:gen_strict}, for each entry $y_i$ in $\yb$, there exists an index $\ell \in [D_i]$ such that for every $g_{i, j} = 1$, the corresponding parameter $\beta_{i, \ell, j} \ne 0$.
Accordingly, we define $y_i' := 1_{y_i = \ell}$.

For the multivariate binary vector $\yb'$, Lemmas \ref{lemm:y_to_w} and \ref{lemm:w_to_z} remain valid, along with parts (ii) and (iii) of Lemma \ref{lemm:law_to_params}.
Consequently, to prove Theorem \ref{theo:gen_strict}, it suffices to establish part (i) of Lemma \ref{lemm:law_to_params} for the general model, as the rest of the proof follows the same reasoning as in Theorem \ref{theo:strict}.
We now prove part (i) of Lemma \ref{lemm:law_to_params}, which states that if there exist parameters $\tilde\Bb$ and $\tilde\Gb$ satisfying the assumptions of Theorem \ref{theo:gen_strict}, and
$$
\law(\yb ~|~ \wb, \Gb, \Bb)
=
\law(\yb ~|~ \wb, \tilde\Gb, \tilde\Bb)
,\quad
\forall \wb \in \{0, 1\}^q
,
$$
then it must hold that $\Gb = \tilde\Gb$ and $(\one ~ \Gb) \circ \Bb = (\one ~ \tilde\Gb) \circ \tilde\Bb$.
The entrywise product $(\one ~ \Gb) \circ \Bb$ was previously defined in \eqref{eq:gen_Bb} and \eqref{eq:gen_entrywise_prod}.

For a categorical $y_i$, recall that its baseline class $D_i$ has parameters fixed as $\beta_{i, D_i, 0} = 0$ and $\bbeta_{i, D_i} = \zero$ to ensure the identifiability of the model. By the conditional independence of the entries $y_i$ given $\wb$, for each $\ell \in [D_i - 1]$ and $\wb \in \{0, 1\}^q$, we have
\begin{align*}
\exp(\beta_{i, \ell, 0} + (\gb_i \circ \bbeta_{i, \ell})^\top \wb)
&=
\frac{\PP(y_i = \ell ~|~ \wb, \Gb, \Bb)}{\PP(y_i = D_i ~|~ \wb, \Gb, \Bb)}
\\&=
\frac{\PP(y_i = \ell ~|~ \wb, \tilde\Gb, \tilde\Bb)}{\PP(y_i = D_i ~|~ \wb, \tilde\Gb, \tilde\Bb)}
\\&=
\exp(\tilde\beta_{i, \ell, 0} + (\tilde\gb_i \circ \tilde\bbeta_{i, \ell})^\top \wb)
.
\end{align*}
For a counted-valued $y_i$, also by the conditional independence of the entries $y_i$ given $\wb$, for each $\wb \in \{0, 1\}^w$, we have
\begin{align*}
\exp(\beta_{i, \infty, 0} + (\gb_i \circ \bbeta_{i, \infty})^\top \wb)
&=
- \log\PP(y_i = 0 ~|~ \wb, \Gb, \Bb)
\\&=
- \log\PP(y_i = 0 ~|~ \wb, \tilde\Gb, \tilde\Bb)
\\&=
\exp(\tilde\beta_{i, \infty, 0} + (\tilde\gb_i \circ \tilde\bbeta_{i, \infty})^\top \wb)
.
\end{align*}
Since the exponential function is strictly increasing, it follows that
$$
\beta_{i, \ell, 0} + (\gb_i \circ \bbeta_{i, \ell})^\top \wb
=
\tilde\beta_{i, \ell, 0} + (\tilde\gb_i \circ \tilde\bbeta_{i, \ell})^\top \wb
,\quad
\forall i \in [p]
,~
\forall \ell \in [D_i]
,\quad
\forall \wb \in \{0, 1\}^q
.
$$
From this, we conclude that
$$
\beta_{i, \ell, 0}
=
\tilde\beta_{i, \ell, 0}
,\quad
\gb_i \circ \bbeta_{i, \ell}
=
\tilde\gb_i \circ \tilde\bbeta_{i, \ell}
,\quad
\forall i \in [p]
,~
\forall \ell \in [D_i]
.
$$
This completes the proof.
\end{proof}

\begin{proof}[Proof of Theorem \ref{theo:gen_generic}]
Using the multivariate binary vector $\yb'$ defined in the proof of Theorem \ref{theo:gen_strict}, both Lemmas \ref{lemm:y_to_w_meas0} and \ref{lemm:w_to_z_meas0} remain valid.
Recall that Lemma \ref{lemm:law_to_params_meas0} is derived from Lemma \ref{lemm:law_to_params} by showing that the set of parameters that violate the conditions in Lemma \ref{lemm:law_to_params} has $(\uplambda \times \upmu)$ measure zero.

For each binary matrix $\Gb$, we observe that the set
$$
\left\{
\Bb:~
\exists i \in [p], j \in [q]~
\text{s.t. }
g_{i, j} = 1~
\text{and }
\forall \ell \in [D_i],~
\beta_{i, \ell, j} = 0
\right\}
$$
is a finite union of linear subspaces, each having Lebesgue measure zero.
Since there are only finitely many possible values of $\Gb$, a union bound argument shows that the set of parameters not satisfying condition (iii) in Theorem \ref{theo:gen_strict} has $(\uplambda \times \upmu)$ measure zero.

As a result, Lemma \ref{lemm:law_to_params_meas0} holds for the general model, allowing Theorem \ref{theo:gen_generic} to be established using the same proof as Theorem \ref{theo:generic}.
\end{proof}

\begin{proof}[Proof of Theorem \ref{theo:post}]
The Lemmas \ref{lemm:prior_supp} and \ref{lemm:quotient_metric} remain valid for the general model. We can adapt the proof of Lemma \ref{lemm:S1_almost_compact} to the general model by redefining the sequence of actively compact sets $\{\sB_m\}_{m = 1}^\infty$ in \eqref{eq:sBm_seq} as
\begin{align*}
\sB_m
:=
\Big\{
(\Ab, \Bb, \bGamma, \Gb) \in \sS_1:~
&
\forall i \in [p],~
\forall j \in [q],~
\forall h \in [d],~
\forall k \in [p_x],~
\forall \ell \in [D_i],
\\&
|\gamma_{j, 0}| \le k,~
|\gamma_{h, k}| \le k,~
\frac{1}{k} \le \alpha_{j, h} \le 1 - \frac{1}{k},~
|\beta_{i, \ell, 0}| \le k,~
|g_{i, j} \beta_{i, \ell, j}| \le k
\Big\}
.
\end{align*}
Therefore, Lemma \ref{lemm:S1_almost_compact} also holds for the general model, allowing Theorem \ref{theo:gen_post} to be established using the same proof as Theorem \ref{theo:post}.
\end{proof}

\subsection{Posterior Computation}\label{supp_ssec:gen_post_comp}

The hierarchical prior distribution specified in Section \ref{ssec:prior} can be naturally extended to the general model by imposing
$$
(\beta_{i, \ell, 0}, \bbeta_{i, \ell})
\stackrel{iid}{\sim}
N(\bbm_\beta, \bV_\beta)
,\quad
\forall i \in [p]
,~
\ell \in [D_i]
$$
while keeping all other components the same as in \eqref{eq:prior_G} and \eqref{eq:prior_cont}.

To perform posterior computation for the general model, we design a Metropolis-with-Gibbs sampler as an extension of the data augmented Gibbs sampler described in Section \ref{supp_ssec:gibbs}.
Specifically, all full conditional distributions remain unchanged, except for those of the latent attributes $\wb^{(1:N)}$ and the parameters $\Gb$ and $\Bb$, which are updated as discussed below.

When the number of latent attributes $q$ is small, we can sample $\wb$ as a block through the full conditional distribution
\begin{align*}
\PP(\wb^{(n)} ~|~ \cdot)
&\propto
\PP(\wb^{(n)} ~|~ z^{(n)}, \Ab) \PP(\yb^{(n)} ~|~ \wb^{(n)}, \Gb, \Bb)
\\&\propto
\prod_{j = 1}^q \left(
\frac{\alpha_{j, z^{(n)}}}{1 - \alpha_{j, z^{(n)}}}
\right)^{w_j^{(n)}}
\prod_{i \in [p]:~ D_i \in \NN} \frac{\exp(\beta_{i, y_i^{(n)}, 0} + (\gb_i \circ \bbeta_{i, y_i^{(n)}})^\top \wb^{(n)})}{\sum_{\ell = 1}^{D_i} \exp(\beta_{i, \ell, 0} + (\gb_i \circ \bbeta_{i, \ell})^\top \wb^{(n)})}
\\&\qquad\cdot
\prod_{i \in [p]:~ D_i = \infty} \frac{\exp(y_i^{(n)} (\bbeta_{i, \infty, 0} + (\gb_i \circ \bbeta_{i, \infty})^\top \wb^{(n)}))}{(y_i^{(n)})! ~ \exp(\exp(\bbeta_{i, \infty, 0} + (\gb_i \circ \bbeta_{i, \infty})^\top \wb^{(n)}))}
,
\end{align*}
which is a categorical distribution over $\{0, 1\}^q$.
When $q$ is large, we can sample $\wb$ entrywise through the full conditional distribution
\begin{align*}
\PP(w_j^{(n)} ~|~ \cdot)
&\propto
\PP(w_j^{(n)} ~|~ z^{(n)}, \Ab) \PP(\yb^{(n)} ~|~ \wb^{(n)}, \Gb, \Bb)
\\&\propto
\left(
\frac{\alpha_{j, z^{(n)}}}{1 - \alpha_{j, z^{(n)}}}
\right)^{w_j^{(n)}}
\prod_{i \in [p]:~ D_i \in \NN} \frac{\exp(\beta_{i, y_i^{(n)}, 0} + (\gb_i \circ \bbeta_{i, y_i^{(n)}})^\top \wb^{(n)})}{\sum_{\ell = 1}^{D_i} \exp(\beta_{i, \ell, 0} + (\gb_i \circ \bbeta_{i, \ell})^\top \wb^{(n)})}
\\&\qquad\cdot
\prod_{i \in [p]:~ D_i = \infty} \frac{\exp(y_i^{(n)} (\bbeta_{i, \infty, 0} + (\gb_i \circ \bbeta_{i, \infty})^\top \wb^{(n)}))}{(y_i^{(n)})! ~ \exp(\exp(\bbeta_{i, \infty, 0} + (\gb_i \circ \bbeta_{i, \infty})^\top \wb^{(n)}))}
,
\end{align*}
which is a Bernoulli distribution.

Recall that the $q \times (p_t + 1)$ matrix $\bTheta$ represents the collection of hyperparameters $\theta_{j, 0}, \btheta_j$ with its $j$th row being the $(p_t + 1)$-dimensional vector $(\theta_{j, 0}, \btheta_j)$.
When $q$ is small, we can sample $\Gb$ in blocks of $\gb_i$, whereas when $q$ is large, we can sample $\Gb$ entrywise.
For a categorical $y_i \in [D_i]$, the full conditional distributions are given by
\begin{align*}
\PP(\gb_i ~|~ \cdot)
&\propto
\PP(\gb_i ~|~ \tb_i, \bTheta)
\prod_{n = 1}^N \PP(\yb^{(n)} ~|~ \wb^{(n)}, \Gb, \Bb)
\\&\propto
\exp\left(
\gb_i^\top \bTheta (1, \tb_i)
\right)
\prod_{n = 1}^N \frac{\exp(\beta_{i, y_i^{(n)}, 0} + (\gb_i \circ \bbeta_{i, y_i^{(n)}})^\top \wb^{(n)})}{\sum_{\ell = 1}^{D_i} \exp(\beta_{i, \ell, 0} + (\gb_i \circ \bbeta_{i, \ell})^\top \wb^{(n)})}
,
\end{align*}
and
\begin{align*}
\PP(g_{i, j} ~|~ \cdot)
&\propto
\PP(g_{i, j} ~|~ \tb_i, \theta_{j, 0}, \btheta_j)
\prod_{n = 1}^N \PP(\yb^{(n)} ~|~ \wb^{(n)}, \Gb, \bbeta)
\\&\propto
\exp\left(
g_{i, j} (\theta_{j, 0} + \tb_i^\top \btheta_j)
\right)
\prod_{n = 1}^N \frac{\exp(\beta_{i, y_i^{(n)}, 0} + (\gb_i \circ \bbeta_{i, y_i^{(n)}})^\top \wb^{(n)})}{\sum_{\ell = 1}^{D_i} \exp(\beta_{i, \ell, 0} + (\gb_i \circ \bbeta_{i, \ell})^\top \wb^{(n)})}
,
\end{align*}
which correspond to a categorical distribution on $\{0, 1\}^q$ and a Bernoulli distribution, respectively.
For a count-valued $y_i$, the full conditional distributions are given by
\begin{align*}
\PP(\gb_i ~|~ \cdot)
&\propto
\PP(\gb_i ~|~ \tb_i, \bTheta)
\prod_{n = 1}^N \PP(\yb^{(n)} ~|~ \wb^{(n)}, \Gb, \Bb)
\\&\propto
\exp\left(
\gb_i^\top \bTheta (1, \tb_i)
\right)
\prod_{n = 1}^N \frac{\exp(y_i^{(n)} (\bbeta_{i, \infty, 0} + (\gb_i \circ \bbeta_{i, \infty})^\top \wb^{(n)}))}{(y_i^{(n)})! ~ \exp(\exp(\bbeta_{i, \infty, 0} + (\gb_i \circ \bbeta_{i, \infty})^\top \wb^{(n)}))}
,
\end{align*}
and
\begin{align*}
\PP(g_{i, j} ~|~ \cdot)
&\propto
\PP(g_{i, j} ~|~ \tb_i, \theta_{j, 0}, \btheta_j)
\prod_{n = 1}^N \PP(\yb^{(n)} ~|~ \wb^{(n)}, \Gb, \bbeta)
\\&\propto
\exp\left(
g_{i, j} (\theta_{j, 0} + \tb_i^\top \btheta_j)
\right)
\prod_{n = 1}^N \frac{\exp(y_i^{(n)} (\bbeta_{i, \infty, 0} + (\gb_i \circ \bbeta_{i, \infty})^\top \wb^{(n)}))}{(y_i^{(n)})! ~ \exp(\exp(\bbeta_{i, \infty, 0} + (\gb_i \circ \bbeta_{i, \infty})^\top \wb^{(n)}))}
,
\end{align*}
which, one again, correspond to a categorical distribution on $\{0, 1\}^q$ and a Bernoulli distribution, respectively.

The parameter $\Bb$ is sampled in blocks of $(\beta_{i, \ell, 0}, \bbeta_{i, \ell})$.
For a categorical $y_i \in [D_i]$, we associate each $y_i^{(n)}$ with augmented variables $\omega_{i, \ell}^{(n)}$ for $\ell \in [D_i]$, sampled from full conditional distributions
$$
\omega_{i, \ell}^{(n)}~|~ \cdot
\sim
\PG\left(
1
,~
\left(
\beta_{i, \ell, 0} + (\gb_i \circ \bbeta_{i, \ell})^\top \wb^{(n)}
\right)
-
\log\left( \sum_{\ell' = 1, \ell \ne \ell}^{D_i} \exp\left(
\beta_{i, \ell', 0} + (\gb_i \circ \bbeta_{i, \ell'})^\top \wb^{(n)}
\right) \right)
\right)
.
$$
Then the full conditional distribution of $(\beta_{i, \ell, 0}, \bbeta_{i, \ell})$ follows
\begin{align*}
&\qquad
\PP(\beta_{i, \ell, 0}, \bbeta_{i, \ell} ~|~ \cdot)
\\&\propto
\PP(\beta_{i, \ell, 0}, \bbeta_{i, \ell}) \prod_{n = 1}^N \PP(y_i^{(n)} ~|~ \wb^{(n)}, \Gb, \Bb)
\\&\propto
\exp\Bigg(
- \frac12 ((\beta_{i, \ell, 0}, \bbeta_{i, \ell}) - \bbm_\beta)^\top \bV_\beta^{-1} ((\beta_{i, \ell, 0}, \bbeta_{i, \ell}) - \bbm_\beta)
\\&\qquad\qquad-
\sum_{n = 1}^N \frac{\omega_{i, \ell}^{(n)}}{2} (\beta_{i, \ell, 0}, \bbeta_{i, \ell})^\top (1, \gb_i \circ \wb^{(n)}) (1, \gb_i \circ \wb^{(n)})^\top (\beta_{i, \ell, 0}, \bbeta_{i, \ell})
\\&\qquad\qquad+
\sum_{n = 1}^N \left(
\log\left(
\sum_{\ell' = 1, \ell' \ne \ell}^{D_i} \exp\left(
\beta_{i, \ell', 0} + \bbeta_{i, \ell'}^\top (\gb_i \circ \wb^{(n)})
\right) \right) \omega_{i, \ell}^{(n)}
+
1_{y_i^{(n)} = \ell} - \frac12
\right)
\\&\hspace{3cm}\cdot
(1, \gb_i \circ \wb^{(n)})^\top (\beta_{i, \ell, 0}, \bbeta_{i, \ell})
\Bigg)
,
\end{align*}
which corresponds to the multivariate normal distribution $N(\bbm_{\beta, i, \ell}, \bV_{\beta, i, \ell})$ with parameters
$$
\bV_{\beta, i, \ell}
=
\left(
\bV_\beta^{-1}
+
\sum_{n = 1}^N \omega_{i, \ell}^{(n)} (1, \gb_i \circ \wb^{(n)}) (1, \gb_i \circ \wb^{(n)})^\top
\right)^{-1}
$$
and
\begin{align*}
\bbm_{\beta, i, \ell}
=
\bV_{\beta, i, \ell} \bV_\beta^{-1} \bbm_\beta
+
\bV_{\beta, i, \ell} \sum_{n = 1}^N \Bigg(
&
\log\left(
\sum_{\ell' = 1, \ell' \ne \ell}^{D_i} \exp\left(
\beta_{i, \ell', 0} + \bbeta_{i, \ell'}^\top (\gb_i \circ \wb^{(n)})
\right) \right) \omega_{i, \ell}^{(n)}
\\&+
1_{y_i^{(n)} = \ell} - \frac12
\Bigg)
(1, \gb_i \circ \wb^{(n)})
.
\end{align*}
For a count-valued $y_i$, a tractable full conditional distribution for $(\beta_{i, \ell, 0}, \bbeta_{i, \ell})$ is no longer available through data augmentation.
However, we can sample it using a
Metropolis-Hastings algorithm.

\section{Related Literature}\label{supp_sec:lit}

In this section, we provide complementary discussions of related literature that was deferred from the main paper for conciseness.

\subsection{Identifiability Theory}\label{supp_ssec:lit_iden}

The notion of generic identifiability was introduced in \citet{allman2009identifiability} and has been established for models having discrete and continuous latent variables.
In addition to strict and generic identifiability, the notions of local identifiability and partial identifiability are often of interest in complicated models.
We refer readers to \citet{allman2009identifiability, drton2011global, anandkumar2013learning, gu2019learning, gu2020partial, chen2020structured, gu2021sufficient, gu2023generic, gu2023dimension, gu2023bayesian, zhou2024bayesian, lee2024new}.

\subsection{Model-based Clustering}\label{supp_ssec:model_clus}

Numerous model-based clustering approaches have been developed in the statistical literature.
Among the most fundamental and widely used is the finite mixture model \citep{mclachlan2000finite, rousseau2011asymptotic, miller2018mixture}, also known as the latent class model \citep{magidson2004latent}.
When allowing the number of latent classes to grow with the number of observations, various Bayesian nonparametric methods have been designed, including Dirichlet process mixture models \citep{ferguson1973bayesian, antoniak1974mixtures, lo1984class, escobar1995bayesian, blei2006variational} and its extensions based on the Pitman-Yor process \citep{pitman1997two, canale2017pitman, lijoi2020pitman, ramirez2024heavy} or other Bayesian nonparametric processes \citep{teh2004sharing, dunson2008kernel, rodriguez2008nested, ren2011logistic, rodriguez2011nonparametric}.
To incorporate the dependency of latent classes on covariate data, a variety of parametric models have been introduced \citep{guo2006latent, chung2006latent, vermunt2010latent}, along with Bayesian nonparametric models that utilize covariate-dependent stick breaking processes \citep{chung2009nonparametric, horiguchi2024tree}.

\subsection{High-dimensional Mixture Models}\label{supp_ssec:lit_mix}

\citet{chandra2023escaping} established that Bayesian inference for the Gaussian mixture model suffers from the curse of dimensionality, leading to unreliable inference of latent clusters as the dimensionality of the data increases. Similarly, in Section \ref{ssec:curse}, we demonstrated that Bayesian inference for Bernoulli mixture model is also affected by the curse of dimensionality when clustering high-dimensional categorical observations.
Both settings assume an arbitrarily fixed prior distribution over all partitions of the $N$ observations $[N]$, which applies to both finite mixture models with a flexible number of clusters and Bayesian nonparametric mixture models.

In contrast, when the true number of clusters is known, the consistency of frequentist clustering approaches in finite mixture models, including the Gaussian mixture model and some extensions to mixtures of multivariate discrete distributions, have been established under appropriate assumptions \citep{loffler2021optimality, ndaoud2022sharp, lyu2025degree, huang2025minimax}.
The difference in these results stems partly from knowing the true number of clusters and partly from fundamental differences between frequentist and Bayesian inferences.

The discrepancy between frequentist and Bayesian inference in high-dimensional settings has been widely recognized \citet{good1992bayes, ritov2014bayesian, lamont2016lindley, ghosal2017fundamentals, hoff2023bayes}.
While our Bayesian hierarchical model provides one approach to improving Bayesian inference in mixture models for high-dimensional data, an alternative promising direction is the development of dimension-dependent Bayesian nonparametric priors, as opposed to conventional dimensional-invariant priors.

\end{document}